\documentclass[12pt]{article}
\usepackage{amsmath}
\usepackage{graphicx}
\usepackage{enumerate}
\usepackage{natbib}
\usepackage{url} 
\usepackage{pdflscape}

\usepackage{amsmath}
\usepackage{graphicx,psfrag,epsf}
\usepackage{enumerate}
\usepackage{natbib}
\usepackage{url} 
\usepackage{amsthm}
\usepackage{booktabs}
\usepackage{framed}  
\usepackage{caption}
\usepackage{pgfplots}
\usepgfplotslibrary{dateplot}
\usetikzlibrary{snakes}
\usepackage{float}
\usepackage{amsfonts}
\usepackage{multirow, booktabs}
\usepackage{cleveref}
\usepackage{subcaption}
\usepackage[ruled,vlined]{algorithm2e}
\usepackage[T1]{fontenc}
\usepackage[utf8]{inputenc}
\usepackage{authblk}
\usepackage[multiple]{footmisc}
\usepackage{blindtext,titlefoot}
\usepackage{sectsty}
\usepackage{xcolor}
\usepackage{tikz}
\usepackage{amsmath}
\usepackage{graphicx}
\usepackage{comment}
\usepackage{amsfonts}
\usepackage{bbm}
\usepackage{amssymb}
\usepackage{setspace}
\usepackage{natbib}
\usepackage[figuresright]{rotating}
\usepackage{adjustbox}
\usepackage{pifont}
\usetikzlibrary{positioning, shapes.geometric}

\usepackage{tabularx}
\usepackage{ragged2e} 
\newcolumntype{L}{>{\RaggedRight}X} 
\usepackage{lipsum} 

\usepackage[ruled,vlined]{algorithm2e}
\usepackage{amsmath, amssymb}
\usepackage{amsfonts, multirow, epsfig, subfig}
\usepackage{graphicx, pdflscape, verbatim, enumerate, colortbl, setspace}
\usepackage{setspace, color,bm}
\usepackage[normalem]{ulem}
\usepackage{cite}
\usepackage{multirow}
\usepackage{booktabs,array}
\usepackage{url}
\usepackage{bbm}

\newtheorem{proposition}{Proposition}
\theoremstyle{definition}
\newtheorem{condition}{Condition}

\newtheorem{lemma}{Lemma}

\newtheorem{theorem}{Theorem}

\theoremstyle{definition}
\newtheorem{remark}{Remark}

\usepackage{enumitem}
\usepackage{mathtools}
\newlist{subcondition}{enumerate}{1}
\setlist[subcondition,1]{
    label=\textup{\textbf{\thecondition\alph*}  },
    ref=\thecondition\alph*,
    leftmargin=1.5em,
    itemsep=0em,
    topsep=0.2em,
    parsep=0em,
    partopsep=0em
}
\makeatletter
\newcommand*{\rom}[1]{\expandafter\@slowromancap\romannumeral #1@}
\makeatother

\begin{document}

\def\spacingset#1{\renewcommand{\baselinestretch}%
{#1}\small\normalsize} \spacingset{1}

\sectionfont{\bfseries\large\sffamily}%
%
\newcommand*\emptycirc[1][1ex]{\tikz\draw (0,0) circle (#1);} 
\newcommand*\halfcirc[1][1ex]{%
  \begin{tikzpicture}
  \draw[fill] (0,0)-- (90:#1) arc (90:270:#1) -- cycle ;
  \draw (0,0) circle (#1);
  \end{tikzpicture}}
\newcommand*\fullcirc[1][1ex]{\tikz\fill (0,0) circle (#1);} 

\subsectionfont{\bfseries\sffamily\normalsize}%
%


\def\spacingset#1{\renewcommand{\baselinestretch}%
{#1}\small\normalsize} \spacingset{1}

\newcommand\Zijian[1]{{\color{magenta}Zijian: ``#1''}}

\begin{center}
    \Large \bf Propensity Score Propagation: A General Framework for Design-Based Inference with Unknown Propensity Scores
\end{center}

\begin{center}
  \large  $\text{Siyu Heng}^{*, \dagger, 1}$, $\text{Yanxin Shen}^{*, 2}$ and $\text{Zijian Guo}^{\dagger, 3}$
\end{center}

\begin{center}
\large   \textit{$^{1}$Department of Biostatistics, New York University}
\end{center}
\begin{center}
\large  \textit{$^{2}$School of Economics, Nankai University}
\end{center}
\begin{center}
  \large \textit{$^{3}$Center for Data Science, Zhejiang University}
\end{center}

\let \thefootnote\relax\footnotetext{$^{*}$Siyu Heng and Yanxin Shen contributed equally to this work.}

 \let \thefootnote\relax\footnotetext{$^{\dagger}$Siyu Heng (siyuheng@nyu.edu) and Zijian Guo (zijguo@zju.edu.cn) are corresponding authors. }

\begin{abstract}

Design-based inference, also known as randomization-based or finite-population inference, provides a principled framework for trustworthy statistical inference. It attributes randomness solely to the design mechanism, such as treatment assignment, survey sampling, or missingness, without imposing super-population distributional or modeling assumptions on the outcome data. From the seminal work of Fisher and Neyman to its recent resurgence, design-based inference has played a central role in causal inference, survey sampling, and missing data analysis. However, its use in many modern applications has been limited by a fundamental obstacle: existing design-based inference theory typically assumes that propensity scores (i.e., design probabilities) are known, whereas they are usually unknown in observational studies, real-world surveys, and missing data problems. We propose \textit{propensity score propagation}, a general framework for valid design-based inference with unknown propensity scores. The framework uses a regeneration-and-union procedure to propagate uncertainty from propensity score estimation into downstream design-based inference, without introducing super-population assumptions about the outcomes. It accommodates both parametric and nonparametric propensity score settings, integrates seamlessly with existing design-based methods developed for known propensity scores, and applies broadly across design-based problems. Theoretical and simulation results show that the proposed framework achieves nominal coverage, even when existing approaches exhibit substantial under-coverage.
\end{abstract}

\noindent


\spacingset{1.75} 

\section{Introduction}

\subsection{Background and Motivation}\label{subsec: background}
Providing valid statistical inference is a central task in statistics. Broadly speaking, most inference problems fall into one of two paradigms: (i) design-based inference, also known as randomization-based or finite-population inference, which targets estimands defined with respect to a fixed set of units, such as all counties in the United States or all students in a school, and in which randomness arises solely from the design mechanism \citep{fisher1937design, hajek1960limiting, rosenbaum2002observational, imbens2015causal, athey2017econometrics, li2017general, abadie2020sampling, zhang2023randomization, ding2024first, dasgupta2026introduction}; and (ii) super-population inference, also known as model-based or sampling-based inference, which targets estimands defined with respect to a hypothetical super-population distribution and treats the observed units as random realizations from that distribution \citep{efron1994introduction, politis1999subsampling, chernozhukov2018double, little2019statistical, vansteelandt2022assumption, kennedy2024semiparametric}.

Design-based inference differs from super-population inference in two key ways. First, it attributes randomness solely to the design mechanism, such as treatment assignment in causal studies or sampling in survey sampling, and does not impose distributional or modeling assumptions on the study units' potential outcomes or responses. Second, it conditions explicitly on the actual study units, making the resulting conclusions directly relevant to the actual studied units rather than to a hypothetical super-population. This distinction is especially crucial when the study units constitute a fixed or nearly exhaustive set, such as all U.S. counties, where applying super-population inference methods can hurt the interpretability (e.g., ill-defined estimands) and validity (e.g., inflated Type-I error rates) of statistical inference \citep{manski2018right, young2019channeling, abadie2020sampling, rambachan2025design}. See \citet{imbens2015causal}, \citet{athey2017econometrics}, \citet{li2023randomization}, \citet{ding2024first}, and \citet{dasgupta2026introduction} for detailed discussions on the connections and differences between design-based and super-population inference. Figure~\ref{fig: super and finite population} illustrates these two fundamental differences using a causal treatment-control comparison as an example.
\begin{figure}
        \centering
        \caption{Illustration of super-population versus design-based (finite-population) inference, using a causal treatment-control comparison as an example. }
 \label{fig: super and finite population}
        \includegraphics[width=0.8\linewidth]{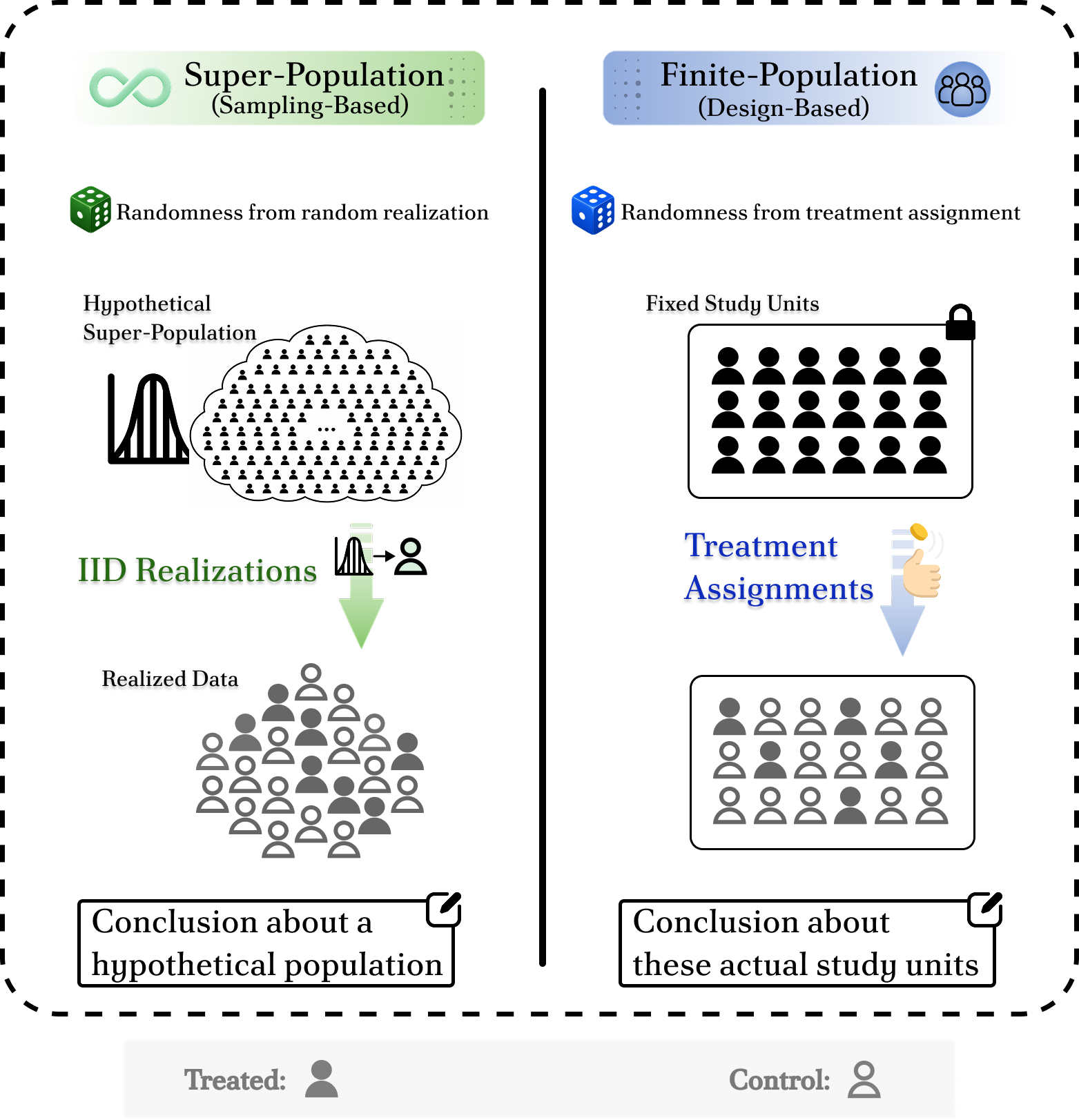}
       
    \end{figure}

The development of design-based inference can be viewed through two intertwined phases, each organized around the role of design probabilities, which are probabilities governing the realization of design variables, such as treatment indicators, sampling inclusion indicators, and missingness indicators \citep{rosenbaum1983central, rosenbaum2002observational, imbens2015causal, ding2024first}. We use \textit{propensity scores} to refer broadly to the design probabilities: treatment assignment probabilities in causal studies, sample-inclusion probabilities in survey sampling, and response probabilities in missing data problems. The first phase, beginning with the seminal work of Fisher and Neyman, asks how known propensity scores can be transformed into valid inference for fixed finite-population quantities without relying on super-population or outcome-modeling assumptions \citep{neyman1923application, fisher1937design, rosenbaum2002observational, imbens2015causal, ding2024first}. In settings where the propensity scores are known by design, uncertainty arises solely from the design mechanism itself, where this kind of uncertainty is referred to as \textit{design-induced uncertainty}. Examples of known propensity score regimes include randomized experiments with prespecified treatment assignment probabilities or probability surveys with prespecified inclusion probabilities. Over the past century, and especially amid the recent resurgence of design-based thinking, this known-design perspective has grown from the classical foundations of randomized experiments and survey sampling into a broad modern toolkit spanning causal inference, survey sampling, time series and panel data analysis, network analysis, and many other areas (e.g., \citealp{hajek1960limiting, hajek1964asymptotic, rosenbaum2002observational, hudgens2008toward, imbens2015causal, aronow2017estimating, athey2017econometrics, li2017general, bojinov2019time, abadie2020sampling, pashley2021conditional, athey2022design, li2023randomization, zhang2023randomization, ding2024first, dasgupta2026introduction}).

The second phase concerns settings in which the propensity scores are unknown, as is common in many modern studies. In observational studies, treatment assignment is not controlled by the investigator \citep{rubin1974estimating, rosenbaum1983central, rosenbaum2002observational, imbens2015causal, ding2024first}; in volunteer, web, and other non-probability surveys, inclusion probabilities are often uncertain \citep{bethlehem2010selection}; and in missing data problems, response probabilities are rarely available and must instead be learned from data \citep{little2019statistical, zhao2019sens_boot}. Figure~\ref{fig: design-based inference problems with known and unknown PS} provides illustrative examples of design-based inference problems with known and unknown propensity scores. The central question in this second phase is therefore whether one can preserve the defining appeal of design-based inference, namely valid uncertainty quantification without outcome modeling or super-population outcome assumptions, when the propensity scores are neither known nor controlled by investigators but must instead be estimated from data, potentially using flexible machine learning methods.

\begin{figure}
    \centering
     \caption{Design-based inference problems with known and unknown propensity scores. }
        \label{fig: design-based inference problems with known and unknown PS}
    \includegraphics[width=0.9\linewidth]{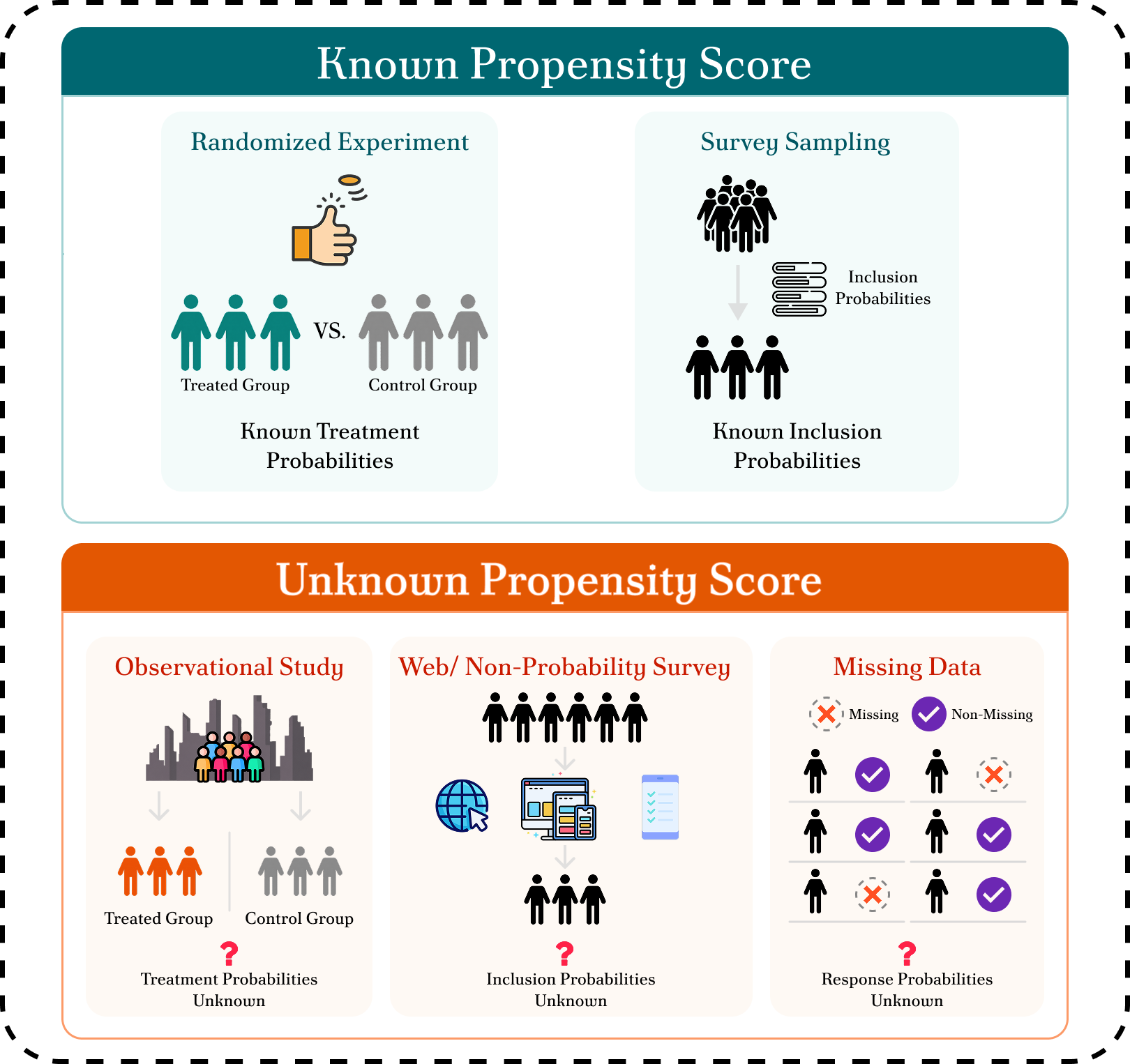}
  
\end{figure}

Existing approaches to this second-phase problem have followed three main paths: plug-in approaches, finite-population M-estimation approaches, and matching-based approaches, each of which has shaped modern practice. Plug-in approaches first estimate the propensity scores and then use them as though they were known. This strategy is simple and modular, but it can understate uncertainty because it treats estimated propensity scores as fixed rather than learned quantities, which may lead to substantial under-coverage \citep{pimentel2024covariate, zhu2025randomization}. 

Finite-population M-estimation approaches take a more integrated route by treating propensity score estimation and target estimand estimation as a joint estimating equation problem, thereby simultaneously accounting for both propensity score estimation uncertainty and design-induced uncertainty when the relevant estimating equations and variance formulas are available \citep{abadie2020sampling, xu2021potential, pimentel2024covariate}. However, existing versions are largely tied to parametric propensity score models, limiting their use with flexible machine learning propensity score estimators and complicating extensions such as sensitivity analysis for hidden bias \citep{zhao2019sens_boot}. 

Matching and stratification approaches, by contrast, seek to reconstruct an as-if randomized design by comparing units with similar covariates and have provided one of the most influential design-based frameworks for observational studies \citep{rosenbaum2002observational, rosenbaum2020design}. However, their validity relies on sufficiently close matching on the relevant covariates, a condition that is difficult to achieve when covariates are multivariate or continuous. Recent work has shown that residual imbalance from inexact matching can render downstream design-based inference asymptotically invalid \citep{pashley2021conditional, savje2022inconsistency, guo2023statistical, pimentel2024covariate, pimentel2024re, zhu2025randomization}.  Moreover, matching-based approaches are primarily tailored to causal inference and do not naturally extend to other design-based problems with unknown propensity scores, such as missing data analysis or non-probability survey sampling. Appendix~B.1 provides a more detailed discussion of the limitations of plug-in, finite-population M-estimation, and matching-based approaches.

In summary, despite the foundational importance, there remains no generally valid and broadly applicable framework for design-based inference with unknown propensity scores, which arise in many real-world studies.

\subsection{Key Challenges, Our Solution, and Its Contributions}\label{subsec: contributions}

The finite-population framework provides robust validity guarantees without requiring distributional or modeling assumptions on the outcome data, and it yields scientifically meaningful inferential conclusions without appealing to a hypothetical super-population. However, this distinctive strength also makes the problem substantially more challenging when propensity scores are unknown and must be estimated from the observed data. 

In super-population inference, bootstrap, resampling, and perturbation-based methods are often used to account for uncertainty from estimating nuisance parameters, such as propensity scores. These methods, however, are generally not directly suited to design-based inference with unknown propensity scores. In design-based causal inference, \citet{Imbens2004Nonparametric} showed that the variance of the sample average treatment effect estimator cannot generally be consistently estimated by the bootstrap without additional super-population assumptions, because ``the estimand itself changes across bootstrap samples.'' Similarly, \citet{ImbensMenzel2021CausalBootstrap} showed that even when the propensity score is a known function of covariates, bootstrap-based confidence sets for causal effects typically require positing a hypothetical super-population distribution for covariates and/or potential outcomes, which conflicts with the design-based principle that covariates and potential outcomes are fixed and randomness arises solely from the treatment assignment mechanism. These limitations underscore the central challenge of conducting valid design-based inference when propensity scores are unknown and must be estimated from the observed data.

In this work, we propose a novel and general framework, called \textit{propensity score propagation}, for valid design-based inference with unknown propensity scores. The central idea is a two-step regeneration-and-union procedure that explicitly propagates uncertainty from propensity score estimation into downstream design-based inference; see Figure~\ref{fig: workflow} for an illustration. 
\begin{figure}[ht]
    \centering
    \caption{\small Illustration of the regeneration-and-union procedure for constructing a design-based confidence interval (CI) under the propensity score (PS) propagation framework. $\mathcal{P}$ denotes oracle/known propensity scores, $\mathcal{C}_{1-\alpha}^{\text{oracle}}$ denotes a $100(1-\alpha)\%$ CI constructed under the oracle knowledge of $\mathcal{P}$, $\widehat{\mathcal{P}}^{(m)}$ denotes the regenerated propensity score estimate in the $m$th regeneration run, and each regenerated CI $\mathcal{C}_{1-\alpha}^{(m)}$ is obtained by replacing $\mathcal{P}$ in $\mathcal{C}_{1-\alpha}^{\text{oracle}}$ with $\widehat{\mathcal{P}}^{(m)}$.}
    \label{fig: workflow}
    \includegraphics[width=0.85\textwidth]{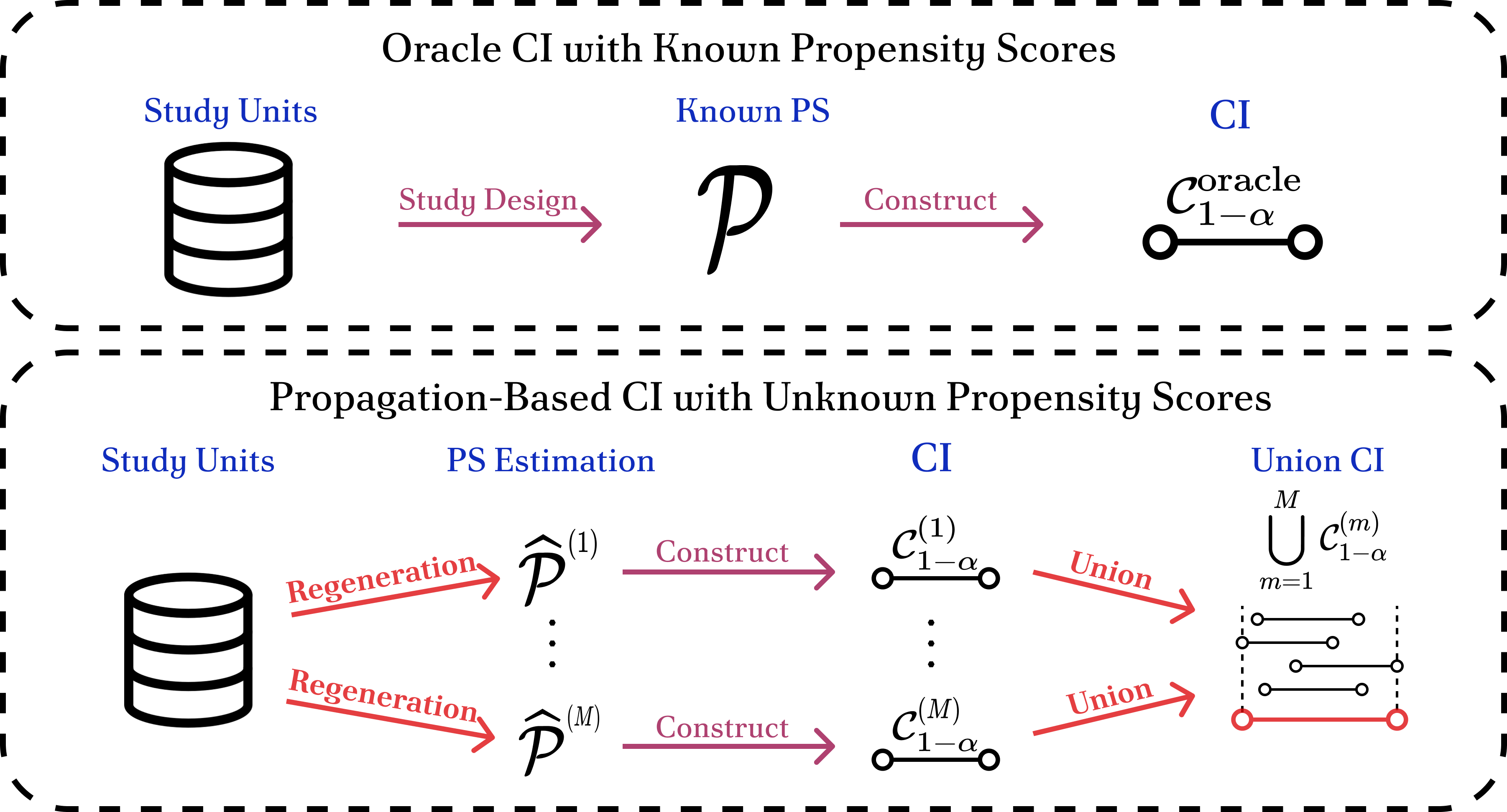} 
\end{figure}

In the regeneration step, we generate multiple plausible versions of propensity score estimates while keeping the finite population fixed. In parametric settings, this is done by regenerating the parameters of a propensity score model, as described in Algorithm~\ref{alg: parametric}. In nonparametric settings, it is implemented through random subsampling when estimating with a flexible propensity score learner, as described in Algorithm~\ref{alg: general}. Each regenerated propensity score estimate is then plugged into an oracle design-based confidence set mapping to produce a regenerated confidence set. In the union step, these regenerated confidence sets are combined by taking their union. Intuitively, the final confidence set attains valid coverage whenever at least one regenerated propensity score estimate, and hence its corresponding confidence set, is sufficiently close to the oracle counterpart. As shown in Section~\ref{sec: theory}, this event occurs with probability approaching one as the number of regeneration runs grows.

In this way, the propensity score propagation framework explicitly accounts for two sources of uncertainty: propensity score estimation uncertainty and design-induced uncertainty. The regeneration step captures propensity score estimation uncertainty by generating many design-compatible and plausible versions of the unknown propensity scores. For each regenerated propensity score estimate, the corresponding design-based confidence interval reflects design-induced uncertainty conditional on that regenerated propensity score estimate. The union step then combines these regenerated confidence intervals into a final confidence set that simultaneously accounts for both layers of uncertainty and restores nominal coverage. 

The contributions of propensity score propagation can be summarized from two complementary perspectives:
\begin{itemize}

\item First, to our knowledge, propensity score propagation provides the first general framework for valid design-based inference with unknown propensity scores. It accommodates both parametric and nonparametric propensity score estimators, including flexible machine-learning-based estimators, and offers a modular approach broadly applicable to causal inference, survey sampling, missing data analysis, and other related settings. In addition, it is readily implementable and integrates seamlessly with existing design-based inference methods developed for known propensity scores.

\item Second, it addresses a broader foundational challenge: valid inference when nuisance functions (e.g., propensity score functions) are learned using flexible machine learning methods. Double/debiased machine learning controls the first-order impact of nuisance estimation through orthogonal scores and cross-fitting, but its guarantees are generally formulated under super-population sampling and need not hold conditional on a fixed finite population \citep{chernozhukov2018double, vansteelandt2022assumption}. Propensity score propagation instead delivers validity for every fixed regular sequence of outcomes and covariates, extending flexible nuisance estimation to a genuinely finite-population, design-based regime; see Table~\ref{tab:contribution} for a comparison.

\end{itemize}

\begin{table}[H]
\centering
\caption{Comparison between double/debiased machine learning \citep{chernozhukov2018double} and propensity score propagation.}
\footnotesize
\begin{tabular}{c|c|c}
 \cline{2-3} 
&   Double/Debiased Machine Learning & Propensity Score Propagation \\
\hline
Objective & \multicolumn{2}{c}{Valid Inference with Nuisance Functions Estimated by Flexible Machine Learning}\\ 
\hline
Target Estimands &  Super-Population Functionals  & Finite-Population Estimands \\
Key Construction & Orthogonal Scores and Cross-Fitting  & Regeneration-and-Union (Our Proposal)  \\
Proof Technique & Semiparametric Efficiency Theory & Isoperimetric Arguments on the Hypercube \\
Validity Guarantee & Validity under Super-Population Sampling & Validity for Fixed Finite Populations  \\
\hline
\end{tabular}
\label{tab:contribution}
\end{table}

The remainder of the paper is organized as follows. Section~\ref{sec: review} reviews design-based inference with known propensity scores and illustrates the limitations of existing approaches for settings with unknown propensity scores. Section~\ref{sec: the general framework} introduces propensity score propagation in both parametric and nonparametric settings. Section~\ref{sec: theory} establishes theoretical guarantees. Section~\ref{sec: applications} presents selected applications to observational causal inference, missing data analysis, difference-in-differences analysis, and instrumental variable analysis, along with extensions to covariate adjustment and sensitivity analysis for hidden bias. Section~\ref{sec: simulation} reports simulation studies, and Section~\ref{sec: data} reanalyzes the data from \citet{heller2010using}. Section~\ref{sec: discussion} concludes. Technical details and additional theoretical and numerical results are provided in the supplementary materials.

\section{Review of Existing Design-Based Inference Methods}\label{sec: review}

\subsection{Design-Based Inference with Known Propensity Scores}\label{subsec: review known}

We review design-based inference when the design distribution is known. Suppose there are $N$ study units indexed by $[N]=\{1,\dots,N\}$, where $N>1$. For each unit $i$, let $Z_i\in\{0,1\}$ denote the design variable, such as a treatment indicator in causal inference or an inclusion indicator in survey sampling, and let $Y_i$ denote the observed outcome. We refer to the true design distribution $\mathcal P$ of $\mathbf Z=(Z_1,\dots,Z_N)$ as the \textit{propensity scores} \citep{rosenbaum1983central, rosenbaum2002observational}. When the design variables are independent, $\mathcal P$ is equivalently represented by the marginal propensity score vector $(p_1,\dots,p_N)$, where $p_i=P(Z_i=1)$. In this case, we slightly abuse notation and write $\mathcal P=(p_1,\dots,p_N)$, which we call the \emph{propensity score vector}.

Given the true propensity scores $\mathcal P$, the realized design variables $\mathbf Z$, and the observed outcomes $\mathbf Y=(Y_{1}, \dots, Y_{N})$, existing design-based inference procedures construct confidence sets for descriptive or causal finite-population quantities. We represent such a procedure abstractly by a mapping $\Lambda$ that outputs an oracle confidence set $\mathcal C_{1-\alpha}^{\text{oracle}}=\Lambda(\mathcal P,\mathbf Z,\mathbf Y)$ for a target estimand $\theta$ defined over the fixed $N$ study units. Under relevant regularity conditions, this oracle confidence set attains the desired $100(1-\alpha)\%$ design-based coverage guarantee under the true propensity scores $\mathcal P$. Here, the qualifier ``oracle'' denotes a confidence set, or later an estimator and variance estimator, constructed as if the true propensity scores $\mathcal P$ were available. Such oracle procedures are feasible when propensity scores are known by design; when propensity scores are unknown, they serve both as benchmarks and as the building blocks that propensity score propagation converts into feasible design-based inference procedures. 

Here, the general confidence set mapping $\Lambda(\mathcal P,\mathbf Z,\mathbf Y)$ is not restricted to Wald-type confidence intervals. It may represent any design-based inference procedure that is valid when $\mathcal P$ is known, such as a Wald-type interval, a confidence set obtained by randomization test inversion, or another non-Wald design-based confidence set.

We consider design-based causal inference in randomized experiments as a canonical example in which propensity scores are known by design and explicit oracle confidence set mappings $\Lambda$ are well established. This setting also provides a useful illustration of the fundamental distinction between design-based and super-population inference. For $i \in [N]$, we use $Z_i\in\{0,1\}$ to denote unit $i$'s treatment assignment and define $\mathbf Z=(Z_1,\ldots,Z_N)$. Under the potential outcomes framework \citep{neyman1923application, rubin1974estimating}, each unit $i$ has potential outcome under treatment (denoted as $Y_{i}(1)$) and that under control (denoted as $Y_{i}(0)$), and the observed outcome is $Y_i=Z_iY_i(1)+(1-Z_i)Y_i(0)$. Design-based causal inference treats $\{(Y_i(0), Y_i(1))\}_{i=1}^N$ as fixed and attributes all randomness to the known randomization distribution of $\mathbf Z$, enabling credible causal conclusions without super-population outcome assumptions \citep{rosenbaum2002observational, imbens2015causal, athey2017econometrics, li2017general, fogarty2018mitigating, ding2024first}. 

Consider the sample average treatment effect (SATE)
$\theta_{\tau}=N^{-1}\sum_{i=1}^N\{Y_i(1)-Y_i(0)\}$. Under a Bernoulli randomized experiment with independent assignments $Z_i\sim\text{Bernoulli}(p_i)$ for prespecified $p_i$, the inverse-probability-weighting (IPW) estimator under $\mathcal P=(p_1,\ldots,p_N)$ is
\begin{equation*}
    \widehat\theta_{\tau, \text{oracle}}
    =\frac{1}{N}\sum_{i=1}^N\widehat{\theta}_{\tau, \text{oracle}, i}
    =\frac{1}{N}\sum_{i=1}^N\left\{
    \frac{Z_iY_i}{p_i}
    -\frac{(1-Z_i)Y_i}{1-p_i}
    \right\},
\end{equation*}
which is unbiased for $\theta_{\tau}$ under the true propensity scores $\mathcal{P}$. A valid design-based variance estimator is
\begin{equation*}
    \widehat{\mathcal V}_{\tau, \text{oracle}}
    =\frac{1}{N(N-1)}\sum_{i=1}^N
    \left(\widehat\theta_{\tau, \text{oracle}, i}
    -\widehat\theta_{\tau, \text{oracle}}\right)^2,
\end{equation*}
which satisfies $E(\widehat{\mathcal V}_{\tau, \text{oracle}})\geq \text{Var}(\widehat\theta_{\tau, \text{oracle}})$ conditional on any fixed sequence of potential outcomes $\{Y_i(0),Y_i(1)\}_{i=1}^N$ \citep{fogarty2018mitigating, ding2024first}. Under suitable finite-population regularity conditions, a finite-population central limit theorem \citep{li2017general, fogarty2018mitigating} yields the design-based $100(1-\alpha)\%$ confidence interval
\begin{equation*}
\mathcal C_{1-\alpha}^{\text{oracle}}(\theta_{\tau})
 =\Lambda_{\tau}(\mathcal P,\mathbf Z,\mathbf Y)=\left[
 \widehat\theta_{\tau, \text{oracle}}
 -z_{1-\alpha/2}\cdot \widehat{\mathcal V}_{\tau, \text{oracle}}^{1/2},\,
 \widehat\theta_{\tau, \text{oracle}}
 +z_{1-\alpha/2}\cdot \widehat{\mathcal V}_{\tau, \text{oracle}}^{1/2}
 \right],
\end{equation*}
where the construction of $\widehat\theta_{\tau, \text{oracle}}$ and $\widehat{\mathcal V}_{\tau, \text{oracle}}$ relies on the knowledge of true propensity scores $\mathcal{P}$, the realized design variables $\mathbf Z$, and the observed outcomes $\mathbf{Y}$. Here, $\Lambda_{\tau}$ is a concrete example of the general oracle confidence-set mapping $\Lambda$. 

Under regularity conditions, it is shown that $\mathcal C_{1-\alpha}^{\text{oracle}}(\theta_{\tau})$ achieves design-based validity (i.e., finite-population validity), that is, for every regular sequence of fixed finite populations $(\mathbf Y(0),\mathbf Y(1))=\{Y_i(0),Y_i(1)\}_{i=1}^N$,
\begin{equation}\label{eqn: design-based validity}
   P_{\mathbf Z}\left(
   \theta_{\tau}\in
   \mathcal C_{1-\alpha}^{\text{oracle}}(\theta_{\tau})
   \mid \mathbf Y(0),\mathbf Y(1)
   \right)
   \geq 1-\alpha-o(1),
\end{equation}
where the probability is taken only over the design distribution of $\mathbf Z$.

Design-based validity in \eqref{eqn: design-based validity} is more demanding than super-population validity: it requires nominal coverage over the true design distribution of $\mathbf Z$ for every regular sequence of fixed finite populations, rather than only after averaging over an assumed distribution of finite populations. To illustrate, a super-population counterpart to the SATE $\theta_{\tau}$ is the population average treatment effect (PATE) $\theta_{\tau}^{\text{sp}}=E\{Y(1)-Y(0)\}$. In contrast to \eqref{eqn: design-based validity}, a super-population confidence set $\mathcal C_{1-\alpha}^{\text{sp}}$ is typically required to satisfy
\begin{equation}\label{eqn: super-population validity}
P_{(\mathbf Z,\mathbf Y(0),\mathbf Y(1))}
\left(
\theta_{\tau}^{\text{sp}}\in\mathcal C_{1-\alpha}^{\text{sp}}
\right)
\geq 1-\alpha-o(1),
\end{equation}
where the probability averages over the joint distribution of the treatment assignments and potential outcomes $(\mathbf Z,\mathbf Y(0),\mathbf Y(1))$. Therefore, super-population validity in \eqref{eqn: super-population validity} can hold on average even when design-based validity in \eqref{eqn: design-based validity} fails for some fixed finite populations $(\mathbf Y(0),\mathbf Y(1))=\{Y_{i}(0), Y_{i}(1)\}_{i=1}^{N}$.

This requirement for coverage conditional on the fixed finite population is a key advantage of design-based validity. Design-based inference directly targets the finite population without positing a larger super-population from which the study units were sampled. This perspective is especially natural when the finite population itself is the object of interest, such as all U.S. counties in a policy evaluation or all students in a particular school \citep{abadie2020sampling}. Accordingly, design-based and super-population inference differ not only in their validity requirements but also in their targets: design-based inference targets a finite-population estimand such as $\theta_{\tau}$, whereas super-population inference typically targets a super-population estimand such as $\theta_{\tau}^{\text{sp}}$.

In some settings, the distinction between the two paradigms lies in the justification rather than in the numerical form of the procedure. In the Bernoulli experiment above, the same IPW estimator and variance estimator can provide a design-based confidence interval for the SATE $\theta_{\tau}$ and, under additional super-population assumptions, a super-population confidence interval for the PATE $\theta_{\tau}^{\text{sp}}$. The super-population interpretation requires assumptions linking the observed units to the target population, such as i.i.d.\ sampling of the potential outcome pairs $(Y_{i}(0), Y_{i}(1))$ from an underlying distribution $(Y(0), Y(1))$ that satisfies some regularity condition. By contrast, the design-based interpretation treats the realized potential outcomes as fixed and relies on the known assignment mechanism and finite-population regularity conditions, without requiring a probability model for the outcomes. Therefore, the coincidence of the interval formulas does not make the validity guarantees equivalent: design-based validity is required to hold for every regular sequence of fixed finite populations, whereas super-population validity holds only after averaging over the assumed super-population outcome distribution.

Design-based causal inference for randomized experiments is one instance of the broader class of design-based inference settings. We briefly present several additional examples of design-based inference here and defer further details to Section~\ref{sec: applications} and Appendix~C. Although these settings differ in their estimands, design variables, and confidence set formulas, they share the same basic structure: once $\mathcal P$ is known, existing design-based inference theory supplies a valid oracle confidence set mapping $\Lambda(\mathcal P,\mathbf Z,\mathbf Y)$ that satisfies the required design-based coverage guarantee.
\begin{itemize}
\item In probability survey sampling, $Z_i$ indicates whether unit $i$ is sampled, and the inclusion probabilities $\mathcal P$ are known by design. Classical estimators, such as the Horvitz-Thompson estimator, together with their design-based variance estimators, provide valid inference for finite-population quantities under known inclusion probabilities \citep{horvitz1952generalization, hajek1960limiting, hajek1964asymptotic}.
\item In difference-in-differences (DID), recent work develops design-based inference by conditioning on an assignment mechanism for treatment adoption, thereby avoiding super-population parallel-trends assumptions when the assignment mechanism is taken as known \citep{athey2022design, rambachan2025design}.
\item In instrumental-variable (IV) analysis, design-based arguments similarly apply when instrument assignment probabilities are known or treated as known, enabling randomization-based inference for causal effects under standard IV conditions \citep{rosenbaum2002observational, imbens2005robust, kang2016full}.
\end{itemize}

\subsection{Illustration of Existing Approaches for Design-Based Inference with Unknown Propensity Scores}\label{subsec: review unkown}

Section~\ref{subsec: review known} reviewed design-based inference with known propensity scores $\mathcal{P}$. In many observational, quasi-experimental, and missing data settings, however, the true propensity scores $\mathcal P$ are unknown and must be estimated from the observed data. The central difficulty is therefore how to account for uncertainty from propensity score estimation when constructing the downstream design-based confidence set, while preserving the finite-population validity emphasized above.

To illustrate this difficulty, we conduct a motivating simulation study in nonparametric propensity score settings. We construct fixed finite populations with $N=1000$ units under two treatment effect settings. Conditional on each fixed finite population, we evaluate coverage and confidence interval length over repeated treatment assignments $\mathbf{Z}$ generated from two propensity score settings: a nonlinear selection model and a nonlinear logistic model. For each finite population, the target estimand is the SATE $\theta_{\tau}=N^{-1}\sum_{i=1}^{N}\{Y_i(1)-Y_i(0)\}$. Detailed data-generating mechanisms and implementation procedures are provided at the beginning of Section~\ref{sec: simulation}. We compare two representative classes of existing approaches: plug-in and matching-based methods. We do not directly compare with finite-population M-estimation approaches here because they require a parametric propensity score model, whereas this motivating study focuses on nonparametric propensity score settings. The plug-in approach estimates propensity scores using the widely used XGBoost \citep{chen2016xgboost} and then treats the estimated propensity scores as known. The matching-based approach uses the widely used optimal full matching \citep{hansen2004full,hansen2006optimal}; see Appendix~D.3 for implementation details. Table~\ref{tab: motivating-simulation} reports the coverage rate, mean estimation bias, and average confidence interval length of each approach, together with those of the oracle design-based confidence interval constructed using the true propensity scores.

\begin{table}[ht]
\centering
\scriptsize
\caption{Motivating simulation results for design-based inference with unknown propensity scores. For each effect setting, propensity score setting, and approach, we report the coverage rate (Coverage), mean bias (Bias), average confidence interval length (Length), coverage rate of the oracle bias-aware confidence interval (Coverage (OBA)), and the average length of the oracle bias-aware confidence interval (Length (OBA)). The row ``Propagation'' is included as a preview of the proposed propensity score propagation framework with $M=100$ regeneration runs; full simulation settings and results are reported in Section~\ref{sec: simulation}. A dash indicates that the quantity is not applicable.}
\resizebox{\textwidth}{!}{
\begin{tabular}{lcccccccccc}
\toprule
\multirow{3}{*}{} 
& \multicolumn{5}{c}{Propensity Score Setting 1}
& \multicolumn{5}{c}{Propensity Score Setting 2} \\
\cmidrule(lr){2-6} \cmidrule(lr){7-11}
& Coverage & Bias & Length & \begin{tabular}{@{}c@{}}Coverage\\(OBA)\end{tabular} & \begin{tabular}{@{}c@{}}Length\\(OBA)\end{tabular}
& Coverage & Bias & Length & \begin{tabular}{@{}c@{}}Coverage\\(OBA)\end{tabular} & \begin{tabular}{@{}c@{}}Length\\(OBA)\end{tabular} \\
\midrule
\multicolumn{11}{l}{\textbf{Effect Setting 1}} \\
\midrule
Plug-in                    
& 0.572 & 0.179 & 0.376 & 0.947 & 0.545
& 0.568 & 0.186 & 0.381 & 0.942 & 0.559 \\
Matching       
& 0.022 & 0.176 & 0.224 & 0.952 & 0.455
& 0.142 & 0.133 & 0.199 & 0.953 & 0.371 \\
Propagation (Our Proposal) & 0.998 & -- & 0.632 & -- & -- & 1.000 & -- & 0.631 & -- & -- \\
Oracle                     
& 0.941 & 0.027 & 0.486 & -- & --
& 0.941 & 0.016 & 0.442 & -- & -- \\
\midrule
\multicolumn{11}{l}{\textbf{Effect Setting 2}} \\
\midrule
Plug-in                    
& 0.739 & 0.174 & 0.423 & 0.948 & 0.562
& 0.670 & 0.192 & 0.427 & 0.943 & 0.599 \\
Matching        
& 0.100 & 0.178 & 0.276 & 0.951 & 0.460
& 0.298 & 0.134 & 0.235 & 0.953 & 0.377 \\
Propagation (Our Proposal) & 1.000 & -- & 0.716 & -- & -- & 1.000 & -- & 0.715 & -- & -- \\
Oracle                     
& 0.955 & 0.027 & 0.518 & -- & --
& 0.945 & 0.016 & 0.473 & -- & -- \\
\bottomrule
\end{tabular}
}
\label{tab: motivating-simulation}
\end{table}

Table~\ref{tab: motivating-simulation} shows that the plug-in approach can suffer from substantial finite-sample estimation bias and severe under-coverage. This behavior is expected: the plug-in approach first estimates $\mathcal P$ and then proceeds as if the estimate were the truth, thereby ignoring uncertainty from propensity score estimation. As a result, the confidence intervals fail to capture all relevant sources of uncertainty and fall well below the nominal coverage level of 0.95 in all four settings.

The matching-based approach also exhibits substantial under-coverage. Its finite-sample bias reflects how inexact matching can introduce non-negligible bias into downstream inference. This pattern is consistent with recent work showing that residual imbalance from inexact matching can induce substantial bias and invalidate design-based inference \citep{guo2023statistical, pimentel2024covariate, zhu2025randomization}.

For the plug-in and matching-based approaches, Table~\ref{tab: motivating-simulation} also reports the coverage rate and average length of the oracle bias-aware (OBA) confidence interval following \citet{armstrong2020bias}. These OBA intervals use oracle knowledge of the true estimation bias and standard error and are therefore unattainable in practice; they are included only as benchmark intervals calibrated to attain nominal coverage at level 0.95. The substantially larger OBA interval lengths relative to the original interval lengths indicate that restoring nominal coverage for plug-in and matching-based approaches would require additional bias and uncertainty corrections based on information unavailable in practice. Section~\ref{sec: simulation} provides a more complete and detailed discussion of these simulation patterns.

For orientation, Table~\ref{tab: motivating-simulation} also includes the proposed propensity score propagation framework with $M=100$ regeneration runs. Across all four settings, propensity score propagation attains nominal coverage while retaining practically sized intervals compared with both the oracle and OBA confidence intervals. The framework and its full theoretical and numerical evaluation are presented in Sections~\ref{sec: the general framework}--\ref{sec: simulation}.

\section{The Propensity Score Propagation Framework}\label{sec: the general framework}

\subsection{Overview and Guiding Principles}\label{subsec: overview of principle}
We now describe the central idea of propensity score propagation: a general \emph{regeneration-and-union} procedure for design-based inference with unknown propensity scores.
\begin{itemize}
\item[(i)]\textbf{Regeneration Step:} Repeatedly regenerate estimated propensity scores $\widehat{\mathcal{P}}^{(1)},\dots,\widehat{\mathcal{P}}^{(M)}$ using a carefully designed regeneration scheme (e.g., those described in the next two subsections). For each $\widehat{\mathcal{P}}^{(m)}$ ($m\in [M]$), construct a design-based confidence set
    \begin{equation*}
    \mathcal{C}_{1-\alpha}^{(m)}=\Lambda(\widehat{\mathcal{P}}^{(m)}, \mathbf{Z}, \mathbf{Y}),
    \end{equation*}
    that is, by plugging $\widehat{\mathcal{P}}^{(m)}$ into the oracle confidence set mapping $\Lambda$ as if $\widehat{\mathcal{P}}^{(m)}$ is the true propensity scores $\mathcal{P}$. 
    \item[(ii)]\textbf{Union Step:} Aggregate the regenerated confidence sets using the following union,
    \begin{equation*}
    \mathcal{C}_{1-\alpha}=\bigcup_{m=1}^{M}\mathcal{C}_{1-\alpha}^{(m)}.
    \end{equation*}
In practice, one may use a restricted union that excludes regeneration runs $m\in[M]$ with extreme propensity score vectors $\widehat{\mathcal P}^{(m)}$, as discussed later in Sections~\ref{subsec: Parametric Propensity Score Propagation} and \ref{subsec: Nonparametric Propensity Score Propagation}.
\end{itemize}

The rationale underlying propensity score propagation is that it handles two distinct sources of uncertainty through two complementary steps. The regeneration step addresses uncertainty about the unknown propensity scores by generating multiple plausible propensity score vectors $\widehat{\mathcal P}^{(1)},\ldots,\widehat{\mathcal P}^{(M)}$. Under ignorability and suitably constructed regeneration mechanisms, Theorems~\ref{thm: uniform approximation} and~\ref{thm: approximation nonparametric} show that, when $M$ is sufficiently large, with high probability, there exists at least one run $m^{*}\in[M]$ such that $\widehat{\mathcal P}^{(m^{*})}$ is sufficiently close to the oracle propensity score vector $\mathcal P$. Consequently, the associated confidence set $\mathcal C_{1-\alpha}^{(m^{*})}$ has the same asymptotic coverage behavior as $\mathcal C_{1-\alpha}^{\text{oracle}}$. Conditional on each $\widehat{\mathcal P}^{(m)}$, the oracle confidence set mapping $\mathcal C_{1-\alpha}^{(m)}=\Lambda(\widehat{\mathcal P}^{(m)},\mathbf Z,\mathbf Y)$ accounts for design-induced uncertainty under that fixed $\widehat{\mathcal P}^{(m)}$. The union step then recombines the two sources of uncertainty: each regenerated confidence set accounts for design-induced uncertainty conditional on one plausible propensity score vector, while the union over regenerated confidence sets incorporates uncertainty from propensity score estimation.

\begin{figure}[h]
    \centering
    \caption{\small Illustration of propensity score propagation ($M=100$) using one simulated dataset. The blue interval is the oracle confidence interval (CI) constructed using the true propensity scores, and the black intervals are regenerated CIs constructed across the $100$ regeneration runs. The arrow marks the regenerated CI closest to the oracle CI among the $100$ regeneration runs. The red interval is the propagation-based CI obtained by taking the union of all regenerated CIs, with the red dashed lines marking its endpoints. The gray dashed line denotes the true estimand value.}
    \label{fig: example}
    \includegraphics[width=1\textwidth]{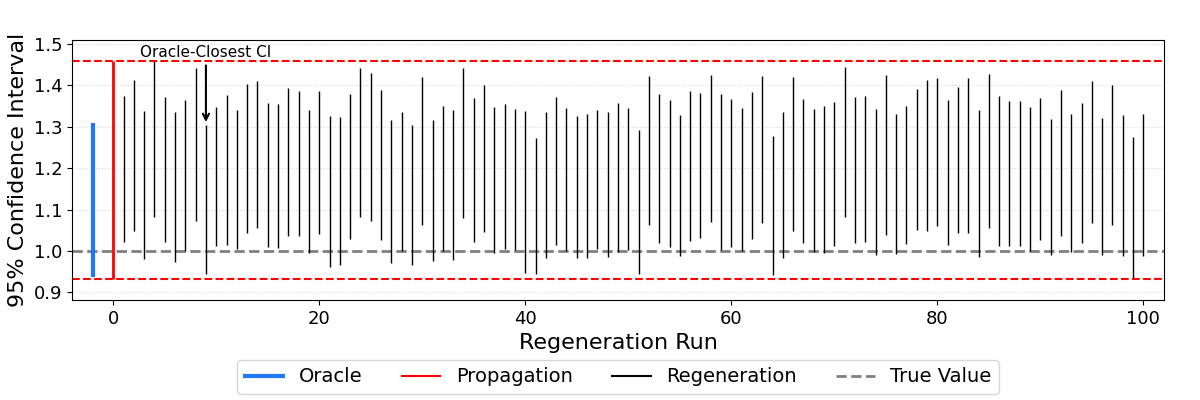}
\end{figure}

Figure~\ref{fig: example} illustrates the key mechanism underlying propensity score propagation. The variation among the regenerated CIs reflects how propensity score estimation uncertainty propagates into downstream design-based inference. With sufficiently many regeneration runs, the procedure is likely to produce at least one regenerated CI close to the oracle CI, as highlighted by the arrow. Because the propagation-based CI takes the union of all regenerated CIs, it contains this oracle-close regenerated CI, which is the key intuition behind its coverage guarantee. In this realization, the propagation-based CI contains the true estimand and is only moderately wider than the oracle CI, illustrating that propagation need not result in excessive conservativeness. This is because regeneration is calibrated to the data-compatible uncertainty in estimating the underlying propensity scores, rather than introducing unrestricted variation in the regenerated propensity scores. Under regularity conditions, the regenerated propensity scores concentrate around the true propensity scores as the sample size grows. If the oracle inference procedure is stable to such local variations in the propensity scores, the regenerated CIs overlap substantially, so taking their union adds only a moderate amount to the oracle CI's length.

In Sections~\ref{subsec: Parametric Propensity Score Propagation} and~\ref{subsec: Nonparametric Propensity Score Propagation}, we present example regeneration mechanisms for parametric and nonparametric propensity score settings, respectively, and establish their coverage and length convergence guarantees in Section~\ref{sec: theory}.

\subsection{Parametric Propensity Score Propagation} \label{subsec: Parametric Propensity Score Propagation}

To illustrate the main ideas, we begin with the parametric propensity score setting. Suppose the design variables $Z_1,\dots,Z_N$ are mutually independent, with propensity scores $p_i=P(Z_i=1\mid \mathbf x_i)$ following a generalized linear model (GLM):
\begin{equation}\label{eqn: GLM}
P(Z_i=1\mid \mathbf x_i)=\Psi(\boldsymbol{\beta}^{T}\mathbf{x}_{i}),
\end{equation}
where $\Psi$ is a known one-to-one inverse link function, $\boldsymbol{\beta}\in \mathbbm{R}^{d}$ is the true propensity score parameter vector, and $\mathbf{x}_{i}=(x_{i1},\dots,x_{id})^{T}$ is a $d$-dimensional covariate vector including an intercept term, with $d$ fixed. This formulation includes many commonly used propensity score models, such as the logistic model with $\Psi(t)=\exp(t)/\{1+\exp(t)\}$, the probit model with $\Psi=\Phi$, where $\Phi$ is the standard normal cumulative distribution function, and the linear probability model with $\Psi(t)=t$, among others.

Let $\widehat{\boldsymbol{\beta}}=(\widehat{\beta}_{1},\dots,\widehat{\beta}_{d})^{T}$ denote the maximum likelihood estimator (MLE) based on the covariate information $\mathbf{X}=(\mathbf{x}_{1}, \dots, \mathbf{x}_{N})$ and realized design variables $\mathbf{Z}=(Z_{1}, \dots, Z_{N})$. By classical MLE theory for GLMs with fixed design \citep{fahrmeir1985consistency}, under the mild regularity conditions specified in Appendix A.3, 
\begin{equation}\label{eqn: asymptotic normality}
  \sqrt{N}\cdot \widehat{\Omega}_{N}^{-\frac{1}{2}} \big(\widehat{\boldsymbol{\beta}}-\boldsymbol{\beta}\big)\xrightarrow{d} N\left(\mathbf{0}_{d\times 1}, \mathbf{I}_{d\times d}\right) \text{ as $N\rightarrow \infty$,}
\end{equation}
where $\widehat{\Omega}_{N}\xrightarrow{p}\Omega$ for some positive definite $d\times d$ matrix $\Omega$. The explicit form of $\widehat{\Omega}_{N}$ is presented in \citet{fahrmeir1985consistency}.

Motivated by (\ref{eqn: asymptotic normality}), we repeatedly regenerate the propensity score model parameters
\begin{equation}\label{eqn: parametric regeneration}
        \widehat{\boldsymbol{\beta}}^{(m)} \overset{\text{iid}}{\sim} N\!\left(\widehat{\boldsymbol{\beta}}, N^{-1}\cdot \widehat{\Omega}_{N} \right), 
    \quad m=1,\dots,M,
\end{equation}
and further construct the regenerated propensity score vectors $\widehat{\mathcal{P}}^{(m)}=(\widehat{p}_{1}^{(m)},\dots, \widehat{p}_{N}^{(m)})=(\Psi(\widehat{\boldsymbol{\beta}}^{(m)T}\mathbf{x}_{1}),\dots, \Psi(\widehat{\boldsymbol{\beta}}^{(m)T}\mathbf{x}_{N}))$ for $m=1,\dots, M$.  

Each draw $\widehat{\boldsymbol{\beta}}^{(m)}$ induces a candidate design distribution through the corresponding regenerated propensity score vector $\widehat{\mathcal P}^{(m)}$. Because these draws are centered on the fitted propensity score model and scaled to reflect the estimated uncertainty in $\widehat{\boldsymbol{\beta}}$, they represent data-compatible versions of the unknown design mechanism. For each $\widehat{\mathcal P}^{(m)}$, we then construct the corresponding design-based confidence set by applying the oracle confidence set mapping $\Lambda$ with $\widehat{\mathcal{P}}^{(m)}$ in place of the oracle propensity score vector $\mathcal{P}$, and finally take the union over all regeneration runs. The detailed procedure for parametric propensity score propagation is described in Algorithm~\ref{alg: parametric}.

\begin{algorithm}  
\caption{Parametric propensity score propagation.} \label{alg: parametric}

\textbf{Input:} The observed data $(\mathbf{Z}, \mathbf{X}, \mathbf{Y})$, a prespecified significance level $\alpha \in (0, 0.5)$, and a prespecified number of regeneration runs $M$. 

\smallskip

\textbf{The Regeneration Step:} In each regeneration run $m \in [M]$, we do the following:

\begin{enumerate}

\item Draw regenerated propensity score model parameters $\widehat{\boldsymbol{\beta}}^{(m)} \sim N(\widehat{\boldsymbol{\beta}}, N^{-1}\cdot \widehat{\Omega}_{N})$.

\item Regenerate the propensity score vector $\widehat{\mathcal{P}}^{(m)}=(\widehat{p}_{1}^{(m)},\dots, \widehat{p}_{N}^{(m)})$ with $\widehat{p}_{i}^{(m)}=\Psi(\widehat{\boldsymbol{\beta}}^{(m)T}\mathbf{x}_{i})$.

\item Calculate the design-based confidence set $\mathcal{C}_{1-\alpha}^{(m)}$ by plugging $\widehat{\mathcal{P}}^{(m)}$ into the confidence set mapping $\Lambda$ established under oracle propensity scores, i.e., setting $\mathcal{C}_{1-\alpha}^{(m)}=\Lambda(\widehat{\mathcal{P}}^{(m)}, \mathbf{Z}, \mathbf{Y})$.

\end{enumerate}

\vspace{-0.3cm}

\textbf{The Union Step:} Take the union of $\mathcal{C}_{1-\alpha}^{(m)}$, i.e., set $\mathcal{C}_{1-\alpha}=\bigcup_{m=1}^{M}\mathcal{C}_{1-\alpha}^{(m)}$. 

\smallskip
\textbf{Output:} A design-based confidence set $\mathcal{C}_{1-\alpha}$ for the target estimand. 
\end{algorithm}

As the number of regeneration runs $M$ increases, the regeneration scheme in (\ref{eqn: parametric regeneration}), as implemented in Algorithm~\ref{alg: parametric}, ensures that, with high probability, at least one regenerated parameter vector $\widehat{\boldsymbol{\beta}}^{(m)}$ lies sufficiently close to the oracle parameter vector $\boldsymbol{\beta}$, or equivalently, that at least one regenerated propensity score vector $\widehat{\mathcal{P}}^{(m)}$ lies sufficiently close to the oracle propensity score vector $\mathcal{P}$. This approximation property is formalized in Theorem~\ref{thm: uniform approximation} in Section~\ref{subsec: theory of parametric PSP}. Consequently, under mild regularity conditions, at least one regenerated confidence set $\mathcal{C}_{1-\alpha}^{(m)}$ achieves the same asymptotic coverage as the oracle confidence set $\mathcal{C}_{1-\alpha}^{\text{oracle}}$. This implies that the union confidence set $\mathcal{C}_{1-\alpha}=\bigcup_{m=1}^{M}\mathcal{C}_{1-\alpha}^{(m)}$ has asymptotic coverage no lower than that of $\mathcal{C}_{1-\alpha}^{\text{oracle}}$, while still maintaining an informative confidence set length, as established in Theorems~\ref{thm: coverage rate parametric} and~\ref{thm: CI length parametric} in Section~\ref{subsec: theory of parametric PSP}.

In practice, the final confidence set $\mathcal{C}_{1-\alpha}$ may be further shortened if the union excludes $\mathcal{C}_{1-\alpha}^{(m)}$ constructed from outlying propensity score estimates $\widehat{\mathcal{P}}^{(m)}$. Specifically, during the union step, we can adopt a restricted-union strategy: rather than taking the union over all $m\in[M]$, we only retain those indices $m$ for which $\widehat{\mathcal{P}}^{(m)}$ is not flagged as an outlier. Let $\widehat{\sigma}_{kk}^{2}$ denote the $k$th diagonal entry of $\widehat{\Omega}_{N}$ in (\ref{eqn: asymptotic normality}), and define
\begin{equation}\label{eqn: definition of the restrictive set}
    \mathcal{M}(\alpha^{\prime})=\Big \{m\in [M]: \max_{k\in [d]}\Big|\big(\widehat{\beta}_{k}^{(m)}-\widehat{\beta}_{k}\big)/\sqrt{\widehat{\sigma}_{kk}^{2}/N}\Big|\leq 1.01\cdot z_{1-\frac{\alpha^{\prime}}{2d}}\Big\},
\end{equation}
in which $\alpha^{\prime}\in (0, \alpha)$ (e.g., setting $\alpha^{\prime}=0.01$ if $\alpha=0.05$) is some prespecified significance level allocated for defining the restricted set $\mathcal{M}(\alpha^{\prime})\subseteq [M]$, and $z_{1-\frac{\alpha^{\prime}}{2d}}$ is the $(1-\frac{\alpha^{\prime}}{2d})$-quantile of standard normal distribution. In the union step of Algorithm~\ref{alg: parametric}, the final confidence set can then be constructed via the restricted union: 
\begin{equation*}
    \widetilde{\mathcal{C}}_{1-\alpha}=\bigcup_{m \in \mathcal{M}(\alpha^{\prime})}\mathcal{C}_{1-\alpha_0}^{(m)},
\end{equation*}
where $\alpha_0=\alpha-\alpha^{\prime}$ is the significance level used for each downstream design-based confidence set. This level adjustment can be viewed as a multiplicity correction that allocates error probability $\alpha^{\prime}$ to the construction of the retained regeneration set $\mathcal{M}(\alpha^{\prime})$ and $\alpha_0=\alpha-\alpha^{\prime}$ to downstream confidence sets.

\subsection{Nonparametric Propensity Score Propagation}\label{subsec: Nonparametric Propensity Score Propagation}

We next consider the nonparametric propensity score setting, in which the propensity score function 
$P(Z=1\mid \mathbf{x})$ is left unspecified and estimated using flexible, data-adaptive learners, such as nonparametric machine learning methods. The regeneration method developed for the parametric setting cannot be directly extended to this general nonparametric setting, since there is no finite-dimensional propensity score parameter to regenerate. Nevertheless, the same regeneration principle can be generalized through repeated subsampling. In this setting, subsampling provides a natural way to generate plausible, data-compatible versions of the learned design mechanism.

The concrete idea is to repeatedly select subsamples of units, fit the propensity score learner using their observed design variables $\mathbf Z$ and covariates $\mathbf X$, and then predict propensity scores for all units in the fixed finite population. Under suitable regularity conditions, propensity score estimates fitted on different subsamples provide plausible versions of the unknown true propensity score vector. Repeating this procedure produces a collection of regenerated propensity score vectors that reflects the uncertainty inherent in flexible propensity score learning, without resampling or perturbing the outcome data and without introducing a super-population outcome model.

We propose the nonparametric propensity score propagation in Algorithm~\ref{alg: general}.
In particular, we use a repeated cross-fitting version of subsampling to avoid in-sample prediction of propensity scores. In each regeneration run, the units are randomly split into two disjoint subsets. The propensity score learner is fitted on one subset and used to predict propensity scores for the other subset, and vice versa. Combining these two sets of out-of-sample predictions yields one regenerated propensity score vector $\widehat{\mathcal P}^{(m)}$ for all $N$ units. Repeating this cross-fitting procedure over independently generated random splits gives $\widehat{\mathcal P}^{(1)},\ldots,\widehat{\mathcal P}^{(M)}$, which are then propagated through the oracle confidence set mapping $\Lambda$ and aggregated by union. Section~\ref{subsec: theory of nonparametric PSP} shows that, under appropriate regularity conditions, the regeneration procedure in Algorithm~\ref{alg: general} ensures that, among sufficiently many regeneration runs, at least one regenerated propensity score vector $\widehat{\mathcal{P}}^{(m)}$ is sufficiently close to the oracle propensity score vector $\mathcal{P}$ with high probability. 

\begin{algorithm}[h]  
\caption{Nonparametric propensity score propagation.} 
\label{alg: general}

\textbf{Input:} The observed data $(\mathbf{Z}, \mathbf{X}, \mathbf{Y})$, a chosen propensity score learner $\mathcal{L}$, a prespecified significance level $\alpha \in (0, 0.5)$, and a prespecified number of regeneration runs $M$. 

\smallskip

\textbf{Regeneration Step:} For each regeneration run $m \in [M]$:
\begin{enumerate}
\item Randomly partition the $N$ units into two disjoint subsets $\mathcal{I}_{A}^{(m)}\subset [N]$ and $\mathcal{I}_{B}^{(m)}=[N]\setminus \mathcal{I}_{A}^{(m)}$. Let $\mathcal{D}_{A}^{(m)}=\{(Z_{i},\mathbf{x}_{i}): i\in \mathcal{I}_{A}^{(m)}\}$ and $\mathcal{D}_{B}^{(m)}=\{(Z_{i},\mathbf{x}_{i}): i\in \mathcal{I}_{B}^{(m)}\}$.

\item Fit the learner $\mathcal{L}$ using $\mathcal{D}_{A}^{(m)}$ and use the fitted model to predict propensity scores for units in $\mathcal{I}_{B}^{(m)}$. Similarly, fit $\mathcal{L}$ using $\mathcal{D}_{B}^{(m)}$ and use the fitted model to predict propensity scores for units in $\mathcal{I}_{A}^{(m)}$. Combining these out-of-sample predictions yields the regenerated propensity score vector $\widehat{\mathcal{P}}^{(m)}$ for all $N$ units.

\item Construct the design-based confidence set $\mathcal C_{1-\alpha}^{(m)}$ by plugging $\widehat{\mathcal P}^{(m)}$ into the oracle confidence set mapping $\Lambda$, i.e., $\mathcal C_{1-\alpha}^{(m)}=\Lambda(\widehat{\mathcal P}^{(m)},\mathbf Z,\mathbf Y)$.

\end{enumerate}
\vspace{-0.3cm}
\textbf{Union Step:} Form the final confidence set $\mathcal C_{1-\alpha}=\bigcup_{m=1}^{M}\mathcal C_{1-\alpha}^{(m)}$.

\textbf{Output:} Design-based confidence set $\mathcal C_{1-\alpha}$ for the target estimand.

\end{algorithm}

In implementation, for both parametric and nonparametric propensity score propagation, we adopt a simple positivity-based safeguard to avoid near-degenerate regenerated propensity score vectors. Under the strict positivity condition in Condition~\ref{cond: positivity assumption} in Section~\ref{subsec: common regularity conditions}, there exists $\delta>0$ such that $p_i\in[\delta,1-\delta]$ for all $i$. It is therefore natural to discard regeneration runs $m$ for which $\min_{i\in[N]}\widehat p_i^{(m)}<\delta'$ or $\max_{i\in[N]}\widehat p_i^{(m)}>1-\delta'$, where $\delta'>0$ is a prespecified small constant satisfying $\delta'<\delta$; for example, one may take $\delta'=10^{-3}$. Under Condition~\ref{cond: positivity assumption} and the theoretical results in Section~\ref{sec: theory}, this safeguard does not alter the asymptotic coverage guarantee of the proposed framework, while helping to avoid non-informative confidence sets arising from regeneration runs that produce nearly degenerate propensity scores.

\subsection{Example Application 1: Design-Based Causal Inference for Observational Studies}\label{subsec: example application in observational studies}

As a concrete illustration of Algorithms~\ref{alg: parametric} and \ref{alg: general}, consider design-based causal inference for the SATE
$\theta_{\tau}=N^{-1}\sum_{i=1}^{N}\{Y_i(1)-Y_i(0)\}$ in an observational study. For each regeneration run $m\in[M]$, based on the formula of the oracle IPW estimator reviewed in Section~\ref{subsec: review known}, we construct a regenerated propensity score vector $\widehat{\mathcal{P}}^{(m)}=(\widehat{p}_{1}^{(m)},\ldots,\widehat{p}_{N}^{(m)})$ and define the resulting regenerated estimator
\[
\widehat{\theta}_{\tau}^{(m)}
=
\frac{1}{N}\sum_{i=1}^{N}\widehat{\theta}_{\tau, i}^{(m)}=\frac{1}{N}\sum_{i=1}^{N}
\left\{\frac{Z_iY_i}{\widehat{p}_{i}^{(m)}}
-
\frac{(1-Z_i)Y_i}{1-\widehat{p}_{i}^{(m)}}\right\},
\]
with the corresponding regenerated design-based variance estimator
\[
\widehat{\mathcal{V}}_{\tau}^{(m)}
=
\frac{1}{N(N-1)}
\sum_{i=1}^{N}
\left\{\widehat{\theta}_{\tau, i}^{(m)}-\widehat{\theta}_{\tau}^{(m)}\right\}^{2}.
\]
The regenerated confidence interval is
\[
\mathcal{C}_{1-\alpha}^{(m)}(\theta_{\tau})
=
\left[
\widehat{\theta}_{\tau}^{(m)}
-
z_{1-\alpha/2}\cdot \{\widehat{\mathcal{V}}_{\tau}^{(m)}\}^{1/2},
\;
\widehat{\theta}_{\tau}^{(m)}
+
z_{1-\alpha/2}\cdot \{\widehat{\mathcal{V}}_{\tau}^{(m)}\}^{1/2}
\right],
\]
and the final propagation-based confidence set for $\theta_{\tau}$ is
\[
\mathcal{C}_{1-\alpha}(\theta_{\tau})
=
\bigcup_{m=1}^{M}\mathcal{C}_{1-\alpha}^{(m)}(\theta_{\tau}),
\]
or its restricted-union analogue. This construction applies whether the regenerated propensity score vectors $\widehat{\mathcal{P}}^{(m)}$ are obtained parametrically through Algorithm~\ref{alg: parametric} or nonparametrically through Algorithm~\ref{alg: general}.

\subsection{Fundamental Distinction from Resampling- and Perturbation-Based Inference}\label{subsec: Fundamental Distinction from Resampling- and Perturbation-Based Inference}

The regeneration-and-union construction may, at a very high level, appear related to bootstrap, subsampling, and perturbation-based procedures, since these methods all introduce auxiliary randomness across repeated runs \citep[e.g.,][]{efron1994introduction, politis1994large, politis1999subsampling, chernozhukov2018double, zhao2019sens_boot, bodory2020finite, ImbensMenzel2021CausalBootstrap, guo2023causal, lin2024consistency, zheng2025perturbed}. The resemblance ends there. Classical bootstrap and subsampling methods typically resample units or observations, and hence the associated outcomes, to approximate a sampling distribution under repeated sampling, often from an implicit or explicit super-population. Perturbation-based methods similarly introduce auxiliary randomness by perturbing the outcome-level estimating equations within super-population inference frameworks, rather than by regenerating the unknown propensity scores for the same fixed finite population. These operations serve a different inferential purpose from the one considered here. Because the finite population is fixed, resampling or perturbing the observed outcomes would either alter the finite-population target or implicitly introduce a super-population sampling interpretation.

The subsampling step in nonparametric propensity score propagation has a distinct role. It is not used to resample the finite population or to approximate the sampling distribution of an outcome-based estimator. Instead, each regeneration run produces a data-compatible propensity score vector for the same fixed finite population. The subsequent inference neither resamples nor perturbs the realized design variables $\mathbf Z$ or the observed outcomes $\mathbf Y$. Therefore, the procedure preserves the finite-population target and avoids imposing a super-population outcome model for $\mathbf Y$. Further discussion of the distinction between propagation and resampling- or perturbation-based inference methods is given in Appendix~B.2.

\subsection{Some Guidance on Choosing the Number of Regeneration Runs $M$}\label{subsec: choose M}

How should one choose the number of propensity score regeneration runs $M$? The theory in Section~\ref{sec: theory} and Appendix~A gives sufficient scaling requirements on $M$ to ensure valid coverage. In the parametric propensity score case, $M \gg \log N$ suffices. In the nonparametric case, the corresponding requirement depends on stability and stochasticity measures of the propensity score learner; see Section~\ref{subsec: theory of nonparametric PSP} and Theorem~S.1 in Appendix~A.4. These theoretical requirements are sufficient, not necessary, and our numerical results suggest that much smaller values of $M$ can be adequate in practice. In particular, in the nonparametric propensity score settings considered in Section~\ref{sec: simulation}, feasible choices such as $M=100$ yield nominal coverage with reasonable confidence interval lengths, even in settings where plug-in and matching-based methods exhibit severe under-coverage. Figure~\ref{fig:simulation_m} further illustrates how the confidence interval length of nonparametric propensity score propagation changes with the number of regeneration runs $M$. As $M$ increases from $100$ to $10{,}000$, the mean confidence interval length remains practically reasonable: its ratio to the oracle confidence interval length increases from 1.30 to 1.69 and then stabilizes. Thus, researchers who prefer a more conservative implementation may choose a larger value of $M$, such as $M=1000$ or $10{,}000$, without substantial concern that the resulting confidence sets or intervals will become prohibitively wide.

\begin{figure}[h]
\centering
\caption{\small Mean confidence interval length across different numbers of regeneration runs $M$, illustrated using Propensity Score Setting 1 and Effect Setting 1 described in Section~\ref{sec: simulation}.}
\label{fig:simulation_m}
\includegraphics[width=0.5\linewidth]{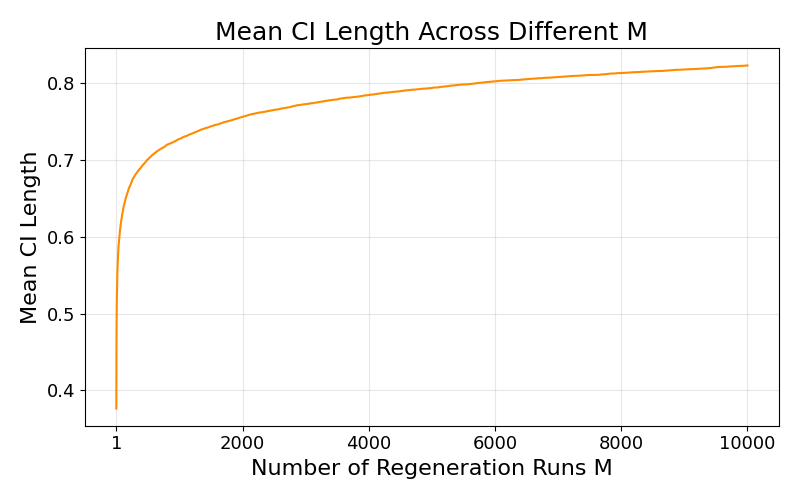}
\end{figure}

In practical applications, we recommend a design-based and outcome-blinded diagnostic for assessing whether the chosen $M$ is sufficiently large. Specifically, one can construct propagation-based confidence sets for covariate contrasts between comparison groups, such as treated versus control units, and check whether these confidence sets include zero. Systematic exclusion of zero suggests that the chosen $M$ may be too small and should be increased. An illustrative data example is provided in Section~\ref{sec: data}.

\section{Theoretical Guarantees of the Propensity Score Propagation Framework}\label{sec: theory}

\subsection{The Setup and Common Regularity Conditions}
\label{subsec: common regularity conditions}

We now study the theoretical properties of propensity score propagation in both parametric and nonparametric propensity score settings. Let $\theta$ denote a generic finite-population estimand for which an oracle design-based inference procedure is available when the true propensity score vector $\mathcal P$ is known. Let $\widehat{\theta}_{\text{oracle}}$ and $\widehat{\mathcal V}_{\text{oracle}}$ denote the corresponding oracle estimator and design-based variance estimator established under $\mathcal{P}$. In this section, we focus on the widely used Wald-type confidence interval
\begin{equation*}\label{eqn: oracle CI under CLT}
\mathcal{C}^{\text{oracle}}_{1-\alpha}(\theta)
=
\left[
\widehat{\theta}_{\text{oracle}}
-
z_{1-\alpha/2}\cdot
\widehat{\mathcal{V}}_{\text{oracle}}^{1/2},
\,
\widehat{\theta}_{\text{oracle}}
+
z_{1-\alpha/2}\cdot
\widehat{\mathcal{V}}_{\text{oracle}}^{1/2}
\right].
\end{equation*}
The interval above already covers a wide range of finite-population problems, including those discussed in Sections~\ref{subsec: review known} and \ref{sec: applications}. We focus on Wald-type intervals to state explicit approximation, coverage, and length guarantees; the general propensity score propagation algorithms themselves, such as Algorithms~\ref{alg: parametric} and \ref{alg: general}, continue to apply to any oracle mapping $\Lambda$ for which analogous validity and continuity conditions can be established.

Let $\widehat{\theta}^{(m)}$ and $\widehat{\mathcal{V}}^{(m)}$ denote the $m$th regenerated estimator and its corresponding design-based variance estimator, respectively, obtained by replacing the oracle propensity scores $\mathcal{P}$ in $\widehat{\theta}_{\text{oracle}}$ and $\widehat{\mathcal{V}}_{\text{oracle}}$ with the $m$th regenerated propensity score vector $\widehat{\mathcal{P}}^{(m)}$. Define the regenerated Wald-type interval
\begin{equation*}\label{eqn: regenerated CI}
\mathcal{C}^{(m)}_{1-\alpha}(\theta)=\left[\widehat{\theta}^{(m)}-z_{1-\alpha/2}\cdot \{\widehat{\mathcal{V}}^{(m)}\}^{1/2},\, \widehat{\theta}^{(m)}+z_{1-\alpha/2}\cdot \{\widehat{\mathcal{V}}^{(m)}\}^{1/2}\right].
\end{equation*}
The propagation-based confidence set is then
\begin{equation}\label{eqn: union CI for theta}
    \mathcal C_{1-\alpha}(\theta)=\bigcup_{m=1}^{M}\mathcal C^{(m)}_{1-\alpha}(\theta),
\end{equation}
with the restricted-union version defined analogously when applicable; for example, Section~\ref{subsec: Parametric Propensity Score Propagation} considers 
\begin{equation*}
    \widetilde{\mathcal{C}}_{1-\alpha}(\theta)=\bigcup_{m \in \mathcal{M}(\alpha^{\prime})}\mathcal{C}_{1-\alpha_0}^{(m)}(\theta),
\end{equation*}
where $\mathcal{M}(\alpha^{\prime})$ is defined in (\ref{eqn: definition of the restrictive set}) and $\alpha_0=\alpha-\alpha^{\prime}$ is the adjusted significance level. 

Before presenting the main results, we state several regularity conditions used throughout Section~\ref{sec: theory}. These conditions are standard in, or directly implied by, existing design-based inference theory with known propensity scores.

\begin{condition}[Finite-Population Regularity Conditions]
\label{condition: finite-population regularity}

\mbox{}\par
\begin{subcondition}
\item\label{cond: positivity assumption} (Strict Positivity of Propensity Scores). There exists a constant $\delta>0$ such that $p_i \in [\delta, 1-\delta]$ for all $i\in[N]$.

\item\label{cond: bounded outcome} (Bounded Outcomes). There exists a constant $C_y$ such that $|Y_i|\leq C_y$ for all $i$.

\item\label{cond: bounded covariates} (Bounded Covariates). There exists a constant $C_x$ such that $\|\mathbf{x}_i\|_{\infty}\leq C_x$ for all $i$.

\end{subcondition}
\end{condition}

\begin{condition}[Validity and Continuity Conditions for Oracle Estimators]
\label{cond: oracle validity conditions}
The oracle estimator $\widehat{\theta}_{\text{oracle}}$ and corresponding variance estimator $\widehat{\mathcal{V}}_{\text{oracle}}$, constructed using the true propensity scores $\mathcal{P}$, satisfy the following conditions:  
\begin{subcondition}
\item\label{cond: normality of oracle estimator} (Established Asymptotic Normality of $\widehat{\theta}_{\text{oracle}}$). 
$\frac{\widehat{\theta}_{\text{oracle}}-\theta}{\sqrt{\text{Var}(\widehat{\theta}_{\text{oracle}})}}\xrightarrow{d}N(0,1)$ as $N\rightarrow\infty$.
\item\label{cond: valid variance} (Established Validity of $\widehat{\mathcal{V}}_{\text{oracle}}$). 
As $N\rightarrow \infty$, $\frac{\widehat{\mathcal{V}}_{\text{oracle}}}{\text{Var}(\widehat{\theta}_{\text{oracle}})}\xrightarrow{p}c_v$ for some constant $c_v\geq 1$, and $\widehat{\mathcal{V}}_{\text{oracle}}^{1/2}\asymp_p N^{-1/2}$. 

\item\label{cond: Lipschitz} \textit{(Lipschitz Continuity Under Strict Positivity and Boundedness).} Let $\widehat{\theta}_{\text{oracle}}$ and $\widehat{\mathcal{V}}_{\text{oracle}}$ denote the oracle estimator and the corresponding design-based variance estimator when the oracle propensity score vector is set to $\mathcal{P}=(p_{1},\dots, p_{N})$, and let $\widehat{\theta}_{\text{oracle}}^{\prime}$ and $\widehat{\mathcal{V}}_{\text{oracle}}^{\prime}$ denote the corresponding quantities when the oracle propensity score vector is set to $\mathcal{P}^{\prime}=(p_{1}^{\prime}, \dots, p_{N}^{\prime})$. Recall that $N^{-1/2}\|\mathcal{P}-\mathcal{P}^{\prime}\|_{2}=\sqrt{N^{-1}\sum_{i=1}^{N}(p_{i}-p_{i}^{\prime})^{2}}$. Then, under bounded outcomes as in Condition~\ref{cond: bounded outcome}, for any fixed $\delta\in(0,0.5)$, there exist finite constants $L_{\theta}^{\prime}(\delta),L_{v}^{\prime}(\delta)>0$ such that $|\widehat{\theta}_{\text{oracle}}-\widehat{\theta}_{\text{oracle}}^{\prime}|\leq L_{\theta}^{\prime}(\delta)N^{-1/2}\|\mathcal{P}-\mathcal{P}^{\prime}\|_{2}$ and $N|\widehat{\mathcal{V}}_{\text{oracle}}-\widehat{\mathcal{V}}_{\text{oracle}}^{\prime}|\leq L_{v}^{\prime}(\delta)N^{-1/2}\|\mathcal{P}-\mathcal{P}^{\prime}\|_{2}$ for any $\mathcal{P},\mathcal{P}^{\prime}\in[\delta,1-\delta]^N$.

\end{subcondition}
\end{condition}
Conditions~\ref{cond: normality of oracle estimator} and~\ref{cond: valid variance} are standard oracle validity conditions in design-based inference and have been established for a broad class of problems (\citealp{imbens2015causal, li2017general, ding2024first}), including those reviewed in Sections~\ref{subsec: background}, \ref{subsec: review known}, and \ref{sec: applications}. Together, they imply that the oracle confidence set $\mathcal{C}^{\text{oracle}}_{1-\alpha}(\theta)$ has asymptotic coverage at least $100(1-\alpha)\%$ for $\theta$. In particular, by invoking finite-population central limit theorems such as \citet{li2017general}, Condition~\ref{cond: normality of oracle estimator} typically follows under mild regularity conditions, while explicit forms of valid design-based variance estimators required in Condition~\ref{cond: valid variance} have been constructed in many settings; see Sections~\ref{subsec: review known} and \ref{sec: applications} and Appendix~A.1 for examples. Condition~\ref{cond: Lipschitz} requires that, when propensity scores are bounded away from zero and one, and outcomes are bounded, the oracle estimator and its variance estimator vary smoothly with the propensity score vector under the normalized $L_2$ norm. We show in Appendix~A.2 that Condition~\ref{cond: Lipschitz} holds for a wide range of commonly used design-based estimators and variance estimators, including those reviewed in Sections~\ref{subsec: review known} and \ref{sec: applications}, via an argument based on the Cauchy-Schwarz inequality.

In Sections~\ref{subsec: theory of parametric PSP} and \ref{subsec: theory of nonparametric PSP}, we establish theoretical guarantees for parametric and nonparametric propensity score propagation, respectively. These results are derived under the maintained ignorability (i.e., no unobserved covariates) assumption \citep{rosenbaum1983central, rosenbaum2002observational}; see Remark~S.1 in Appendix~B.3. In Section~5.3, we discuss an extension that incorporates sensitivity analysis for hidden bias arising from unobserved covariates (i.e., unmeasured confounding).

\subsection{Theoretical Guarantees of Parametric Propensity Score Propagation}\label{subsec: theory of parametric PSP}

For parametric propensity score propagation, we first collect the regularity conditions used below for the GLM in \eqref{eqn: GLM}. These conditions mainly concern the smoothness of the GLM link function and the asymptotic normality of the GLM maximum likelihood estimator (MLE). The latter property holds under mild regularity conditions in classical fixed-design GLM theory \citep{fahrmeir1985consistency}; see Appendix A.3 for details.

\begin{condition}[Parametric Propensity Score Regularity Conditions]
\label{condition: Parametric Propensity Score Regularity Conditions}
We consider the following regularity conditions in the parametric propensity score setting:
\begin{subcondition}
\item\label{cond: validity of GLM} (Correct Specification and Smoothness of the GLM for Propensity Scores). We assume that $\{Z_{i}\}_{i=1}^{N}$ are independent and that, for each unit $i$, the propensity score satisfies $P(Z_{i}=1\mid \mathbf{x}_{i})=\Psi(\boldsymbol{\beta}^{T}\mathbf{x}_{i})$, where $\boldsymbol{\beta}\in\mathbbm{R}^{d}$ is the true parameter vector, and the inverse link function $\Psi: \mathbbm{R} \rightarrow (0,1)$ is one-to-one and twice continuously differentiable (i.e., $\Psi \in C^{2}$), with $|\Psi^{\prime}| \leq C_{\Psi^{\prime}}$ for some positive constant $C_{\Psi^{\prime}}>0$. 

\item\label{cond: MLE for GLM} (Asymptotic Normality of the MLE for the GLM). The MLE $\widehat{\boldsymbol{\beta}}$ exists with probability tending to one and satisfies $\sqrt{N}\cdot \widehat{\Omega}_{N}^{-\frac{1}{2}} \big(\widehat{\boldsymbol{\beta}}-\boldsymbol{\beta}\big)\xrightarrow{d} N\left(\mathbf{0}_{d\times 1}, \mathbf{I}_{d\times d}\right)$, where $\widehat{\Omega}_{N}$ is the covariance estimator used in the regeneration step and $\widehat\Omega_N\xrightarrow{p}\Omega$ for some fixed positive definite matrix $\Omega$. 
\end{subcondition}
\end{condition}

For parametric propensity score propagation, we first present Theorem~\ref{thm: uniform approximation}, which shows that, with sufficiently many regeneration runs, at least one regenerated propensity score vector is uniformly close to the oracle propensity score vector. Recall that $\widehat{\mathcal{P}}^{(m)}=(\widehat{p}_{1}^{(m)},\dots,\widehat{p}_{N}^{(m)})$ denotes the $m$th regenerated propensity score vector, $\mathcal{P}=(p_{1},\dots,p_{N})$ denotes the oracle propensity score vector, and $\|\widehat{\mathcal{P}}^{(m)}-\mathcal{P}\|_{\infty}=\max_{i\in[N]}|\widehat{p}_{i}^{(m)}-p_i|$.

\begin{theorem}\label{thm: uniform approximation}
Let $\alpha\in (0,0.5)$ be a prespecified significance level. Under Conditions~\ref{cond: bounded covariates} and \ref{condition: Parametric Propensity Score Regularity Conditions}, for any $\alpha^{\prime}\in (0, \alpha)$, there exists a constant $C_{\alpha^{\prime}}>0$ such that
\begin{equation*}
 \liminf_{N\rightarrow \infty }\lim_{M\rightarrow \infty} P\left(\min_{m\in \mathcal{M}(\alpha^{\prime})}\left\|  \widehat{\mathcal{P}}^{(m)}-\mathcal{P}\right\|_{\infty}  \leq  \frac{C_{\alpha^{\prime}}}{\sqrt{N}}\left(\frac{\log N}{M}\right)^{\frac{1}{d}}   \right)\geq 1-\alpha^{\prime},
\end{equation*} 
which also implies that 
\begin{equation*}
    \liminf_{N\rightarrow \infty }\lim_{M\rightarrow \infty} P\left(\min_{m\in [M]}\left\| \widehat{\mathcal{P}}^{(m)}-\mathcal{P}\right\|_{\infty}  \leq  \frac{C_{\alpha^{\prime}}}{\sqrt{N}} \left(\frac{\log N}{M}\right)^{\frac{1}{d}}  \right)\geq 1-\alpha^{\prime}.
\end{equation*}
\end{theorem}
Theorem~\ref{thm: uniform approximation} shows that, with sufficiently many regeneration runs, at least one regenerated propensity score vector is uniformly close to the oracle propensity score vector. Combining this result with Conditions~\ref{cond: positivity assumption}, \ref{cond: bounded outcome}, and \ref{cond: Lipschitz}, we obtain Theorem~\ref{thm: approximation theorem of parametric propensity score propagation estimator}.

\begin{theorem}\label{thm: approximation theorem of parametric propensity score propagation estimator}
Let $\alpha\in (0, 0.5)$ be a prespecified significance level. Under Conditions~\ref{condition: finite-population regularity}, \ref{cond: Lipschitz}, and \ref{condition: Parametric Propensity Score Regularity Conditions}, for any $\alpha^{\prime}\in (0,\alpha)$, there exists a constant $\widetilde{C}_{\alpha^{\prime}}>0$ such that
\begin{equation*}
   \liminf_{N\rightarrow \infty }\lim_{M\rightarrow \infty} P\left(\min_{m\in \mathcal{M}(\alpha^{\prime})}\left|\widehat{\theta}^{(m)}-\widehat{\theta}_{\text{oracle}}\right|  \leq  \frac{\widetilde{C}_{\alpha^{\prime}}}{\sqrt{N}} \left(\frac{\log N}{M}\right)^{\frac{1}{d}}  \right)\geq 1-\alpha^{\prime},
\end{equation*}
which also implies that 
\begin{equation*}
    \liminf_{N\rightarrow \infty }\lim_{M\rightarrow \infty} P\left(\min_{m\in [M]}\left|\widehat{\theta}^{(m)}-\widehat{\theta}_{\text{oracle}}\right|  \leq  \frac{\widetilde{C}_{\alpha^{\prime}}}{\sqrt{N}} \left(\frac{\log N}{M}\right)^{\frac{1}{d}}  \right)\geq 1-\alpha^{\prime}.
\end{equation*}
\end{theorem}

Theorem~\ref{thm: approximation theorem of parametric propensity score propagation estimator} shows that, when $M\gg \log N$, the oracle-closest regeneration run $m^{*}$ satisfies $|\widehat{\theta}^{(m^{*})}-\widehat{\theta}_{\text{oracle}}|=\min_{m\in[M]}|\widehat{\theta}^{(m)}-\widehat{\theta}_{\text{oracle}}|=o_p(N^{-1/2})$. Thus, at least one regenerated estimator $\widehat{\theta}^{(m^{*})}$ is asymptotically indistinguishable from the oracle estimator $\widehat{\theta}_{\text{oracle}}$ and the confidence interval $\mathcal{C}_{1-\alpha}^{(m^{*})}(\theta)$ constructed from this oracle-closest run has the same asymptotic coverage rate as $\mathcal{C}_{1-\alpha}^{\text{oracle}}(\theta)$. Since $\mathcal{C}_{1-\alpha}^{(m^{*})}(\theta)\subseteq \mathcal{C}_{1-\alpha}(\theta)=\bigcup_{m=1}^{M}\mathcal{C}_{1-\alpha}^{(m)}(\theta)$, the coverage of $\mathcal{C}_{1-\alpha}(\theta)$ follows from the validity of $\mathcal{C}_{1-\alpha}^{\text{oracle}}(\theta)$. The same reasoning applies to the restricted-union confidence set $\widetilde{\mathcal{C}}_{1-\alpha}(\theta)$ by replacing the minimum over $[M]$ with the minimum over $\mathcal{M}(\alpha^{\prime})$, since Theorem~\ref{thm: approximation theorem of parametric propensity score propagation estimator} establishes the corresponding approximation for both cases. Theorem~\ref{thm: coverage rate parametric} formalizes this intuition and provides rigorous coverage guarantees for both the unrestricted- and restricted-union confidence sets.

\begin{theorem}\label{thm: coverage rate parametric}
Let $\{M_N\}_{N\geq 1}$ be a fixed positive sequence satisfying $\lim_{N\to\infty}(\log N)/M_N=0$. Under Conditions~\ref{condition: finite-population regularity}--\ref{condition: Parametric Propensity Score Regularity Conditions}, if the number of regeneration runs satisfies $M\geq M_N$, then $\liminf_{N\to\infty} P\big(\theta \in \mathcal{C}_{1-\alpha}(\theta)\big)\ge 1-\alpha$ and $\liminf_{N\to\infty} P\big(\theta \in \widetilde{\mathcal{C}}_{1-\alpha}(\theta)\big)\ge 1-\alpha$.
\end{theorem}

Given the coverage guarantee, the next question is whether, and at what rate, the length of the confidence set $\mathcal{C}_{1-\alpha}(\theta)$, as well as the length of its restricted-union variant $\widetilde{\mathcal{C}}_{1-\alpha}(\theta)$, converges to zero as $N\rightarrow\infty$. The following Theorem~\ref{thm: CI length parametric} shows that the lengths of $\mathcal{C}_{1-\alpha}(\theta)$ and $\widetilde{\mathcal{C}}_{1-\alpha}(\theta)$ are of orders $\sqrt{\frac{\log M}{N}}$ and $\frac{1}{\sqrt{N}}$, respectively.

\begin{theorem}\label{thm: CI length parametric} 
Let $\mu_{L}$ denote Lebesgue measure on $\mathbbm{R}$. Under Conditions~\ref{condition: finite-population regularity}--\ref{condition: Parametric Propensity Score Regularity Conditions}, we have $\mu_{L}\{\mathcal{C}_{1-\alpha}(\theta)\}=O_{p}\Big(\sqrt{\frac{\log M}{N}} \Big)$ and $\mu_{L}\{\widetilde{\mathcal{C}}_{1-\alpha}(\theta)\}=O_{p}\Big(\frac{1}{\sqrt{N}} \Big)$. That is, for any $\epsilon>0$, there exists a constant $C_{\epsilon}>0$, such that
\begin{equation*}
  \liminf_{N\rightarrow \infty}\lim_{M\rightarrow \infty}  P\left(\mu_{L}\big\{\mathcal{C}_{1-\alpha}(\theta)\big\}\leq C_{\epsilon}\cdot \sqrt{\frac{\log M}{N}} \right)\geq 1-\epsilon 
\end{equation*}
and
\begin{equation*}
   \liminf_{N\rightarrow \infty}\lim_{M\rightarrow \infty}  P\left(\mu_{L}\big\{\widetilde{\mathcal{C}}_{1-\alpha}(\theta)\big\}\leq C_{\epsilon}\cdot \frac{1}{\sqrt{N}} \right)\geq 1-\epsilon.
\end{equation*}
\end{theorem}

In other words, Theorem~\ref{thm: CI length parametric} shows that the restricted-union confidence set $\widetilde{\mathcal{C}}_{1-\alpha}(\theta)$ has the same $N^{-1/2}$ length rate as the oracle confidence set $\mathcal{C}^{\text{oracle}}_{1-\alpha}(\theta)$, while the unrestricted-union confidence set $\mathcal{C}_{1-\alpha}(\theta)$ has the nearly oracle rate $\sqrt{(\log M)/N}$. Since Theorem~\ref{thm: coverage rate parametric} only requires $M$ to grow faster than $\log N$ for coverage, $M$ can grow slowly enough for the unrestricted-union confidence set $\mathcal{C}_{1-\alpha}(\theta)$ to remain close to the oracle $N^{-1/2}$ rate, while the restricted-union confidence set $\widetilde{\mathcal{C}}_{1-\alpha}(\theta)$ attains that rate directly. This theoretical message is consistent with the parametric simulation studies: $M=100$ already attains nominal coverage with moderate interval lengths, about 1.14--1.67 times the oracle length across the considered settings; see Appendix~D.1 for details.

\begin{remark}[Distinction Between Finite-Population and Finite-Sample Validity]
Throughout the paper, \emph{finite-population inference} refers to the inferential perspective under which the potential outcomes and covariates are treated as fixed, and randomness arises solely from the design variables \citep{li2017general, xu2021potential, cohen2022gaussian}, such as treatment assignment indicators, sampling inclusion indicators, and missingness indicators. This is distinct from \emph{finite-sample validity}, which typically refers to exact or non-asymptotic guarantees at a fixed sample size. Our methodology and theory are finite-population in the sense that they do not impose super-population or outcome-modeling assumptions, although their formal theoretical guarantees are asymptotic, in the sense that they apply to sequences of finite populations whose sizes grow. We use this finite-population asymptotic formulation primarily to simplify the theoretical statements, without carrying problem-specific finite-sample remainder terms.
\end{remark}

\subsection{Theoretical Guarantees of Nonparametric Propensity Score Propagation}\label{subsec: theory of nonparametric PSP}

We next provide a theoretical justification for nonparametric propensity score propagation in Algorithm~\ref{alg: general}. For each regeneration run $m\in[M]$, let $\mathbf{S}^{(m)}=(S^{(m)}_1,\ldots,S^{(m)}_N)\in \{0,1\}^{N}$ denote the subsample partition: $S_i^{(m)}=1$ indicates that unit $i$ is assigned to the first subsample $\mathcal{I}^{(m)}_A$, whereas $S_i^{(m)}=0$ indicates assignment to the second subsample $\mathcal{I}^{(m)}_B$, with $P(S_i^{(m)}=1)=1/2$. Equivalently, $\mathbf{S}^{(m)}$ is uniformly distributed over the partition hypercube $\{0,1\}^{N}$. Given $\mathbf{S}^{(m)}$, let $\widehat{\mathcal{P}}^{(m)}=(\widehat p^{(m)}_1,\ldots,\widehat p^{(m)}_N)$ and $\widehat{\theta}^{(m)}$ denote the regenerated propensity score vector and the resulting regenerated estimator of the target estimand $\theta$, respectively.
\begin{condition}[Stability with Respect to Random Partitions]
\label{cond: Lipschitz condition of propensity score learner}
Let $d_H(\mathbf{S}_1,\mathbf{S}_2)=\sum_{i=1}^N\mathbbm{1}\{S_{1i}\neq S_{2i}\}$ be the Hamming distance between two partition indicators. Let $\widehat{\mathcal{P}}_1$ and $\widehat{\mathcal{P}}_2$ be the estimated propensity score vectors produced by Algorithm~\ref{alg: general} under $\mathbf{S}_1$ and $\mathbf{S}_2$, respectively. Suppose there exist sequences $L_{s,N}>0$ and $c_N\to0$ such that $\sup_N L_{s,N}<\infty$ and $P\left(\max_{\mathbf{S}_1\neq\mathbf{S}_2}\frac{N^{-1/2}\|\widehat{\mathcal{P}}_1-\widehat{\mathcal{P}}_2\|_2}{N^{-1}d_H(\mathbf{S}_1,\mathbf{S}_2)}\leq L_{s,N}\right)\geq1-c_N$.
\end{condition}

Condition~\ref{cond: Lipschitz condition of propensity score learner} requires the propensity score learner to be stable with respect to the random partition \citep{bousquet2002stability}: the root-mean-square difference between two estimated propensity score vectors (i.e., $N^{-1/2}\|\widehat{\mathcal{P}}_1-\widehat{\mathcal{P}}_2\|_2$) is controlled by the fraction of partition indicators on which the corresponding partitions differ (i.e., $N^{-1}d_H(\mathbf{S}_1,\mathbf{S}_2)$). For the coverage guarantee, it suffices that \(L_{s,N}\) remain bounded, which ensures that small perturbations of the partition have only proportionally small effects on the estimates. For the confidence set length convergence in Theorem~\ref{thm: CI length nonparametric}, we impose the stronger requirement \(L_{s,N}\to0\), so that the learner's sensitivity to the partition vanishes asymptotically. This strengthening remains reasonable since estimates obtained under different partitions should converge toward the same true propensity score vector, although consistency alone does not imply the stated uniform Lipschitz bound.

To state the next condition, we introduce some relevant definitions. Let $\alpha\in (0,0.5)$ be a prespecified significance level. Given any $\alpha_{*}\in (0, \alpha)$, let $\mathcal{I}_{1-\alpha_{*}}(\widehat{\theta}_{\text{oracle}})=\big[L_{1-\alpha_{*}}, \, U_{1-\alpha_{*}}\big]=\big[\theta-z_{1-\alpha_{*}/2}\cdot \widehat{\mathcal{V}}_{\text{oracle}}^{1/2},\, \theta+z_{1-\alpha_{*}/2}\cdot \widehat{\mathcal{V}}_{\text{oracle}}^{1/2}\big]$ represent the interval for capturing the variation of the oracle estimator $\widehat{\theta}_{\text{oracle}}$ with confidence $100(1-\alpha_{*})\%$.

\begin{condition}[Nondegenerate Support of Regenerated Estimator]
\label{cond: sufficient support of regenerated estimators} For any fixed $\alpha_{*}\in (0,\alpha)$, there exists some positive sequence $\gamma_{\alpha_{*}, N}>0$ such that, conditional on the realized $\mathbf{Z}$, with probability tending to one, we have $P(\widehat{\theta}^{(m)}< L_{1-\alpha_{*}} \mid \mathbf{Z})\geq \gamma_{\alpha_{*}, N}$ and $P(\widehat{\theta}^{(m)}> U_{1-\alpha_{*}} \mid \mathbf{Z})\geq \gamma_{\alpha_{*}, N}$. In addition, assume that the tail-mass lower bound $\gamma_{\alpha_{*},N}$ and the middle mass $\rho_{\alpha_{*},N}=1-2\gamma_{\alpha_{*},N}$ are non-negligible on the partition hypercube $\{0,1\}^{N}$, in the sense that $\gamma_{\alpha_{*},N}\gg N^{1/4}2^{-N/2}$ and $\rho_{\alpha_{*},N}\gg 2^{-N}$, where $a_{N}\gg b_{N}$ means $a_{N}/b_{N}\rightarrow\infty$ as $N\rightarrow \infty$.
\end{condition}

Because $\mathbf S^{(m)}$ is uniformly distributed over the $2^N$ partition indicators in $\{0,1\}^N$, conditional on the realized $\mathbf Z$, the probabilities $P(\widehat{\theta}^{(m)}< L_{1-\alpha_*}\mid\mathbf Z)$ and $P(\widehat{\theta}^{(m)}> U_{1-\alpha_*}\mid\mathbf Z)$ equal the respective fractions of partition indicators producing estimates below $L_{1-\alpha_*}$ and above $U_{1-\alpha_*}$. Therefore, Condition~\ref{cond: sufficient support of regenerated estimators} imposes only a mild non-degeneracy requirement on the regeneration distribution of $\widehat{\theta}^{(m)}$. In particular, it does not require either tail to contain a non-vanishing fraction of all $2^{N}$ partitions, since $\gamma_{\alpha_{*},N}$ is allowed to converge to zero exponentially fast. For example, $\gamma_{\alpha_{*}, N}=a^{-N}$ for some $1<a<\sqrt{2}$ satisfies this condition. Similarly, $\rho_{\alpha_{*},N}\gg2^{-N}$ only requires that, among the $2^{N}$ partition indicators, the number of partition indicators for which the resulting estimate lies within $\mathcal{I}_{1-\alpha_{*}}(\widehat{\theta}_{\text{oracle}})=[L_{1-\alpha_{*}},\, U_{1-\alpha_{*}}]$, namely $\rho_{\alpha_*,N}\cdot 2^N$, diverges as $N\to\infty$.

Recall that $m^{*}\in [M]$ denotes the oracle-closest regeneration run satisfying $|\widehat{\theta}^{(m^*)}-\widehat{\theta}_{\text{oracle}}|=\min_{m\in[M]}|\widehat{\theta}^{(m)}-\widehat{\theta}_{\text{oracle}}|$. Let $\widehat{\mathcal{P}}^{(m^{*})}=(\widehat{p}_{1}^{(m^{*})}, \dots, \widehat{p}_{N}^{(m^{*})})$ denote the regenerated propensity score vector from this oracle-closest run $m^{*}$. 

\begin{condition}[Consistency of the Propensity Score Learner]
\label{cond: consistency of propensity score learner}
The regenerated propensity score vector $\widehat{\mathcal{P}}^{(m^{*})}$ is mean-square consistent for the true propensity score vector $\mathcal{P}$. That is, $N^{-1}\sum_{i=1}^{N}\big(\widehat{p}_{i}^{(m^{*})}-p_{i}\big)^{2}=o_{p}(1)$, or equivalently, $N^{-1/2}\|\widehat{\mathcal{P}}^{(m^{*})}-\mathcal{P}\|_{2}=o_{p}(1)$.
\end{condition}

Condition~\ref{cond: consistency of propensity score learner} is a standard consistency condition for the propensity score learner, such as a nonparametric machine learning method, imposed only on the oracle-closest regeneration run. It is therefore weaker than requiring uniform consistency of the propensity score learner over all random partitions or over all regeneration runs. Overall, Conditions~\ref{cond: Lipschitz condition of propensity score learner}--\ref{cond: consistency of propensity score learner} serve as convenient sufficient conditions for the generic theoretical results and may be replaced by learner-specific regularity conditions in specific settings.

\begin{theorem}
\label{thm: approximation nonparametric}
Let $\alpha\in(0,0.5)$ and suppose that Conditions~\ref{condition: finite-population regularity}, \ref{cond: oracle validity conditions}, \ref{cond: Lipschitz condition of propensity score learner}, and \ref{cond: sufficient support of regenerated estimators} hold. Then, for any $\alpha_*\in(0,\alpha)$, there exists a deterministic sequence $\varepsilon_{N,M}$ satisfying $\lim_{N\rightarrow \infty} \lim_{M\to\infty}
\frac{\varepsilon_{N,M}}{N^{-1/2}}=0$ and
\begin{equation*}
    \liminf_{N\to\infty}\lim_{M\to\infty}P\left(\min_{m\in[M]}\left|\widehat{\theta}^{(m)}-\widehat{\theta}_{\text{oracle}}\right|\leq \varepsilon_{N,M}\right)\geq1-\alpha_{*}.
\end{equation*}
\end{theorem}

Theorem~\ref{thm: approximation nonparametric} plays the role of a nonparametric counterpart to Theorem~\ref{thm: approximation theorem of parametric propensity score propagation estimator}. It shows that, among sufficiently many regeneration runs $M$, at least one regenerated estimator is close to the oracle estimator with high probability. The proof in Appendix~A.4 gives an example choice of $\varepsilon_{N,M}$ in terms of $L_{s,N}$, $\gamma_{\alpha_*,N}$, $\rho_{\alpha_*,N}$, $N$, and $M$, using an isoperimetric argument on the partition hypercube $\{0,1\}^N$. Building on Theorem~\ref{thm: approximation nonparametric}, Theorem~\ref{thm: coverage rate nonparametric} provides a coverage guarantee for design-based confidence sets. 

\begin{theorem}\label{thm: coverage rate nonparametric}
Under Conditions 1, 2, and 4--6, we have $\liminf_{N\to\infty}\lim_{M\to\infty}P\bigl(\theta \in \mathcal{C}_{1-\alpha}(\theta)\bigr)\ge 1-\alpha$, where $\mathcal{C}_{1-\alpha}(\theta)$ is the propagation-based confidence set defined in (\ref{eqn: union CI for theta}).
\end{theorem}

The above theorem establishes coverage as $M\to\infty$, while the numerical evaluation of finite-$M$ behavior is examined in Section~\ref{sec: simulation}; there, moderate values such as $M=100$ already achieve at least nominal coverage, even when plug-in and matching-based approaches exhibit substantial under-coverage. We next provide a sufficient condition under which the length of the resulting confidence set converges to zero.

\begin{theorem}\label{thm: CI length nonparametric}
Consider the setup of Theorem~\ref{thm: approximation nonparametric}. Suppose further that: (i) the stability constants in Condition~\ref{cond: Lipschitz condition of propensity score learner} satisfy $L_{s,N}\to0$; and (ii) for the sequence of regeneration numbers $M=M_N$, the regenerated variance estimators satisfy $\max_{m\in[M_N]}\widehat{\mathcal{V}}^{(m)}=o_p(1)$. Then $\mu_L\{\mathcal{C}_{1-\alpha}(\theta)\}\xrightarrow{p}0$ as $N\to\infty$.
\end{theorem}

As will be illustrated in the simulation studies in Section~\ref{sec: simulation}, the confidence set $\mathcal{C}_{1-\alpha}(\theta)$ produced by nonparametric propensity score propagation (Algorithm~\ref{alg: general}) has a practically reasonable length. With a moderate number of regeneration runs (e.g., $M=100$), its length is typically only 1.3 to 1.6 times that of the oracle confidence set $\mathcal{C}_{1-\alpha}^{\text{oracle}}(\theta)$ across the simulation settings in Section~\ref{sec: simulation}. The simulations also show that, even when $M=10{,}000$, the length stabilizes at a practically reasonable level, around 1.69 times the oracle confidence interval length; see Figure~\ref{fig:simulation_m}. In practice, the length of $\mathcal{C}_{1-\alpha}(\theta)$ can be further reduced using a restricted-union strategy or covariate adjustment; see Remark~S.4 in Appendix~B.3 and Section~\ref{subsec: covariate adjustment} for details.

\section{Some Additional Examples and Extensions of Propensity Score Propagation}\label{sec: applications}

\subsection{Additional Example Applications of Propensity Score Propagation}\label{subsec: additional applications}

In the previous sections, we used design-based causal inference for observational studies to illustrate the main idea of propensity score propagation. We now demonstrate that our framework applies broadly to a wide range of problems. 

\textit{Example Application 2 (Design-Based Inference for Missing Data):} Design-based inference for missing data is closely connected to classical design-based inference for probability survey sampling, but the key difference is that propensity scores are known in the latter and typically unknown in the former. Consider $N$ fixed units. Let $y_i$ denote the outcome of unit $i$, let $Z_i\in\{0,1\}$ indicate whether $y_i$ is observed, and let $Y_i$ be the observed outcome, where $Y_i=y_i$ if $Z_i=1$ and $Y_i=\text{``NA''}$ otherwise; write $\mathbf Z=(Z_1,\ldots,Z_N)$, $\mathbf y=(y_1,\ldots,y_N)$, and $\mathbf Y=(Y_1,\ldots,Y_N)$. The target estimand is the finite-population mean $\theta_{\bar y}=N^{-1}\sum_{i=1}^{N}y_i$. In classical probability survey sampling, $Z_i$ represents the sampling inclusion indicator, and the inclusion probabilities $p_{i}=P(Z_i=1)$ are typically known by design. Under independent sampling, such as Poisson sampling, the Horvitz-Thompson estimator $\widehat{\theta}_{\bar{y}, \text{oracle}}=N^{-1}\sum_{i=1}^{N}\frac{Z_iY_i}{p_i}$ is unbiased for $\theta_{\bar y}$ (where we define $0\times\text{``NA''}=0$), and a corresponding design-based variance estimator is $\widehat{\mathcal V}_{\bar{y},\text{oracle}}=N^{-2}\sum_{i=1}^{N}\frac{1-p_i}{p_i^{2}}Z_iY_i^{2}$. Standard finite-population central limit theorems then yield the oracle design-based confidence interval $\mathcal{C}_{1-\alpha}^{\text{oracle}}=\Lambda_{\bar{y}}(\mathcal{P}, \mathbf{Z}, \mathbf{Y})=\big[\widehat{\theta}_{\bar{y}, \text{oracle}}-z_{1-\alpha/2}\cdot\widehat{\mathcal V}_{\bar{y},\text{oracle}}^{1/2},\, \widehat{\theta}_{\bar{y}, \text{oracle}}+z_{1-\alpha/2}\cdot\widehat{\mathcal V}_{\bar{y},\text{oracle}}^{1/2}\big]$ when the probabilities $\mathcal{P}=(p_1,\dots, p_{N})$ are known \citep{horvitz1952generalization, hajek1960limiting, hajek1964asymptotic, li2017general}. In missing data problems, the same oracle construction would apply if the non-missingness probabilities were known; however, these probabilities are usually unknown in practice. Under the ignorability assumption and independence of missingness across units, let $p_i=P(Z_i=1\mid \mathbf{x}_{i})$ denote the non-missingness probability (propensity score) of unit $i$, which is an unknown function of observed covariates $\mathbf{x}_i$. Propensity score propagation can then be applied by treating $\mathbf Z$ as the design variables, regenerating plausible non-missingness probability vectors $\widehat{\mathcal P}^{(1)},\ldots,\widehat{\mathcal P}^{(M)}$, constructing the Horvitz-Thompson-based design-based confidence interval under each regenerated vector, and reporting the final confidence set as the union of these regenerated intervals.

\textit{Example Application 3 (Design-Based Difference-in-Differences Analysis):} As mentioned in Section~\ref{subsec: review known}, existing design-based inference frameworks for DID studies, such as \citet{athey2022design}, rely on known propensity scores, for example, those arising from complete randomization of treatment assignments or from conditioning on an exactly matched design. This known-propensity-score assumption is conceptually clean but often strong in practice, and the proposed propensity score propagation framework allows us to relax this assumption by allowing the propensity scores to be unknown. Consider the canonical DID setup in which all units are untreated at baseline $(t=0)$, and at $t=1$, some units receive treatment while others remain in control. Let $Z_i\in\{0,1\}$ denote the treatment indicator at $t=1$, let $Y_{it}$ be the observed outcome for unit $i\in[N]$ at time $t\in \{0,1\}$, and let $Y_{it}(1)$ and $Y_{it}(0)$ be the corresponding potential outcomes under treatment and control. Under the assumption of no time-varying unmeasured confounding, let $p_{i}=P(Z_{i}=1\mid \mathbf{x}_{i})$ denote the propensity score for each unit $i$ (assuming independence of treatment uptake across units). If $p_{i}$ were known, it is easy to show that the weighted DID estimator $\widehat{\theta}_{\text{DID}, \text{oracle}}=N^{-1}\sum_{i=1}^{N}\widehat{\theta}_{\text{DID}, \text{oracle}, i}=N^{-1}\sum_{i=1}^{N}\big\{\frac{Z_{i}}{p_{i}}(Y_{i1}-Y_{i0})-\frac{1-Z_{i}}{1-p_{i}}(Y_{i1}-Y_{i0})\big\}$ is an unbiased estimator for $\theta_{t=1}=N^{-1}\sum_{i=1}^{N}\{Y_{i1}(1)-Y_{i1}(0)\}$ (see Remark S.7 in Appendix B.3), and $\widehat{\mathcal{V}}_{\text{DID}, \text{oracle}}=\frac{1}{N(N-1)}\sum_{i=1}^{N}(\widehat{\theta}_{\text{DID}, \text{oracle}, i}-\widehat{\theta}_{\text{DID}, \text{oracle}})^{2}$ is a valid, design-based variance estimator for $\widehat{\theta}_{\text{DID}, \text{oracle}}$. Invoking the finite-population central limit theorem (\citealp{li2017general}), we can obtain an oracle design-based confidence interval $\mathcal{C}^{\text{oracle}}_{1-\alpha}=\Lambda_{\text{DID}}(\mathcal{P}, \mathbf{Z}, \mathbf{Y})=\big[\widehat{\theta}_{\text{DID}, \text{oracle}}-z_{1-\alpha/2}\cdot \widehat{\mathcal{V}}_{\text{DID}, \text{oracle}}^{1/2}, \, \widehat{\theta}_{\text{DID}, \text{oracle}}+z_{1-\alpha/2}\cdot \widehat{\mathcal{V}}^{1/2}_{\text{DID}, \text{oracle}}\big]$. We then apply propensity score propagation to extend the oracle DID procedure to unknown propensity scores. For each regeneration run $m\in[M]$, we obtain $\widehat{\mathcal P}^{(m)}=(\widehat p_1^{(m)},\dots,\widehat p_N^{(m)})$, construct $\mathcal C^{(m)}_{1-\alpha}$ by substituting $\widehat p_i^{(m)}$ for $p_i$ in $\mathcal C^{\text{oracle}}_{1-\alpha}$, and form the final confidence set $\mathcal C_{1-\alpha}=\bigcup_{m=1}^{M}\mathcal C^{(m)}_{1-\alpha}$.

\textit{Example Application 4 (Design-Based Instrumental Variable Analysis in Observational Studies):} The main arguments in this paper are developed under the ignorability assumption, whereas instrumental variable (IV) methods provide an important alternative approach for addressing unmeasured confounding in observational studies \citep{angrist1996identification, rosenbaum2002observational, kang2016full, rosenbaum2020design}. Consider a finite population of $N$ units. For each unit $i$, let $\mathbf{x}_{i}$ denote observed covariates, $Z_i\in\{0,1\}$ a binary IV, $D_i\in\{0,1\}$ the observed treatment, and $Y_i$ the observed outcome. Under the potential outcomes framework and the exclusion restriction, let $D_i(1)$ and $D_i(0)$ denote the potential treatment variables under $Z_i=1$ and $Z_i=0$, and let $Y_i(1)$ and $Y_i(0)$ denote the potential outcomes under $Z_i=1$ and $Z_i=0$. Under IV relevance and monotonicity, i.e., $N^{-1}\sum_{i=1}^{N}\{D_i(1)-D_i(0)\}\neq 0$ and $D_i(1)\geq D_i(0)$ for all $i$, the sample complier average treatment effect is $\theta_{\text{IV}}=\sum_{i=1}^{N}\{Y_i(1)-Y_i(0)\}/\sum_{i=1}^{N}\{D_i(1)-D_i(0)\}$ \citep{angrist1996identification}. Let $\mathcal{F}_{N}=\big\{(D_i(1),D_i(0),Y_i(1),Y_i(0)):i=1,\dots,N\big\}$. Under IV validity and independent IV assignment across units, $P(Z_i=1\mid \mathbf{x}_i,\mathcal{F}_{N})=P(Z_i=1\mid \mathbf{x}_i)=p_i$. For any candidate value $\theta_0\in\mathbbm{R}$, define $Y_i^{\theta_0}=Y_i-\theta_0D_i$ and $Y_i^{\theta_0}(z)=Y_i(z)-\theta_0D_i(z)$ for $z=0,1$. Then testing $H_{\theta_0}:\theta_{\text{IV}}=\theta_0$ is equivalent to testing $N^{-1}\sum_{i=1}^{N}\{Y_i^{\theta_0}(1)-Y_i^{\theta_0}(0)\}=0$. This motivates the oracle IV statistic $T_{\text{IV}}^{\theta_0}=N^{-1}\sum_{i=1}^{N}T_{\text{IV},i}^{\theta_0}=N^{-1}\sum_{i=1}^{N}\big\{\frac{Z_i}{p_{i}}Y_i^{\theta_0}-\frac{1-Z_i}{1-p_{i}}Y_i^{\theta_0}\big\}$, with variance estimator $\widehat{\mathcal{V}}_{\text{IV}}^{\theta_0}=\frac{1}{N(N-1)}\sum_{i=1}^{N}\big(T_{\text{IV},i}^{\theta_0}-T_{\text{IV}}^{\theta_0}\big)^2$. Following arguments similar to those in \citet{kang2016full}, under $H_{\theta_0}$, $T_{\text{IV}}^{\theta_0}$ is unbiased for zero and $\widehat{\mathcal{V}}_{\text{IV}}^{\theta_0}$ is conservative for its design-based variance. Therefore, an oracle two-sided $p$-value is $p_{\text{IV}}^{\text{oracle}}(\theta_0)=2\big\{1-\Phi\big(|T_{\text{IV}}^{\theta_0}|/\sqrt{\widehat{\mathcal{V}}_{\text{IV}}^{\theta_0}}\big)\big\}$, and the oracle confidence set is $\mathcal{C}_{1-\alpha}^{\text{oracle}}=\Lambda_{\text{IV}}(\mathcal{P},\mathbf{Z},\mathbf{D},\mathbf{Y})=\big\{\theta_0\in\mathbbm{R}:p_{\text{IV}}^{\text{oracle}}(\theta_0)\geq\alpha\big\}$, obtained by inverting the oracle tests over candidate values of $\theta_0$. When the IV assignment probabilities $\mathcal{P}=(p_1,\dots,p_N)$ are unknown, as in observational IV studies, propensity score propagation can be applied by replacing $\mathcal{P}$ with each regenerated probability vector $\widehat{\mathcal{P}}^{(m)}$ and computing the corresponding regenerated $p$-value $p_{\text{IV}}^{(m)}(\theta_0)$. Following the same principle as the union step in propensity score propagation, we define the propagation-based $p$-value as $p_{\text{IV}}(\theta_0)=\max_{m\in[M]}p_{\text{IV}}^{(m)}(\theta_0)$. The corresponding propagation-based confidence set is obtained by test inversion:
$\mathcal{C}_{1-\alpha}=\{\theta_0\in\mathbbm{R}:p_{\text{IV}}(\theta_0)\geq\alpha\}$.

\subsection{Extension to Design-Based Inference with Covariate Adjustment}\label{subsec: covariate adjustment}

Covariates can often be used to improve the efficiency of design-based inference. In randomized experiments, covariate adjustment uses pre-treatment covariates to reduce residual variation while preserving design-based validity; see, for example, \citet{rosenbaum2002covariance}, \citet{lin2013agnostic}, \citet{fogarty2018mitigating}, \citet{abadie2020sampling}, \citet{li2020rerandomization}, \citet{liu2020regression}, \citet{zhao2021covariate}, \citet{ding2024first}, and \citet{heng2025design}. Such procedures can be accommodated in our framework by allowing the oracle confidence-set mapping $\Lambda$ to depend on the observed covariates $\mathbf X$. Since propensity score propagation is not tied to any particular form of $\Lambda$, any covariate-adjusted design-based confidence set satisfying the relevant validity and regularity conditions, such as those discussed in Section~\ref{subsec: common regularity conditions}, can be used within the propensity score propagation procedure.

For example, consider inference for the SATE $\theta_{\tau}$ in the Bernoulli randomized experiments reviewed in Section~\ref{subsec: review known}. The usual design-based variance estimator $\widehat{\mathcal V}_{\tau,\text{oracle}}$ can be conservative when the individual treatment effects $Y_i(1)-Y_i(0)$ are heterogeneous \citep{fogarty2018mitigating, ding2024first}. Covariate adjustment seeks to reduce this conservativeness by using observed covariates $\mathbf{X}$ to explain part of the variation in the unit-level treatment effect contributions, so that the variance estimate is constructed from the remaining unexplained variation. Specifically, following \citet{fogarty2018mitigating}, let $Q$ be a fixed, full-column-rank $N\times L$ matrix, with $N>L$, whose columns contain covariate information, including an intercept term. For example, if the covariate matrix $\mathbf X$ is taken to include an intercept column, one may set $Q=\mathbf{X}$. More generally, $Q$ may also include nonlinear functions of the covariates, such as interaction or quadratic terms. Let $H_Q=Q(Q^TQ)^{-1}Q^T$ be the corresponding projection matrix, with $h_{Qii}$ denoting its $i$th diagonal element. Given the oracle unit-level treatment effect contribution $\widehat{\theta}_{\tau,\text{oracle},i}$, define $\widehat{\theta}_{\tau, *,i}=\widehat{\theta}_{\tau,\text{oracle},i}/\sqrt{1-h_{Qii}}$ and $\widehat{\boldsymbol{\theta}}_{\tau, *}=(\widehat{\theta}_{\tau, *,1},\ldots,\widehat{\theta}_{\tau, *, N})$. Under the known or oracle propensity scores $\mathcal P$, the covariate-adjusted design-based variance estimator is 
\begin{equation*}
    \widehat{\mathcal V}_{\tau,\text{oracle}}(Q)=\frac{1}{N^{2}}\widehat{\boldsymbol{\theta}}_{\tau, *}(\mathbf I_{N\times N}-H_Q)\widehat{\boldsymbol{\theta}}_{\tau, *}^T.
\end{equation*}
Intuitively, $H_Q$ captures the part of the adjusted unit-level contributions that is linearly predictable from the covariates in $Q$, while $\mathbf I_{N\times N}-H_Q$ retains the residual variation. The estimator therefore removes covariate-explained treatment effect heterogeneity from the conservative variance bound. Importantly, this projection is only a working adjustment device rather than a modeling assumption: it does not require the true relationship between the covariates and treatment effects to be linear. Consequently, when $Q$ captures meaningful treatment effect heterogeneity, the adjustment can reduce excess conservativeness and yield shorter confidence intervals without compromising design-based validity, even if the true heterogeneity is not linear in the covariates \citep{fogarty2018mitigating,ding2024first}. The resulting oracle covariate-adjusted confidence interval is
\begin{align*}
    \mathcal C_{1-\alpha}^{\text{oracle}}(Q)&=\Lambda_{\tau}(\mathcal P,\mathbf Z,\mathbf Y,\mathbf X)\\
    &=\left[\widehat{\theta}_{\tau,\text{oracle}}-z_{1-\alpha/2}
    \cdot \{\widehat{\mathcal V}_{\tau,\text{oracle}}(Q)\}^{\frac{1}{2}},\,\widehat{\theta}_{\tau,\text{oracle}}+z_{1-\alpha/2}\cdot \{\widehat{\mathcal V}_{\tau,\text{oracle}}(Q)\}^{\frac{1}{2}}\right].
\end{align*}

Propensity score propagation can directly use this oracle mapping $\Lambda_{\tau}(\mathcal P,\mathbf Z,\mathbf Y,\mathbf X)$ to improve the efficiency of design-based causal inference in observational studies: for each regenerated propensity score vector $\widehat{\mathcal{P}}^{(m)}$, construct the corresponding regenerated covariate-adjusted confidence interval $\mathcal{C}_{1-\alpha}^{(m)}(Q)=\Lambda_{\tau}(\widehat{\mathcal{P}}^{(m)}, \mathbf{Z}, \mathbf{Y}, \mathbf{X})$, and then the propagation-based covariate-adjusted confidence set is $\mathcal{C}_{1-\alpha}(Q)=\bigcup_{m=1}^{M}\mathcal{C}_{1-\alpha}^{(m)}(Q)$. Therefore, existing oracle validity results for covariate-adjusted design-based inference, such as those in \citet{fogarty2018mitigating} and \citet{zhu2025randomization}, can be combined with the propagation-based arguments in Section~\ref{sec: theory} to account for propensity score estimation uncertainty while retaining the efficiency gains from covariate adjustment.

\subsection{Extension to Sensitivity Analysis for Hidden Bias} \label{subsec: extensions to sensitivity analysis}

In many settings, hidden bias due to unobserved covariates (unmeasured confounders) may distort statistical inference. Sensitivity analysis is therefore often used to assess the robustness of inference results to possible unmeasured confounding \citep{rosenbaum2002observational, imbens2015causal, zhao2019sens_boot}. Researchers typically report a \textit{sensitivity set}, defined as a confidence set that remains valid uniformly over all confounding mechanisms allowed by a prespecified sensitivity model and a prespecified magnitude of hidden bias. Our framework naturally accommodates such sensitivity analyses. For example, consider the widely used marginal sensitivity model for causal inference and missing data problems \citep{tan2006distributional, zhao2019sens_boot}, which assumes that, for some $\Gamma \geq 1$,
\begin{equation}\label{eqn: sensitivity model}
\Gamma^{-1}\leq \frac{e^{*}(\mathbf{x},\mathbf{u})/\{1-e^{*}(\mathbf{x},\mathbf{u})\}}{e(\mathbf{x})/\{1-e(\mathbf{x})\}}\leq\Gamma
\quad \text{for all }(\mathbf{x},\mathbf{u})
\end{equation}
where $\mathbf{x}$ and $\mathbf{u}$ denote the observed and unobserved covariates, respectively, $e(\mathbf{x})=P(Z=1\mid\mathbf{x})$ is the marginal probability of treatment assignment (or non-missingness) given the observed covariates, henceforth referred to as the \textit{marginal propensity score}, and $e^{*}(\mathbf{x},\mathbf{u})=P(Z=1\mid\mathbf{x},\mathbf{u})$ is the corresponding probability given both observed and unobserved covariates, henceforth referred to as the \textit{true propensity score}. The prespecified sensitivity parameter $\Gamma$ controls the magnitude of hidden bias: $\Gamma=1$ corresponds to no unmeasured confounding, whereas larger $\Gamma$ values allow greater departures from this assumption.

The key to constructing a sensitivity set is that, when the marginal propensity scores $e(\mathbf{x}_{i})$ are known, the sensitivity model~\eqref{eqn: sensitivity model} determines the admissible range of each true propensity score $e^{*}(\mathbf{x}_{i},\mathbf{u}_{i})$. Specifically, model~\eqref{eqn: sensitivity model} implies the following sensitivity bounds for each unit $i$:
\begin{equation}\label{eqn: sensitivity bounds}
\frac{e(\mathbf{x}_i)}
{e(\mathbf{x}_i)+\Gamma\{1-e(\mathbf{x}_i)\}}
\leq
e^{*}(\mathbf{x}_i,\mathbf{u}_i)
\leq
\frac{\Gamma e(\mathbf{x}_i)}
{1-e(\mathbf{x}_i)+\Gamma e(\mathbf{x}_i)}.
\end{equation}
Therefore, the oracle marginal propensity score vector $\mathcal{E}=(e(\mathbf{x}_1),\ldots,e(\mathbf{x}_N))$ determines a collection of true propensity score vectors $\mathcal{P}=(e^{*}(\mathbf{x}_{1},\mathbf{u}_{1}),\ldots,e^{*}(\mathbf{x}_{N},\mathbf{u}_{N}))$ compatible with the sensitivity model~\eqref{eqn: sensitivity model}, or equivalently, the sensitivity bounds~\eqref{eqn: sensitivity bounds}. For many design-based problems, the existing sensitivity analysis literature provides oracle sensitivity sets given $\mathcal{E}$, defined as confidence sets for the target estimand that are uniformly valid over all true propensity score vectors $\mathcal{P}$ satisfying the sensitivity bounds~\eqref{eqn: sensitivity bounds}. We denote such a construction by $\mathcal{S}^{\text{oracle}}_{1-\alpha,\Gamma}=\psi(\mathcal{E},\mathbf{Z},\mathbf{Y},\Gamma)$; see Appendix~E for detailed examples of such constructions for representative problems.

In practice, however, the marginal propensity score vector $\mathcal{E}$ is usually unknown and must be estimated from the observed data. Researchers can apply the propensity score propagation framework to an established oracle sensitivity set mapping $\psi(\mathcal{E},\mathbf{Z},\mathbf{Y},\Gamma)$ to conduct a valid sensitivity analysis. Specifically, in regeneration run $m$, we replace $\mathcal{E}$ with its regenerated estimate $\widehat{\mathcal{E}}^{(m)}$ and compute $\mathcal{S}^{(m)}_{1-\alpha,\Gamma}=\psi(\widehat{\mathcal{E}}^{(m)},\mathbf{Z},\mathbf{Y},\Gamma)$. The final propagation-based sensitivity set is $\mathcal{S}_{1-\alpha,\Gamma}=\bigcup_{m=1}^{M}\mathcal{S}^{(m)}_{1-\alpha,\Gamma}$.

\section{Simulation Studies}\label{sec: simulation}

Using design-based causal inference for the SATE $\theta_{\tau}$ in observational studies as an example (i.e., the example application in Section~\ref{subsec: example application in observational studies}), we conduct simulation studies to compare the performance of our proposed framework with the existing nonparametric approaches for design-based inference with unknown propensity scores: the plug-in approach and a matching-based approach built on the widely used optimal full matching design (\citealp{rosenbaum1991characterization, hansen2004full, hansen2006optimal}; see Remark~S.3 in Appendix B.3 and Remark S.10 in Appendix D.3). We consider a finite-population dataset with $N=1000$ units, of which the covariates $\mathbf{x}_{i}=(x_{i1},\dots, x_{i5})^{T}$ are generated as follows: $(x_{i1}, x_{i2}, x_{i3})^{T} \overset{\text{iid}}{\sim} \mathcal{N}((0,0,0),\mathbf{I}_{3 \times 3})$, $x_{i4}\overset{\text{iid}}{\sim} \text{Laplace}(0, \sqrt{2}/2)$, $x_{i5}\overset{\text{iid}}{\sim} \text{Laplace}(0, \sqrt{2}/2)$, and the potential outcomes $(Y_{i}(0), Y_{i}(1))$ are generated from the following process: we first generate $Y_{i}(0)=0.15 x_{i1}^3 + 0.15 |x_{i2}| + 0.1 x_{i3}^3 + 0.3 |x_{i4}| + 0.2 x_{i5} + 0.1\epsilon_{i}^{y}$, where $\epsilon_{i}^{y} \overset{\text{iid}}{\sim} N(0,1)$, and $Y_{i}(1)$ is then generated based on the following two treatment effect mechanisms: $Y_{i}(1)=Y_{i}(0)+1 + 0.3 \text{sin}(x_{i2}) + 0.2 x_{i4} + 0.1 x_{i5}$ (Effect Setting 1) and $Y_{i}(1)=Y_{i}(0)+1 + 0.3|x_{i1}| + 0.1\text{tanh}(x_{i5})$ (Effect Setting 2). This data-generating process is used solely to automatically and transparently construct a fixed finite population of covariates $\mathbf{X}=(\mathbf{x}_{1},\dots, \mathbf{x}_{N})$ and potential outcomes $(\mathbf{Y}(0), \mathbf{Y}(1))=\{(Y_{i}(0), Y_{i}(1)): i=1,\dots, N\}$. Importantly, our framework is design-based and does not rely on any distributional assumptions on $(\mathbf{X}, \mathbf{Y}(0), \mathbf{Y}(1))$; once generated, these values are treated as fixed characteristics of the study population. Given this fixed finite population, we generate simulated treatment assignment vectors $\mathbf{Z}=(Z_{1},\dots, Z_{N})$. In each run, the treatment indicator of unit $i$ is drawn according to one of the following two propensity score settings: (i) Setting 1 (Nonlinear Selection Model): $Z_{i}=\mathbbm{1}\{\phi(\mathbf{x}_{i})-0.5>\epsilon_{i}^{z}\}$, where $\phi(\mathbf{x}_{i})=0.1 x_{i1}^3 + 0.3x_{i2} + 0.2 \text{log}(x_{i3}^2) + 0.1 x_{i4} + 0.2 x_{i5} + 0.1 |x_{i1} x_{i2}| + 0.3 (x_{i2} x_{i4})^2$ and $\epsilon_{i}^{z} \overset{\text{iid}}{\sim} N(0,1)$; and (ii) Setting 2 (Nonlinear Logistic Model): $\text{logit} \ P(Z_{i}=1\mid \mathbf{x}_{i})=0.1 x_{i1}^3 + 0.3x_{i2} + 0.2 \text{log}(x_{i3}^2) + 0.1 x_{i4} + 0.2 x_{i5} + 0.2 |x_{i1} x_{i2}| + 0.4 (x_{i3} x_{i4})^2 + 0.1 (x_{i2} x_{i4})^2 - 1$. For each fixed finite population of $(\mathbf{X}, \mathbf{Y}(0), \mathbf{Y}(1))$, the coverage rate is defined over the distribution of treatment assignments. All reported simulation results are then averaged over 1000 realizations of the treatment assignment.

For a fair comparison across methods, all estimated propensity scores used in the IPW estimator and its design-based variance estimator, as described in Sections~\ref{subsec: review known} and \ref{subsec: Nonparametric Propensity Score Propagation}, are obtained using XGBoost \citep{chen2016xgboost}, a widely used gradient boosting algorithm, under both the plug-in and propensity score propagation approaches. For the matching-based approach, we use the commonly adopted post-matching Neyman estimator (see Remarks S.8 in Appendix B.3 for details), together with the corresponding design-based variance estimator \citep{imbens2015causal, fogarty2018mitigating}; this represents the prevailing strategy for conducting design-based causal inference for $\theta_{\tau}$ after matching. When implementing the propensity score propagation framework, we use the nonparametric version described in Algorithm~\ref{alg: general}, setting the number of regeneration runs $M=100$. Table~\ref{tab: simulation} reports, for each approach, the empirical coverage rate of the resulting design-based confidence interval for the SATE $\theta_{\tau}$ (Coverage), the mean finite-sample bias of the target estimator (Bias), the mean confidence interval length (Length), and the ratio of this length to that of the oracle confidence interval constructed with the oracle propensity scores ($\frac{\text{Length}}{\text{Length (Oracle)}}$). In addition, for plug-in and matching-based approaches, we also consider an oracle bias-aware (OBA) confidence interval. Specifically, for a given estimator $T$ (e.g., plug-in or matching estimators) for $\theta_{\tau}$, we assume that $(T-\theta_{\tau})/\text{SE}(T)\xrightarrow{d} N(b,1)$ as $N\rightarrow \infty$, where $\text{SE}(T)$ and $b$ are the standard error and asymptotic bias of $T$, respectively. Following \citet{armstrong2020bias}, we leverage the oracle knowledge of $|E(T)-\theta_{\tau}|$ and $\widehat{\text{SE}}(T)$ (i.e., the sample standard error of $T$ across the 1000 treatment assignment realizations) to construct an OBA confidence interval based on $T$ as $\mathcal{C}^{\text{OBA}}_{1-\alpha}=(T-\chi, T+\chi)$, with $\chi=\widehat{\text{SE}}(T)\cdot \sqrt{Q_{1-\alpha}\big(|E(T)-\theta_{\tau}|^{2}/ \widehat{\text{SE}}^{2}(T)\big)}$, where $Q_{1-\alpha}(\cdot)$ denotes the $100(1-\alpha)\%$-quantile of the chi-squared distribution with one degree of freedom. Then, for the OBA confidence intervals based on plug-in and matching estimators, we report the coverage rate (Coverage (OBA)), the mean confidence interval length (Length (OBA)), and the ratio of the mean confidence interval length of propensity score propagation to that of the OBA confidence interval ($\frac{\text{Length (Propagation)}}{\text{Length (OBA)}}$). 

\begin{table}[ht]
\scriptsize
\caption{For each setting and each approach, we report the coverage rate (Coverage), the mean bias of the target estimator (Bias), the average confidence interval length (Length), and the ratio of this length to that of the oracle confidence interval ($\frac{\text{Length}}{\text{Length (Oracle)}}$). For the oracle bias-aware (OBA) confidence interval for plug-in and matching estimators, we report the corresponding coverage rate (Coverage (OBA)), the average confidence interval length (Length (OBA)), and the ratio of the mean confidence interval length of propensity score propagation to that of the oracle bias-aware (OBA) confidence interval ($\frac{\text{Length (Propagation)}}{\text{Length (OBA)}}$). }
\resizebox{\textwidth}{!}{
\begin{tabular}{cccccccc}
\toprule

\multirow{2}{*}{Effect Setting 1}&\multicolumn{7}{c}{Propensity Score Setting 1}  \\ 
\cmidrule(rl){2-8} 
&  Coverage & Bias  & Length & $\frac{\text{Length}}{\text{Length (Oracle)}}$ & Coverage (OBA) & Length (OBA) & $\frac{\text{Length (Propagation)}}{\text{Length (OBA)}}$ \\
\midrule
Plug-in                    & 0.572 & 0.179  & 0.376 & 0.774 & 0.947 & 0.545 & 1.160\\
Matching      & 0.022 & 0.176  & 0.224 & 0.461 & 0.952 & 0.455 & 1.389 \\
Propagation (Our Proposal)  & 0.998 & - & 0.632 & 1.300 & - &   -   &    -  \\
Oracle                     & 0.941 & 0.027 & 0.486 & 1.000 & - & - & - \\
\midrule
\midrule

\multirow{2}{*}{Effect Setting 1} &\multicolumn{7}{c}{Propensity Score Setting 2}  \\ 
\cmidrule(rl){2-8}
&  Coverage & Bias  & Length & $\frac{\text{Length}}{\text{Length (Oracle)}}$ & Coverage (OBA) & Length (OBA) & $\frac{\text{Length (Propagation)}}{\text{Length (OBA)}}$ \\
\midrule
Plug-in                    & 0.568 & 0.186  & 0.381 & 0.862 & 0.942 & 0.559 & 1.129 \\
Matching     & 0.142 & 0.133  & 0.199 & 0.450 & 0.953 & 0.371 & 1.701 \\
Propagation (Our Proposal)  & 1.000 & - & 0.631 & 1.428 & - &    -   &   -   \\
Oracle                     & 0.941 & 0.016  & 0.442 & 1.000 & - & - & - \\
\midrule
\midrule

\multirow{2}{*}{Effect Setting 2}&\multicolumn{7}{c}{Propensity Score Setting 1}  \\ 
\cmidrule(rl){2-8} 
&  Coverage & Bias  & Length & $\frac{\text{Length}}{\text{Length (Oracle)}}$ & Coverage (OBA) & Length (OBA) & $\frac{\text{Length (Propagation)}}{\text{Length (OBA)}}$ \\
\midrule
Plug-in                    & 0.739 & 0.174  & 0.423 & 0.817 & 0.948 & 0.562 & 1.274 \\
Matching    & 0.100 & 0.178  & 0.276 & 0.533 & 0.951 & 0.460 & 1.557 \\
Propagation (Our Proposal)  & 1.000 & - & 0.716 & 1.382 & - & - & - \\
Oracle                     & 0.955 & 0.027 & 0.518 & 1.000 & - & - & - \\
\midrule
\midrule

\multirow{2}{*}{Effect Setting 2}&\multicolumn{7}{c}{Propensity Score Setting 2}  \\ 
\cmidrule(rl){2-8} 
&  Coverage & Bias  & Length & $\frac{\text{Length}}{\text{Length (Oracle)}}$ & Coverage (OBA) & Length (OBA) & $\frac{\text{Length (Propagation)}}{\text{Length (OBA)}}$ \\
\midrule
Plug-in                    & 0.670 & 0.192  & 0.427 & 0.903 & 0.943 & 0.599 & 1.194 \\
Matching      & 0.298 & 0.134  & 0.235 & 0.497 & 0.953 & 0.377 & 1.897 \\
Propagation (Our Proposal) & 1.000 & - & 0.715 & 1.512 & - & - & - \\
Oracle                     & 0.945 & 0.016 & 0.473 & 1.000 & - & - & - \\
\bottomrule
\end{tabular}
}
\label{tab: simulation}
\end{table}

Table~\ref{tab: simulation} reports the simulation results across the propensity score and treatment effect settings. Across the simulation scenarios, the coverage rate of the proposed propensity score propagation framework attains at least the nominal 95\% coverage rate, whereas all conventional approaches exhibit substantial under-coverage. The plug-in approach systematically falls well below the nominal 95\% coverage rate (achieving only 57\%--74\% coverage), due to treating estimated propensity scores as fixed and known, which induces finite-sample bias and leads to a substantial underestimation of uncertainty. The matching-based approaches also suffer from severe under-coverage (achieving only 2\%--30\% coverage), largely due to residual bias arising from inexact matching on covariates; see Remark~\ref{rem: matching} and Appendix B.1 for details. In contrast, our framework maintains valid coverage across all simulation scenarios, demonstrating that the regeneration-and-union procedure effectively propagates uncertainty from the propensity score estimation step and restores inference validity. Although the resulting confidence intervals obtained via the propensity score propagation framework are necessarily wider (since they account for both layers of uncertainty without requiring any distributional or modeling assumptions on outcome data), their lengths remain practical, with confidence interval length ratios typically between 1.3 and 1.6 relative to the oracle confidence interval $\mathcal{C}_{1-\alpha}^{\text{oracle}}$ (i.e., the unachievable benchmark confidence interval based on true propensity scores).

In addition to the oracle confidence interval $\mathcal{C}_{1-\alpha}^{\text{oracle}}$, the OBA confidence interval $\mathcal{C}^{\text{OBA}}_{1-\alpha}$ provides another useful benchmark for interpreting the length of the confidence intervals produced by propensity score propagation. By construction, $\mathcal{C}^{\text{OBA}}_{1-\alpha}$ leverages oracle knowledge of the finite-sample bias of the plug-in or matching estimators, together with infeasible information about their finite-sample dispersion, and thus represents an idealized and unattainable form of bias correction. Importantly, the confidence intervals obtained via propensity score propagation are practically sized compared to these OBA intervals, with the ratio $\frac{\text{Length (Propagation)}}{\text{Length (OBA)}}$ ranging from 1.1 to 1.9 across the simulation settings. This comparison highlights a key strength of the proposed propensity score propagation framework: without relying on oracle bias information or other infeasible quantities, it delivers valid design-based inference with confidence intervals whose lengths remain practically comparable to those obtained from idealized bias-aware procedures.

Appendix~D provides additional simulation results for a different sample size and for parametric propensity score settings; the main insights regarding propensity score propagation in Section~\ref{sec: simulation} remain unchanged.

\begin{remark}[Some Remarks on Matching]\label{rem: matching}
The matching-based approaches in our simulation studies and data application use optimal full matching with a propensity score caliper \citep{hansen2004full, hansen2006optimal, rosenbaum2020design}, a widely used matching design that retains all study units. This ensures that all approaches considered in this paper target the same finite-population estimand. In the simulation studies, the resulting matched datasets pass conventional balance diagnostics: the average post-matching absolute standardized mean difference is below 0.035 for each of the five covariates, well below the commonly used threshold of 0.10 for assessing adequate covariate balance after matching \citep{small2024protocols}; see Appendix~D.3 for details. Therefore, the under-coverage of the matching-based approaches in our simulations cannot be attributed to poor practical implementation of matching or to a clear failure of conventional post-matching balance diagnostics. Rather, it reflects limitations of design-based inference after inexact matching, consistent with recent work showing that such procedures can suffer from substantial bias and under-coverage even when matched datasets pass standard balance checks \citep{pimentel2024covariate, pimentel2024re, zhu2025randomization}. See Appendices~B.1 and D.3 for further discussion.
\end{remark}

\section{Data Application}\label{sec: data}

We illustrate the practical performance of propensity score propagation by reanalyzing the observational study of \citet{heller2010using}, which examines the causal effect of beginning college at a two-year community college, rather than a four-year college, on total educational attainment, measured by total years of education. Following \citet{heller2010using}, we restrict attention to the 1,819 students whose baseline test scores are at least 55, the median score among students who began at a four-year college. For these students, attending a two-year community college need not reflect a lack of access to a four-year college, making it meaningful to ask what their educational attainment would have been had they instead begun at a four-year college \citep{heller2010using}. Because the scientific question concerns this particular group of students, rather than a broader population for which the two-year versus four-year college contrast may be less interpretable, design-based inference targeting a finite-population causal estimand is especially natural.

In the original dataset, however, students who began at a two-year community college and those who began at a four-year college exhibit substantial imbalance in observed covariates. Design-based approaches are appealing in this setting not only because they yield more interpretable finite-population estimands, but also because they avoid additional distributional or modeling assumptions on the outcome data \citep{rosenbaum2020design, small2024protocols}. They also support outcome-blinded validity checks, such as covariate balance diagnostics formulated as falsification tests; see Figure~\ref{fig:balance_real}.

We consider three nonparametric design-based inference approaches for this observational study: the matching-based approach, the plug-in approach, and nonparametric propensity score propagation. Their implementation follows the same procedures as those described in the simulation studies in Section~\ref{sec: simulation}. Figure~\ref{fig:balance_real} presents the covariate balance diagnostics for these three approaches in the real-data analysis. Specifically, the confidence intervals reported in Figure~\ref{fig:balance_real} are obtained by replacing the outcome variable in the confidence interval formula with each of the 20 covariates considered in \citet{heller2010using}, including region, parents' education level, race, baseline test scores, and family income, among others; see Table~S.5 in Appendix D.3 for the complete list. Because these pre-treatment covariates cannot be causally affected by treatment assignment, the SATE on each of them is exactly zero. Therefore, each resulting confidence interval provides a rigorous falsification test for the corresponding design-based inference approach: if a confidence interval fails to cover zero, this indicates inadequate balance for that covariate and raises concerns about the validity of the resulting inference for the outcome variable. 

\begin{figure}[h]
    \caption{\small Covariate balance diagnostics formulated as falsification tests for the three approaches considered in the real-data analysis: the matching-based approach, the plug-in approach, and propensity score propagation with $M=100$ regeneration runs. Each covariate is normalized to $[0,1]$ to place all covariates on a common scale.}
    \hspace*{-1.0cm}
    \centering
    \includegraphics[width=1.15\linewidth]{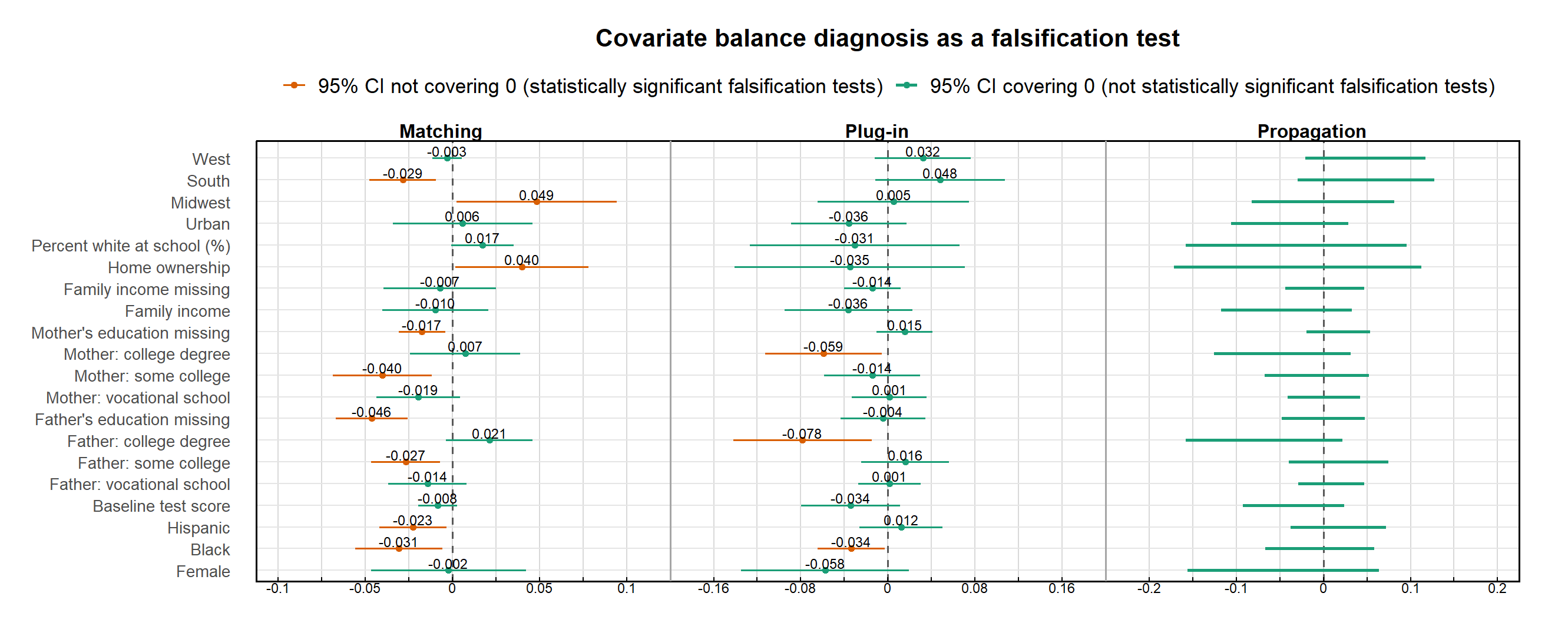}
    \label{fig:balance_real}
\end{figure}

As shown in Figure~\ref{fig:balance_real}, both the matching-based and plug-in approaches leave noticeable residual imbalance for several important covariates, which may introduce substantial bias into inference for the outcome. After observing evident covariate imbalance, \citet{heller2010using} considered improving post-matching balance by either discarding control units that are difficult to match well with treated units or matching on only a subset of covariates at a certain stage of the matching procedure. However, neither strategy is ideal: the former changes the study population and hence alters the causal estimand, thereby complicating the causal interpretation, while the latter may induce confounding bias by excluding relevant covariates during matching. In contrast, under propensity score propagation (e.g., with $M=100$ regeneration runs), all confidence intervals for covariate differences cover zero without discarding any units or covariates, although the intervals are somewhat wider because they appropriately reflect propensity score estimation uncertainty. Moreover, if a chosen value such as $M=100$ failed this outcome-blinded falsification test, one could increase $M$ and reassess covariate balance before conducting the outcome analysis. This provides outcome-blinded guidance for choosing $M$, which is an additional practical advantage of propensity score propagation.

\begin{table}[h]
    \centering
    \caption{The 95\% confidence intervals produced by propensity score propagation, and their corresponding lengths, for different numbers of regeneration runs $M$. In particular, when $M=1$, propensity score propagation reduces to the plug-in approach.}
    \small
    \begin{tabular}{lcc}
    \toprule
         & 95\% CI & Length\\
         \midrule
       Plug-in (i.e., Propagation with $M=1$) & $[-3.324, +0.018]$ & $3.342$ \\
       Propagation ($M=100$) & $[-3.807, +0.408]$ & $4.215$\\
       Propagation ($M=250$) & $[-3.860, +0.535]$ & $4.395$\\
       Propagation ($M=500$) & $[-3.954, +0.793]$ & $4.747$\\
       Propagation ($M=1000$) & $[-3.954, +0.876]$ & $4.830$\\
        \bottomrule
    \end{tabular}
    
    \label{tab: data results}
\end{table}

Table~\ref{tab: data results} reports the 95\% confidence intervals for the SATE produced by nonparametric propensity score propagation (Algorithm~\ref{alg: general}), together with their corresponding lengths, under different numbers of regeneration runs $M$. In particular, when $M=1$, propensity score propagation reduces to the plug-in approach based on a single cross-fitted propensity score estimate. Table~\ref{tab: data results} shows that the confidence intervals produced by propensity score propagation remain reasonably sized relative to the plug-in approach, even when $M$ is large, with interval lengths increasing only moderately from $3.342$ for the plug-in approach to $4.215$, $4.395$, $4.747$, and $4.830$ for $M=100$, $250$, $500$, and $1000$, respectively.

\section{Discussion}\label{sec: discussion}

Providing valid inference when nuisance functions are unknown and must be estimated from the observed data, possibly using flexible machine learning methods, is a central challenge in modern statistics. In super-population settings, semiparametric methods such as double/debiased machine learning address this challenge through orthogonal scores, thereby limiting the first-order impact of nuisance estimation error on inference for low-dimensional targets \citep{chernozhukov2018double, vansteelandt2022assumption}. In contrast, in design-based or finite-population inference, one does not impose distributional assumptions on the outcomes; validity must hold for fixed potential outcome and covariate realizations, with randomness arising only from the treatment assignment, sampling, or missingness mechanism. When the propensity scores are unknown, valid inference must therefore account for uncertainty from propensity score estimation without introducing a super-population outcome model. To the best of our knowledge, propensity score propagation provides the first general framework for doing so.

The present paper studies propagation for one particular nuisance object: the propensity score vector. When this vector is known, valid design-based inference is available through an oracle procedure. When it is unknown, the proposed framework regenerates data-compatible versions of the propensity score vector, applies the oracle procedure to each regenerated vector, and combines the resulting confidence sets. This perspective raises a broader question: for which other nuisance objects can uncertainty be regenerated separately from the outcome data and then propagated through an oracle procedure, without strengthening assumptions on the outcomes?
\begin{equation*} \large \begin{gathered} \substack{\text{Valid inference under}\\ \text{known nuisance objects}} \quad + \quad \substack{\text{Propagation}} \quad \Rightarrow \quad \substack{\text{Valid inference under}\\ \text{unknown nuisance objects}}  \end{gathered} \end{equation*}
Answering this question would extend propagation beyond propensity scores and could establish it as a general reduction principle: valid oracle inference under known nuisance objects can be converted into valid feasible inference when those nuisance objects are unknown and must be estimated from data, potentially using flexible nonparametric or machine-learning methods.

Another important direction is to extend the framework to more complex design-based estimands and assignment structures, including clustered designs, longitudinal treatment regimes, multi-valued or factorial treatments, and settings with network interference (e.g., \citealp{aronow2017estimating, bojinov2019time, athey2022design, li2023randomization, zhang2023randomization, zubizarreta2023handbook, dasgupta2026introduction}). These problems involve richer propensity score objects and may therefore motivate new implementations of the basic principle of propensity score propagation.

\section*{Acknowledgement}

The authors thank Rui Gao for helpful discussions and contributions during the early stages of the simulation studies. AI tools were used to polish the language, assist with the preparation of selected figures, and help produce basic functions in \textsf{R} code. The authors take full responsibility for the study design, methodology, analysis, interpretation, and presentation of all content.

\section*{Supplementary Materials}

The supplementary materials include technical proofs, additional theoretical and numerical results, more data details, and further discussions.

\bibliographystyle{apalike}
\bibliography{references}

@article{vansteelandt2022assumption,
  title={Assumption-lean inference for generalised linear model parameters},
  author={Vansteelandt, Stijn and Dukes, Oliver},
  journal={Journal of the Royal Statistical Society: Series B (Statistical Methodology)},
  volume={84},
  number={3},
  pages={657--685},
  year={2022},
  publisher={Oxford University Press}
}

@article{cohen2022gaussian,
  title={Gaussian prepivoting for finite population causal inference},
  author={Cohen, Peter L and Fogarty, Colin B},
  journal={Journal of the Royal Statistical Society: Series B (Statistical Methodology)},
  volume={84},
  number={2},
  pages={295--320},
  year={2022},
  publisher={Oxford University Press}
}

@article{li2024sensitivity,
  title={Sensitivity analysis for observational studies with flexible matched designs},
  author={Li, Xinran},
  journal={Biometrika},
  year    = {2025},
  volume  = {112},
  number  = {4},
  pages   = {asaf069}
}

@book{dasgupta2026introduction,
  title     = {Introduction to Modern Randomization-Based Design and Analysis for Causal Inference},
  author    = {Dasgupta, Tirthankar and Rubin, Donald B.},
  year      = {2026},
  edition   = {1},
  publisher = {Chapman and Hall/CRC}
}

@article{savje2022inconsistency,
  title={On the inconsistency of matching without replacement},
  author={S{\"a}vje, Fredrik},
  journal={Biometrika},
  volume={109},
  number={2},
  pages={551--558},
  year={2022},
  publisher={Oxford University Press}
}

@article{imbens2005robust,
  title={Robust, accurate confidence intervals with a weak instrument: quarter of birth and education},
  author={Imbens, Guido W and Rosenbaum, Paul R},
  journal={Journal of the Royal Statistical Society: Series A (Statistics in Society)},
  volume={168},
  number={1},
  pages={109--126},
  year={2005},
  publisher={Oxford University Press}
}

@article{angrist1996identification,
  title={Identification of causal effects using instrumental variables},
  author={Angrist, Joshua D and Imbens, Guido W and Rubin, Donald B},
  journal={Journal of the American Statistical Association},
  volume={91},
  number={434},
  pages={444--455},
  year={1996},
  publisher={Taylor \& Francis}
}

@article{haberman1977maximum,
  title={Maximum likelihood estimates in exponential response models},
  author={Haberman, Shelby J},
  journal={The Annals of Statistics},
  year = {1977}, 
  volume = {5}, 
  number = {5}, 
  pages = {815--841}
}

@article{rosenbaum1984conditional,
  title={Conditional permutation tests and the propensity score in observational studies},
  author={Rosenbaum, Paul R},
  journal={Journal of the American Statistical Association},
  volume={79},
  number={387},
  pages={565--574},
  year={1984},
  publisher={Taylor \& Francis}
}

@article{bickel2010subsampling,
  title={Subsampling methods for genomic inference},
  author={Bickel, Peter J and Boley, Nathan and Brown, James B and Huang, Haiyan and Zhang, Nancy R},
  journal={The Annals of Applied Statistics},
  pages={1660--1697},
  year={2010},
  publisher={JSTOR}
}

@article{silber2013characteristics,
  title={Characteristics associated with differences in survival among black and white women with breast cancer},
  author={Silber, Jeffrey H and Rosenbaum, Paul R and Clark, Amy S and Giantonio, Bruce J and Ross, Richard N and Teng, Yun and Wang, Min and Niknam, Bijan A and Ludwig, Justin M and Wang, Wei and others},
  journal={JAMA},
  volume={310},
  number={4},
  year={2013}
}

@article{small2024protocols,
  title={Protocols for observational studies: Methods and open problems},
  author={Small, Dylan S},
  journal={Statistical Science},
  volume={39},
  number={4},
  pages={519--554},
  year={2024},
  publisher={Institute of Mathematical Statistics}
}

@incollection{kennedy2024semiparametric,
  author    = {Edward H. Kennedy},
  title     = {Semiparametric Doubly Robust Targeted Double Machine Learning: A Review},
  booktitle = {Handbook of Statistical Methods for Precision Medicine},
  publisher = {Chapman \& Hall/CRC},
  year      = {2024},
  pages     = {207--236}
}

@article{chernozhukov2018double,
  title={Double/debiased machine learning for treatment and structural parameters},
  author={Chernozhukov, Victor and Chetverikov, Denis and Demirer, Mert and Duflo, Esther and Hansen, Christian and Newey, Whitney and Robins, James},
 journal = {The Econometrics Journal},
  volume  = {21},
  number  = {1},
  pages   = {C1--C68},
  year    = {2018}
}

@article{bodory2020finite,
  title={The finite sample performance of inference methods for propensity score matching and weighting estimators},
  author={Bodory, Hugo and Camponovo, Lorenzo and Huber, Martin and Lechner, Michael},
  journal={Journal of Business \& Economic Statistics},
  volume={38},
  number={1},
  pages={183--200},
  year={2020},
  publisher={Taylor \& Francis}
}

@book{politis1999subsampling,
  author    = {Dimitris N. Politis and Joseph P. Romano and Michael Wolf},
  title     = {Subsampling},
  year      = {1999},
  publisher = {Springer},
  address   = {New York, NY},
  series    = {Springer Series in Statistics}
}

@article{bousquet2002stability,
  title={Stability and generalization},
  author={Bousquet, Olivier and Elisseeff, Andr{\'e}},
  journal={Journal of Machine Learning Research},
  volume={2},
  number={Mar},
  pages={499--526},
  year={2002}
}

@article{politis1994large,
  title={Large sample confidence regions based on subsamples under minimal assumptions},
  author={Politis, Dimitris N and Romano, Joseph P},
  journal={The Annals of Statistics},
  year = {1994},
  volume = {22},
  number = {4}, 
  pages = {2031--2050}
}

@article{shi2024some,
  title={Some theoretical foundations for the design and analysis of randomized experiments},
  author={Shi, Lei and Li, Xinran},
  journal={Journal of Causal Inference},
  volume={12},
  number={1},
  pages={20230067},
  year={2024},
  publisher={De Gruyter}
}

@article{shi2026berry,
  title={{Berry-Esseen bounds for design-based causal inference with possibly diverging treatment levels and varying group sizes}},
  author={Shi, Lei and Ding, Peng},
  journal={The Annals of Statistics},
  volume={54},
  number={1},
  pages={324--349},
  year={2026},
  publisher={Institute of Mathematical Statistics}
}

@book{cassel1977foundations,
  title={Foundations of Inference in Survey Sampling},
  author={Cassel, Claes-Magnus and Sarndal, Carl-Erik and Wretman, Jan H{\aa}kan},
  publisher = {John Wiley},
  address   = {New York, N.Y.},
  year={1977}
}

@article{lunceford2004stratification,
  title={Stratification and weighting via the propensity score in estimation of causal treatment effects: a comparative study},
  author={Lunceford, Jared K and Davidian, Marie},
  journal={Statistics in Medicine},
  volume={23},
  number={19},
  pages={2937--2960},
  year={2004},
  publisher={Wiley Online Library}
}

@article{hanley1982meaning,
  title={{The meaning and use of the area under a receiver operating characteristic (ROC) curve.}},
  author={Hanley, James A and McNeil, Barbara J},
  journal={Radiology},
  volume={143},
  number={1},
  pages={29--36},
  year={1982}
}

@article{armstrong2020bias,
  title={Bias-aware inference in regularized regression models},
  author={Armstrong, Timothy B and Koles{\'a}r, Michal and Kwon, Soonwoo},
  journal={arXiv preprint arXiv:2012.14823},
  year={2020}
}

@article{hajek1964asymptotic,
  title={Asymptotic theory of rejective sampling with varying probabilities from a finite population},
  author={H{\'a}jek, Jaroslav},
  journal={The Annals of Mathematical Statistics},
  volume={35},
  number={4},
  pages={1491--1523},
  year={1964},
  publisher={Institute of Mathematical Statistics}
}

@article{Eldan2025Isoperimetric,
	author = {Eldan, Ronen and Kindler, Guy and Lifshitz, Noam and Minzer, Dor},
	journal = {Discrete Analysis},
	doi = {10.19086/da.142095},
	year = {2025},
	title = {Isoperimetric {Inequalities} {Made} {Simpler}},
}

@book{o2014analysis,
  title={{Analysis of Boolean Functions}},
  author={O'Donnell, Ryan},
  year={2014},
  publisher={Cambridge University Press}
}

@article{zheng2025perturbed,
  title={Perturbed Double Machine Learning: Nonstandard Inference Beyond the Parametric Length},
  author={Zheng, Mengchu and Bonvini, Matteo and Guo, Zijian},
  journal={arXiv preprint arXiv:2511.01222},
  year={2025}
}

@article{lin2024consistency,
  title={On the consistency of bootstrap for matching estimators},
  author={Lin, Ziming and Han, Fang},
  journal={Biometrika},
  pages={asag005},
  year={2026},
  publisher={Oxford University Press}
}

@article{fahrmeir1985consistency,
  title={Consistency and asymptotic normality of the maximum likelihood estimator in generalized linear models},
  author={Fahrmeir, Ludwig and Kaufmann, Heinz},
  journal={The Annals of Statistics},
  volume={13},
  number={1},
  pages={342--368},
  year={1985},
  publisher={Institute of Mathematical Statistics}
}

@article{guo2023causal,
  title={Causal inference with invalid instruments: post-selection problems and a solution using searching and sampling},
  author={Guo, Zijian},
  journal={Journal of the Royal Statistical Society: Series B (Statistical Methodology)},
  volume={85},
  number={3},
  pages={959--985},
  year={2023},
  publisher={Oxford University Press US}
}

@book{efron1994introduction,
  title={An Introduction to the Bootstrap},
  author={Efron, Bradley and Tibshirani, Robert J},
  year={1994},
  publisher={Chapman and Hall/CRC}
}

@article{heng2025nonbipartiteDID,
  title={A Non-Bipartite Matching Framework for Difference-in-Differences with General Treatment Types},
  author={Heng, Siyu and Huang, Yuan and Kang, Hyunseung},
  journal={arXiv preprint arXiv:2511.21973},
  year={2025}
}

@article{horvitz1952generalization,
  title={A generalization of sampling without replacement from a finite universe},
  author={Horvitz, Daniel G and Thompson, Donovan J},
  journal={Journal of the American Statistical Association},
  volume={47},
  number={260},
  pages={663--685},
  year={1952},
  publisher={Taylor \& Francis}
}

@article{pashley2021conditional,
  title={Conditional as-if analyses in randomized experiments},
  author={Pashley, Nicole E and Basse, Guillaume W and Miratrix, Luke W},
  journal={Journal of Causal Inference},
  volume={9},
  number={1},
  pages={264--284},
  year={2021},
  publisher={De Gruyter}
}

@article{pimentel2024re,
  title={Re-evaluating the impact of hormone replacement therapy on heart disease using match-adaptive randomization inference},
  author={Pimentel, Samuel D and Yu, Ruoqi},
  journal={arXiv preprint arXiv:2403.01330},
  year={2024}
}

@article{zhang2023randomization,
  title={What is a randomization test?},
  author={Zhang, Yao and Zhao, Qingyuan},
  journal={Journal of the American Statistical Association},
  volume={118},
  number={544},
  pages={2928--2942},
  year={2023},
  publisher={Taylor \& Francis}
}

@article{bethlehem2010selection,
  title={Selection bias in web surveys},
  author={Bethlehem, Jelke},
  journal={International Statistical Review},
  volume={78},
  number={2},
  pages={161--188},
  year={2010},
  publisher={Wiley Online Library}
}

@book{little2019statistical,
  title={{Statistical Analysis with Missing Data}},
  author={Little, Roderick JA and Rubin, Donald B},
  year={2019},
  publisher={John Wiley \& Sons}
}

@article{ivanova2022randomization,
  title={Randomization tests in clinical trials with multiple imputation for handling missing data},
  author={Ivanova, Anastasia and Lederman, Seth and Stark, Philip B and Sullivan, Gregory and Vaughn, Ben},
  journal={Journal of Biopharmaceutical Statistics},
  volume={32},
  number={3},
  pages={441--449},
  year={2022},
  publisher={Taylor \& Francis}
}

@article{heussen2023randomization,
  title={Randomization-Based Inference for Clinical Trials with Missing Outcome Data},
  author={Heussen, Nicole and Hilgers, Ralf-Dieter and Rosenberger, William F and Tan, Xiao and Uschner, Diane},
  journal={Statistics in Biopharmaceutical Research},
volume={16},
  number={4},
  pages={456--467},
  year={2024},
  publisher={Taylor \& Francis}
}

@article{li2025randomization,
  title={Randomization Inference with Sample Attrition},
  author={Li, Xinran and Sheng, Peizan and Yu, Zeyang},
  journal={arXiv preprint arXiv:2507.00795},
  year={2025}
}

@article{rosenbaum2002covariance,
  title={Covariance adjustment in randomized experiments and observational studies},
  author={Rosenbaum, Paul R},
  journal={Statistical Science},
  volume={17},
  number={3},
  pages={286--327},
  year={2002},
  publisher={Institute of Mathematical Statistics}
}

@article{dasgupta2015causal,
  title={Causal inference from 2K factorial designs by using potential outcomes},
  author={Dasgupta, Tirthankar and Pillai, Natesh S and Rubin, Donald B},
  journal={Journal of the Royal Statistical Society: Series B (Statistical Methodology)},
  volume={77},
  number={4},
  pages={727--753},
  year={2015},
  publisher={Oxford University Press}
}

@article{little2004model,
  title={To model or not to model? Competing modes of inference for finite population sampling},
  author={Little, Roderick J},
  journal={Journal of the American Statistical Association},
  volume={99},
  number={466},
  pages={546--556},
  year={2004},
  publisher={Taylor \& Francis}
}

@article{heng2025design,
  title={Design-based causal inference with missing outcomes: Missingness mechanisms, imputation-assisted randomization tests, and covariate adjustment},
  author={Heng, Siyu and Zhang, Jiawei and Feng, Yang},
  journal={Journal of the American Statistical Association},
  volume={121},
  number={553},
  pages={312--325},
  year={2026}
}

@article{manski2018right,
  title={{How do right-to-carry laws affect crime rates? Coping with ambiguity using bounded-variation assumptions}},
  author={Manski, Charles F and Pepper, John V},
  journal={Review of Economics and Statistics},
  volume={100},
  number={2},
  pages={232--244},
  year={2018},
  publisher={MIT Press}
}

@article{sarndal1978design,
  title={Design-based and model-based inference in survey sampling [with discussion and reply]},
  author={S{\"a}rndal, Carl-Erik and Thomsen, Ib and Hoem, Jan M and Lindley, DV and Barndorff-Nielsen, O and Dalenius, Tore},
  journal={Scandinavian Journal of Statistics},
  year    = {1978},
  volume  = {5},
  number  = {1},
  pages   = {27--52}
}

@article{ImbensMenzel2021CausalBootstrap,
  title        = {A Causal Bootstrap},
  author       = {Imbens, Guido and Menzel, Konrad},
  journal      = {The Annals of Statistics},
  volume       = {49},
  number       = {3},
  pages        = {1460--1488},
  year         = {2021}
}

@article{rosenbaum1989optimal,
  title={Optimal matching for observational studies},
  author={Rosenbaum, Paul R},
  journal={Journal of the American Statistical Association},
  volume={84},
  number={408},
  pages={1024--1032},
  year={1989},
  publisher={Taylor \& Francis}
}

@article{rambachan2025design,
 title={Design-based uncertainty for quasi-experiments},
  author={Rambachan, Ashesh and Roth, Jonathan},
  journal={Journal of the American Statistical Association},
  volume={121},
  number={553},
  pages={477--491},
  year={2026},
  publisher={Taylor \& Francis}
}

@article{zhu2025randomization,
  title={Randomization-Based Inference for Average Treatment Effect in Inexactly Matched Observational Studies},
  author={Zhu, Jianan and Zhang, Jeffrey and Guo, Zijian and Heng, Siyu},
  journal={arXiv preprint arXiv:2308.02005},
  year={2025}
}

@article{athey2022design,
  title={Design-based analysis in difference-in-differences settings with staggered adoption},
  author={Athey, Susan and Imbens, Guido W},
  journal={Journal of Econometrics},
  volume={226},
  number={1},
  pages={62--79},
  year={2022},
  publisher={Elsevier}
}

@article{xu2021potential,
  title={Potential outcomes and finite-population inference for M-estimators},
  author={Xu, Ruonan},
  journal={The Econometrics Journal},
  volume={24},
  number={1},
  pages={162--176},
  year={2021},
  publisher={Oxford University Press}
}

@article{abadie2020sampling,
  title={Sampling-based versus design-based uncertainty in regression analysis},
  author={Abadie, Alberto and Athey, Susan and Imbens, Guido W and Wooldridge, Jeffrey M},
  journal={Econometrica},
  volume={88},
  number={1},
  pages={265--296},
  year={2020},
  publisher={Wiley Online Library}
}

@article{tan2006distributional,
  title={A distributional approach for causal inference using propensity scores},
  author={Tan, Zhiqiang},
  journal={Journal of the American Statistical Association},
  volume={101},
  number={476},
  pages={1619--1637},
  year={2006},
  publisher={Taylor \& Francis}
}

@article{rosenbaum1983central,
  title={The central role of the propensity score in observational studies for causal effects},
  author={Rosenbaum, Paul R and Rubin, Donald B},
  journal={Biometrika},
  volume={70},
  number={1},
  pages={41--55},
  year={1983},
  publisher={Oxford University Press}
}

@article{imbens2004nonparametric,
  title={Nonparametric estimation of average treatment effects under exogeneity: A review},
  author={Imbens, Guido W},
  journal={Review of Economics and Statistics},
  volume={86},
  number={1},
  pages={4--29},
  year={2004},
  publisher={MIT Press 238 Main St., Suite 500, Cambridge, MA 02142-1046, USA journals~…}
}

@article{guo2024statistical,
  title={Statistical inference for maximin effects: Identifying stable associations across multiple studies},
  author={Guo, Zijian},
  journal={Journal of the American Statistical Association},
  volume={119},
  number={547},
  pages={1968--1984},
  year={2024},
  publisher={Taylor \& Francis}
}

@article{rosenbaum1991characterization,
  title={A characterization of optimal designs for observational studies},
  author={Rosenbaum, Paul R},
  journal={Journal of the Royal Statistical Society: Series B (Methodological)},
  volume={53},
  number={3},
  pages={597--610},
  year={1991},
  publisher={Wiley Online Library}
}

@incollection{athey2017econometrics,
  title={The econometrics of randomized experiments},
  author={Athey, Susan and Imbens, Guido W},
  booktitle={Handbook of Economic Field Experiments},
  volume={1},
  pages={73--140},
  year={2017},
  publisher={Elsevier}
}

@article{young2019channeling,
  title={{Channeling Fisher: Randomization tests and the statistical insignificance of seemingly significant experimental results}},
  author={Young, Alwyn},
  journal={The Quarterly Journal of Economics},
  volume={134},
  number={2},
  pages={557--598},
  year={2019},
  publisher={Oxford University Press}
}

@article{li2017general,
  title={General forms of finite population central limit theorems with applications to causal inference},
  author={Li, Xinran and Ding, Peng},
  journal={Journal of the American Statistical Association},
  volume={112},
  number={520},
  pages={1759--1769},
  year={2017},
  publisher={Taylor \& Francis}
}

@article{hajek1960limiting,
  title={Limiting distributions in simple random sampling from a finite population},
  author={H{\'a}jek, Jaroslav},
  journal={Publications of the Mathematical Institute of the Hungarian Academy of Sciences},
  volume={5},
  pages={361--374},
  year={1960}
}

@book{imbens2015causal,
  title={Causal Inference in Statistics, Social, and Biomedical Sciences},
  author={Imbens, Guido W and Rubin, Donald B},
  year={2015},
  publisher={Cambridge University Press}
}

@article{fogarty2020studentized,
  title={Studentized sensitivity analysis for the sample average treatment effect in paired observational studies},
  author={Fogarty, Colin B},
  journal={Journal of the American Statistical Association},
volume={115},
  number={531},
  pages={1518--1530},
  year={2020},
  publisher={Taylor \& Francis}
}

@article{zhao2018randomization,
  title={Randomization-based causal inference from split-plot designs},
  author={Zhao, Anqi and Ding, Peng and Mukerjee, Rahul and Dasgupta, Tirthankar},
journal={The Annals of Statistics},
  volume={46},
  number={5},
  pages={1876--1903},
  year={2018}
}

@article{neyman1923application,
  title={{On the application of probability theory to agricultural experiments. Essay on principles. Section 9. (translated and edited by {D.} {M.} {D}abrowska and {T.} {P.} {S}peed)}},
  author={Neyman, Jerzy S},
  journal={Statistical Science},
  volume={(1990) 5},
  pages={465--480},
  year={1923}
}

@article{basse2024randomization,
  title={Randomization tests for peer effects in group formation experiments},
  author={Basse, Guillaume and Ding, Peng and Feller, Avi and Toulis, Panos},
  journal={Econometrica},
  volume={92},
  number={2},
  pages={567--590},
  year={2024},
  publisher={Wiley Online Library}
}

@article{mukerjee2018using,
  title={Using standard tools from finite population sampling to improve causal inference for complex experiments},
  author={Mukerjee, Rahul and Dasgupta, Tirthankar and Rubin, Donald B},
  journal={Journal of the American Statistical Association},
  volume={113},
  number={522},
  pages={868--881},
  year={2018},
  publisher={Taylor \& Francis}
}

@inproceedings{chen2016xgboost,
  title={{XGBoost: A scalable tree boosting system}},
  author={Chen, Tianqi and Guestrin, Carlos},
  booktitle={Proceedings of the 22nd ACM SIGKDD International Conference on Knowledge Discovery and Data Mining},
  pages={785--794},
  year={2016}
}

@book{ding2024first,
  title={A First Course in Causal Inference},
  author={Ding, Peng},
  year={2024},
  publisher={CRC Press}
}

@article{guo2023statistical,
  title={On the statistical role of inexact matching in observational studies},
  author={Guo, Kevin and Rothenh{\"a}usler, Dominik},
  journal={Biometrika},
  volume={110},
  number={3},
  pages={631--644},
  year={2023},
  publisher={Oxford University Press}
}

@article{rubin1974estimating,
  title={Estimating causal effects of treatments in randomized and nonrandomized studies.},
  author={Rubin, Donald B},
  journal={Journal of Educational Psychology},
  volume={66},
  number={5},
  pages={688},
  year={1974},
  publisher={American Psychological Association}
}

@book{fisher1937design,
  title={The Design of Experiments},
  publisher={Edinburgh and London: Oliver and Boyd},
  author={Fisher, Ronald A},
  year={1935}
}

@article{hansen2006optimal,
  title={Optimal full matching and related designs via network flows},
  author={Hansen, Ben B and Klopfer, Stephanie Olsen},
  journal={Journal of Computational and Graphical Statistics},
  volume={15},
  number={3},
  pages={609--627},
  year={2006},
  publisher={Taylor \& Francis}
}

@article{zubizarreta2012using,
  title={Using mixed integer programming for matching in an observational study of kidney failure after surgery},
  author={Zubizarreta, Jos{\'e} R},
  journal={Journal of the American Statistical Association},
  volume={107},
  number={500},
  pages={1360--1371},
  year={2012},
  publisher={Taylor \& Francis Group}
}

@book{rosenbaum2002observational,
  author={Rosenbaum, Paul R},
  title={Observational Studies},
  year={2002},
  publisher={Springer}
}

@article{zhao2019sens_boot,
  title={Sensitivity analysis for inverse probability weighting estimators via the percentile bootstrap},
  author={Zhao, Qingyuan and Small, Dylan S and Bhattacharya, Bhaswar B},
  journal={Journal of the Royal Statistical Society: Series B (Statistical Methodology)},
  year    = {2019},
  volume  = {81},
  number  = {4},
  pages   = {735--761}
}

@article{zhao2019sensitivityvalue,
  title={On sensitivity value of pair-matched observational studies},
  author={Zhao, Qingyuan},
  journal={Journal of the American Statistical Association},
  volume={114},
  number={526},
  pages={713--722},
  year={2019},
  publisher={Taylor \& Francis}
}

@article{hansen2004full,
  title={Full matching in an observational study of coaching for the {SAT}},
  author={Hansen, Ben B},
  journal={Journal of the American Statistical Association},
  volume={99},
  number={467},
  pages={609--618},
  year={2004},
  publisher={Taylor \& Francis}
}

@article{li2023randomization,
  title={Randomization-based test for censored outcomes: a new look at the Logrank test},
  author={Li, Xinran and Small, Dylan S},
  journal={Statistical Science},
  volume={38},
  number={1},
  pages={92--107},
  year={2023},
  publisher={Institute of Mathematical Statistics}
}

@article{kang2016full,
  title={Full matching approach to instrumental variables estimation with application to the effect of malaria on stunting},
  author={Kang, Hyunseung and Kreuels, Benno and May, J{\"u}rgen and Small, Dylan S},
  journal={The Annals of Applied Statistics},
  volume={10},
  number={1},
  pages={335--364},
  year={2016},
  publisher={Institute of Mathematical Statistics}
}

@article{fogarty2018mitigating,
  title={On mitigating the analytical limitations of finely stratified experiments},
  author={Fogarty, Colin B},
  journal={Journal of the Royal Statistical Society: Series B (Statistical Methodology)},
  volume={80},
  number={5},
  pages={1035--1056},
  year={2018},
  publisher={Wiley Online Library}
}

@article{ma2020robust,
  title={Robust inference using inverse probability weighting},
  author={Ma, Xinwei and Wang, Jingshen},
  journal={Journal of the American Statistical Association},
  volume={115},
  number={532},
  pages={1851--1860},
  year={2020},
  publisher={Taylor \& Francis}
}

@article{rosenbaum1987sensitivity,
  title={Sensitivity analysis for certain permutation inferences in matched observational studies},
  author={Rosenbaum, Paul R},
  journal={Biometrika},
  volume={74},
  number={1},
  pages={13--26},
  year={1987},
  publisher={Oxford University Press}
}

@article{pimentel2024covariate,
  title={Covariate-adaptive randomization inference in matched designs},
  author={Pimentel, Samuel D and Huang, Yaxuan},
  journal={Journal of the Royal Statistical Society: Series B (Statistical Methodology)},
  year    = {2024},
  volume  = {86},
  number  = {5},
  pages   = {1312--1338}
}

@book{rosenbaum2020design,
  title={Design of Observational Studies (Second Edition)},
  author={Rosenbaum, Paul R},
  year={2020},
  publisher={Springer}
}

@article{crump2009dealing,
  title={Dealing with limited overlap in estimation of average treatment effects},
  author={Crump, Richard K and Hotz, V Joseph and Imbens, Guido W and Mitnik, Oscar A},
  journal={Biometrika},
  volume={96},
  number={1},
  pages={187--199},
  year={2009},
  publisher={Oxford University Press}
}

@article{fogarty2023testing,
  title={Testing weak nulls in matched observational studies},
  author={Fogarty, Colin B},
  journal={Biometrics},
  volume={79},
  number={3},
  pages={2196--2207},
  year={2023},
  publisher={Wiley Online Library}
}

@article{lin2013agnostic,
  title={{Agnostic notes on regression adjustments to experimental data: Reexamining Freedman's critique}},
  author={Lin, Winston},
   journal = {The Annals of Applied Statistics},
  year    = {2013},
  volume  = {7},
  number  = {1},
  pages   = {295--318}
}

@article{li2020rerandomization,
  title={Rerandomization and regression adjustment},
  author={Li, Xinran and Ding, Peng},
  journal={Journal of the Royal Statistical Society: Series B (Statistical Methodology)},
  volume={82},
  number={1},
  pages={241--268},
  year={2020},
  publisher={Oxford University Press}
}

@article{zhao2021covariate,
  title={Covariate-adjusted Fisher randomization tests for the average treatment effect},
  author={Zhao, Anqi and Ding, Peng},
  journal={Journal of Econometrics},
  volume={225},
  number={2},
  pages={278--294},
  year={2021},
  publisher={Elsevier}
}

@article{liu2020regression,
  title={Regression-adjusted average treatment effect estimates in stratified randomized experiments},
  author={Liu, Hanzhong and Yang, Yuehan},
  journal={Biometrika},
  volume={107},
  number={4},
  pages={935--948},
  year={2020},
  publisher={Oxford University Press}
}

@article{bojinov2019time,
  title={Time series experiments and causal estimands: exact randomization tests and trading},
  author={Bojinov, Iavor and Shephard, Neil},
  journal={Journal of the American Statistical Association},
  volume={114},
  number={528},
  pages={1665--1682},
  year={2019},
  publisher={Taylor \& Francis}
}

@article{heller2010using,
  title={Using the cross-match test to appraise covariate balance in matched pairs},
  author={Heller, Ruth and Rosenbaum, Paul R and Small, Dylan S},
  journal={The American Statistician},
  volume={64},
  number={4},
  pages={299--309},
  year={2010},
  publisher={Taylor \& Francis}
}

@article{aronow2017estimating,
  author = {Aronow, Peter M. and Samii, Cyrus},
  title = {Estimating average causal effects under general interference, with application to a social network experiment},
  journal = {The Annals of Applied Statistics},
  year = {2017},
  volume = {11},
  number = {4},
  pages = {1912--1947}
}

@article{rouse1995democratization,
  title={Democratization or diversion? {The effect of community colleges on educational attainment}},
  author={Rouse, Cecilia Elena},
  journal={Journal of Business \& Economic Statistics},
  volume={13},
  number={2},
  pages={217--224},
  year={1995},
  publisher={Taylor \& Francis}
}

@book{zubizarreta2023handbook,
  title={{Handbook of Matching and Weighting Adjustments for Causal Inference}},
  author={Zubizarreta, Jos{\'e} R and Stuart, Elizabeth A and Small, Dylan S and Rosenbaum, Paul R},
  year={2023},
  publisher={CRC Press}
}

@article{hudgens2008toward,
  title={Toward causal inference with interference},
  author={Hudgens, Michael G and Halloran, M Elizabeth},
  journal={Journal of the American Statistical Association},
  volume={103},
  number={482},
  pages={832--842},
  year={2008},
  publisher={Taylor \& Francis}
}

\newpage

\begin{center}
    \Large \bf Supplementary Materials for ``Propensity Score Propagation: A General Framework for Design-Based Inference with Unknown Propensity Scores'' 
\end{center}

\renewcommand{\thetheorem}{S.\arabic{theorem}}

\renewcommand{\theassumption}{S.\arabic{assumption}}

\renewcommand{\theproposition}{S.\arabic{proposition}}

\renewcommand{\thelemma}{S.\arabic{lemma}}

\renewcommand{\thecorollary}{S.\arabic{corollary}}

\renewcommand{\thecondition}{S.\arabic{condition}}

\renewcommand{\theremark}{S.\arabic{remark}}

\renewcommand{\thetable}{S.\arabic{table}}

\renewcommand{\thefigure}{S.\arabic{figure}}

\setcounter{theorem}{0}
\setcounter{condition}{0}
\setcounter{assumption}{0}
\setcounter{proposition}{0}
\setcounter{corollary}{0}
\setcounter{lemma}{0}
\setcounter{remark}{0}
\setcounter{table}{0}
\setcounter{figure}{0}

\begingroup
\allowdisplaybreaks

\section*{Appendix A: Proofs and Additional Theoretical Results}

\subsection*{A.1: Examples of Established Validity Results for Oracle Confidence Sets}

In the main text, a key ingredient for establishing the validity (i.e., coverage guarantee) of propensity score propagation is the established validity of the Wald-type oracle confidence interval $\mathcal{C}^{\text{oracle}}_{1-\alpha}(\theta)=\big[\widehat{\theta}_{\text{oracle}}-z_{1-\alpha/2}\cdot \widehat{\mathcal{V}}_{\text{oracle}}^{1/2},\, \widehat{\theta}_{\text{oracle}}+z_{1-\alpha/2}\cdot \widehat{\mathcal{V}}_{\text{oracle}}^{1/2}\big]$ for the target estimand $\theta$. There is an extensive literature on constructing such oracle confidence intervals $\mathcal{C}^{\text{oracle}}_{1-\alpha}(\theta)$ for a wide range of target estimands in settings such as design-based causal inference (e.g., \citealp{neyman1923application, rosenbaum2002observational, rosenbaum2020design, lin2013agnostic, dasgupta2015causal, imbens2015causal, athey2017econometrics, li2017general, zhao2018randomization, fogarty2018mitigating, abadie2020sampling, li2023randomization, basse2024randomization, ding2024first, dasgupta2026introduction,  shi2026berry}; among many others) and design-based inference for survey sampling (e.g., \citealp{hajek1960limiting, hajek1964asymptotic, cassel1977foundations, sarndal1978design, little2004model, mukerjee2018using}, among many others). A standard route for establishing the validity of $\mathcal{C}^{\text{oracle}}_{1-\alpha}(\theta)$ is to show: (i) the asymptotic normality of $\widehat{\theta}_{\text{oracle}}$ (as stated in Condition~2a in the main text), which can often be established using finite-population central limit theorems (\citealp{li2017general, shi2024some}) under mild regularity conditions; and (ii) the validity, or asymptotic conservativeness, of the design-based variance estimator $\widehat{\mathcal{V}}_{\text{oracle}}$ (as stated in Condition~2b in the main text).

For example, following the arguments in \citet{zhu2025randomization}, which may be viewed as extending the arguments in \citet{fogarty2018mitigating} from perfect randomization to biased randomization, Condition 1 in the main text, together with Condition~\ref{cond: convergence of finite-population means} below, provide a set of sufficient conditions under which Condition~2a in the main text holds for the oracle estimator $\widehat{\theta}_{\tau, \text{oracle}}$ of the sample average treatment effect $\theta_{\tau}$, and Condition~2b in the main text holds for the corresponding oracle design-based variance estimator $\widehat{\mathcal{V}}_{\tau, \text{oracle}}$.

\begin{condition}[Convergence of Finite-Population Means] \label{cond: convergence of finite-population means}
Let $\mu_{i}=E(\widehat{\theta}_{\tau, \text{oracle}, i})$ and $\nu_{i}=\text{Var}(\widehat{\theta}_{\tau, \text{oracle}, i})$. Then, as $N\rightarrow \infty$, we have: (i) $N^{-1}\sum_{i=1}^{N}\mu_{i}$ and $N^{-1}\sum_{i=1}^{N}\mu_{i}^{2}$ converge to finite limits; and (ii) $N^{-1}\sum_{i=1}^{N}\nu_{i}$ converges to a finite positive limit.
\end{condition}

Therefore, Condition 1 in the main text and Condition~\ref{cond: convergence of finite-population means} are sufficient for establishing the validity of the oracle confidence interval $\big[\widehat{\theta}_{\tau, \text{oracle}}-z_{1-\alpha/2}\cdot \widehat{\mathcal{V}}_{\tau, \text{oracle}}^{1/2}, \, \widehat{\theta}_{\tau, \text{oracle}}+z_{1-\alpha/2}\cdot \widehat{\mathcal{V}}_{\tau, \text{oracle}}^{1/2} \big]$ for $\theta_{\tau}$. For additional examples of established results on the construction and validity of oracle confidence intervals or confidence sets, that is, intervals or sets based on oracle propensity scores, see the aforementioned works, as well as textbooks and survey articles on design-based inference under known or oracle propensity scores (e.g., \citealp{imbens2015causal, athey2017econometrics, li2017general, ding2024first, shi2024some}).

\subsection*{A.2: Some Example Results on the Lipschitz Continuity of Oracle Estimators}

A key ingredient in the proofs of Theorems~2--7 is the Lipschitz continuity of $\widehat{\theta}_{\text{oracle}}$ and $N\cdot \widehat{\mathcal{V}}_{\text{oracle}}$ with respect to $\mathcal{P}$ under the normalized $L_{2}$ norm, that is, $N^{-1/2}\|\cdot\|_{2}$, as stated in Condition~2c in the main text. As will be shown in this section, this condition holds naturally for a wide range of commonly used choices of $\widehat{\theta}_{\text{oracle}}$ and $N\cdot \widehat{\mathcal{V}}_{\text{oracle}}$, including those reviewed in Sections~2.1 and 5.1 of the main text. For proving Theorems~2--4, however, it suffices to impose the following weaker condition, which requires Lipschitz continuity of $\widehat{\theta}_{\text{oracle}}$ and $N\cdot \widehat{\mathcal{V}}_{\text{oracle}}$ with respect to $\mathcal{P}$ under the $L_{\infty}$ norm.

\begin{condition}[Lipschitz Continuity of $\widehat{\theta}_{\text{oracle}}$ and $N\cdot \widehat{\mathcal{V}}_{\text{oracle}}$ with Respect to $\mathcal{P}$ in $\|\cdot\|_{\infty}$] \label{cond: Lipschitz infinity} Let $\widehat{\theta}_{\text{oracle}}$ and $\widehat{\theta}_{\text{oracle}}^{\prime}$ denote the values of the oracle estimator when the oracle propensity score vector is set to $\mathcal{P}$ and $\mathcal{P}^{\prime}$, respectively, and define $\widehat{\mathcal{V}}_{\text{oracle}}$ and $\widehat{\mathcal{V}}_{\text{oracle}}^{\prime}$ analogously. Then, under bounded outcomes, namely Condition~1b, and for any fixed $\delta\in(0,0.5)$, there exist bounded constants $L_{\theta}(\delta),L_{v}(\delta)>0$ such that, for any $\mathcal{P},\mathcal{P}^{\prime}\in[\delta,1-\delta]^N$, we have $|\widehat{\theta}_{\text{oracle}}-\widehat{\theta}_{\text{oracle}}^{\prime}|\leq L_{\theta}(\delta)\|\mathcal{P}-\mathcal{P}^{\prime}\|_{\infty}$ and $N|\widehat{\mathcal{V}}_{\text{oracle}}-\widehat{\mathcal{V}}_{\text{oracle}}^{\prime}|\leq L_{v}(\delta)\|\mathcal{P}-\mathcal{P}^{\prime}\|_{\infty}$. 
\end{condition} 

We next show that the oracle IPW estimator $\widehat{\theta}_{\tau, \text{oracle}}$ and its corresponding oracle design-based variance estimator $\widehat{\mathcal{V}}_{\tau, \text{oracle}}$ satisfy Condition~\ref{cond: Lipschitz infinity}. Similar arguments apply to many other commonly used design-based estimators and their corresponding variance estimators, including those reviewed in Sections~2.1 and 5.1.

\begin{proposition}\label{prop: lipschitz in L infinity}
Condition~\ref{cond: Lipschitz infinity} holds for the oracle IPW estimator $\widehat{\theta}_{\tau, \text{oracle}}$ and the corresponding oracle design-based variance estimator $\widehat{\mathcal{V}}_{\tau, \text{oracle}}$.  
\end{proposition}

\vspace{-0.5cm}

\begin{proof}: Let $\widehat{\theta}_{\tau, \text{oracle}}$ (or $\widehat{\theta}_{\tau, \text{oracle}}^{\prime}$) and $\widehat{\mathcal{V}}_{\tau, \text{oracle}}$ (or $\widehat{\mathcal{V}}_{\tau, \text{oracle}}^{\prime}$) denote the values of $\widehat{\theta}_{\tau, \text{oracle}}$ and $\widehat{\mathcal{V}}_{\tau, \text{oracle}}$ when the oracle propensity score vector is set to $\mathcal{P}=(p_{1}, \dots, p_{N})$ (or $\mathcal{P}^{\prime}=(p_{1}^{\prime}, \dots, p_{N}^{\prime})$), respectively. Under bounded outcomes, there exists a finite constant $C_{y}>0$ such that $|Y_{i}|\leq C_{y}$ for all $i$. Then, for any $\mathcal{P}, \mathcal{P}^{\prime}\in [\delta, 1-\delta]^{N}$, we have
\begin{align*}
\left|\widehat{\theta}_{\tau, \text{oracle}, i}-\widehat{\theta}_{\tau, \text{oracle}, i}^{\prime}\right|&\leq \left|\frac{Z_{i}Y_{i}}{p_{i}}-\frac{Z_{i}Y_{i}}{p^{\prime}_{i}}\right|+\left|\frac{(1-Z_{i})Y_{i}}{1-p_{i}}-\frac{(1-Z_{i})Y_{i}}{1-p^{\prime}_{i}}\right|\\
&\leq \frac{Z_{i}|Y_{i}|}{p_{i}p^{\prime}_{i}} \max_{1\leq i\leq N}\left|p_i-p^{\prime}_{i}\right|+\frac{(1-Z_{i})|Y_{i}|}{(1-p_{i})(1-p^{\prime}_{i})} \max_{1\leq i\leq N}\left|p_i-p^{\prime}_{i}\right|\\
&\leq \frac{C_{y}}{\delta^{2}} \left\|\mathcal{P}-\mathcal{P}^{\prime}  \right\|_{\infty}.
\end{align*}
Therefore, we have
\begin{equation*}
    \left|\widehat{\theta}_{\tau, \text{oracle}}-\widehat{\theta}_{\tau, \text{oracle}}^{\prime}\right|\leq \frac{1}{N}\sum_{i=1}^{N}\left|\widehat{\theta}_{\tau, \text{oracle}, i}-\widehat{\theta}_{\tau, \text{oracle}, i}^{\prime}\right|\leq \frac{C_{y}}{\delta^{2}} \left\|\mathcal{P}-\mathcal{P}^{\prime}  \right\|_{\infty}.
\end{equation*}
Also, note that
\begin{align*}
   &\quad \left|\widehat{\mathcal{V}}_{\tau, \text{oracle}}-\widehat{\mathcal{V}}_{\tau, \text{oracle}}^{\prime}\right|\\
   &=\left | \frac{1}{N(N-1)}\sum_{i=1}^{N}\left (\widehat{\theta}_{\tau, \text{oracle}, i}-\widehat{\theta}_{\tau, \text{oracle}}\right)^{2}- \frac{1}{N(N-1)}\sum_{i=1}^{N}\left (\widehat{\theta}_{\tau, \text{oracle}, i}^{\prime}-\widehat{\theta}_{\tau, \text{oracle}}^{\prime}\right)^{2}  \right | \\
    &\leq \frac{1}{N(N-1)}\sum_{i=1}^{N}\left | \left (\widehat{\theta}_{\tau, \text{oracle}, i}-\widehat{\theta}_{\tau, \text{oracle}}\right)^{2}- \left (\widehat{\theta}_{\tau, \text{oracle}, i}^{\prime}-\widehat{\theta}_{\tau, \text{oracle}}^{\prime}\right)^{2}  \right |\\
    &=\frac{1}{N(N-1)}\sum_{i=1}^{N}\Big\{\left | \left (\widehat{\theta}_{\tau, \text{oracle}, i}-\widehat{\theta}_{\tau, \text{oracle}}\right) -\left (\widehat{\theta}_{\tau, \text{oracle}, i}^{\prime}-\widehat{\theta}_{\tau, \text{oracle}}^{\prime}\right) \right |\\
    &\quad \quad \quad  \quad \quad \quad \quad \quad \quad \times \left | \left (\widehat{\theta}_{\tau, \text{oracle}, i}-\widehat{\theta}_{\tau, \text{oracle}}\right) +\left (\widehat{\theta}_{\tau, \text{oracle}, i}^{\prime}-\widehat{\theta}_{\tau, \text{oracle}}^{\prime}\right) \right | \Big\},
\end{align*}
where we have
\begin{align*}
    \left | \left (\widehat{\theta}_{\tau, \text{oracle}, i}-\widehat{\theta}_{\tau, \text{oracle}}\right) -\left (\widehat{\theta}_{\tau, \text{oracle}, i}^{\prime}-\widehat{\theta}_{\tau, \text{oracle}}^{\prime}\right) \right |
   &\leq \left | \widehat{\theta}_{\tau, \text{oracle}, i}-\widehat{\theta}_{\tau, \text{oracle}, i}^{\prime}\right| + \left |\widehat{\theta}_{\tau, \text{oracle}}-\widehat{\theta}_{\tau, \text{oracle}}^{\prime} \right |\\
   &\leq \frac{2C_{y}}{\delta^{2}} \left\|\mathcal{P}-\mathcal{P}^{\prime}  \right\|_{\infty},
\end{align*}
and 
\begin{align*}
     \left | \left (\widehat{\theta}_{\tau, \text{oracle}, i}-\widehat{\theta}_{\tau, \text{oracle}}\right) +\left (\widehat{\theta}_{\tau, \text{oracle}, i}^{\prime}-\widehat{\theta}_{\tau, \text{oracle}}^{\prime}\right) \right |&\leq \left | \widehat{\theta}_{\tau, \text{oracle}, i} \right|+\left|\widehat{\theta}_{\tau, \text{oracle}}\right| +\left |\widehat{\theta}_{\tau, \text{oracle}, i}^{\prime}\right|+\left|\widehat{\theta}_{\tau, \text{oracle}}^{\prime} \right | \\
    &\leq \frac{4C_{y}}{\delta}.
\end{align*}
Therefore, when $N\geq 2$, we have 
\begin{equation*}
   \left|\widehat{\mathcal{V}}_{\tau, \text{oracle}}-\widehat{\mathcal{V}}_{\tau, \text{oracle}}^{\prime}\right|\leq \frac{8C_{y}^{2}}{\delta^{3}}\cdot \frac{1}{N-1} \left\|\mathcal{P}-\mathcal{P}^{\prime}  \right\|_{\infty}\leq \frac{16C_{y}^{2}}{\delta^{3}}\cdot \frac{1}{N} \left\|\mathcal{P}-\mathcal{P}^{\prime}  \right\|_{\infty}.
\end{equation*}
That is, when $N\geq 2$, we have
\begin{equation*}
   N\left|\widehat{\mathcal{V}}_{\tau, \text{oracle}}-\widehat{\mathcal{V}}_{\tau, \text{oracle}}^{\prime}\right|\leq \frac{16C_{y}^{2}}{\delta^{3}}\cdot \left\|\mathcal{P}-\mathcal{P}^{\prime}  \right\|_{\infty}.
\end{equation*}
Therefore, the desired conclusion follows from setting $L_{\theta}(\delta)=\frac{C_{y}}{\delta^{2}}$ and $L_{v}(\delta)=\frac{16C_{y}^{2}}{\delta^{3}}$.
\end{proof}

A key condition used in the proofs of Theorems~5--7 is the Lipschitz continuity of $\widehat{\theta}_{\text{oracle}}$ and $N\cdot\widehat{\mathcal{V}}_{\text{oracle}}$ with respect to $\mathcal{P}$ under the normalized $\|\cdot\|_{2}$ norm, as stated in Condition~2c in the main text and formalized in Condition~\ref{cond: Lipschitz in L2 norm}. Since $N^{-1/2}\|\mathcal{P}-\mathcal{P}^{\prime}\|_{2}\leq \|\mathcal{P}-\mathcal{P}^{\prime}\|_{\infty}$, Condition~\ref{cond: Lipschitz in L2 norm} (i.e., Condition~2c in the main text) is stronger than Condition~\ref{cond: Lipschitz infinity}. Nevertheless, it holds naturally for a wide range of commonly used choices of $\widehat{\theta}_{\text{oracle}}$ and $N\cdot\widehat{\mathcal{V}}_{\text{oracle}}$, including those reviewed in Sections~2.1 and 5.1. We next show that the oracle IPW estimator $\widehat{\theta}_{\tau, \text{oracle}}$ and its corresponding oracle design-based variance estimator $\widehat{\mathcal{V}}_{\tau, \text{oracle}}$ satisfy Condition~\ref{cond: Lipschitz in L2 norm}. Similar arguments apply to many other design-based estimators and their corresponding variance estimators, including those reviewed in Sections~2.1 and 5.1.

\begin{condition}[The Detailed Version of Condition~2c in the Main Text]\label{cond: Lipschitz in L2 norm} 
 Under the setup and notations in Condition~\ref{cond: Lipschitz infinity}, we further assume that for bounded outcomes (i.e., Condition~1b) and for any fixed $\delta\in (0,0.5)$, there exist bounded constants $L^{\prime}_{\theta}(\delta), L^{\prime}_{v}(\delta)>0$ such that $|\widehat{\theta}_{\text{oracle}}-\widehat{\theta}_{\text{oracle}}^{\prime}|\leq L^{\prime}_{\theta}(\delta) \cdot N^{-1/2} \|\mathcal{P}-\mathcal{P}^{\prime}\|_{2}$ and $N|\widehat{\mathcal{V}}_{\text{oracle}}-\widehat{\mathcal{V}}^{\prime}_{\text{oracle}}|\leq L^{\prime}_{v}(\delta) \cdot N^{-1/2} \| \mathcal{P}-\mathcal{P}^{\prime}\|_{2}$ for any $\mathcal{P}, \mathcal{P}^{\prime}\in [\delta, 1-\delta]^{N}$.
\end{condition}

\vspace{-0.5cm}

\begin{proposition}\label{prop: lipschitz in L2 norm}
Condition~\ref{cond: Lipschitz in L2 norm} holds for the oracle IPW estimator $\widehat{\theta}_{\tau, \text{oracle}}$ and corresponding oracle design-based variance estimator $\widehat{\mathcal{V}}_{\tau, \text{oracle}}$.  
\end{proposition}

\vspace{-0.5cm}
\begin{proof}: Under the notations used in the proof of Proposition~\ref{prop: lipschitz in L infinity}, we have
\begin{align*}
  &\quad \  \left|\widehat{\theta}_{\tau, \text{oracle}}-\widehat{\theta}_{\tau, \text{oracle}}^{\prime}\right|\\
     &\leq \frac{1}{N}\sum_{i=1}^{N}\left | \widehat{\theta}_{\tau, \text{oracle}, i}-\widehat{\theta}_{\tau, \text{oracle}, i}^{\prime}\right|\\
&\leq\frac{1}{N}\sum_{i=1}^{N}\left|\frac{Z_{i}Y_{i}}{p_{i}}-\frac{Z_{i}Y_{i}}{p^{\prime}_{i}}\right|+\frac{1}{N}\sum_{i=1}^{N}\left|\frac{(1-Z_{i})Y_{i}}{1-p_{i}}-\frac{(1-Z_{i})Y_{i}}{1-p^{\prime}_{i}}\right|\\
&= \frac{1}{N}\sum_{i=1}^{N}Z_{i}|Y_{i}|\left|\frac{p_{i}-p^{\prime}_{i}}{{p_{i}p^{\prime}_{i}}}\right|+\frac{1}{N}\sum_{i=1}^{N}(1-Z_{i})|Y_{i}| \left|\frac{p_{i}-p^{\prime}_{i}}{{\left(1-p_{i}\right)\left(1-p^{\prime}_{i}\right)}}\right|\\
&\leq \sqrt{\frac{1}{N}\sum_{i=1}^{N}Z^{2}_{i}Y^{2}_{i} }\sqrt{\frac{1}{N}\sum_{i=1}^{N}\left|\frac{p_{i}-p^{\prime}_{i}}{{p_{i}p^{\prime}_{i}}}\right|^{2} }\\
&\quad \quad +\sqrt{\frac{1}{N}\sum_{i=1}^{N}\left(1-Z_{i}\right)^{2}Y^{2}_{i}}\sqrt{\frac{1}{N}\sum_{i=1}^{N}\left|\frac{p_{i}-p^{\prime}_{i}}{{\left(1-p_{i}\right)\left(1-p^{\prime}_{i}\right)}}\right|^{2}} \\
&\leq \frac{2C_{y}}{\delta^{2}\sqrt{N}}\left\|\mathcal{P}- \mathcal{P}^{\prime} \right\|_{2}.
\end{align*}
Also, note that
\begin{align*}
   &\quad \left|\widehat{\mathcal{V}}_{\tau, \text{oracle}}-\widehat{\mathcal{V}}_{\tau, \text{oracle}}^{\prime}\right|\\
   &=\left | \frac{1}{N(N-1)}\sum_{i=1}^{N}\left (\widehat{\theta}_{\tau, \text{oracle}, i}-\widehat{\theta}_{\tau, \text{oracle}}\right)^{2}- \frac{1}{N(N-1)}\sum_{i=1}^{N}\left (\widehat{\theta}_{\tau, \text{oracle}, i}^{\prime}-\widehat{\theta}_{\tau, \text{oracle}}^{\prime}\right)^{2}  \right | \\
    &\leq \frac{1}{N(N-1)}\sum_{i=1}^{N}\left | \left (\widehat{\theta}_{\tau, \text{oracle}, i}-\widehat{\theta}_{\tau, \text{oracle}}\right)^{2}- \left (\widehat{\theta}_{\tau, \text{oracle}, i}^{\prime}-\widehat{\theta}_{\tau, \text{oracle}}^{\prime}\right)^{2}  \right |\\
    &=\frac{1}{N(N-1)}\sum_{i=1}^{N}\Big\{\left | \left (\widehat{\theta}_{\tau, \text{oracle}, i}-\widehat{\theta}_{\tau, \text{oracle}}\right) -\left (\widehat{\theta}_{\tau, \text{oracle}, i}^{\prime}-\widehat{\theta}_{\tau, \text{oracle}}^{\prime}\right) \right |\\
    &\quad \quad \quad  \quad \quad \quad \quad \quad \quad \times \left | \left (\widehat{\theta}_{\tau, \text{oracle}, i}-\widehat{\theta}_{\tau, \text{oracle}}\right) +\left (\widehat{\theta}_{\tau, \text{oracle}, i}^{\prime}-\widehat{\theta}_{\tau, \text{oracle}}^{\prime}\right) \right | \Big\}\\
       &\leq \frac{1}{N(N-1)}\sum_{i=1}^{N}\Big\{\left | \left (\widehat{\theta}_{\tau, \text{oracle}, i}-\widehat{\theta}_{\tau, \text{oracle}}\right) -\left (\widehat{\theta}_{\tau, \text{oracle}, i}^{\prime}-\widehat{\theta}_{\tau, \text{oracle}}^{\prime}\right) \right |\\
    &\quad \quad \quad  \quad \quad \quad \quad \quad \quad \times \left ( \left |\widehat{\theta}_{\tau, \text{oracle}, i}\right|+\left|\widehat{\theta}_{\tau, \text{oracle}}\right| +\left |\widehat{\theta}_{\tau, \text{oracle}, i}^{\prime}\right|+\left|\widehat{\theta}_{\tau, \text{oracle}}^{\prime}\right| \right ) \Big\}\\
      &\leq \frac{4C_{y}}{\delta N(N-1)}\sum_{i=1}^{N}\left | \left (\widehat{\theta}_{\tau, \text{oracle}, i}-\widehat{\theta}_{\tau, \text{oracle}}\right) -\left (\widehat{\theta}_{\tau, \text{oracle}, i}^{\prime}-\widehat{\theta}_{\tau, \text{oracle}}^{\prime}\right) \right |\\
    &\leq \frac{4C_{y}}{\delta(N-1)}\left\{ \left |\widehat{\theta}_{\tau, \text{oracle}}-\widehat{\theta}_{\tau, \text{oracle}}^{\prime} \right |+\frac{1}{N}\sum_{i=1}^{N}\left | \widehat{\theta}_{\tau, \text{oracle}, i}-\widehat{\theta}_{\tau, \text{oracle}, i}^{\prime}\right| \right\}\\
    &\leq \frac{16C^{2}_{y}}{\delta^{3}(N-1)}\frac{1}{\sqrt{N}}\left\|\mathcal{P}- \mathcal{P}^{\prime} \right\|_{2}.
\end{align*}
Therefore, when $N\geq 2$, we have 
\begin{equation*}
   \left|\widehat{\mathcal{V}}_{\tau, \text{oracle}}-\widehat{\mathcal{V}}_{\tau, \text{oracle}}^{\prime}\right|\leq \frac{16C^{2}_{y}}{\delta^{3}(N-1)}\frac{1}{\sqrt{N}}\left\|\mathcal{P}- \mathcal{P}^{\prime} \right\|_{2}\leq  \frac{32C^{2}_{y}}{\delta^{3}N}\frac{1}{\sqrt{N}}\left\|\mathcal{P}- \mathcal{P}^{\prime} \right\|_{2}.
\end{equation*}
That is, when $N\geq 2$, we have
\begin{equation*}
  N \left|\widehat{\mathcal{V}}_{\tau, \text{oracle}}-\widehat{\mathcal{V}}_{\tau, \text{oracle}}^{\prime}\right|\leq  \frac{32C^{2}_{y}}{\delta^{3}}\frac{1}{\sqrt{N}}\left\|\mathcal{P}- \mathcal{P}^{\prime} \right\|_{2}.
\end{equation*}
Therefore, the desired conclusion follows from setting $L^{\prime}_{\theta}(\delta)=\frac{2C_{y}}{\delta^{2}}$ and $L^{\prime}_{v}(\delta)=\frac{32C^{2}_{y}}{\delta^{3}}$.
\end{proof}

\subsection*{A.3: Proofs for Parametric Propensity Score Propagation}

To prove the theoretical results related to parametric propensity score propagation, we require the asymptotic normality of MLE for GLM (i.e., Condition 3b in the main text). Such asymptotic normality has been shown to hold under mild regularity conditions by the classic MLE theory for GLMs under fixed design \citep{haberman1977maximum, fahrmeir1985consistency}. For example, following \citet{fahrmeir1985consistency}, Condition 3 in the main text can be implied by Conditions~\ref{cond: Regularity of Covariates}--\ref{cond: Convergence of the Sample Version of the Fisher Information Matrix} listed below.

\begin{condition}[Regularity of Covariates]\label{cond: Regularity of Covariates} The support $B_{x}$ of covariates $\mathbf{x}_{i}$ is bounded (i.e., $\|\mathbf{x}_{i}\|_{\infty}\leq C_{x}$ for some constant $C_{x}>0$), connected, and convex. Also, the sample covariance matrix of the $d$-dimensional covariates $\mathbf{x}_{i}$ has full rank for any sufficiently large $N$.
\end{condition}

\begin{condition}[Smoothness of the Link Function]\label{cond: Smoothness of the Link Function} The link function $\Psi: \mathbbm{R} \rightarrow (0,1)$ is one-to-one and twice continuously differentiable (i.e., $\Psi \in C^{2}$), and $|\Psi^{\prime}| \leq C_{\Psi^{\prime}}$ for some positive constant $C_{\Psi^{\prime}}>0$. 
\end{condition}

\begin{condition}[Regularity of the Fisher Information Matrix]\label{cond: Regularity of the Fisher Information Matrix} The scaled Fisher information matrix $\mathcal{I}_{N}=\mathcal{I}_{N}(\boldsymbol{\beta})$ is positive definite for any $\boldsymbol{\beta}\in B_{\boldsymbol{\beta}}$ and sufficiently large $N$, where $B_{\boldsymbol{\beta}}\subset \mathbbm{R}^{d}$ is an open, convex admissible set containing the true parameters $\boldsymbol{\beta}$ as an interior point. Also, as $N\rightarrow \infty$, we have $N\cdot \lambda_{\min}(\mathcal{I}_{N})\rightarrow \infty$ and $N^{-1}\cdot \text{tr}(\mathbf{X}^{T}\mathcal{I}_{N}\mathbf{X})\rightarrow 0$, where $\lambda_{\min}(\cdot)$ and $\text{tr}(\cdot)$ represent the smallest eigenvalue and the trace of some matrix, respectively. In addition, $\mathcal{I}_{N}=\mathcal{I}_{N}(\boldsymbol{\beta})$ is locally stable around the true parameters: for every fixed \(\delta>0\),
\[
\max_{\boldsymbol{\beta}'\in \mathcal N_N(\delta)}
\left\|
\mathcal{I}_N^{-1/2} \mathcal{I}_N(\boldsymbol{\beta}') \mathcal{I}_N^{-1/2} -\mathbf{I}_{d \times d}
\right\|\to 0,
\]
where
\[
\mathcal N_N(\delta)
=
\left\{
\boldsymbol{\beta}'\in B_{\boldsymbol{\beta}}:
\left\| \sqrt{N}\cdot \mathcal{I}_N^{1/2}(\boldsymbol{\beta}'-\boldsymbol{\beta})\right\|\leq \delta
\right\}.
\]
This local stability condition is the analogue, under the present notation, of Condition (N) in Section 3 of \citet{fahrmeir1985consistency}.
    
\end{condition}

\begin{condition}[Convergence of the Fisher Information Matrix]\label{cond: Convergence of the Sample Version of the Fisher Information Matrix} Let $\mathcal{I}_{N}=\mathcal{I}_{N}(\boldsymbol{\beta})$ denote the (scaled) Fisher information matrix given the $N$ study units, and let $\widehat{\mathcal{I}}_{N}=\mathcal{I}_{N}(\widehat{\boldsymbol{\beta}})$ denote the sample version of the (scaled) Fisher information matrix based on $(\mathbf{Z}, \mathbf{X})$. Then, as $N\rightarrow \infty$, $\mathcal{I}_{N}$ converges to some positive definite matrix and $\widehat{\mathcal{I}}_{N}$ converges in probability to some positive definite $d\times d$ matrix. Consequently, the covariance estimator used in the parametric regeneration step satisfies: $\widehat\Omega_N=\widehat{\mathcal{I}}_{N}^{-1}\xrightarrow{p}\Omega$, where $\Omega$ is a positive definite matrix.

\end{condition}

We first state and prove the following key lemma (Lemma~\ref{lem: rate of beta parametric}), which establishes a design-based upper bound for $\min_{m\in [M]} \|\widehat{\boldsymbol{\beta}}^{(m)}-\boldsymbol{\beta}\|_{\infty}$ (for the unrestricted-union case) and $\min_{m\in \mathcal{M}(\alpha^{\prime})} \|\widehat{\boldsymbol{\beta}}^{(m)}-\boldsymbol{\beta}\|_{\infty}$ (for the restricted-union case), in terms of the sample size $N$, the number of regeneration runs $M$, and the covariate dimension $d$ (including the intercept term). The proof of Lemma~\ref{lem: rate of beta parametric} builds on an extension of a sampling argument in \citet{guo2023causal}--originally developed for instrumental variable regressions in super-population settings--to the setting of parametric propensity score models under fixed design, with classic fixed-design MLE theory for GLMs \citep{fahrmeir1985consistency} providing the necessary asymptotic control.

\begin{lemma}\label{lem: rate of beta parametric}
For an arbitrary but fixed $\alpha^{\prime}\in (0, \alpha)$, we define a positive constant 
\begin{equation*}
    \widetilde{C}_{\alpha^{\prime}}=2^{\frac{1-d}{d}} \left(2\pi \right)^{\frac{1}{2}}\left\{\lambda_{\max}\left(\Omega \right)+\frac{1}{2}\lambda_{\min}(\Omega) \right\}^{\frac{1}{2}} \exp\left \{\frac{3\lambda_{\max}(\Omega)}{2\lambda_{\min}(\Omega)}\cdot z^{2}_{1-\frac{\alpha^{\prime}}{2d}} \right\},
\end{equation*}
where $\lambda_{\max}\left(\Omega \right)$ and $\lambda_{\min}\left(\Omega \right)$ denote the largest and smallest eigenvalues of $\Omega$, respectively. Under the ignorability assumption (i.e., no unobserved covariates) and Condition~3, we have
\begin{equation*}
    \liminf_{N\rightarrow \infty }\lim_{M\rightarrow \infty} P\left(\min_{m\in [M]}\left\| \widehat{\boldsymbol{\beta}}^{(m)}-\boldsymbol{\beta}\right\|_{\infty}  \leq \frac{\widetilde{C}_{\alpha^{\prime}}}{\sqrt{N}} \cdot \left(\frac{\log N}{M}\right)^{\frac{1}{d}}  \right)\geq 1-\alpha^{\prime}
\end{equation*}
and
\begin{equation*}
    \liminf_{N\rightarrow \infty }\lim_{M\rightarrow \infty} P\left(\min_{m\in \mathcal{M}(\alpha^{\prime})}\left\| \widehat{\boldsymbol{\beta}}^{(m)}-\boldsymbol{\beta}\right\|_{\infty}  \leq \frac{\widetilde{C}_{\alpha^{\prime}}}{\sqrt{N}} \cdot \left(\frac{\log N}{M}\right)^{\frac{1}{d}}   \right)\geq 1-\alpha^{\prime}.
\end{equation*}
\end{lemma}

\begin{proof}[Proof of Lemma~\ref{lem: rate of beta parametric}:]

Recall that in the finite-population (design-based) inference framework, $\mathbf{X}$ are fixed covariates and $\mathbf{Z}$ are random treatment indicators. Therefore, throughout the proof (and the proofs of other results), all the arguments are conditional on fixed $\mathbf{X}$ (i.e., adopt the fixed design regime). For notational simplicity, we omit the notation of ``conditional on $\mathbf{X}$'' in most places. Still, readers should be aware that all the distributions, probabilities, and events are conditional on the given $\mathbf{X}$. 

Define the event $A_{1}=\{\mathbf{Z}: \|\widehat{\Omega}_{N}-\Omega\|_{2}<c_{1}\}$, where $\|\widehat{\Omega}_{N}-\Omega\|_{2}$ is the spectral norm of $\widehat{\Omega}_{N}-\Omega$ and $c_{1}\in (0, \frac{\lambda_{\min}(\Omega)}{2})$ is some small positive constant. Then, for $\mathbf{Z}\in A_{1}$, we have 
\begin{equation*}
  2\lambda_{\max}(\Omega)\cdot \mathbf{I}_{d \times d} \succ  \Omega + c_{1}\cdot \mathbf{I}_{d \times d}  \succ \widehat{\Omega}_{N} \succ \Omega - c_{1}\cdot \mathbf{I}_{d \times d}  \succ \frac{\lambda_{\min}(\Omega)}{2}\cdot \mathbf{I}_{d \times d}.
\end{equation*} 
 Define $\widehat{\boldsymbol{\epsilon}}=\sqrt{N}(\widehat{\boldsymbol{\beta}}-\boldsymbol{\beta})$ and $\widehat{\boldsymbol{\epsilon}}^{(m)}= \sqrt{N}(\widehat{\boldsymbol{\beta}}-\widehat{\boldsymbol{\beta}}^{(m)})$ for $m \in [M]$. Recall that $\widehat{\boldsymbol{\epsilon}}^{(m)}\mid \mathbf{Z} \overset{\text{iid}}{\sim} N(\mathbf{0}, \widehat{\Omega}_{N})$. Therefore, if we let $f\left(\widehat{\boldsymbol{\epsilon}}^{(m)}=\mathbf{t}\mid \mathbf{Z} \right)$ denote the conditional density of $\widehat{\boldsymbol{\epsilon}}^{(m)}$ given $\mathbf{Z}$, we have: for all $\mathbf{Z}\in A_{1}$, 
\begin{align}
     f\left(\widehat{\boldsymbol{\epsilon}}^{(m)}=\mathbf{t}\mid \mathbf{Z} \right)&=\left(2\pi \right)^{-\frac{d}{2}}|\widehat{\Omega}_{N}|^{-\frac{1}{2}}\exp\left (-\frac{1}{2}\mathbf{t}^{T}\widehat{\Omega}_{N}^{-1} \mathbf{t}  \right ) \nonumber \\
     &\geq \left(2\pi \right)^{-\frac{d}{2}}|\Omega + c_{1}\cdot\mathbf{I}_{d \times d}|^{-\frac{1}{2}}\exp\left \{-\frac{1}{2}\mathbf{t}^{T}\left(\Omega - c_{1}\cdot\mathbf{I}_{d \times d} \right)^{-1} \mathbf{t} \right \}:=\widetilde{f}(\mathbf{t}).
\end{align}
Also, for any fixed $\alpha^{\prime}\in (0,\alpha)$, we define the event 
\begin{equation*}
  A_{2} = \left \{\mathbf{Z}: \max_{k\in [d]} \left|\frac{\widehat{\beta}_{k}-\beta_{k}}{\sqrt{\widehat{\sigma}_{kk}^{2}/N}}\right|\leq z_{1-\frac{\alpha^{\prime}}{2d}} \right\},
\end{equation*}
where $\widehat{\sigma}_{kk}^{2}$ is the $k$th diagonal entry of $\widehat{\Omega}_{N}$. Let $\sigma_{kk}^{2}$ denote the $k$th diagonal entry of $\Omega$. Under Condition~3, the asymptotic normality (4) holds, so we have $\liminf_{N\rightarrow \infty}P(A_{2})\geq 1-\alpha^{\prime}$ by the union bound. Also, because Condition~3b states that $\widehat{\Omega}_{N}\xrightarrow{p} \Omega$ as $N\rightarrow \infty$, we have $\liminf_{N\rightarrow \infty}P(A_{1})=1$. Therefore, we have
\begin{align*}
    \liminf_{N\rightarrow\infty} P(A_{1}\cap A_{2})&\geq  \liminf_{N\rightarrow\infty}\left\{P(A_{1})+P(A_{2})-1\right\}\geq 1-\alpha^{\prime}.
\end{align*}
Also, for any $\mathbf{Z}\in A_{1}\cap A_{2}$, we have
\begin{align*}
    \widehat{\boldsymbol{\epsilon}}^{T}\left(\Omega - c_{1}\cdot\mathbf{I}_{d \times d} \right)^{-1} \widehat{\boldsymbol{\epsilon}}&\leq \frac{2}{\lambda_{\min}(\Omega)}\widehat{\boldsymbol{\epsilon}}^{T}\mathbf{I}_{d \times d}\widehat{\boldsymbol{\epsilon}}\\
    &=\frac{2}{\lambda_{\min}(\Omega)}\sum_{k=1}^{d}N\left(\widehat{\beta}_{k}-\beta_{k}\right)^{2}\\
    &\leq \frac{2}{\lambda_{\min}(\Omega)}\sum_{k=1}^{d}\widehat{\sigma}_{kk}^{2} \cdot z^{2}_{1-\frac{\alpha^{\prime}}{2d}}\\
    &\leq \frac{2}{\lambda_{\min}(\Omega)}\left(\sum_{k=1}^{d}\sigma_{kk}^{2}+ d\cdot \frac{\lambda_{\min}(\Omega)}{2}\right) \cdot z^{2}_{1-\frac{\alpha^{\prime}}{2d}}\\
    &\leq  \frac{3d\lambda_{\max}(\Omega)}{\lambda_{\min}(\Omega)}\cdot z^{2}_{1-\frac{\alpha^{\prime}}{2d}}.
\end{align*}
Therefore, for any $\mathbf{Z}\in A_{1}\cap A_{2}$, we have 
\begin{align*}
\widetilde{f}\left(\widehat{\boldsymbol{\epsilon}}\right)&=\left(2\pi \right)^{-\frac{d}{2}}|\Omega + c_{1}\cdot\mathbf{I}_{d \times d}|^{-\frac{1}{2}}\exp\left \{-\frac{1}{2}\widehat{\boldsymbol{\epsilon}}^{T}\left(\Omega - c_{1}\cdot\mathbf{I}_{d \times d} \right)^{-1} \widehat{\boldsymbol{\epsilon}} \right \}\\
    &\geq \left(2\pi \right)^{-\frac{d}{2}}\left\{\lambda_{\max}\left(\Omega \right)+\frac{1}{2}\lambda_{\min}(\Omega) \right\}^{-\frac{d}{2}} \exp\left \{-\frac{3d\lambda_{\max}(\Omega)}{2\lambda_{\min}(\Omega)}\cdot z^{2}_{1-\frac{\alpha^{\prime}}{2d}} \right\}:=\widetilde{c}(\alpha^{\prime}).
\end{align*}
If $m\notin \mathcal{M}(\alpha^{\prime})$, there exists some $k_{0}\in [d]$ such that
\begin{equation*}
    \left|\frac{\widehat{\beta}_{k_{0}}^{(m)}-\widehat{\beta}_{k_{0}}}{\sqrt{\widehat{\sigma}_{k_{0}k_{0}}^{2}/N}}\right|> 1.01\cdot z_{1-\frac{\alpha^{\prime}}{2d}}.
\end{equation*}
Meanwhile, for any $\mathbf{Z}\in A_{2}$, we have
\begin{equation*}
    \left|\frac{\widehat{\beta}_{k_{0}}^{(m)}-\beta_{k_{0}}}{\sqrt{\widehat{\sigma}_{k_{0}k_{0}}^{2}/N}}\right|>\left |\left|\frac{\widehat{\beta}_{k_{0}}^{(m)}-\widehat{\beta}_{k_{0}}}{\sqrt{\widehat{\sigma}_{k_{0}k_{0}}^{2}/N}}\right|- \left|\frac{\widehat{\beta}_{k_{0}}-\beta_{k_{0}}}{\sqrt{\widehat{\sigma}_{k_{0}k_{0}}^{2}/N}}\right| \right | > 0.01\cdot z_{1-\frac{\alpha^{\prime}}{2d}}.
\end{equation*}
Therefore, for any $\mathbf{Z}\in A_{1}\cap A_{2}$, we have 
\begin{equation*}
    \min_{m\notin \mathcal{M}(\alpha^{\prime}) } \left \|\widehat{\boldsymbol{\epsilon}}^{(m)}-\widehat{\boldsymbol{\epsilon}}   \right \|_{\infty}\geq 0.01 \cdot \widehat{\sigma}_{k_{0}k_{0}}\cdot z_{1-\frac{\alpha^{\prime}}{2d}}.
\end{equation*}
We define 
\begin{equation*}
    \widetilde{\epsilon}(N, M, \alpha^{\prime})=\frac{1}{2}\left\{\frac{2\log N}{\widetilde{c}(\alpha^{\prime})M}\right\}^{\frac{1}{d}}, 
\end{equation*}
where 
\begin{equation*}
   \widetilde{c}(\alpha^{\prime})=  \left(2\pi \right)^{-\frac{d}{2}}\left\{\lambda_{\max}\left(\Omega \right)+\frac{1}{2}\lambda_{\min}(\Omega) \right\}^{-\frac{d}{2}} \exp\left \{-\frac{3d\lambda_{\max}(\Omega)}{2\lambda_{\min}(\Omega)}\cdot z^{2}_{1-\frac{\alpha^{\prime}}{2d}} \right\}.
\end{equation*}
Therefore, for sufficiently large $M$ (under each fixed and sufficiently large $N$), we have $\widetilde{\epsilon}(N, M, \alpha^{\prime})<0.01 \cdot \widehat{\sigma}_{k_{0}k_{0}}\cdot z_{1-\frac{\alpha^{\prime}}{2d}}$. Hence, for sufficiently large $N$ and $M$, for any $\mathbf{Z}\in A_{1}\cap A_{2}$, we have 
\begin{equation*}
    \min_{m\notin \mathcal{M}(\alpha^{\prime}) } \left \|\widehat{\boldsymbol{\epsilon}}^{(m)}-\widehat{\boldsymbol{\epsilon}}   \right \|_{\infty}> \widetilde{\epsilon}(N, M, \alpha^{\prime}).
\end{equation*}
In other words, for sufficiently large $N$ and $M$, for any $\mathbf{Z}\in A_{1}\cap A_{2}$, we have
\begin{align}\label{eqn: equivalence between restricted and unrestricted}
P\left(\min_{m\in [M]}\left\|\widehat{\boldsymbol{\epsilon}}^{(m)}-\widehat{\boldsymbol{\epsilon}}\right\|_{\infty}\leq \widetilde{\epsilon}(N, M, \alpha^{\prime})\mid \mathbf{Z} \right) = P\left(\min_{m\in \mathcal{M}(\alpha^{\prime}) }\left\|\widehat{\boldsymbol{\epsilon}}^{(m)}-\widehat{\boldsymbol{\epsilon}}\right\|_{\infty}\leq \widetilde{\epsilon}(N, M, \alpha^{\prime})\mid \mathbf{Z} \right).
\end{align}

Note that
\begin{align}\label{eqn: independence inequality}
    &\quad    P\left(\min_{m\in [M]}\left\|\widehat{\boldsymbol{\epsilon}}^{(m)}-\widehat{\boldsymbol{\epsilon}}\right\|_{\infty}\leq \widetilde{\epsilon}(N, M, \alpha^{\prime})\mid \mathbf{Z} \right)\nonumber \\
    &= 1-P\left(\min_{m\in [M]}\left\|\widehat{\boldsymbol{\epsilon}}^{(m)}-\widehat{\boldsymbol{\epsilon}}\right\|_{\infty}> \widetilde{\epsilon}(N, M, \alpha^{\prime})\mid \mathbf{Z} \right)\nonumber \\
        &=1-\prod_{m=1}^{M}\left\{1- P\left(\left\|\widehat{\boldsymbol{\epsilon}}^{(m)}-\widehat{\boldsymbol{\epsilon}}\right\|_{\infty}\leq \widetilde{\epsilon}(N, M, \alpha^{\prime})\mid \mathbf{Z} \right)\right\}\nonumber \\
        &\geq 1-\exp\left\{ -\sum_{m=1}^{M}P\left(\left\|\widehat{\boldsymbol{\epsilon}}^{(m)}-\widehat{\boldsymbol{\epsilon}}\right\|_{\infty}\leq \widetilde{\epsilon}(N, M, \alpha^{\prime})\mid \mathbf{Z} \right)  \right\}.
\end{align}
Meanwhile, we have
\begin{align*}
   &\quad P\left(\left\|\widehat{\boldsymbol{\epsilon}}^{(m)}-\widehat{\boldsymbol{\epsilon}}\right\|_{\infty}\leq \widetilde{\epsilon}(N, M, \alpha^{\prime})\mid \mathbf{Z} \right)\cdot \mathbbm{1}\left\{\mathbf{Z}\in A_{1}\cap A_{2} \right\} \\
   & = \int f\left(\widehat{\boldsymbol{\epsilon}}^{(m)}=\mathbf{t}\mid \mathbf{Z} \right)\cdot \mathbbm{1}\left\{ \left\| \mathbf{t}-\widehat{\boldsymbol{\epsilon}} \right\|_{\infty}\leq \widetilde{\epsilon}(N, M, \alpha^{\prime}) \right\} d\mathbf{t}\cdot \mathbbm{1}\left\{\mathbf{Z}\in A_{1}\cap A_{2} \right\} \\
    &\geq \int \widetilde{f}\left(\mathbf{t}\mid \mathbf{Z} \right)\cdot \mathbbm{1}\left\{ \left\| \mathbf{t}-\widehat{\boldsymbol{\epsilon}} \right\|_{\infty}\leq \widetilde{\epsilon}(N, M, \alpha^{\prime}) \right\} d\mathbf{t}\cdot \mathbbm{1}\left\{\mathbf{Z}\in A_{1}\cap A_{2} \right\}\\
    &=\int \widetilde{f}\left(\widehat{\boldsymbol{\epsilon}}\mid \mathbf{Z} \right)\cdot \mathbbm{1}\left\{ \left\| \mathbf{t}-\widehat{\boldsymbol{\epsilon}}\right\|_{\infty}\leq \widetilde{\epsilon}(N, M, \alpha^{\prime}) \right\} d\mathbf{t}\cdot \mathbbm{1}\left\{\mathbf{Z}\in A_{1}\cap A_{2} \right\}\\
    &\quad \quad + \int \left\{\widetilde{f}\left(\mathbf{t}\mid \mathbf{Z} \right)-\widetilde{f}\left(\widehat{\boldsymbol{\epsilon}}\mid \mathbf{Z} \right)\right\}\cdot \mathbbm{1}\left\{ \left\| \mathbf{t}-\widehat{\boldsymbol{\epsilon}} \right\|_{\infty}\leq \widetilde{\epsilon}(N, M, \alpha^{\prime}) \right\} d\mathbf{t}\cdot \mathbbm{1}\left\{\mathbf{Z}\in A_{1}\cap A_{2} \right\},
\end{align*}
where 
\begin{align}\label{eqn: bound the first term}
 &\quad \int \widetilde{f}\left(\widehat{\boldsymbol{\epsilon}}\mid \mathbf{Z} \right)\cdot \mathbbm{1}\left\{ \left\| \mathbf{t}-\widehat{\boldsymbol{\epsilon}}\right\|_{\infty}\leq \widetilde{\epsilon}(N, M, \alpha^{\prime}) \right\} d\mathbf{t}\cdot \mathbbm{1}\left\{\mathbf{Z}\in A_{1}\cap A_{2} \right\}\nonumber\\
 &\geq \widetilde{c}(\alpha^{\prime})\int  \mathbbm{1}\left\{ \left\| \mathbf{t}-\widehat{\boldsymbol{\epsilon}}\right\|_{\infty}\leq \widetilde{\epsilon}(N, M, \alpha^{\prime}) \right\} d\mathbf{t}\cdot \mathbbm{1}\left\{\mathbf{Z}\in A_{1}\cap A_{2} \right\}\nonumber \\
 &\geq \widetilde{c}(\alpha^{\prime})\cdot \left\{2\widetilde{\epsilon}(N, M, \alpha^{\prime})\right\}^{d}\cdot \mathbbm{1}\left\{\mathbf{Z}\in A_{1}\cap A_{2} \right\}.
\end{align}
Also, by the mean value theorem, there exists some $\xi\in (0,1)$ such that 
$\widetilde{f}\left(\mathbf{t}\mid \mathbf{Z} \right)-\widetilde{f}\left(\widehat{\boldsymbol{\epsilon}}\mid \mathbf{Z} \right)=\{\nabla \widetilde{f}(\widehat{\boldsymbol{\epsilon}}+\xi(\mathbf{t}-\widehat{\boldsymbol{\epsilon}}))\}^{T}(\mathbf{t}-\widehat{\boldsymbol{\epsilon}})$, where we have 
\begin{equation*}
    \nabla \widetilde{f}(\mathbf{v})=-\left(2\pi \right)^{-\frac{d}{2}}|\Omega + c_{1}\cdot\mathbf{I}_{d \times d}|^{-\frac{1}{2}}\exp\left \{-\frac{1}{2}\mathbf{v}^{T}\left(\Omega - c_{1}\cdot\mathbf{I}_{d \times d} \right)^{-1} \mathbf{v} \right \}\left(\Omega - c_{1}\cdot\mathbf{I}_{d \times d} \right)^{-1} \mathbf{v}.
\end{equation*}
Because $\|\nabla \widetilde{f}\|_{2}$ is upper bounded by some constant $C_{\nabla}>0$, we have 
\begin{equation*}
    \left | \widetilde{f}\left(\mathbf{t}\mid \mathbf{Z} \right)-\widetilde{f}\left(\widehat{\boldsymbol{\epsilon}}\mid \mathbf{Z} \right) \right |\leq \left\| \nabla \widetilde{f}(\widehat{\boldsymbol{\epsilon}}+\xi(\mathbf{t}-\widehat{\boldsymbol{\epsilon}})) \right \|_{2}\left\|\mathbf{t}-\widehat{\boldsymbol{\epsilon}} \right\|_{2}\leq C_{\nabla}\sqrt{d} \left\|\mathbf{t}-\widehat{\boldsymbol{\epsilon}} \right\|_{\infty}.
\end{equation*}
Then, we have
\begin{align*}
     &\quad \left | \int \left\{\widetilde{f}\left(\mathbf{t}\mid \mathbf{Z} \right)-\widetilde{f}\left(\widehat{\boldsymbol{\epsilon}}\mid \mathbf{Z} \right)\right\}\cdot \mathbbm{1}\left\{ \left\| \mathbf{t}-\widehat{\boldsymbol{\epsilon}} \right\|_{\infty}\leq \widetilde{\epsilon}(N, M, \alpha^{\prime}) \right\} d\mathbf{t}\cdot \mathbbm{1}\left\{\mathbf{Z}\in A_{1}\cap A_{2} \right\}\right| \\
     &\leq  \int \left |\widetilde{f}\left(\mathbf{t}\mid \mathbf{Z} \right)-\widetilde{f}\left(\widehat{\boldsymbol{\epsilon}}\mid \mathbf{Z} \right)\right|\cdot \mathbbm{1}\left\{ \left\| \mathbf{t}-\widehat{\boldsymbol{\epsilon}} \right\|_{\infty}\leq \widetilde{\epsilon}(N, M, \alpha^{\prime}) \right\} d\mathbf{t}\cdot \mathbbm{1}\left\{\mathbf{Z}\in A_{1}\cap A_{2} \right\} \\
     &\leq  \int C_{\nabla}\sqrt{d} \left\|\mathbf{t}-\widehat{\boldsymbol{\epsilon}} \right\|_{\infty}\cdot \mathbbm{1}\left\{ \left\| \mathbf{t}-\widehat{\boldsymbol{\epsilon}} \right\|_{\infty}\leq \widetilde{\epsilon}(N, M, \alpha^{\prime}) \right\} d\mathbf{t}\cdot \mathbbm{1}\left\{\mathbf{Z}\in A_{1}\cap A_{2} \right\}  \\
     &\leq C_{\nabla} \sqrt{d}\ \widetilde{\epsilon}(N, M, \alpha^{\prime}) \int  \mathbbm{1}\left\{ \left\| \mathbf{t}-\widehat{\boldsymbol{\epsilon}} \right\|_{\infty}\leq \widetilde{\epsilon}(N, M, \alpha^{\prime}) \right\} d\mathbf{t}\cdot \mathbbm{1}\left\{\mathbf{Z}\in A_{1}\cap A_{2} \right\}  \\
   &=C_{\nabla} \sqrt{d}\ \widetilde{\epsilon}(N, M, \alpha^{\prime}) \cdot \left\{2\widetilde{\epsilon}(N, M, \alpha^{\prime})\right\}^{d}\cdot \mathbbm{1}\left\{\mathbf{Z}\in A_{1}\cap A_{2} \right\}.
\end{align*}
Because $\widetilde{c}(\alpha^{\prime})>0$ is a constant and $\lim_{M\rightarrow \infty}\widetilde{\epsilon}(N, M, \alpha^{\prime})=0$ for sufficiently large $N$, there exists some $M_{0}\in \mathbbm{Z}^{+}$ such that for any $M>M_{0}$, we have $C_{\nabla}\sqrt{d}\ \widetilde{\epsilon}(N, M, \alpha^{\prime})<\frac{1}{2}\widetilde{c}(\alpha^{\prime})$. That is, for sufficiently large $N$ and $M$, we have
\begin{align}\label{eqn: bound the second term}
     &\quad \left | \int \left\{\widetilde{f}\left(\mathbf{t}\mid \mathbf{Z} \right)-\widetilde{f}\left(\widehat{\boldsymbol{\epsilon}}\mid \mathbf{Z} \right)\right\}\cdot \mathbbm{1}\left\{ \left\| \mathbf{t}-\widehat{\boldsymbol{\epsilon}} \right\|_{\infty}\leq \widetilde{\epsilon}(N, M, \alpha^{\prime}) \right\} d\mathbf{t}\cdot \mathbbm{1}\left\{\mathbf{Z}\in A_{1}\cap A_{2} \right\}\right|\nonumber \\
   &\leq \frac{1}{2}\widetilde{c}(\alpha^{\prime}) \cdot \left\{2\widetilde{\epsilon}(N, M, \alpha^{\prime})\right\}^{d}\cdot \mathbbm{1}\left\{\mathbf{Z}\in A_{1}\cap A_{2} \right\}.
\end{align}
Combining (\ref{eqn: bound the first term}) and (\ref{eqn: bound the second term}), for sufficiently large $N$ and $M$, we have
\begin{align}\label{eqn: bounding integral dif}
   &\quad P\left(\left\|\widehat{\boldsymbol{\epsilon}}^{(m)}-\widehat{\boldsymbol{\epsilon}}\right\|_{\infty}\leq \widetilde{\epsilon}(N, M, \alpha^{\prime})\mid \mathbf{Z} \right)\cdot \mathbbm{1}\left\{\mathbf{Z}\in A_{1}\cap A_{2} \right\}\nonumber \\
   &\geq \int \widetilde{f}\left(\widehat{\boldsymbol{\epsilon}}\mid \mathbf{Z} \right)\cdot \mathbbm{1}\left\{ \left\| \mathbf{t}-\widehat{\boldsymbol{\epsilon}}\right\|_{\infty}\leq \widetilde{\epsilon}(N, M, \alpha^{\prime}) \right\} d\mathbf{t}\cdot \mathbbm{1}\left\{\mathbf{Z}\in A_{1}\cap A_{2} \right\}\nonumber \\
    &\quad \quad + \int \left\{\widetilde{f}\left(\mathbf{t}\mid \mathbf{Z} \right)-\widetilde{f}\left(\widehat{\boldsymbol{\epsilon}}\mid \mathbf{Z} \right)\right\}\cdot \mathbbm{1}\left\{ \left\| \mathbf{t}-\widehat{\boldsymbol{\epsilon}} \right\|_{\infty}\leq \widetilde{\epsilon}(N, M, \alpha^{\prime}) \right\} d\mathbf{t}\cdot \mathbbm{1}\left\{\mathbf{Z}\in A_{1}\cap A_{2} \right\}\nonumber\\
    &\geq \frac{1}{2}\widetilde{c}(\alpha^{\prime}) \cdot \left\{2\widetilde{\epsilon}(N, M, \alpha^{\prime})\right\}^{d}\cdot \mathbbm{1}\left\{\mathbf{Z}\in A_{1}\cap A_{2} \right\}.
\end{align}
Therefore, for any $\alpha^{\prime}\in (0, \alpha)$, for sufficiently large $N$ and $M$, we have
\begin{align}\label{eqn: bounding the min probability by exp}
      &\quad\  P\left(\min_{m\in \mathcal{M}(\alpha^{\prime}) }\left\|\sqrt{N}\left(\widehat{\boldsymbol{\beta}}^{(m)}-\boldsymbol{\beta}\right)\right\|_{\infty}\leq \widetilde{\epsilon}(N, M, \alpha^{\prime})\mid \mathbf{Z} \right)\cdot \mathbbm{1}\left\{\mathbf{Z}\in A_{1}\cap A_{2} \right\}\nonumber \\
      &=P\left(\min_{m\in [M]}\left\|\sqrt{N}\left(\widehat{\boldsymbol{\beta}}^{(m)}-\boldsymbol{\beta}\right)\right\|_{\infty}\leq \widetilde{\epsilon}(N, M, \alpha^{\prime})\mid \mathbf{Z} \right)\cdot \mathbbm{1}\left\{\mathbf{Z}\in A_{1}\cap A_{2} \right\} \text{ (by (\ref{eqn: equivalence between restricted and unrestricted}))}\nonumber \\
      &=P\left(\min_{m\in [M]}\left\|\widehat{\boldsymbol{\epsilon}}^{(m)}-\widehat{\boldsymbol{\epsilon}}\right\|_{\infty}\leq \widetilde{\epsilon}(N, M, \alpha^{\prime})\mid \mathbf{Z} \right)\cdot \mathbbm{1}\left\{\mathbf{Z}\in A_{1}\cap A_{2} \right\}\nonumber\\
        &\geq 1-\exp\left\{ -\sum_{m=1}^{M}P\left(\left\|\widehat{\boldsymbol{\epsilon}}^{(m)}-\widehat{\boldsymbol{\epsilon}}\right\|_{\infty}\leq \widetilde{\epsilon}(N, M, \alpha^{\prime})\mid \mathbf{Z} \right) \cdot \mathbbm{1}\left\{\mathbf{Z}\in A_{1}\cap A_{2} \right\}  \right\} \text{ (by (\ref{eqn: independence inequality}))}\nonumber \\
        &\geq 1-\exp\left\{ -\frac{M}{2}\widetilde{c}(\alpha^{\prime})\cdot \left\{2\widetilde{\epsilon}(N, M, \alpha^{\prime})\right\}^{d}\cdot \mathbbm{1}\left\{\mathbf{Z}\in A_{1}\cap A_{2} \right\} \right\}\text{ (by (\ref{eqn: bounding integral dif}))}\nonumber \\
        &=\left [ 1-\exp\left\{ -\frac{M}{2}\widetilde{c}(\alpha^{\prime})\cdot \left\{2\widetilde{\epsilon}(N, M, \alpha^{\prime})\right\}^{d} \right\} \right ] \cdot \mathbbm{1}\left\{\mathbf{Z}\in A_{1}\cap A_{2} \right\}.
\end{align}
Then, for any $\alpha^{\prime}\in (0, \alpha)$, for sufficiently large $N$ and $M$, we obtain
\begin{align*}
      &\quad  P\left(\min_{m\in \mathcal{M}(\alpha^{\prime}) }\left\|\sqrt{N}\left(\widehat{\boldsymbol{\beta}}^{(m)}-\boldsymbol{\beta}\right)\right\|_{\infty}\leq \widetilde{\epsilon}(N, M, \alpha^{\prime}) \right)\nonumber\\
      &=E_{\mathbf{Z}}\left \{ P\left(\min_{m\in \mathcal{M}(\alpha^{\prime})}\left\|\widehat{\boldsymbol{\epsilon}}^{(m)}-\widehat{\boldsymbol{\epsilon}}\right\|_{\infty}\leq \widetilde{\epsilon}(N, M, \alpha^{\prime})\mid \mathbf{Z} \right)\right\}\nonumber \\
      &\geq E_{\mathbf{Z}}\left \{ P\left(\min_{m\in \mathcal{M}(\alpha^{\prime})}\left\|\widehat{\boldsymbol{\epsilon}}^{(m)}-\widehat{\boldsymbol{\epsilon}}\right\|_{\infty}\leq \widetilde{\epsilon}(N, M, \alpha^{\prime})\mid \mathbf{Z} \right)\cdot \mathbbm{1}\left\{\mathbf{Z}\in A_{1}\cap A_{2} \right\} \right\}\nonumber \\
        &\geq E_{\mathbf{Z}}\left\{ \left [ 1-\exp\left\{ -\frac{M}{2}\widetilde{c}(\alpha^{\prime})\cdot \left\{2\widetilde{\epsilon}(N, M, \alpha^{\prime})\right\}^{d} \right\} \right ] \cdot \mathbbm{1}\left\{\mathbf{Z}\in A_{1}\cap A_{2} \right\}\right\} \text{ (by (\ref{eqn: bounding the min probability by exp}))}\nonumber\\
        &=\left [ 1-\exp\left\{ -\frac{M}{2}\widetilde{c}(\alpha^{\prime})\cdot \left\{2\widetilde{\epsilon}(N, M, \alpha^{\prime})\right\}^{d} \right\} \right ]\cdot P\left(A_{1}\cap A_{2}\right) \nonumber.
\end{align*}
By plugging $\widetilde{\epsilon}(N, M, \alpha^{\prime})=\frac{1}{2}\left\{\frac{2\log N}{\widetilde{c}(\alpha^{\prime})M}\right\}^{\frac{1}{d}}=\widetilde{C}_{\alpha^{\prime}}\cdot (\frac{\log N}{M})^{\frac{1}{d}} $ into (\ref{eqn: bounding the min probability by exp}), for sufficiently large $N$, we have 
\begin{equation*}
    \lim_{M\rightarrow \infty} P\left(\min_{m\in \mathcal{M}(\alpha^{\prime}) }\sqrt{N}\left\| \widehat{\boldsymbol{\beta}}^{(m)}-\boldsymbol{\beta}\right\|_{\infty}  \leq \widetilde{C}_{\alpha^{\prime}}\cdot \left(\frac{\log N}{M}\right)^{\frac{1}{d}}   \right)\geq \left(1-\frac{1}{N}\right) \cdot P\left(A_{1}\cap A_{2}\right),
\end{equation*}
which implies that
\begin{align*}
  \liminf_{N\rightarrow \infty }  \lim_{M\rightarrow \infty}P\left(\min_{m\in \mathcal{M}(\alpha^{\prime}) }\sqrt{N}\left\| \widehat{\boldsymbol{\beta}}^{(m)}-\boldsymbol{\beta}\right\|_{\infty}  \leq \widetilde{C}_{\alpha^{\prime}} \cdot \left(\frac{\log N}{M}\right)^{\frac{1}{d}}   \right)&\geq \liminf_{N\rightarrow \infty} P\left(A_{1}\cap A_{2}\right)\\
  &\geq 1-\alpha^{\prime},
\end{align*}
and 
\begin{align*}
 &\quad \  \liminf_{N\rightarrow \infty }  \lim_{M\rightarrow \infty}P\left(\min_{m\in [M]}\sqrt{N}\left\| \widehat{\boldsymbol{\beta}}^{(m)}-\boldsymbol{\beta}\right\|_{\infty}  \leq \widetilde{C}_{\alpha^{\prime}} \cdot \left(\frac{\log N}{M}\right)^{\frac{1}{d}}   \right)\\
 &\geq \liminf_{N\rightarrow \infty }  \lim_{M\rightarrow \infty}P\left(\min_{m\in \mathcal{M}(\alpha^{\prime}) }\sqrt{N}\left\| \widehat{\boldsymbol{\beta}}^{(m)}-\boldsymbol{\beta}\right\|_{\infty}  \leq \widetilde{C}_{\alpha^{\prime}} \cdot \left(\frac{\log N}{M}\right)^{\frac{1}{d}}   \right)\\
   &\geq 1-\alpha^{\prime}.
\end{align*}
Therefore, the desired results have been established. 
\end{proof}

\vspace{-0.5cm}

\begin{proof}[Proof of Theorem~1:]
We set 
\begin{equation*}
   C_{\alpha^{\prime}}=d\cdot C_{\Psi^{\prime}}C_{x}2^{\frac{1-d}{d}} \left(2\pi \right)^{\frac{1}{2}}\left\{\lambda_{\max}\left(\Omega \right)+\frac{1}{2}\lambda_{\min}(\Omega) \right\}^{\frac{1}{2}} \exp\left \{\frac{3\lambda_{\max}(\Omega)}{2\lambda_{\min}(\Omega)}\cdot z^{2}_{1-\frac{\alpha^{\prime}}{2d}} \right\}.
\end{equation*}
The desired conclusion immediately follows from Lemma~\ref{lem: rate of beta parametric} and the following argument: 
    \begin{align}\label{eqn: bound p dif by beta dif}
    \left \|\widehat{\mathcal{P}}^{(m)}-\mathcal{P}   \right \|_{\infty}
        &=\max_{i \in [N]} \left|\widehat{p}^{(m)}_{i}-p_{i} \right | \nonumber \\
        &=\max_{i \in [N]}\left | \Psi\left(\widehat{\boldsymbol{\beta}}^{(m)T}\mathbf{x}_{i}\right) -\Psi\left(\boldsymbol{\beta}^{T}\mathbf{x}_{i}
        \right) \right | \nonumber \\
        &\leq C_{\Psi^{\prime}}\max_{i \in [N]}\left|\widehat{\boldsymbol{\beta}}^{(m)T}\mathbf{x}_{i} -\boldsymbol{\beta}^{T}\mathbf{x}_{i}\right| \nonumber \\
        &\leq C_{\Psi^{\prime}} C_{x} \left \|\widehat{\boldsymbol{\beta}}^{(m)}-\boldsymbol{\beta}\right \|_{1} \nonumber\\
        &\leq d \cdot C_{\Psi^{\prime}} C_{x} \left \|\widehat{\boldsymbol{\beta}}^{(m)}-\boldsymbol{\beta}\right \|_{\infty}.
    \end{align}
\end{proof}

\vspace{-0.5cm}

\begin{proof}[Proof of Theorem~2:]
The result follows immediately from Theorem~1 and Conditions~1a, 1b, and 2c in the main text, noting that Condition~2c is stronger than Condition~\ref{cond: Lipschitz infinity}.
\end{proof}

\vspace{-0.5cm}

\begin{proof}[Proof of Theorem~3:] Set $M\geq M_{N}$, where $\lim_{N\to\infty}(\log N)/M_N=0$. Let $m^{*}\in \mathcal{M}(\alpha^{\prime})\subseteq [M]$ such that $\min_{m\in \mathcal{M}(\alpha^{\prime})} \|\widehat{\boldsymbol{\beta}}^{(m)}-\boldsymbol{\beta}\|_{\infty}=\|\widehat{\boldsymbol{\beta}}^{(m^{*})}-\boldsymbol{\beta}\|_{\infty}$. Note that $\|\widehat{\mathcal{P}}^{(m^{*})}-\mathcal{P} \|_{\infty}=\max_{i\in [N]}|\Psi(\widehat{\boldsymbol{\beta}}^{(m^{*})T}\mathbf{x}_{i})- \Psi(\boldsymbol{\beta}^{T}\mathbf{x}_{i})|\leq C_{\Psi^{\prime}}\max_{i\in [N]} |\widehat{\boldsymbol{\beta}}^{(m^{*})T}\mathbf{x}_{i}-  \boldsymbol{\beta}^{T}\mathbf{x}_{i}|\leq C_{\Psi^{\prime}}C_{x}\|\widehat{\boldsymbol{\beta}}^{(m^{*})}-\boldsymbol{\beta}\|_{1}\leq d \cdot C_{\Psi^{\prime}}C_{x}\|\widehat{\boldsymbol{\beta}}^{(m^{*})}-\boldsymbol{\beta}\|_{\infty}$, and each $p_{i}=\Psi(\boldsymbol{\beta}^{T}\mathbf{x}_{i})\in [\delta, 1-\delta]$ for some small constant $\delta>0$. For notational simplicity, we define the following events:
\begin{equation*}
    A_{3}=\left\{\min_{m\in \mathcal{M}(\alpha^{\prime}) }\left\|  \widehat{\mathcal{P}}^{(m)}-\mathcal{P}\right\|_{\infty}  \leq \frac{C_{\alpha^{\prime}}}{\sqrt{N}} \cdot \left(\frac{\log N}{M}\right)^{\frac{1}{d}}\right\}
\end{equation*}
and 
\begin{equation*}
    A_{4}=\left\{\widehat{p}^{(m^{*})}_{i}\in \left[\frac{\delta}{2}, 1-\frac{\delta}{2}\right] \text{ for all $i\in [N]$}\right\}.
\end{equation*}
It is clear that, as $N\rightarrow \infty$, we have $P(A_{3}\cap A_{4}^{c})\rightarrow 0$. Therefore, by Theorem~1, we have 
\begin{equation*}
    \liminf_{N\rightarrow \infty } P\left(A_{4}\right)\geq \liminf_{N\rightarrow \infty } P\left(A_{3}\right) \geq 1-\alpha^{\prime}.
\end{equation*}
Meanwhile, because $\widehat{\mathcal{V}}_{\text{oracle}}^{1/2}\asymp_{p}N^{-1/2}$, there exists some constant $\nu_{1}>0$ such that $\lim_{N\rightarrow \infty}P\big(\widehat{\mathcal{V}}_{\text{oracle}}^{1/2}\geq \nu_{1}/\sqrt{N}\big)=1$. Therefore, if we define event $A_{5}=\big\{\widehat{\mathcal{V}}_{\text{oracle}}^{1/2}\geq \nu_{1}/\sqrt{N}\big\}$, we have $\lim_{N\rightarrow \infty}P(A_{5})=1$. Meanwhile, we define $A_{6}=\big\{|\widehat{\theta}^{(m^{*})}-\widehat{\theta}_{\text{oracle}}|\leq L_{\theta}(\delta/2)\cdot \|\mathcal{P}^{(m^{*})}-\mathcal{P}\|_{\infty}\big\}$ and $A_{7}=\big\{N|\widehat{\mathcal{V}}^{(m^{*})}-\widehat{\mathcal{V}}_{\text{oracle}}|\leq L_{v}(\delta/2)\cdot \|\mathcal{P}^{(m^{*})}-\mathcal{P}\|_{\infty}\big\}$. Under Conditions 1a, 1b, and \ref{cond: Lipschitz infinity} (i.e., a weaker condition than Condition 2c), we have $\lim_{N\rightarrow \infty}P(A_{6})=1$ and $\lim_{N\rightarrow \infty}P(A_{7})=1$. Then, under events $A_{5}$ and $A_{7}$, we have
\begin{align*}
    \left|\{\widehat{\mathcal{V}}^{(m^{*})}\}^{1/2}-\{\widehat{\mathcal{V}}_{\text{oracle}}\}^{1/2}\right|&=\frac{\left|\widehat{\mathcal{V}}^{(m^{*})}-\widehat{\mathcal{V}}_{\text{oracle}}\right|}{\{\widehat{\mathcal{V}}^{(m^{*})}\}^{1/2}+\{\widehat{\mathcal{V}}_{\text{oracle}}\}^{1/2}}\\
    &\leq \frac{\sqrt{N}}{\nu_{1}}\left|\widehat{\mathcal{V}}^{(m^{*})}-\widehat{\mathcal{V}}_{\text{oracle}}\right|\\
    &\leq \frac{L_{v}(\delta/2)}{\nu_{1}}\frac{1}{\sqrt{N}}\left\|\widehat{\mathcal{P}}^{(m^{*})}-\mathcal{P}  \right\|_{\infty}.
\end{align*}
Let $L_{1-\alpha}^{(m^{*})}=\widehat{\theta}^{(m^{*})}-z_{1-\alpha/2}\cdot \{\widehat{\mathcal{V}}^{(m^{*})}\}^{1/2}$ and $U_{1-\alpha}^{(m^{*})}=\widehat{\theta}^{(m^{*})}+z_{1-\alpha/2}\cdot \{\widehat{\mathcal{V}}^{(m^{*})}\}^{1/2}$ represent the left and right end points of $\mathcal{C}_{1-\alpha}^{(m^{*})}(\theta)$, respectively. Let $L_{1-\alpha}^{\text{oracle}}=\widehat{\theta}_{\text{oracle}}-z_{1-\alpha/2}\cdot \{\widehat{\mathcal{V}}_{\text{oracle}}\}^{1/2}$ and $U_{1-\alpha}^{\text{oracle}}=\widehat{\theta}_{\text{oracle}}+z_{1-\alpha/2}\cdot \{\widehat{\mathcal{V}}_{\text{oracle}}\}^{1/2}$ represent the left and right end points of $\mathcal{C}_{1-\alpha}^{\text{oracle}}(\theta)$, respectively. Let $\Delta_{L, N}=|L_{1-\alpha}^{(m^{*})}-L_{1-\alpha}^{\text{oracle}} |$ and $\Delta_{U, N}=|U_{1-\alpha}^{(m^{*})}-U_{1-\alpha}^{\text{oracle}}|$. Under events $A_{5}$--$A_{7}$, we have
\begin{align*}
   \frac{\Delta_{L, N}}{\{\widehat{\mathcal{V}}_{\text{oracle}}\}^{1/2}}&=\frac{\left|L_{1-\alpha}^{(m^{*})}-L_{1-\alpha}^{\text{oracle}} \right|}{\{\widehat{\mathcal{V}}_{\text{oracle}}\}^{1/2}}\\
   &\leq \frac{\sqrt{N}}{\nu_{1}} \left | \left(\widehat{\theta}^{(m^{*})}-z_{1-\alpha/2}\cdot \{\widehat{\mathcal{V}}^{(m^{*})}\}^{1/2} \right)-\left( \widehat{\theta}_{\text{oracle}}-z_{1-\alpha/2}\cdot \{\widehat{\mathcal{V}}_{\text{oracle}}\}^{1/2}\right)\right|\\
    &\leq \frac{\sqrt{N}}{\nu_{1}}\left | \widehat{\theta}^{(m^{*})}-\widehat{\theta}_{\text{oracle}}\right|+z_{1-\alpha/2}\cdot \frac{\sqrt{N}}{\nu_{1}} \left | \{\widehat{\mathcal{V}}^{(m^{*})}\}^{1/2}- \{\widehat{\mathcal{V}}_{\text{oracle}}\}^{1/2}\right|\\
    &\leq \frac{L_{\theta}(\delta/2)}{\nu_{1}}\cdot \sqrt{N} \left\|\widehat{\mathcal{P}}^{(m^{*})}-\mathcal{P}  \right\|_{\infty}+z_{1-\alpha/2}\cdot \frac{L_{v}(\delta/2)}{\nu_{1}^{2}}\cdot \left\|\widehat{\mathcal{P}}^{(m^{*})}-\mathcal{P}  \right\|_{\infty},
\end{align*}
and 
\begin{align*}
   \frac{\Delta_{U, N}}{\{\widehat{\mathcal{V}}_{\text{oracle}}\}^{1/2}}&=\frac{\left|U_{1-\alpha}^{(m^{*})}-U_{1-\alpha}^{\text{oracle}} \right|}{\{\widehat{\mathcal{V}}_{\text{oracle}}\}^{1/2}}\\
   &\leq \frac{\sqrt{N}}{\nu_{1}} \left | \left(\widehat{\theta}^{(m^{*})}+z_{1-\alpha/2}\cdot \{\widehat{\mathcal{V}}^{(m^{*})}\}^{1/2} \right)-\left( \widehat{\theta}_{\text{oracle}}+z_{1-\alpha/2}\cdot \{\widehat{\mathcal{V}}_{\text{oracle}}\}^{1/2}\right)\right|\\
    &\leq \frac{\sqrt{N}}{\nu_{1}}\left | \widehat{\theta}^{(m^{*})}-\widehat{\theta}_{\text{oracle}}\right|+z_{1-\alpha/2}\cdot \frac{\sqrt{N}}{\nu_{1}} \left | \{\widehat{\mathcal{V}}^{(m^{*})}\}^{1/2}- \{\widehat{\mathcal{V}}_{\text{oracle}}\}^{1/2}\right|\\
    &\leq \frac{L_{\theta}(\delta/2)}{\nu_{1}}\cdot \sqrt{N} \left\|\widehat{\mathcal{P}}^{(m^{*})}-\mathcal{P}  \right\|_{\infty}+z_{1-\alpha/2}\cdot \frac{L_{v}(\delta/2)}{\nu_{1}^{2}}\cdot \left\|\widehat{\mathcal{P}}^{(m^{*})}-\mathcal{P}  \right\|_{\infty}.
\end{align*}
Therefore, combining the above results with event $A_{3}$, we have $\max\left\{\frac{\Delta_{L, N}}{\{\widehat{\mathcal{V}}_{\text{oracle}}\}^{1/2}}, \frac{\Delta_{U, N}}{\{\widehat{\mathcal{V}}_{\text{oracle}}\}^{1/2}}\right\}\rightarrow 0$ as $N\rightarrow \infty$ (recall that $M\geq M_{N}$ satisfies $\lim_{N\rightarrow \infty}\frac{\log N}{M_{N}}=0$). That is, for any $\epsilon>0$, we have 
\begin{align*}
&\quad \ \liminf_{N\rightarrow\infty} P\left(\max\left\{\frac{\Delta_{L, N}}{\{\widehat{\mathcal{V}}_{\text{oracle}}\}^{1/2}}, \frac{\Delta_{U, N}}{\{\widehat{\mathcal{V}}_{\text{oracle}}\}^{1/2}}\right\}\leq \epsilon \right)\\
&\geq \liminf_{N\rightarrow \infty}P\left(A_{3}\cap A_{4} \cap A_{5} \cap A_{6}\cap A_{7} \right)\\
&= \liminf_{N\rightarrow \infty}P\left(A_{3}\right) \ \text{(because $\lim_{N\rightarrow \infty}P(A_{3}\cap A^{c}_{4})= 0$ and $ \liminf_{N\rightarrow \infty}P\left(A_{5} \cap A_{6} \cap A_{7}\right)=1$)}\\
&\geq 1-\alpha^{\prime} \ \text{(by Theorem~1)}.
\end{align*}
Note that 
\begin{equation*}
    \left[L_{1-\alpha}^{\text{oracle}}+\Delta_{L, N}, U_{1-\alpha}^{\text{oracle}}-\Delta_{U, N}\right]\subseteq \mathcal{C}_{1-\alpha}^{(m^{*})}(\theta)\subseteq \mathcal{C}_{1-\alpha}(\theta),
\end{equation*}
and $\theta\in \left[L_{1-\alpha}^{\text{oracle}}+\Delta_{L, N}, U_{1-\alpha}^{\text{oracle}}-\Delta_{U, N}\right]$ if and only if 
\begin{equation*}
    \frac{\widehat{\theta}_{\text{oracle}}-\theta}{\{\widehat{\mathcal{V}}_{\text{oracle}}\}^{1/2}}\leq z_{1-\alpha/2}-\frac{\Delta_{L, N}}{\{\widehat{\mathcal{V}}_{\text{oracle}}\}^{1/2}} \quad \text{ and } \quad  \frac{\widehat{\theta}_{\text{oracle}}-\theta}{\{\widehat{\mathcal{V}}_{\text{oracle}}\}^{1/2}}\geq -z_{1-\alpha/2}+\frac{\Delta_{U, N}}{\{\widehat{\mathcal{V}}_{\text{oracle}}\}^{1/2}}.
\end{equation*}
Therefore, for any $\epsilon>0$, we have
\begin{align*}
&\quad \ P\left(\theta\in \mathcal{C}_{1-\alpha}(\theta) \right)\\
&\geq P\left(\theta\in \mathcal{C}_{1-\alpha}^{(m^{*})}(\theta) \right)\\
&\geq P\left(\theta\in \left[L_{1-\alpha}^{\text{oracle}}+\Delta_{L, N}, U_{1-\alpha}^{\text{oracle}}-\Delta_{U, N}\right] \right)\\
 &=P\left(-z_{1-\alpha/2}+\frac{\Delta_{U, N}}{\{\widehat{\mathcal{V}}_{\text{oracle}}\}^{1/2}}\leq \frac{\widehat{\theta}_{\text{oracle}}-\theta}{\{\widehat{\mathcal{V}}_{\text{oracle}}\}^{1/2}}\leq z_{1-\alpha/2}-\frac{\Delta_{L, N}}{\{\widehat{\mathcal{V}}_{\text{oracle}}\}^{1/2}} \right)\\
 &= P\left(-z_{1-\alpha/2}+\frac{\Delta_{U, N}}{\{\widehat{\mathcal{V}}_{\text{oracle}}\}^{1/2}}\leq \frac{\widehat{\theta}_{\text{oracle}}-\theta}{\{\widehat{\mathcal{V}}_{\text{oracle}}\}^{1/2}}\leq z_{1-\alpha/2}-\frac{\Delta_{L, N}}{\{\widehat{\mathcal{V}}_{\text{oracle}}\}^{1/2}}, \right. \\
&\quad \quad \quad \quad \quad \quad \quad \quad \quad \quad \quad \quad \quad \quad \quad \quad \quad \quad \left. \max\left\{\frac{\Delta_{L, N}}{\{\widehat{\mathcal{V}}_{\text{oracle}}\}^{1/2}}, \frac{\Delta_{U, N}}{\{\widehat{\mathcal{V}}_{\text{oracle}}\}^{1/2}}\right\}\leq \epsilon \right)\\
 &\quad \quad +P\left(-z_{1-\alpha/2}+\frac{\Delta_{U, N}}{\{\widehat{\mathcal{V}}_{\text{oracle}}\}^{1/2}}\leq \frac{\widehat{\theta}_{\text{oracle}}-\theta}{\{\widehat{\mathcal{V}}_{\text{oracle}}\}^{1/2}}\leq z_{1-\alpha/2}-\frac{\Delta_{L, N}}{\{\widehat{\mathcal{V}}_{\text{oracle}}\}^{1/2}}, \right. \\
&\quad \quad \quad \quad \quad \quad \quad \quad \quad \quad \quad \quad \quad \quad \quad \quad \quad \quad \left. \max\left\{\frac{\Delta_{L, N}}{\{\widehat{\mathcal{V}}_{\text{oracle}}\}^{1/2}}, \frac{\Delta_{U, N}}{\{\widehat{\mathcal{V}}_{\text{oracle}}\}^{1/2}}\right\}>\epsilon  \right)\\
 &\geq P\left(-z_{1-\alpha/2}+\epsilon \leq \frac{\widehat{\theta}_{\text{oracle}}-\theta}{\{\widehat{\mathcal{V}}_{\text{oracle}}\}^{1/2}}\leq z_{1-\alpha/2}-\epsilon, \max\left\{\frac{\Delta_{L, N}}{\{\widehat{\mathcal{V}}_{\text{oracle}}\}^{1/2}}, \frac{\Delta_{U, N}}{\{\widehat{\mathcal{V}}_{\text{oracle}}\}^{1/2}}\right\}\leq \epsilon \right)\\
 &\geq  P\left(-z_{1-\alpha/2}+\epsilon \leq \frac{\widehat{\theta}_{\text{oracle}}-\theta}{\{\widehat{\mathcal{V}}_{\text{oracle}}\}^{1/2}}\leq z_{1-\alpha/2}-\epsilon\right)-P\left(\max\left\{\frac{\Delta_{L, N}}{\{\widehat{\mathcal{V}}_{\text{oracle}}\}^{1/2}}, \frac{\Delta_{U, N}}{\{\widehat{\mathcal{V}}_{\text{oracle}}\}^{1/2}}\right\}> \epsilon \right).
\end{align*}
In addition, under Conditions 2a and 2b, for any $\alpha^{\prime}\in (0, \alpha)$ and sufficiently small $\epsilon>0$, we have
\begin{align*}
&\quad \ \liminf_{N\rightarrow \infty}  P\left(\theta\in \mathcal{C}_{1-\alpha}(\theta) \right)\\
  &\geq \liminf_{N\rightarrow \infty} P\left(-z_{1-\alpha/2}+\epsilon \leq \frac{\widehat{\theta}_{\text{oracle}}-\theta}{\{\widehat{\mathcal{V}}_{\text{oracle}}\}^{1/2}}\leq z_{1-\alpha/2}-\epsilon\right)\\
  &\quad \quad \quad \quad - \limsup_{N\rightarrow \infty} P\left(\max\left\{\frac{\Delta_{L, N}}{\{\widehat{\mathcal{V}}_{\text{oracle}}\}^{1/2}}, \frac{\Delta_{U, N}}{\{\widehat{\mathcal{V}}_{\text{oracle}}\}^{1/2}}\right\}> \epsilon \right)\\
  &\geq \liminf_{N\rightarrow \infty} P\left(\frac{\widehat{\theta}_{\text{oracle}}-\theta}{\{\widehat{\mathcal{V}}_{\text{oracle}}\}^{1/2}}\leq z_{1-\alpha/2}-\epsilon \right)-\limsup_{N\rightarrow \infty}P\left(\frac{\widehat{\theta}_{\text{oracle}}-\theta}{\{\widehat{\mathcal{V}}_{\text{oracle}}\}^{1/2}}< -z_{1-\alpha/2}+\epsilon \right)\\
 &\quad \quad + \liminf_{N\rightarrow \infty} P\left(\max\left\{\frac{\Delta_{L, N}}{\{\widehat{\mathcal{V}}_{\text{oracle}}\}^{1/2}}, \frac{\Delta_{U, N}}{\{\widehat{\mathcal{V}}_{\text{oracle}}\}^{1/2}}\right\}\leq \epsilon \right)-1\\
 &= \liminf_{N\rightarrow \infty} P\left(\frac{\widehat{\theta}_{\text{oracle}}-\theta}{\{\text{Var}(\widehat{\theta}_{\text{oracle}})\}^{1/2}}\frac{\{\text{Var}(\widehat{\theta}_{\text{oracle}})\}^{1/2}}{\{\widehat{\mathcal{V}}_{\text{oracle}}\}^{1/2}}\leq z_{1-\alpha/2}-\epsilon \right)\\
 &\quad \quad -\limsup_{N\rightarrow \infty}P\left(\frac{\widehat{\theta}_{\text{oracle}}-\theta}{\{\text{Var}(\widehat{\theta}_{\text{oracle}})\}^{1/2}}\frac{\{\text{Var}(\widehat{\theta}_{\text{oracle}})\}^{1/2}}{\{\widehat{\mathcal{V}}_{\text{oracle}}\}^{1/2}}< -z_{1-\alpha/2}+\epsilon \right)\\
 &\quad \quad + \liminf_{N\rightarrow \infty} P\left(\max\left\{\frac{\Delta_{L, N}}{\{\widehat{\mathcal{V}}_{\text{oracle}}\}^{1/2}}, \frac{\Delta_{U, N}}{\{\widehat{\mathcal{V}}_{\text{oracle}}\}^{1/2}}\right\}\leq \epsilon \right)-1\\
 &\geq \Phi\left( z_{1-\alpha/2}-\epsilon\right)-\Phi\left(-z_{1-\alpha/2}+\epsilon\right)-\alpha^{\prime}.
\end{align*}
Setting $\epsilon\rightarrow 0$ and $\alpha^{\prime}\rightarrow 0$, by the continuity of $\Phi$, we have 
\begin{align*}
  \liminf_{N\rightarrow \infty} P\left(\theta\in \mathcal{C}_{1-\alpha}(\theta) \right)\geq \Phi\left(z_{1-\alpha/2}\right)-\Phi\left(-z_{1-\alpha/2}\right)=1-\alpha. 
\end{align*}
Also, by setting $\alpha^{\prime}\in (0, \alpha)$ as some fixed value and $\epsilon\rightarrow 0$, as well as replacing $\alpha$ in each $\mathcal{C}^{(m)}_{1-\alpha}(\theta)$ with $\alpha_{0}=\alpha-\alpha^{\prime}$, we have 
\begin{align*}
  \liminf_{N\rightarrow \infty} P\left(\theta\in \widetilde{\mathcal{C}}_{1-\alpha}(\theta) \right)\geq \Phi\left(z_{1-\alpha_{0}/2}\right)-\Phi\left(-z_{1-\alpha_{0}/2}\right)-\alpha^{\prime}=1-\alpha_{0}-\alpha^{\prime}=1-\alpha. 
\end{align*}\end{proof}

Let $\widehat{\mathcal{P}}=(\widehat{p}_{1},\dots,\widehat{p}_{N})=(\Psi(\widehat{\boldsymbol{\beta}}^{T}\mathbf{x}_{1}),\dots,\Psi(\widehat{\boldsymbol{\beta}}^{T}\mathbf{x}_{N}))$, and let $\widehat{\theta}$ and $\widehat{\mathcal{V}}$ denote the plug-in estimator and its corresponding plug-in variance estimator, obtained by replacing $\mathcal{P}$ in $\widehat{\theta}_{\text{oracle}}$ and $\widehat{\mathcal{V}}_{\text{oracle}}$ with $\widehat{\mathcal{P}}$, respectively.

\begin{lemma}\label{lemma: max beta dif parametric}
   Under Conditions 1c and 3, we have $\max_{m\in [M]}  \|\widehat{\mathcal{P}}^{(m)}-\widehat{\mathcal{P}} \|_{\infty}=O_{p}\Big(\sqrt{\frac{\log M}{N}}\Big)$ and $\max_{m\in \mathcal{M}(\alpha^{\prime})}   \|\widehat{\mathcal{P}}^{(m)}-\widehat{\mathcal{P}} \|_{\infty}=O_{p}\Big(\frac{1}{\sqrt{N}}\Big)$.
\end{lemma}

\vspace{-1cm}

\begin{proof}: Recall that $\widehat{\beta}_{k}^{(m)}\overset{\text{iid}}{\sim} N(\widehat{\beta}_{k}, \widehat{\sigma}^{2}_{kk}/N)$, where $\widehat{\sigma}^{2}_{kk}\xrightarrow{p} \sigma^{2}_{kk}>0$ for all $k\in [d]$. Let $\sigma^{2}_{*}=2\max_{k \in [d]}\sigma^{2}_{kk}>0$. Define the event $A_{8}=\{\max_{k \in [d]}\widehat{\sigma}^{2}_{kk}\leq \sigma^{2}_{*}\}$, where we have $\lim_{N\rightarrow \infty}P(A_{8})=1$. Then, for any $t>0$, we have 
\begin{align*}
    P\left(\max_{m\in [M]} \left \|\widehat{\boldsymbol{\beta}}^{(m)}-\widehat{\boldsymbol{\beta}}\right \|_{\infty}\geq t \mid A_{8} \right)&=P\left(\max_{m\in [M], k\in [d]} \left |\widehat{\beta}^{(m)}_{k}-\widehat{\beta}_{k}\right |\geq t \mid A_{8} \right)\\
    &\leq \sum_{m=1}^{M}\sum_{k=1}^{d}P\left(\left |\widehat{\beta}^{(m)}_{k}-\widehat{\beta}_{k}\right |\geq t \mid A_{8} \right)\\
    &\leq  2Md\exp\left(-\frac{N t^{2}}{2\sigma^{2}_{*}}   \right) \quad \text{(by the Gaussian tail bound)}.
\end{align*}
Therefore, for any $t>0$, we have 
\begin{align*}
   &\quad \ P\left(\max_{m\in [M]} \left \|\widehat{\boldsymbol{\beta}}^{(m)}-\widehat{\boldsymbol{\beta}}\right \|_{\infty}\geq t \right)\\
   &= P\left(A_{8}\right) P\left(\max_{m\in [M]} \left \|\widehat{\boldsymbol{\beta}}^{(m)}-\widehat{\boldsymbol{\beta}}\right \|_{\infty}\geq t \mid A_{8} \right)+P\left(A_{8}^{c}\right) P\left(\max_{m\in [M]} \left \|\widehat{\boldsymbol{\beta}}^{(m)}-\widehat{\boldsymbol{\beta}}\right \|_{\infty}\geq t \mid A_{8}^{c} \right)\\
   &\leq P\left(\max_{m\in [M]} \left \|\widehat{\boldsymbol{\beta}}^{(m)}-\widehat{\boldsymbol{\beta}}\right \|_{\infty}\geq t \mid A_{8} \right)+P\left(A_{8}^{c}\right). 
\end{align*}
Setting $t=C\sqrt{\frac{\log (2Md)}{N}}$, we have 
\begin{equation*}
    P\left(\max_{m\in [M]} \left \|\widehat{\boldsymbol{\beta}}^{(m)}-\widehat{\boldsymbol{\beta}}\right \|_{\infty}\geq C\sqrt{\frac{\log (2Md)}{N}} \right) \leq \left(2Md \right)^{1-\frac{C^{2}}{2\sigma_{*}^{2}}}+P\left(A_{8}^{c}\right). 
\end{equation*}
For any $\epsilon>0$, by choosing sufficiently large $C>0$, we have $\left(2Md \right)^{1-\frac{C^{2}}{2\sigma_{*}^{2}}}\leq \epsilon/2$ for any $M$ and $P\left(A_{8}^{c}\right)\leq \epsilon/2$ for any $M$ and for any sufficiently large $N$. That is, for any $\epsilon$, there exists $C_{\epsilon}>0$, such that 
    \begin{equation*}
    P\left(\max_{m\in [M]} \left \|\widehat{\boldsymbol{\beta}}^{(m)}-\widehat{\boldsymbol{\beta}}\right \|_{\infty}\geq C_{\epsilon}\sqrt{\frac{\log (2Md)}{N}} \right) \leq \epsilon \text{ for any sufficiently large $N$ and $M$}. 
\end{equation*}
That is, we have 
\begin{equation*}
    \max_{m\in [M]} \left \|\widehat{\boldsymbol{\beta}}^{(m)}-\widehat{\boldsymbol{\beta}}\right \|_{\infty}=O_{p}\left(\sqrt{\frac{\log (2Md)}{N}}\right)=O_{p}\left(\sqrt{\frac{\log M}{N}}\right).
\end{equation*}
Therefore, we have 
\begin{equation*}
   \max_{m\in [M]} \left \|\widehat{\mathcal{P}}^{(m)}-\widehat{\mathcal{P}}\right \|_{\infty}\leq C_{\Psi^{\prime} }C_{x} \max_{m\in [M]} \left \|\widehat{\boldsymbol{\beta}}^{(m)}-\widehat{\boldsymbol{\beta}}\right \|_{\infty}=O_{p}\left(\sqrt{\frac{\log M}{N}}\right).
\end{equation*}
Also, when $A_{8}$ happens, for any $M$ and for any $m\in \mathcal{M}(\alpha^{\prime})$, we have 
\begin{equation*}
    \left \|\widehat{\boldsymbol{\beta}}^{(m)}-\widehat{\boldsymbol{\beta}}\right \|_{\infty}\leq \left \|\widehat{\boldsymbol{\beta}}^{(m)}-\widehat{\boldsymbol{\beta}}\right \|_{1}\leq 1.01\cdot z_{1-\frac{\alpha^{\prime}}{2d}}\cdot \frac{\sum_{k=1}^{d}\widehat{\sigma}_{kk}}{\sqrt{N}}\leq 1.01\cdot z_{1-\frac{\alpha^{\prime}}{2d}}\cdot \frac{d \cdot\sigma_{*}}{\sqrt{N}},
\end{equation*}
which implies that 
\begin{equation*}
    \max_{m\in \mathcal{M}(\alpha^{\prime})}\left \|\widehat{\boldsymbol{\beta}}^{(m)}-\widehat{\boldsymbol{\beta}}\right \|_{\infty}\leq 1.01\cdot z_{1-\frac{\alpha^{\prime}}{2d}}\cdot \frac{d \cdot \sigma_{*}}{\sqrt{N}} \text{ for any $M$}.
\end{equation*}
Since $\lim_{N\rightarrow \infty}P(A_{8})=1$, we have 
\begin{equation*}
    \max_{m\in \mathcal{M}(\alpha^{\prime})}\left \|\widehat{\boldsymbol{\beta}}^{(m)}-\widehat{\boldsymbol{\beta}}\right \|_{\infty}=O_{p}\left(\frac{1}{\sqrt{N}}\right).
\end{equation*}
Therefore, we have 
\begin{equation*}
   \max_{m\in \mathcal{M}(\alpha^{\prime}) } \left \|\widehat{\mathcal{P}}^{(m)}-\widehat{\mathcal{P}}\right \|_{\infty}\leq C_{\Psi^{\prime} }C_{x} \max_{m\in \mathcal{M}(\alpha^{\prime})} \left \|\widehat{\boldsymbol{\beta}}^{(m)}-\widehat{\boldsymbol{\beta}}\right \|_{\infty}=O_{p}\left(\frac{1}{\sqrt{N}}\right).
\end{equation*}\end{proof}

\vspace{-0.5cm}

\begin{proof}[Proof of Theorem~4:] Following arguments similar to those in (\ref{eqn: bound p dif by beta dif}), we have $\|\widehat{\mathcal{P}}-\mathcal{P} \|_{\infty}\leq d \cdot C_{\Psi^{\prime}} C_{x} \|\widehat{\boldsymbol{\beta}}-\boldsymbol{\beta} \|_{\infty}$. By the consistency of $\widehat{\boldsymbol{\beta}}$, as well as the strict positivity assumption stated in Condition 1a, we have $\lim_{N\rightarrow\infty}P\big(\widehat{p}_{i}\in[\delta/2,1-\delta/2]\big)=1$. Then, under Conditions 1a, 1b, and \ref{cond: Lipschitz infinity} (a weaker condition than Condition~2c), we have $\lim_{N\rightarrow \infty} P\big(|\widehat{\theta }-\widehat{\theta}_{\text{oracle}}|\leq L_{\theta}(\delta/2)\cdot \|\widehat{\mathcal{P}}-\mathcal{P}\|_{\infty}\big)=1$ and $\lim_{N\rightarrow \infty} P\big(N|\widehat{\mathcal{V}}-\widehat{\mathcal{V}}_{\text{oracle}}|\leq L_{v}(\delta/2)\cdot \|\widehat{\mathcal{P}}-\mathcal{P}\|_{\infty}\big)=1$, which implies that $\lim_{N\rightarrow \infty} P\big(|\widehat{\theta}-\widehat{\theta}_{\text{oracle}}|\leq L_{\theta}(\delta/2)\cdot d \cdot C_{\Psi^{\prime}} C_{x} \|\widehat{\boldsymbol{\beta}}-\boldsymbol{\beta} \|_{\infty}\big)=1$ and $\lim_{N\rightarrow \infty} P\big(N|\widehat{\mathcal{V}}-\widehat{\mathcal{V}}_{\text{oracle}}|\leq L_{v}(\delta/2)\cdot d \cdot C_{\Psi^{\prime}} C_{x} \|\widehat{\boldsymbol{\beta}}-\boldsymbol{\beta} \|_{\infty}\big)=1$. Together with the facts that $\|\widehat{\boldsymbol{\beta}}-\boldsymbol{\beta} \|_{\infty}=O_{p}\left(\frac{1}{\sqrt{N}}\right)$ and $\lim_{N\rightarrow \infty}P(A_{5})=1$, we have (i) $|\widehat{\theta}-\widehat{\theta}_{\text{oracle}}|=O_{p}\left(\frac{1}{\sqrt{N}}\right)$ and (ii) there exists $\widetilde{\nu}_{1}, \widetilde{\nu}_{2}>0$ such that $\liminf_{N\rightarrow\infty}P( \widetilde{\nu}_{1}^{2}/N\leq \widehat{\mathcal{V}}\leq  \widetilde{\nu}_{2}^{2}/N)=\liminf_{N\rightarrow\infty}P(\widetilde{\nu}_{1}/\sqrt{N}\leq \widehat{\mathcal{V}}^{1/2}\leq \widetilde{\nu}_{2}/\sqrt{N})=1$ (i.e., $\widehat{\mathcal{V}}$ is nondegenerate).

Let $\mu_{L}$ denote the Lebesgue measure on $\mathbbm{R}$. Since
\begin{align*}
  &\quad \  \left |\widehat{\theta}^{(m)}-z_{1-\alpha/2}\cdot \{\widehat{\mathcal{V}}^{(m)}\}^{1/2} -\theta  \right |\\
  &\leq \left |\widehat{\theta}^{(m)}-\widehat{\theta}     \right|+\left |\widehat{\theta} -\theta \right|+ z_{1-\alpha/2}\cdot \left | \{\widehat{\mathcal{V}}^{(m)}\}^{1/2}-\{\widehat{\mathcal{V}}\}^{1/2}  \right |+z_{1-\alpha/2}\cdot \{\widehat{\mathcal{V}}\}^{1/2} 
\end{align*}
and
\begin{align*}
  &\quad \  \left |\widehat{\theta}^{(m)}+z_{1-\alpha/2}\cdot \{\widehat{\mathcal{V}}^{(m)}\}^{1/2} -\theta  \right |\\
  &\leq \left |\widehat{\theta}^{(m)}-\widehat{\theta}     \right|+\left |\widehat{\theta} -\theta \right|+ z_{1-\alpha/2}\cdot \left | \{\widehat{\mathcal{V}}^{(m)}\}^{1/2}-\{\widehat{\mathcal{V}}\}^{1/2}  \right |+z_{1-\alpha/2}\cdot \{\widehat{\mathcal{V}}\}^{1/2},
\end{align*}
we have
\begin{align}\label{eqn: bound the length parametric unrestricted}
  &\quad \ \mu_{L}\big\{\mathcal{C}_{1-\alpha}(\theta)\big\}\nonumber\\
  &\leq \max_{m\in [M]}\left |\widehat{\theta}^{(m)}-z_{1-\alpha/2}\cdot \{\widehat{\mathcal{V}}^{(m)}\}^{1/2} -\theta   \right |+\max_{m\in [M]}\left |\widehat{\theta}^{(m)}+z_{1-\alpha/2}\cdot \{\widehat{\mathcal{V}}^{(m)}\}^{1/2} -\theta   \right |\nonumber\\
   &\leq 2\max_{m\in [M]} \left |\widehat{\theta}^{(m)}-\widehat{\theta}     \right|+2\left |\widehat{\theta} -\theta     \right|+ 2 z_{1-\alpha/2}\cdot \max_{m\in [M]} \left | \{\widehat{\mathcal{V}}^{(m)}\}^{1/2}-\{\widehat{\mathcal{V}}\}^{1/2} \right |+2z_{1-\alpha/2}\cdot \{\widehat{\mathcal{V}}\}^{1/2},
\end{align}
and, similarly, 
\begin{align}\label{eqn: bound the length parametric restricted}
 &\quad \ \mu_{L}\big\{\widetilde{\mathcal{C}}_{1-\alpha}(\theta)\big\} \nonumber\\
  &\leq \max_{m\in \mathcal{M}(\alpha^{\prime})}\left |\widehat{\theta}^{(m)}-z_{1-\alpha/2}\cdot \{\widehat{\mathcal{V}}^{(m)}\}^{1/2} -\theta   \right |+\max_{m\in \mathcal{M}(\alpha^{\prime})}\left |\widehat{\theta}^{(m)}+z_{1-\alpha/2}\cdot \{\widehat{\mathcal{V}}^{(m)}\}^{1/2} -\theta   \right |\nonumber\\
   &\leq 2\max_{m\in \mathcal{M}(\alpha^{\prime})} \left |\widehat{\theta}^{(m)}-\widehat{\theta}     \right|+2\left |\widehat{\theta} -\theta     \right|+ 2 z_{1-\alpha/2}\cdot \max_{m\in \mathcal{M}(\alpha^{\prime})} \left | \{\widehat{\mathcal{V}}^{(m)}\}^{1/2}-\{\widehat{\mathcal{V}}\}^{1/2} \right |\nonumber \\
   &\quad \quad \quad \quad \quad \quad \quad \quad \quad \quad \quad \quad \quad \quad \quad \quad \quad \quad +2z_{1-\alpha/2}\cdot \{\widehat{\mathcal{V}}\}^{1/2}.
\end{align}
For the term $\max_{m\in [M]} \left |\widehat{\theta}^{(m)}-\widehat{\theta} \right|$ and $\max_{m\in \mathcal{M}(\alpha^{\prime})} \left |\widehat{\theta}^{(m)}-\widehat{\theta} \right|$, define the events
\begin{align*}
      & A_{9}=\left\{ \max_{m\in [M]}\left |\widehat{\theta}^{(m)}-\widehat{\theta}     \right|
    \leq L_{\theta}(\delta^{\prime})\times  \max_{m\in [M]} \left \| \widehat{\mathcal{P}}^{(m)}-\widehat{\mathcal{P}} \right \|_{\infty}\right\}, \\
 &A_{10}= \left\{\max_{m\in \mathcal{M}(\alpha^{\prime})}\left |\widehat{\theta}^{(m)}-\widehat{\theta}     \right|
    \leq L_{\theta}(\delta^{\prime})\times \max_{m\in \mathcal{M}(\alpha^{\prime})} \left \| \widehat{\mathcal{P}}^{(m)}-\widehat{\mathcal{P}} \right \|_{\infty}\right\},
\end{align*}
where the regularization constant $\delta^{\prime}$ mentioned in Section 3.3 is chosen sufficiently small, for example, such that $\delta^{\prime}<\delta/2$. Then, we have $P\big(\widehat{p}_{i}^{(m)}\in[\delta^{\prime},1-\delta^{\prime}]\big)=1$ and $\lim_{N\rightarrow\infty}P\big(\widehat{p}_{i}\in[\delta^{\prime},1-\delta^{\prime}]\big)\geq \lim_{N\rightarrow\infty}P\big(\widehat{p}_{i}\in[\delta/2,1-\delta/2]\big)=1$. Therefore, under Conditions 1a, 1b, and 2c, we have $\lim_{N\rightarrow\infty}P(A_{9})=1$ and $\lim_{N\rightarrow\infty}P(A_{10})=1$. Combining these results with Lemma~\ref{lemma: max beta dif parametric}, we have $\max_{m\in [M]}\left |\widehat{\theta}^{(m)}-\widehat{\theta}\right|=O_{p}\Big(\sqrt{\frac{\log M}{N}}\Big)$ and $\max_{m\in \mathcal{M}(\alpha^{\prime})}\left |\widehat{\theta}^{(m)}-\widehat{\theta}\right|=O_{p}\left(\frac{1}{\sqrt{N}}\right)$. 

Then, we define the event
\begin{equation*}
    A_{11}=\big\{\widetilde{\nu}_{1}/\sqrt{N}\leq \widehat{\mathcal{V}}^{1/2}\leq \widetilde{\nu}_{2}/\sqrt{N}\big\},
\end{equation*}
for which we have $\lim_{N\rightarrow \infty}P(A_{11})=1$. Under Condition~\ref{cond: Lipschitz}, when the event $A_{11}$ happens, for any $m$, we have 
\begin{align*}
    \left|\{\widehat{\mathcal{V}}^{(m)}\}^{1/2}-\widehat{\mathcal{V}}^{1/2}\right|&=\frac{\left|\widehat{\mathcal{V}}^{(m)}-\widehat{\mathcal{V}}\right|}{\{\widehat{\mathcal{V}}^{(m)}\}^{1/2}+\widehat{\mathcal{V}}^{1/2}}\\
    &\leq \frac{\sqrt{N}}{\widetilde{\nu}_{1}}\left|\widehat{\mathcal{V}}^{(m)}-\widehat{\mathcal{V}}\right|\\
    &\leq \frac{L_{v}(\delta^{\prime})}{\widetilde{\nu}_{1}}\frac{1}{\sqrt{N}}\left \| \widehat{\mathcal{P}}^{(m)}-\widehat{\mathcal{P}} \right \|_{\infty},
\end{align*}
which implies that
\begin{align*}
  \max_{m\in [M]}  \left|\{\widehat{\mathcal{V}}^{(m)}\}^{1/2}-\widehat{\mathcal{V}}^{1/2}\right|&\leq \frac{L_{v}(\delta^{\prime})}{\widetilde{\nu}_{1}}\frac{1}{\sqrt{N}} \max_{m\in [M]}\left \| \widehat{\mathcal{P}}^{(m)}-\widehat{\mathcal{P}} \right \|_{\infty},\\
   \max_{m\in \mathcal{M}(\alpha^{\prime})}  \left|\{\widehat{\mathcal{V}}^{(m)}\}^{1/2}-\widehat{\mathcal{V}}^{1/2}\right|&\leq \frac{L_{v}(\delta^{\prime})}{\widetilde{\nu}_{1}}
   \frac{1}{\sqrt{N}} \max_{m\in \mathcal{M}(\alpha^{\prime})}\left \| \widehat{\mathcal{P}}^{(m)}-\widehat{\mathcal{P}} \right \|_{\infty}.
\end{align*}
Therefore, by Lemma~\ref{lemma: max beta dif parametric}, we have $\max_{m\in [M]}  \left|\{\widehat{\mathcal{V}}^{(m)}\}^{1/2}-\widehat{\mathcal{V}}^{1/2}\right|=O_{p}\Big(\frac{\sqrt{\log M}}{N}\Big)$ and $\max_{m\in \mathcal{M}(\alpha^{\prime})}  \left|\{\widehat{\mathcal{V}}^{(m)}\}^{1/2}-\widehat{\mathcal{V}}^{1/2}\right|=O_{p}\Big(\frac{1}{N}\Big)$. 

Meanwhile, because $|\widehat{\theta}-\widehat{\theta}_{\text{oracle}}|=O_{p}\Big(\frac{1}{\sqrt{N}}\Big)$, $|\widehat{\theta}_{\text{oracle}}-\theta|=O_{p}\Big(\frac{1}{\sqrt{N}}\Big)$, and $|\widehat{\theta}-\theta|\leq |\widehat{\theta}_{\text{oracle}}-\theta|+|\widehat{\theta}-\widehat{\theta}_{\text{oracle}}|$, we have $|\widehat{\theta}-\theta|=O_{p}\Big(\frac{1}{\sqrt{N}}\Big)$. Also, in the earlier part of the proof, we have shown that $\widehat{\mathcal{V}}^{1/2}=O_{p}\Big(\frac{1}{\sqrt{N}}\Big)$. Combining these facts with $\max_{m\in [M]}  \left|\widehat{\theta}^{(m)}-\widehat{\theta}\right|=O_{p}\Big(\sqrt{\frac{\log M}{N}}\Big)$, $\max_{m\in [M]}  \left|\{\widehat{\mathcal{V}}^{(m)}\}^{1/2}-\widehat{\mathcal{V}}^{1/2}\right|=O_{p}\Big(\frac{\sqrt{\log M}}{N}\Big)$, $\max_{m\in \mathcal{M}(\alpha^{\prime})}  \left|\widehat{\theta}^{(m)}-\widehat{\theta}\right|=O_{p}\Big(\frac{1}{\sqrt{N}}\Big)$, and $\max_{m\in \mathcal{M}(\alpha^{\prime})}  \left|\{\widehat{\mathcal{V}}^{(m)}\}^{1/2}-\widehat{\mathcal{V}}^{1/2}\right|=O_{p}\Big(\frac{1}{N}\Big)$, as well as (\ref{eqn: bound the length parametric unrestricted}) and (\ref{eqn: bound the length parametric restricted}), we have shown that $\mu_{L}\big\{\mathcal{C}_{1-\alpha}(\theta)\big\}=O_{p}\left(\sqrt{\frac{\log M}{N}}\right)$ and $\mu_{L}\big\{\widetilde{\mathcal{C}}_{1-\alpha}(\theta)\big\}=O_{p}\Big(\frac{1}{\sqrt{N}}\Big)$. That is, for any $\epsilon>0$, there exists some $C_{\epsilon}>0$, such that
\begin{equation*}
  \liminf_{N\rightarrow \infty} \lim_{M\rightarrow \infty} P\left(\mu_{L}\big\{\mathcal{C}_{1-\alpha}(\theta)\big\}\leq C_{\epsilon}\cdot \sqrt{\frac{\log M}{N}} \right)\geq 1-\epsilon 
\end{equation*}
and
\begin{equation*}
   \liminf_{N\rightarrow \infty} \lim_{M\rightarrow \infty}  P\left(\mu_{L}\big\{\widetilde{\mathcal{C}}_{1-\alpha}(\theta)\big\}\leq C_{\epsilon}\cdot \frac{1}{\sqrt{N}} \right)\geq 1-\epsilon.
\end{equation*}\end{proof}

\subsection*{A.4: Proofs for Nonparametric Propensity Score Propagation}

Theorem~\ref{thm: approximation nonparametric full} below is the detailed version of Theorem~5 in the main text. 

\begin{theorem}[The Detailed Version of Theorem 5]\label{thm: approximation nonparametric full}
Let $\alpha\in (0,0.5)$ be a prespecified significance level. Consider the ignorability assumption and Conditions 1, 2, 4, and 5 in the main text. Define $\varphi(t)=t\exp(t)$ for $t\geq 0$, and let $\varphi^{-1}$ denote its inverse function. Let $\lfloor x \rfloor$ denote the largest integer no greater than $x$. For any $\alpha_{*}\in (0,\alpha)$ and sufficiently large $N$ and $M$, define $K_{N,M}=\left\lfloor\min\left\{\frac{\rho_{\alpha_{*}, N}M}{2\varphi^{-1}(\rho_{\alpha_{*}, N}M\sqrt N/2)},\, \rho_{\alpha_{*}, N}\cdot 2^{N-1}\right\}\right\rfloor\geq 1$ and $\varepsilon_{N,M}=L^{\prime}_{\theta}(\delta^{\prime}) L_{s,N}\left\{\frac{\rho_{\alpha_{*}, N}/K_{N,M}+3\cdot 2^{-N}}{(\gamma_{\alpha_{*}, N}-2^{-N})^2}+\frac{1}{N}\right\}$. Then $\lim_{N\rightarrow \infty} \lim_{M\to\infty}\frac{\varepsilon_{N,M}}{N^{-1/2}}=0$, and 
\begin{equation*} \liminf_{N\rightarrow \infty}\lim_{M\rightarrow \infty} P\left(\min_{m \in [M]} \left|\widehat{\theta}^{(m)}-\widehat{\theta}_{\text{oracle}}\right|\leq  \varepsilon_{N, M}\right)\geq 1-\alpha_{*}. \end{equation*}
\end{theorem}

As a preliminary step toward proving Theorem 5 (i.e., Theorem~\ref{thm: approximation nonparametric full}), we recall Lemma~\ref{lem: edge lower bound}, which follows directly from the Poincaré inequality on the hypercube; see Chapter~2.3 of \citet{o2014analysis} or \citet{Eldan2025Isoperimetric}.

\begin{lemma}[Edge-Isoperimetric Inequality for the Hypercube]\label{lem: edge lower bound}
   Let $\mu_{s}$ be the uniform measure over $\{0,1\}^{N}$. That is, we have $\mu_{s}(\mathbf{S})=2^{-N}$ for any $\mathbf{S}\in \{0,1\}^{N}$. For any given $\mathcal{S} \subset \{0,1\}^{N}$, let $\partial_e \mathcal{S}$ be its edge boundary:
   \begin{equation*}
     \partial_e \mathcal{S} :=\left \{\left\{\mathbf{S}_{1},\mathbf{S}_{2}\right \}: \mathbf{S}_{1},\mathbf{S}_{2}\in \{0,1\}^{N}, d_{H}(\mathbf{S}_{1}, \mathbf{S}_{2})=1, \mathbf{S}_{1}\in \mathcal{S}, \mathbf{S}_{2}\notin \mathcal{S} \right \}.
   \end{equation*}
Then, we have 
\begin{equation*}
\frac{|\partial_e \mathcal{S}|}{2^N} \geq \mu_{s}(\mathcal{S})\left\{1-\mu_{s}(\mathcal{S})\right\}.
\end{equation*}
\end{lemma}

Then, we state and prove the following three-set isoperimetric inequality for the hypercube with the Hamming distance:

\begin{lemma}[A Three-Set Isoperimetric Inequality for the Hypercube]\label{lemma: 3-set}
Let $\mu_{s}$ be the uniform measure over $\{0,1\}^{N}$. Let $\mathcal{S}_1,\mathcal{S}_2,\mathcal{S}_3$ form a partition of $\{0,1\}^N$, and let $d_H(\mathcal{S}_1,\mathcal{S}_2)=\min\{d_{H}(\mathbf{S}_{1}, \mathbf{S}_{2}): \mathbf{S}_{1}\in \mathcal{S}_{1}, \mathbf{S}_{2}\in \mathcal{S}_{2} \}$ denote the Hamming distance between $\mathcal{S}_{1}$ and $\mathcal{S}_{2}$, where we have $d_H(\mathcal{S}_1,\mathcal{S}_2)\geq 1$. Then, we have
\begin{equation*}
    \mu_{s}(\mathcal{S}_3) \geq \frac{1}{N} \cdot\left\{ d_H(\mathcal{S}_1,\mathcal{S}_2)-1\right\} \cdot \mu_{s}(\mathcal{S}_1) \mu_{s}(\mathcal{S}_2).
\end{equation*}
\end{lemma}
\vspace{-0.5cm}

\begin{proof}: Note that when $d_H(\mathcal{S}_1,\mathcal{S}_2)=1$, the desired result naturally holds. Therefore, throughout the proof, we consider the case when $d_H(\mathcal{S}_1,\mathcal{S}_2)\geq 2$.

For the given $\mathcal{S}_{1}\subset \{0,1\}^{N}$, let $B_{k}(\mathcal{S}_{1})=\{\mathbf{S}\in \{0,1\}^{N}: d_{H}(\mathbf{S}, \mathcal{S}_{1})\leq k\}$ denote the closed Hamming ball expansion of $\mathcal{S}_{1}$ with radius $k\geq 0$. Let $d_{\Delta}= d_H(\mathcal{S}_1,\mathcal{S}_2)\geq 2$. Then, we have $B_{k}(\mathcal{S}_{1})\cap \mathcal{S}_{2}=\emptyset$ for $k=0,1,\dots,d_{\Delta}-1$. Therefore, we have
\begin{equation}\label{eq:S1t-upper}
\mu_{s}\left (B_{k}(\mathcal{S}_{1})\right )\leq 1-\mu_{s}(\mathcal{S}_2) \text{ for any $k=0,1,\dots,d_{\Delta}-1$.}
\end{equation}
Because each vertex of $\{0,1\}^{N}$ has degree $N$, for any $\mathcal{S}\subset \{0,1\}^{N}$, the vertex external boundary 
\begin{equation*}
    \partial_v \mathcal{S}:=B_{1}(\mathcal{S})\setminus \mathcal{S}=\left\{\mathbf{S}\notin \mathcal{S}: d_{H}(\mathbf{S}, \mathcal{S})=1    \right\}
\end{equation*}
satisfies
\begin{equation}\label{eq:edge-to-vertex}
|\partial_v \mathcal{S}|\ \geq \frac{|\partial_e \mathcal{S}|}{N}.
\end{equation}
Therefore, for any $k=0,1,\dots, d_{\Delta}-2$, we obtain
\begin{align}\label{eqn: bounding dif in Bk}
  \mu_{s}\left(B_{k+1}(\mathcal{S}_{1})\right)-\mu_{s}\left(B_{k}(\mathcal{S}_{1})\right)&= \mu_{s}\left(\partial_v B_{k}(\mathcal{S}_{1})\right)\nonumber\\
  &=\frac{|\partial_v B_{k}(\mathcal{S}_{1})|}{2^N} \nonumber \\
&\geq \frac{1}{N}\cdot\frac{|\partial_e B_{k}(\mathcal{S}_{1})|}{2^N} \quad \text{(by (\ref{eq:edge-to-vertex}))}\nonumber \\
&\geq \frac{1}{N}\cdot \mu_{s}\left(B_{k}(\mathcal{S}_{1})\right)\left\{1-\mu_{s}\left(B_{k}(\mathcal{S}_{1})\right)\right\}  \text{ (by Lemma~\ref{lem: edge lower bound}). }
\end{align}
Meanwhile, note that $\mu_{s}(B_{k}(\mathcal{S}_{1}))\geq \mu_{s}(\mathcal{S}_{1})$ for any $k=0,1,\dots, d_{\Delta}-1$. Combining this fact with (\ref{eq:S1t-upper}), the inequality (\ref{eqn: bounding dif in Bk}) further implies: for any $k=0,1,\dots, d_{\Delta}-2$,
\begin{equation}\label{eq:step-growth}
\mu_{s}\left(B_{k+1}(\mathcal{S}_{1})\right)-\mu_{s}\left(B_{k}(\mathcal{S}_{1})\right)\geq \frac{1}{N}\cdot \mu_{s}(\mathcal{S}_1) \mu_{s}(\mathcal{S}_2).
\end{equation}
Summing \eqref{eq:step-growth} over $k=0,1,\dots,d_{\Delta}-2$, we have
\begin{align}\label{eq:r-growth}
\mu_{s}\left(B_{d_{\Delta}-1}(\mathcal{S}_{1})\right)-\mu_{s}\left(B_{0}(\mathcal{S}_{1})\right)&=\mu_{s}\left(B_{d_{\Delta}-1}(\mathcal{S}_{1})\right)-\mu_{s}\left(\mathcal{S}_{1}\right)\nonumber\\
&\geq \frac{1}{N}\cdot (d_{\Delta}-1)\cdot \mu_{s}(\mathcal{S}_1) \mu_{s}(\mathcal{S}_2).\nonumber 
\end{align}
Finally, since $ B_{d_{\Delta}-1}(\mathcal{S}_{1})\cap \mathcal{S}_{2}=\emptyset$, we have $B_{d_{\Delta}-1}(\mathcal{S}_{1})\setminus \mathcal{S}_{1} \subseteq \mathcal{S}_3$. Therefore, we have
\begin{equation*}
    \mu_{s}(\mathcal{S}_{3})\geq \mu_{s}\left(B_{d_{\Delta}-1}(\mathcal{S}_{1})\right)-\mu_{s}\left(\mathcal{S}_{1}\right)\geq \frac{1}{N}\cdot (d_{\Delta}-1)\cdot \mu_{s}(\mathcal{S}_1) \mu_{s}(\mathcal{S}_2).
\end{equation*}
That is, the desired result has been established. 
\end{proof}

Conditional on the observed data, including the observed design variables $\mathbf{Z}$, we define the mapping $\eta:\{0,1\}^{N}\rightarrow \mathbbm{R}$, where $\eta(\mathbf{S}^{(m)})=\widehat{\theta}^{(m)}$ for $\mathbf{S}^{(m)}\in \{0,1\}^{N}$. That is, the mapping $\eta=\eta_{1}\circ \eta_{2}$ is the composition of the following two mappings: (i) $\eta_{1}: \{0,1\}^{N}\rightarrow \mathbbm{R}^{N}$, where $\eta_{1}(\mathbf{S}^{(m)})=\widehat{\mathcal{P}}^{(m)}=(\widehat{p}_{1}^{(m)}, \dots, \widehat{p}_{N}^{(m)})$; and (ii) $\eta_{2}: \mathbbm{R}^{N}\rightarrow \mathbbm{R}$, where $\eta_{2}(\widehat{\mathcal{P}}^{(m)})=\widehat{\theta}^{(m)}$. Let $\mu_{s}$ denote the uniform measure over $\{0,1\}^{N}$. That is, we have $\mu_{s}(\mathbf{S})=2^{-N}$ for any $\mathbf{S}=(S_{1},\dots. S_{N})\in \{0,1\}^{N}$. Let $\widehat{\theta}_{\mathbf{S}}=\eta(\mathbf{S})$ denote the corresponding regenerated estimator under $\mathbf{S}\in \{0,1\}^{N}$. For any fixed $\alpha_{*}\in (0,\alpha)$, define the event $A_{\alpha_{*}}=\big \{P(\widehat{\theta}^{(m)}< L_{1-\alpha_{*}} \mid \mathbf{Z})\geq \gamma_{\alpha_{*}, N}\big\}\cap \big \{P(\widehat{\theta}^{(m)}> U_{1-\alpha_{*}} \mid \mathbf{Z})\geq \gamma_{\alpha_{*}, N}\big\}$.

\begin{lemma}\label{lem: partition construction}

Suppose that Condition~5 in the main text holds. Let \(K\) be an integer satisfying $1\le K\le \left\lfloor \rho_{\alpha_{*}, N}\cdot 2^{N-1}\right\rfloor$. Then, conditional on the observed data, there exists a partition $\{G_0, G_1, \dots, G_K, G_{K+1}\}$ of the $N$-dimensional hypercube $\{0,1\}^{N}$ with the following properties: (i) $\cup_{k=0}^{K+1} G_{k}=\{0,1\}^{N}$; (ii) $G_{i}\cap G_{j}=\emptyset$;  (iii) $\gamma_{\alpha_{*}, N}-2^{-N}
    \le
    \mu_s(G_0)=\mu_s(G_{K+1})
    \le
    \gamma_{\alpha_{*}, N} $; (iv) for each \(k=1,\dots, K\), we have $\frac{\rho_{\alpha_{*}, N}}{K}-2^{-N}
    \le
    \mu_s(G_k)
    \le
    \frac{\rho_{\alpha_{*}, N}}{K}+3\cdot 2^{-N}$; and (v) over the event $A_{\alpha_{*}}$, we have $\max_{\mathbf{S}\in G_0}\widehat{\theta}_{\mathbf{S}}<L_{1-\alpha_{*}}$ and $\min_{\mathbf{S}\in G_{K+1}}\widehat{\theta}_{\mathbf{S}}>U_{1-\alpha_{*}}$.
\end{lemma}

\vspace{-0.5cm}

\begin{proof}: All arguments in the construction are conditional on the observed data, including the observed design variables $\mathbf{Z}$. Let $\prec$ denote lexicographical order on \(\{0,1\}^{N} \). Then, we reorder the $2^{N}$ vertices $\mathbf{S}_{1}, \mathbf{S}_{2}, \dots, \mathbf{S}_{2^{N}}$ as $\mathbf{S}_{(1)}, \mathbf{S}_{(2)}, \dots, \mathbf{S}_{(2^{N})}$ so that $\eta (\mathbf{S}_{(1)})\leq \eta (\mathbf{S}_{(2)})\le\cdots\le \eta(\mathbf{S}_{(2^N)})$, with ties broken by \(\prec\). Let $n_0=\lfloor \gamma_{\alpha_{*}, N}\cdot 2^{N}\rfloor$. We define $G_0=\big\{\mathbf{S}_{(1)},\ldots,\mathbf{S}_{(n_0)}\big\}$ and $G_{K+1}=\big\{\mathbf{S}_{(2^N-n_0+1)},\ldots,\mathbf{S}_{(2^N)}\big\}$, and partition the remaining ordered vertices $\mathbf{S}_{(n_0+1)},\ldots,\mathbf{S}_{(2^{N}-n_0)}$ into \(K\) consecutive and non-overlapping blocks \(G_1,\ldots,G_K\) whose cardinalities differ by at most one. Because \(\mu_s\) is the uniform measure (i.e., each vertex has mass \(2^{-N}\)), we have 
\[
    \gamma_{\alpha_{*}, N}-2^{-N}
    \le
    \mu_s(G_0)=\mu_s(G_{K+1})
    \le
    \gamma_{\alpha_{*}, N} .
\]
Then, we have 
\[
   \sum_{k=1}^{K}\mu_{s}(G_{k})= 1-2\mu_s(G_0)
    =
    \rho_{\alpha_{*}, N}+2\{\gamma_{\alpha_{*}, N}-\mu_s(G_0)\}\in [\rho_{\alpha_{*}, N}, \, \rho_{\alpha_{*}, N}+2^{-N+1}].
\]
Since $\{\mu_{s}(G_{k})\}_{k=1}^{K}$ differ by at most $2^{-N}$, we have
\[
    \frac{\rho_{\alpha_{*}, N}}{K}-2^{-N}
    \le
    \mu_s(G_k)
    \le
    \frac{\rho_{\alpha_{*}, N}}{K}+3\cdot 2^{-N}, \quad k=1,\dots, K.
\]
Over the event \(A_{\alpha_{*}}\), we have
\[
    \mu_s(\{\mathbf{S}:\widehat{\theta}_{\mathbf{S}}<L_{1-\alpha_{*}}\})=P(\widehat{\theta}^{(m)}< L_{1-\alpha_{*}} \mid \mathbf{Z})\geq \gamma_{\alpha_{*}, N}\ge \mu_s(G_0),
\]
which implies that $\max_{\mathbf{S}\in G_0}\widehat{\theta}_{\mathbf{S}}<L_{1-\alpha_{*}}$. Similarly, we can show that $\min_{\mathbf{S}\in G_{K+1}}\widehat{\theta}_{\mathbf{S}}>U_{1-\alpha_{*}}$. 

Putting all the above results together, we have finished the construction. 
\end{proof}

Now, we are closer to proving Theorem~5 (i.e., Theorem~\ref{thm: approximation nonparametric full}). The key idea of the remaining proof is to apply Lemma~\ref{lemma: 3-set}, a three-set isoperimetric inequality on the hypercube, as well as the above construction for ordered partition of the hypercube, to adapt the isoperimetric argument developed in \citet{zheng2025perturbed} from the Euclidean space $\mathbbm{R}^{N}$ to the discrete hypercube $\{0,1\}^{N}$, which corresponds to the space of partition indicators $\mathbf{S}\in \{0,1\}^{N}$. To establish the final proof of Theorem 5, we will first state and prove Lemma~\ref{lem: bounds on min dif}.

Let \(K\) be an integer
satisfying $1\le K\le \left\lfloor \rho_{\alpha_{*}, N}\cdot 2^{N-1}\right\rfloor$. Let the event
\[
    A_{G}(K)
    =
    \bigcap_{k=1}^{K}
    \bigcup_{m=1}^M
    \big\{\mathbf{S}^{(m)}\in G_k\big\}
\]
denote the event that every middle block $\{G_{k}\}_{k=1}^{K}$ is hit by at least one regenerated partition indicator $\mathbf{S}^{(m)}\in \{0,1\}^{N}$. In addition, we define the event \begin{equation*}
    A_{\theta}
    =
    \left\{
    \max_{\mathbf{S}_1\neq \mathbf{S}_2}
    \frac{
        |\widehat\theta_{\mathbf{S}_1}-\widehat\theta_{\mathbf{S}_2}|
    }{
        N^{-1}d_H(\mathbf{S}_1,\mathbf{S}_2)
    }
    \le
    L'_\theta(\delta')L_{s,N}
    \right\}
\end{equation*}
and the event $A(\widehat\theta_{\rm oracle})=\big\{\widehat{\theta}_{\rm oracle}\in \mathcal{I}_{1-\alpha_{*}}(\widehat{\theta}_{\text{oracle}})\big\}$.

\begin{lemma}\label{lem: bounds on min dif}
Consider the blocks $G_{0}, G_{1}, \dots, G_{K}, G_{K+1}$ constructed in Lemma~\ref{lem: partition construction}. Over the event $A_{G}(K)\cap A_{\theta}\cap A_{\alpha_{*}}
    \cap A(\widehat\theta_{\rm oracle})$, we have
\[
    \min_{m\in[M]}
    \left|\widehat\theta^{(m)}-\widehat\theta_{\rm oracle}\right|
    \le L'_\theta(\delta')L_{s,N}
    \left\{
        \frac{\rho_{\alpha_{*}, N}/K+3\cdot 2^{-N}}
             {(\gamma_{\alpha_{*}, N}-2^{-N})^2}
        +
        \frac1N
    \right\}.
\]
\end{lemma}

\vspace{-0.5cm}
\begin{proof}: For each \(k\in\{1,\ldots,K\}\), we define $G_{k}^{-}=\bigcup_{j=0}^{k-1}G_j$ and $G_{k}^{+}=\bigcup_{j=k+1}^{K+1}G_j$. For each \(k\in\{1,\ldots,K\}\), note that
\[
    \mu_s(G_{k}^{-})\ge \mu_s(G_0)\ge \gamma_{\alpha_{*}, N}-2^{-N},
    \qquad
    \mu_s(G_{k}^{+})\ge \mu_s(G_{K+1})\ge \gamma_{\alpha_{*}, N}-2^{-N},
\]
and
\[
    \mu_s(G_k)\le \frac{\rho_{\alpha_{*}, N}}{K}+3\cdot 2^{-N}.
\]
Define $\Delta_k
    =
    \min_{\mathbf{S}\in G_{k}^{+}}\widehat{\theta}_{\mathbf{S}}
    -
    \max_{\mathbf{S}\in G_{k}^{-}}\widehat{\theta}_{\mathbf{S}}=\min_{\mathbf{S}\in G_{k}^{+}}\eta(\mathbf{S})
    -
    \max_{\mathbf{S}\in G_{k}^{-}}\eta(\mathbf{S})$. The blocks are consecutive in the \(\eta\)-ordering, so \(\Delta_k\ge0\). Over the event
\(A_{\theta}\), for every \(\mathbf{S}_{1}\in G_{k}^{-}\) and \(\mathbf{S}_{2}\in G_{k}^{+}\), we have
\[
  \Delta_{k}\leq  \left|\widehat{\theta}_{\mathbf{S}_{1}}-\widehat{\theta}_{\mathbf{S}_{2}}\right|
    \le
    L'_\theta(\delta')L_{s,N}\,N^{-1}d_H(\mathbf{S}_{1},\mathbf{S}_{2}),
\]
which implies that
\[
    d_H(G_{k}^{-},G_{k}^{+})
    \ge
    \frac{N\Delta_k}
         {L'_\theta(\delta')L_{s,N}} .
\]
Applying Lemma~\ref{lemma: 3-set} to the blocks $G_{k}^{-}, G_k, G_{k}^{+}$
gives
\[
    \mu_s(G_k)
    \ge
    \frac1N
    \left\{
        d_H(G_{k}^{-},G_{k}^{+})-1
    \right\}
    \mu_s(G_{k}^{-})\mu_s(G_{k}^{+}) .
\]
Combining the last two inequalities, we have
\[
    \mu_s(G_k)
    \ge
    \frac1N
    \left\{
        \frac{N\Delta_k}
             {L'_\theta(\delta')L_{s,N}}
        -1
    \right\}
    \mu_s(G_{k}^{-})\mu_s(G_{k}^{+}).
\]
That is, over the event $A_{\theta}$, we have
\[
    \Delta_k
    \le
    L'_\theta(\delta')L_{s,N}
    \left\{
        \frac{\mu_s(G_k)}{\mu_s(G_{k}^{-})\mu_s(G_{k}^{+})}
        +
        \frac1N
    \right\}.
\]
Using the bounds on \(\mu_s(G_k),\mu_s(G_{k}^{-}),\mu_s(G_{k}^{+})\) established at the beginning of the proof, over the event $A_{\theta}$, we have
\begin{equation}\label{eqn: upper bounds on delta_k}
        \Delta_k
    \le
    L'_\theta(\delta')L_{s,N}
    \left\{
        \frac{\rho_{\alpha_{*}, N}/K+3\cdot 2^{-N}}
             {(\gamma_{\alpha_{*}, N}-2^{-N})^2}
        +
        \frac1N
    \right\}.
\end{equation}
For each \(k=1,\dots, K\), define the interval $J_k=
    \left[
        \max_{\mathbf{S}\in G_{k}^{-}}\eta(\mathbf{S}),\,
        \min_{\mathbf{S}\in G_{k}^{+}}\eta(\mathbf{S})
    \right]$. Then, we have
\[
    \left[
        \max_{\mathbf{S}\in G_0}\eta(\mathbf{S}),\,
        \min_{\mathbf{S}\in G_{K+1}}\eta(\mathbf{S})
    \right]\subseteq \bigcup_{k=1}^{K}J_{k}.
\]
Over the event $A_{\alpha_{*}}$, by Lemma~\ref{lem: partition construction}, we have
\[
    \big[L_{1-\alpha_{*}}, U_{1-\alpha_{*}}\big]
    \subset
     \left[
        \max_{\mathbf{S}\in G_0}\eta(\mathbf{S}),\,
        \min_{\mathbf{S}\in G_{K+1}}\eta(\mathbf{S})
    \right] \subseteq \bigcup_{k=1}^{K}J_{k}.
\]
Therefore, over the event $A_{\alpha_{*}}\cap A(\widehat{\theta}_{\text{oracle}})$,
there exists \(k_{0}\in\{1,\ldots,K\}\) such that $\widehat\theta_{\rm oracle}\in J_{k_{0}}$. Also, over the event $A_{G}(K)$, there exists some \(m_{0}\), such that $\mathbf{S}^{(m_{0})}\in G_{k_{0}}$, which implies that $\eta(\mathbf{S}^{(m_{0})})\in J_{k_{0}}$.
Since the length of \(J_k\) equals \(\Delta_k\), by (\ref{eqn: upper bounds on delta_k}), we have: over the event $A_{G}(K)\cap A_{\theta}\cap A_{\alpha_{*}}
    \cap A(\widehat\theta_{\rm oracle})$,
\[
    \left|\widehat\theta^{(m_{0})}-\widehat\theta_{\rm oracle}\right|
    =
    \left|\eta(\mathbf{S}^{(m_{0})})-\widehat\theta_{\rm oracle}\right|\leq \Delta_k
    \le
    L'_\theta(\delta')L_{s,N}
    \left\{
        \frac{\rho_{\alpha_{*}, N}/K+3\cdot 2^{-N}}
             {(\gamma_{\alpha_{*}, N}-2^{-N})^2}
        +
        \frac1N
    \right\}.
\]
Therefore, over the event $A_{G}(K)\cap A_{\theta}\cap A_{\alpha_{*}}
    \cap A(\widehat\theta_{\rm oracle})$, we have
\begin{equation*}
     \min_{m\in[M]}
    \left|\widehat\theta^{(m)}-\widehat\theta_{\rm oracle}\right|
    \le L'_\theta(\delta')L_{s,N}
    \left\{
        \frac{\rho_{\alpha_{*}, N}/K+3\cdot 2^{-N}}
             {(\gamma_{\alpha_{*}, N}-2^{-N})^2}
        +
        \frac1N
    \right\}.
\end{equation*}
\end{proof}

Putting all the above lemmas and constructions together, we are ready to finish the proof of Theorem 5.

\begin{proof}[Proof of Theorem 5:] By Conditions 2c and 4, we have $\lim_{N\rightarrow \infty} P(A_{\theta})=1$. By Conditions 2a and 2b, the oracle confidence interval $\mathcal{C}_{1-\alpha_{*}}^{\text{oracle}}(\theta)$ has asymptotic coverage at least \(1-\alpha_{*}\), and therefore $\liminf_{N\to\infty}P(A(\widehat\theta_{\rm oracle}))\ge 1-\alpha_{*}$. By Condition~5, we have $\lim_{N\rightarrow \infty} P(A_{\alpha_{*}})=1$. Therefore, we have 
\begin{equation}\label{eqn: lower bound probability for three sets}
\liminf_{N\rightarrow \infty}P(A_{\theta} \cap A(\widehat\theta_{\rm oracle}) \cap A_{\alpha_{*}})\geq 1-\alpha_{*}. 
\end{equation}
For the \(K_{N, M}\) defined in Theorem 5, consider the ordered blocks (i.e., a partition of $\{0, 1\}^{N}$) constructed in Lemma~\ref{lem: partition construction}: $G_0,\ldots,G_{K_{N,M}+1}$. Let
\[
   A_{G}(K_{N, M})
    =
    \bigcap_{k=1}^{K_{N,M}}
    \bigcup_{m=1}^M
    \big\{\mathbf{S}^{(m)}\in G_k\big\}.
\]
Then, we can obtain that, 
\begin{align*}
  P\left(A_{G}(K_{N, M})\mid \mathbf{Z} \right) &=P\left(\cap_{k=1}^{K_{N, M}}\cup_{m=1}^{M}\left\{\mathbf{S}^{(m)}\in G_k \right\}\mid \mathbf{Z} \right)\\
    &=1-P\left(\cup_{k=1}^{K_{N, M}}\cap_{m=1}^{M}\left\{\mathbf{S}^{(m)}\notin G_k \right\}\mid \mathbf{Z}   \right)\\
    &\geq 1-\sum_{k=1}^{K_{N, M}}P\left(\cap_{m=1}^{M}\left\{\mathbf{S}^{(m)}\notin G_k\right\}\mid \mathbf{Z}  \right)\\
    &=1-\sum_{k=1}^{K_{N, M}}\left\{ P\left(\mathbf{S}^{(m)}\notin G_k \mid \mathbf{Z}  \right)\right\}^{M}\\
    &=1-\sum_{k=1}^{K_{N, M}}\left\{1-P\left(\mathbf{S}^{(m)}\in G_k \mid \mathbf{Z}  \right)  \right\}^{M}\\
   &= 1-\sum_{k=1}^{K_{N, M}}\left\{1-\mu_{s} (G_{k})   \right\}^{M}\\
   &\geq 1-\sum_{k=1}^{K_{N, M}}\exp\left\{ -M\cdot \mu_{s} (G_{k}) \right\}.
\end{align*}
By Lemma~\ref{lem: partition construction} and the fact that $K_{N,M}\le \rho_{\alpha_{*}, N}\cdot 2^{N-1}$, we have
\[
    \mu_s(G_k)
    \ge
    \frac{\rho_{\alpha_{*}, N}}{K_{N,M}}-2^{-N}\geq  \frac{\rho_{\alpha_{*}, N}}{2K_{N,M}}.
\]
Therefore, we have
\[
P\left(A_{G}(K_{N, M})\mid \mathbf{Z} \right)
    \ge 1-
    K_{N,M}\exp\left\{
        -\frac{M\rho_{\alpha_{*}, N}}{2K_{N,M}}
    \right\}.
\]
Since $K\exp\big\{ -\frac{M\rho_{\alpha_{*}, N}}{2K} \big\}$ is strictly monotonically increasing in $K$ (for any fixed $M>0$) and strictly monotonically decreasing in $M$ (for any fixed $K>0$), we have
\[
    K_{N,M}
    \le
    \frac{\rho_{\alpha_{*}, N}M}
         {2\varphi^{-1}(\rho_{\alpha_{*}, N}M\sqrt{N}/2)} \Rightarrow  K_{N,M}\exp\left\{
        -\frac{M\rho_{\alpha_{*}, N}}{2K_{N,M}}
    \right\} \leq \frac{1}{\sqrt{N}}.
\]
That is, for sufficiently large $M$, we have 
\[
P\left(A_{G}(K_{N, M})\mid \mathbf{Z} \right)
    \ge 1-\frac{1}{\sqrt{N}},
\]
which implies that 
\begin{equation*}
  P\left(A_{G}(K_{N, M})\right)=E_{\mathbf{Z}}\left\{ P\left(A_{G}(K_{N, M})\mid \mathbf{Z} \right) \right\}\geq  1-\frac{1}{\sqrt{N}}.
\end{equation*}
Therefore, by (\ref{eqn: lower bound probability for three sets}), we have
\begin{align*}
  &\quad \ \liminf_{N\rightarrow \infty}\lim_{M\rightarrow \infty} P\left(\min_{m \in [M]} \left|\widehat{\theta}^{(m)}-\widehat{\theta}_{\text{oracle}}\right|\leq   \varepsilon_{N, M}\right)\\
  &\geq \liminf_{N\rightarrow \infty}\lim_{M\rightarrow \infty}  P\left(A_{G}(K_{N, M})\cap A_{\theta} \cap A(\widehat\theta_{\rm oracle}) \cap A_{\alpha_{*}}  \right)\\
  &\geq \liminf_{N\rightarrow \infty}\lim_{M\rightarrow \infty} \left\{P\left(A_{G}(K_{N, M})\right)+P\left(A_{\theta} \cap A(\widehat\theta_{\rm oracle}) \cap A_{\alpha_{*}}  \right) -1\right\}\\
  &\geq \liminf_{N\rightarrow \infty}\lim_{M\rightarrow \infty} P\left(A_{G}(K_{N, M})\right)-\alpha_{*}\\
  &\geq \liminf_{N\rightarrow \infty}\left\{ 1-\frac{1}{\sqrt{N}}\right\}-\alpha_{*}=1-\alpha_{*}.
\end{align*}
It remains only to show that $ \lim_{N\to\infty}\lim_{M\rightarrow \infty}
   \frac{\varepsilon_{N,M}}{N^{-1/2}}=0$. Since
\begin{equation*}
    \lim_{M\rightarrow \infty}\frac{\rho_{\alpha_{*}, N}M}{2\varphi^{-1}(\rho_{\alpha_{*}, N}M\sqrt N/2)}=\infty,
\end{equation*}
we have
\begin{align*}
  \lim_{M\rightarrow \infty}  K_{N,M}&= \lim_{M\rightarrow \infty} \left\lfloor\min\left\{\frac{\rho_{\alpha_{*}, N}M}{2\varphi^{-1}(\rho_{\alpha_{*}, N}M\sqrt N/2)},\, \rho_{\alpha_{*}, N}\cdot 2^{N-1}\right\}\right\rfloor\\
  &= \left\lfloor \rho_{\alpha_{*}, N}\cdot 2^{N-1}\right\rfloor\\
  &\geq \rho_{\alpha_{*}, N}\cdot 2^{N-2}.
\end{align*}
Therefore, we have
    \begin{align*}
        \lim_{M\to\infty}\frac{\varepsilon_{N,M}}{N^{-1/2}}&=\lim_{M\rightarrow \infty} \sqrt{N}L'_\theta(\delta')L_{s,N}\left\{\frac{\rho_{\alpha_{*}, N}/K_{N,M}+3\cdot 2^{-N}}{(\gamma_{\alpha_{*}, N}-2^{-N})^2}+\frac{1}{N}\right\}\\
        &\leq   \sqrt{N} L'_\theta(\delta')L_{s,N}\left\{\frac{7\cdot 2^{-N}}{(\gamma_{\alpha_{*}, N}-2^{-N})^2}+\frac{1}{N}\right\}.
    \end{align*}
Since \(L_{s,N}\) is bounded under Condition~4 and
\(\sqrt{N}\cdot 2^{-N}=o(\gamma_{\alpha_{*}, N}^2)\), the right-hand side of the above inequality converges to zero as
\(N\to\infty\). This proves the desired result.
\end{proof}

\vspace{-0.5cm}

Next, we prove Theorems 6 and 7 in the main text.

\begin{proof}[Proof of Theorem~6:] Let $m^{*}\in [M]$ such that $\min_{m \in [M]} |\widehat{\theta}^{(m)}-\widehat{\theta}_{\text{oracle}}|=|\widehat{\theta}^{(m^{*})}-\widehat{\theta}_{\text{oracle}}|$. Therefore, applying Theorem~5, we have 
\begin{equation}\label{eqn: bound tail prob of min dif in tau m and tau oracle scaled by sqrtN}
       \liminf_{N\rightarrow \infty} \lim_{M\rightarrow \infty} P\left(\min_{m \in [M]} \sqrt{N}\left|\widehat{\theta}^{(m)}-\widehat{\theta}_{\text{oracle}}\right|\leq \epsilon \right)\geq 1-\alpha_{*} \text{  for any $\epsilon>0$}. 
\end{equation}
By Condition~6, there exist sequences $c_{N}$ and $\epsilon_N$ such that $c_{N}\to 0$ and $\epsilon_N\to 0$ as $N\to\infty$, and
$ P\bigl(N^{-1/2}\|\widehat{\mathcal{P}}^{(m^{*})}-\mathcal{P}\|_2 \leq \epsilon_N\bigr)\geq 1-c_{N}$. Let $A_{12}=\{N^{-1/2}\|\widehat{\mathcal{P}}^{(m^{*})}- \mathcal{P} \|_{2}\leq \epsilon_{N} \}$, $A_{13}=\big\{N|\widehat{\mathcal{V}}^{(m^{*})}-\widehat{\mathcal{V}}_{\text{oracle}}|\leq L^{\prime}_{v}(\delta^{\prime})\cdot N^{-1/2}\|\widehat{\mathcal{P}}^{(m^{*})}-\mathcal{P}\|_{2}\big\}$, and $A_{14}=\{\sqrt{N}|\widehat{\theta}^{(m^{*})}-\widehat{\theta}_{\text{oracle}}|\leq \epsilon\}$. Therefore, when $A_{5}\cap A_{12}\cap A_{13}$ happens, under $\delta^{\prime}<\delta$ (recall that $\delta^{\prime}$ is defined in Section 3.3 of the main text), we have 
\begin{align*}
    \left|\{\widehat{\mathcal{V}}^{(m^{*})}\}^{1/2}-\widehat{\mathcal{V}}_{\text{oracle}}^{1/2}\right|&=\frac{ \left|\widehat{\mathcal{V}}^{(m^{*})}-\widehat{\mathcal{V}}_{\text{oracle}}\right|}{\{\widehat{\mathcal{V}}^{(m^{*})}\}^{1/2}+\widehat{\mathcal{V}}_{\text{oracle}}^{1/2}}\\
    &\leq \frac{\sqrt{N}}{\nu_{1}} \left|\widehat{\mathcal{V}}^{(m^{*})}-\widehat{\mathcal{V}}_{\text{oracle}}\right|\\
    &\leq \frac{L^{\prime}_{v}(\delta^{\prime}) }{\nu_{1} N}\left\|\widehat{\mathcal{P}}^{(m^{*})}-\mathcal{P}  \right \|_{2}\\
    &\leq \frac{L^{\prime}_{v}(\delta^{\prime})\epsilon_{N} }{\nu_{1} \sqrt{N}}.
\end{align*}
Let $L_{1-\alpha}^{(m^{*})}=\widehat{\theta}^{(m^{*})}-z_{1-\alpha/2}\cdot \{\widehat{\mathcal{V}}^{(m^{*})}\}^{1/2}$ and $U_{1-\alpha}^{(m^{*})}=\widehat{\theta}^{(m^{*})}+z_{1-\alpha/2}\cdot \{\widehat{\mathcal{V}}^{(m^{*})}\}^{1/2}$ denote the left and right end points of $\mathcal{C}_{1-\alpha}^{(m^{*})}(\theta)$, respectively. Let $L_{1-\alpha}^{\text{oracle}}=\widehat{\theta}_{\text{oracle}}-z_{1-\alpha/2}\cdot \widehat{\mathcal{V}}_{\text{oracle}}^{1/2}$ and $U_{1-\alpha}^{\text{oracle}}=\widehat{\theta}_{\text{oracle}}+z_{1-\alpha/2}\cdot \widehat{\mathcal{V}}_{\text{oracle}}^{1/2}$ denote the left and right end points of $\mathcal{C}_{1-\alpha}^{\text{oracle}}(\theta)$, respectively. Let $\Delta_{L, N}=|L_{1-\alpha}^{(m^{*})}-L_{1-\alpha}^{\text{oracle}} |$ and $\Delta_{U, N}=|U_{1-\alpha}^{(m^{*})}-U_{1-\alpha}^{\text{oracle}}|$. When $A_{5}\cap A_{12}\cap A_{13}\cap A_{14}$ happens, we have
\begin{align*}
   \frac{\Delta_{L, N}}{\widehat{\mathcal{V}}_{\text{oracle}}^{1/2}}&=\frac{\left|L_{1-\alpha}^{(m^{*})}-L_{1-\alpha}^{\text{oracle}} \right|}{\widehat{\mathcal{V}}_{\text{oracle}}^{1/2}}\\
   &\leq \frac{\sqrt{N}}{\nu_{1}} \left | \left(\widehat{\theta}^{(m^{*})}-z_{1-\alpha/2}\cdot \{\widehat{\mathcal{V}}^{(m^{*})}\}^{1/2} \right)-\left( \widehat{\theta}_{\text{oracle}}-z_{1-\alpha/2}\cdot \widehat{\mathcal{V}}_{\text{oracle}}^{1/2}\right)\right|\\
    &\leq \frac{\sqrt{N}}{\nu_{1}}\left | \widehat{\theta}^{(m^{*})}-\widehat{\theta}_{\text{oracle}}\right|+z_{1-\alpha/2}\cdot \frac{\sqrt{N}}{\nu_{1}} \left | \{\widehat{\mathcal{V}}^{(m^{*})}\}^{1/2}- \widehat{\mathcal{V}}_{\text{oracle}}^{1/2}\right|\\
    &\leq \frac{\epsilon}{\nu_{1}}+z_{1-\alpha/2}\cdot \frac{L^{\prime}_{v}(\delta^{\prime})\epsilon_{N} }{\nu^{2}_{1}},
\end{align*}
and 
\begin{align*}
   \frac{\Delta_{U, N}}{\widehat{\mathcal{V}}_{\text{oracle}}^{1/2}}&=\frac{\left|U_{1-\alpha}^{(m^{*})}-U_{1-\alpha}^{\text{oracle}} \right|}{\widehat{\mathcal{V}}_{\text{oracle}}^{1/2}}\\
   &\leq \frac{\sqrt{N}}{\nu_{1}} \left | \left(\widehat{\theta}^{(m^{*})}+z_{1-\alpha/2}\cdot \{\widehat{\mathcal{V}}^{(m^{*})}\}^{1/2} \right)-\left( \widehat{\theta}_{\text{oracle}}+z_{1-\alpha/2}\cdot \widehat{\mathcal{V}}_{\text{oracle}}^{1/2}\right)\right|\\
       &\leq \frac{\sqrt{N}}{\nu_{1}}\left | \widehat{\theta}^{(m^{*})}-\widehat{\theta}_{\text{oracle}}\right|+z_{1-\alpha/2}\cdot \frac{\sqrt{N}}{\nu_{1}} \left | \{\widehat{\mathcal{V}}^{(m^{*})}\}^{1/2}- \widehat{\mathcal{V}}_{\text{oracle}}^{1/2}\right|\\
    &\leq \frac{\epsilon}{\nu_{1}}+z_{1-\alpha/2}\cdot \frac{L^{\prime}_{v}(\delta^{\prime})\epsilon_{N} }{\nu^{2}_{1}}.
\end{align*}
Therefore, over $A_{5}\cap A_{12}\cap A_{13}\cap A_{14}$, we have $\max\left\{\frac{\Delta_{L, N}}{\widehat{\mathcal{V}}_{\text{oracle}}^{1/2}}, \frac{\Delta_{U, N}}{\widehat{\mathcal{V}}_{\text{oracle}}^{1/2}}\right\}\leq \frac{2\epsilon}{\nu_{1}}$ for any sufficiently large $N$ such that $z_{1-\alpha/2}\cdot \frac{L^{\prime}_{v}(\delta^{\prime})\epsilon_{N}}{\nu_{1}}\leq \epsilon$. Meanwhile, we have $P(A_{12})\geq 1-c_{N}$. Also, we have $\lim_{N\rightarrow \infty}P(A_{5})=1$ under $\widehat{\mathcal{V}}^{1/2}_{\text{oracle}}\asymp_{p}N^{-1/2}$ and $\lim_{N\rightarrow \infty}P(A_{13})=1$ under Conditions 1a, 1b, and 2c. Then, we have 
\begin{align*}
\liminf_{N\rightarrow\infty} P\left(\max\left\{\frac{\Delta_{L, N}}{\widehat{\mathcal{V}}_{\text{oracle}}^{1/2}}, \frac{\Delta_{U, N}}{\widehat{\mathcal{V}}_{\text{oracle}}^{1/2}}\right\}\leq \frac{2\epsilon}{\nu_{1}} \right)
&\geq \liminf_{N\rightarrow \infty}P\left(A_{5}\cap A_{12}\cap A_{13}\cap A_{14}\right)\\
&\geq \liminf_{N\rightarrow \infty}P\left(A_{12}\cap A_{14}\right)\\
&\geq \liminf_{N\rightarrow \infty}\left\{P\left(A_{12}\right)+P\left(A_{14}\right)-1\right\} \\
 &\geq \liminf_{N\rightarrow \infty}\left\{P\left(A_{14}\right)-c_{N}\right\} \\
  &\geq 1-\alpha_{*}\quad \text{(by (\ref{eqn: bound tail prob of min dif in tau m and tau oracle scaled by sqrtN}) and $\lim_{N\rightarrow \infty}c_{N}=0$)}.
\end{align*}
Note that $\left[L_{1-\alpha}^{\text{oracle}}+\Delta_{L, N}, U_{1-\alpha}^{\text{oracle}}-\Delta_{U, N}\right]\subseteq \mathcal{C}_{1-\alpha}^{(m^{*})}(\theta)\subseteq \mathcal{C}_{1-\alpha}(\theta)$, and $\theta\in \left[L_{1-\alpha}^{\text{oracle}}+\Delta_{L, N}, U_{1-\alpha}^{\text{oracle}}-\Delta_{U, N}\right]$ if and only if  $\frac{\widehat{\theta}_{\text{oracle}}-\theta}{\widehat{\mathcal{V}}_{\text{oracle}}^{1/2}}\leq z_{1-\alpha/2}-\frac{\Delta_{L, N}}{\widehat{\mathcal{V}}_{\text{oracle}}^{1/2}}$ and $\frac{\widehat{\theta}_{\text{oracle}}-\theta}{\widehat{\mathcal{V}}_{\text{oracle}}^{1/2}}\geq -z_{1-\alpha/2}+\frac{\Delta_{U, N}}{\widehat{\mathcal{V}}_{\text{oracle}}^{1/2}}$. Therefore, we have
\begin{align*}
&\quad \ P\left(\theta\in \mathcal{C}_{1-\alpha}(\theta) \right)\\
&\geq P\left(\theta\in \mathcal{C}_{1-\alpha}^{(m^{*})}(\theta) \right)\\
&\geq P\left(\theta\in \left[L_{1-\alpha}^{\text{oracle}}+\Delta_{L, N}, U_{1-\alpha}^{\text{oracle}}-\Delta_{U, N}\right] \right)\\
 &=P\left(-z_{1-\alpha/2}+\frac{\Delta_{U, N}}{\widehat{\mathcal{V}}_{\text{oracle}}^{1/2}}\leq \frac{\widehat{\theta}_{\text{oracle}}-\theta}{\widehat{\mathcal{V}}_{\text{oracle}}^{1/2}}\leq z_{1-\alpha/2}-\frac{\Delta_{L, N}}{\widehat{\mathcal{V}}_{\text{oracle}}^{1/2}} \right)\\
 &= P\left(-z_{1-\alpha/2}+\frac{\Delta_{U, N}}{\widehat{\mathcal{V}}_{\text{oracle}}^{1/2}}\leq \frac{\widehat{\theta}_{\text{oracle}}-\theta}{\widehat{\mathcal{V}}_{\text{oracle}}^{1/2}}\leq z_{1-\alpha/2}-\frac{\Delta_{L, N}}{\widehat{\mathcal{V}}_{\text{oracle}}^{1/2}}, \max\left\{\frac{\Delta_{L, N}}{\widehat{\mathcal{V}}_{\text{oracle}}^{1/2}}, \frac{\Delta_{U, N}}{\widehat{\mathcal{V}}_{\text{oracle}}^{1/2}}\right\}\leq \frac{2\epsilon}{\nu_{1}} \right)\\
 &\quad \quad +P\left(-z_{1-\alpha/2}+\frac{\Delta_{U, N}}{\widehat{\mathcal{V}}_{\text{oracle}}^{1/2}}\leq \frac{\widehat{\theta}_{\text{oracle}}-\theta}{\widehat{\mathcal{V}}_{\text{oracle}}^{1/2}}\leq z_{1-\alpha/2}-\frac{\Delta_{L, N}}{\widehat{\mathcal{V}}_{\text{oracle}}^{1/2}}, \max\left\{\frac{\Delta_{L, N}}{\widehat{\mathcal{V}}_{\text{oracle}}^{1/2}}, \frac{\Delta_{U, N}}{\widehat{\mathcal{V}}_{\text{oracle}}^{1/2}}\right\}>\frac{2\epsilon}{\nu_{1}}  \right)\\
 &\geq P\left(-z_{1-\alpha/2}+\frac{2\epsilon}{\nu_{1}} \leq \frac{\widehat{\theta}_{\text{oracle}}-\theta}{\widehat{\mathcal{V}}_{\text{oracle}}^{1/2}}\leq z_{1-\alpha/2}-\frac{2\epsilon}{\nu_{1}}, \max\left\{\frac{\Delta_{L, N}}{\widehat{\mathcal{V}}_{\text{oracle}}^{1/2}}, \frac{\Delta_{U, N}}{\widehat{\mathcal{V}}_{\text{oracle}}^{1/2}}\right\}\leq \frac{2\epsilon}{\nu_{1}} \right)\\
 &\geq  P\left(-z_{1-\alpha/2}+\frac{2\epsilon}{\nu_{1}} \leq \frac{\widehat{\theta}_{\text{oracle}}-\theta}{\widehat{\mathcal{V}}_{\text{oracle}}^{1/2}}\leq z_{1-\alpha/2}-\frac{2\epsilon}{\nu_{1}}\right)- P\left(\max\left\{\frac{\Delta_{L, N}}{\widehat{\mathcal{V}}_{\text{oracle}}^{1/2}}, \frac{\Delta_{U, N}}{\widehat{\mathcal{V}}_{\text{oracle}}^{1/2}}\right\}> \frac{2\epsilon}{\nu_{1}} \right).
\end{align*}
Therefore, under Conditions 2a and 2b, for any sufficiently small $\epsilon>0$, we have
\begin{align*}
&\quad \ \liminf_{N\rightarrow \infty} \lim_{M\rightarrow \infty} P\left(\theta\in \mathcal{C}_{1-\alpha}(\theta) \right)\\
  &\geq \liminf_{N\rightarrow \infty}  P\left(-z_{1-\alpha/2}+\frac{2\epsilon}{\nu_{1}} \leq \frac{\widehat{\theta}_{\text{oracle}}-\theta}{\widehat{\mathcal{V}}_{\text{oracle}}^{1/2}}\leq z_{1-\alpha/2}-\frac{2\epsilon}{\nu_{1}}\right)\\
  &\quad \quad \quad - \limsup_{N\rightarrow \infty}\lim_{M\rightarrow \infty} P\left(\max\left\{\frac{\Delta_{L, N}}{\widehat{\mathcal{V}}_{\text{oracle}}^{1/2}}, \frac{\Delta_{U, N}}{\widehat{\mathcal{V}}_{\text{oracle}}^{1/2}}\right\}> \frac{2\epsilon}{\nu_{1}} \right)\\
  &\geq \liminf_{N\rightarrow \infty} P\left(\frac{\widehat{\theta}_{\text{oracle}}-\theta}{\widehat{\mathcal{V}}_{\text{oracle}}^{1/2}}\leq z_{1-\alpha/2}-\frac{2\epsilon}{\nu_{1}} \right)-\limsup_{N\rightarrow \infty}P\left(\frac{\widehat{\theta}_{\text{oracle}}-\theta}{\widehat{\mathcal{V}}_{\text{oracle}}^{1/2}}< -z_{1-\alpha/2}+\frac{2\epsilon}{\nu_{1}} \right)\\
 &\quad \quad + \liminf_{N\rightarrow \infty}\lim_{M\rightarrow \infty} P\left(\max\left\{\frac{\Delta_{L, N}}{\widehat{\mathcal{V}}_{\text{oracle}}^{1/2}}, \frac{\Delta_{U, N}}{\widehat{\mathcal{V}}_{\text{oracle}}^{1/2}}\right\}\leq \frac{2\epsilon}{\nu_{1}} \right)-1\\
 &= \liminf_{N\rightarrow \infty} P\left(\frac{\widehat{\theta}_{\text{oracle}}-\theta}{\{\text{Var}(\widehat{\theta}_{\text{oracle}})\}^{1/2}}\frac{\{\text{Var}(\widehat{\theta}_{\text{oracle}})\}^{1/2}}{\{\widehat{\mathcal{V}}_{\text{oracle}}\}^{1/2}}\leq z_{1-\alpha/2}-\frac{2\epsilon}{\nu_{1}} \right)\\
 &\quad \quad -\limsup_{N\rightarrow \infty}P\left(\frac{\widehat{\theta}_{\text{oracle}}-\theta}{\{\text{Var}(\widehat{\theta}_{\text{oracle}})\}^{1/2}}\frac{\{\text{Var}(\widehat{\theta}_{\text{oracle}})\}^{1/2}}{\{\widehat{\mathcal{V}}_{\text{oracle}}\}^{1/2}}< -z_{1-\alpha/2}+\frac{2\epsilon}{\nu_{1}} \right) \\
 &\quad \quad + \liminf_{N\rightarrow \infty}\lim_{M\rightarrow \infty} P\left(\max\left\{\frac{\Delta_{L, N}}{\widehat{\mathcal{V}}_{\text{oracle}}^{1/2}}, \frac{\Delta_{U, N}}{\widehat{\mathcal{V}}_{\text{oracle}}^{1/2}}\right\}\leq \frac{2\epsilon}{\nu_{1}} \right)-1\\
 &\geq \Phi\left( z_{1-\alpha/2}-\frac{2\epsilon}{\nu_{1}}\right)-\Phi\left(-z_{1-\alpha/2}+\frac{2\epsilon}{\nu_{1}}\right)-\alpha_{*}.
\end{align*}
Setting $\epsilon\rightarrow 0$, by the continuity of $\Phi$, we have: for any $\alpha_{*}\in (0,\alpha)$, we have
\begin{align*}
  \liminf_{N\rightarrow \infty} \lim_{M\rightarrow \infty}P\left(\theta\in \mathcal{C}_{1-\alpha}(\theta) \right)\geq \Phi\left(z_{1-\alpha/2}\right)-\Phi\left(-z_{1-\alpha/2}\right)-\alpha_{*}=1-\alpha-\alpha_{*}.
\end{align*}
In other words, we have $\liminf_{N\rightarrow \infty}\lim_{M\rightarrow \infty} P(\theta\in \mathcal{C}_{1-\alpha}(\theta) )\geq 1-\alpha-\alpha_{*}$ for any $\alpha_{*}\in (0,\alpha)$. Therefore, setting $\alpha_{*} \rightarrow 0$, we have 
\begin{align*}
  \liminf_{N\rightarrow \infty}\lim_{M\rightarrow \infty} P\left(\theta\in \mathcal{C}_{1-\alpha}(\theta) \right)\geq 1-\alpha. 
\end{align*}\end{proof}

\vspace{-0.5cm}

\begin{proof}[Proof of Theorem~7:]

 Note that 
\begin{align}\label{eqn: bound the length nonparametric restricted}
 &\quad \  \mu_{L}\big\{\mathcal{C}_{1-\alpha}(\theta)\big\}\nonumber\\
 &\leq \max_{m_{1}, m_{2}\in [M_{N}] }\left |\left(\widehat{\theta}^{(m_{1})}-z_{1-\alpha/2}\cdot \{\widehat{\mathcal{V}}^{(m_{1})}\}^{1/2} \right)- \left(\widehat{\theta}^{(m_{2})}+z_{1-\alpha/2}\cdot \{\widehat{\mathcal{V}}^{(m_{2})}\}^{1/2}\right) \right |\nonumber\\
   &\leq \max_{m_{1}, m_{2}\in [M_{N}] }\left |\widehat{\theta}^{(m_{1})}-\widehat{\theta}^{(m_{2})}\right |+2\times z_{1-\alpha/2}\times \max_{m\in [M_{N}] }\{\widehat{\mathcal{V}}^{(m)}\}^{1/2}.
\end{align}
When the event $A_{\theta}$ happens, we have 
\begin{align}\label{eqn: bound on tau dif by uniform lip constant}
    \max_{m_{1}, m_{2}\in [M_{N}] }\left |\widehat{\theta}^{(m_{1})}-\widehat{\theta}^{(m_{2})}\right |
    &\leq \frac{L^{\prime}_{\theta}(\delta^{\prime})L_{s,N} }{N}\max_{m_{1}, m_{2}\in  [M_{N}]} d_{H}\left(\mathbf{S}^{(m_{1})}, \mathbf{S}^{(m_{2})} \right)\nonumber\\
    &\leq L^{\prime}_{\theta}(\delta^{\prime})L_{s,N}.
\end{align}
Recall that $\lim_{N\rightarrow \infty}P(A_{\theta})=1$ under the considered conditions. Combining (\ref{eqn: bound on tau dif by uniform lip constant}) and the assumption that the Lipschitz constants $L_{s, N}$ go to zero as $N\rightarrow \infty$, we have 
\begin{equation}\label{eqn: tau dif to zero nonparametric}
    \max_{m_{1}, m_{2}\in [M_{N}] }\left |\widehat{\theta}^{(m_{1})}-\widehat{\theta}^{(m_{2})}\right |\xrightarrow{p} 0 \text{ as $N\rightarrow \infty$}.
\end{equation}
Based on (\ref{eqn: bound the length nonparametric restricted}), (\ref{eqn: tau dif to zero nonparametric}), and the assumption that $\max_{m\in[M_N]} \widehat{\mathcal{V}}^{(m)} = o_p(1)$ (i.e., the uniform vanishing condition of regenerated variance estimators), we have $\mu_{L}\{\mathcal{C}_{1-\alpha}(\theta)\}\xrightarrow{p} 0$ as $N\rightarrow\infty$. \end{proof}

\section*{Appendix B: Additional Discussions and Remarks}

\subsection*{B.1: Limitations of Existing Approaches}

As discussed in the main text, existing approaches to design-based inference with unknown propensity scores can be broadly grouped into three classes: plug-in approaches, finite-population M-estimation approaches, and matching-based approaches. We now describe their limitations in more technical terms.

\begin{itemize}
    \item \textbf{Plug-in Approaches:} Given the established design-based inference methods under known propensity scores, a common and straightforward approach when the propensity scores $\mathcal{P}$ are unknown is to substitute an estimate $\widehat{\mathcal{P}}$, obtained from $(\mathbf{Z}, \mathbf{X})$, into the oracle confidence set mapping $\Lambda$, yielding the plug-in confidence set $\mathcal{C}_{1-\alpha}^{\text{plug-in}} = \Lambda(\widehat{\mathcal{P}}, \mathbf{Z}, \mathbf{Y})$. However, as shown in previous work (e.g., \citealp{pimentel2024covariate, zhu2025randomization}) and in the simulation studies of Section~6, this approach can suffer from severe under-coverage because it treats $\widehat{\mathcal{P}}$ as if it were fixed and known, thereby failing to account for the discrepancy between $\widehat{\mathcal{P}}$ and the oracle propensity scores $\mathcal{P}$ and, more fundamentally, the uncertainty introduced by propensity score estimation.

\item \textbf{Finite-Population M-Estimation Approaches:} A more principled alternative to the plug-in approach is finite-population M-estimation, which treats the estimation of the propensity score model parameters and the target estimand as a joint estimating equation problem and derives a corresponding design-based sandwich-type variance estimator that accounts for both sources of uncertainty \citep{abadie2020sampling, xu2021potential}. In principle, this avoids the main deficiency of the plug-in approach, namely, ignoring uncertainty from estimating $\mathcal{P}$. However, existing finite-population M-estimation approaches are largely restricted to parametric propensity score models, which limits their applicability when propensity scores are estimated using flexible nonparametric or machine-learning-based learners. In addition, the resulting variance formulas are often model-specific and technically cumbersome, making them difficult to transport across different design-based inference problems or to extend to settings such as sensitivity analysis; see \citet{zhao2019sens_boot} for a detailed discussion.

\item \textbf{Matching-Based Approaches:} In observational studies, a common nonparametric alternative to plug-in approaches is matching, which preprocesses the observed data into matched sets based on covariate similarity and then conducts design-based inference under an assumed post-matching randomization mechanism. For notational simplicity, we focus on matched pairs; the same issues arise for more general matched sets, as discussed in Remark~\ref{rem: general matching} in Appendix~B.3. Suppose that there are $I$ matched pairs, each containing one treated unit and one control unit. Following the classic framework for matched observational studies \citep{rosenbaum2002observational, rosenbaum2020design}, let $Z_{ij}\in\{0,1\}$ denote the treatment indicator and $\mathbf{x}_{ij}$ the observed covariates for unit $j=1,2$ in matched pair $i=1,\dots,I$, and let $\mathbf{Z}=(Z_{11},Z_{12},\dots,Z_{I1},Z_{I2})$ and $\mathbf{X}=(\mathbf{x}_{11},\mathbf{x}_{12},\dots,\mathbf{x}_{I1},\mathbf{x}_{I2})$. Let $\mathcal{Z}=\{(i1,i2):i=1,\dots,I\}$ denote the structure of the formed matched pairs, encoding which two units are paired to form each matched pair. If matching is exact on covariates, that is, $\mathbf{x}_{i1}=\mathbf{x}_{i2}$ for all $i$, then under the ignorability condition (i.e., no unobserved covariates), the post-matching treatment assignment is uniform within each pair: 
\begin{equation*}
    \text{$P(Z_{i1}=1,Z_{i2}=0\mid \mathcal{Z},\mathbf{X})=P(Z_{i1}=0,Z_{i2}=1\mid \mathcal{Z},\mathbf{X})=1/2$ for $i=1,\dots,I$. }
\end{equation*}
This can be viewed as a matched-pair version of the post-matching uniform propensity condition \citep{rosenbaum2002observational, rosenbaum2020design}. Existing work often treats matched observational studies as paired or stratified randomized experiments and applies established design-based inference methods under this post-matching uniform propensity condition. However, this approach faces several important limitations. First, exact matching is rarely achievable in practice, especially with continuous or multivariate covariates. As shown by recent work \citep{guo2023statistical,pimentel2024covariate,zhu2025randomization}, inexact matching can leave statistically meaningful residual bias and render downstream design-based inference based on the post-matching uniform propensity condition asymptotically invalid. More generally, \citet{savje2022inconsistency} proves that, when the fraction of units with propensity scores greater than $0.5$ is non-negligible, matching without replacement (the typical matching framework for design-based inference; \citealp{rosenbaum2002observational, rosenbaum2020design}) generally yields post-matching propensity score discrepancies that are bounded away from zero. As shown in the simulation studies in Section~6 (see Tables~2 and 3), the matching-based approach suffers from severe under-coverage across several canonical settings. Second, existing post-matching design-based analyses typically condition on the realized matched pairs, treating $\mathcal{Z}$ as fixed, even though the matching procedure is performed after observing the realized treatment assignments $\mathbf{Z}$ and therefore generally produces matched pairs that depend on $\mathbf{Z}$. In other words, who is matched with whom, as encoded by $\mathcal{Z}$, usually depends on the realized treatment assignments $\mathbf{Z}$. This fact, known as $Z$-dependence in the matching literature, introduces an additional source of randomness that is typically ignored by existing post-matching design-based analyses \citep{pashley2021conditional, pimentel2024covariate, pimentel2024re}. As shown by \citet{pimentel2024covariate} and \citet{pimentel2024re}, ignoring this dependence can further invalidate downstream design-based inference. Third, matching is primarily tailored to causal inference in observational studies. In other design-based settings, such as missing data problems or survey sampling, matching is more often used as a problem-specific imputation or preprocessing device than as a general foundation for rigorous design-based inference with unknown propensity scores. This substantially limits its applicability across the broader class of design-based inference problems considered in this paper.
\end{itemize}

\subsection*{B.2: Connections with and Fundamental Differences from the Super-Population Inference Literature}

As emphasized in the main text, the proposed propensity score propagation framework is intrinsically different from existing simulation-based inference approaches developed under super-population frameworks, including bootstrap, resampling, subsampling, and perturbation-based methods (e.g., \citealp{efron1994introduction, politis1994large, politis1999subsampling, chernozhukov2018double, zhao2019sens_boot, bodory2020finite, ImbensMenzel2021CausalBootstrap, guo2023causal, guo2024statistical, lin2024consistency, zheng2025perturbed}; among many others). Among these, the closest conceptual analogues are recent perturbation-based inference methods, including \citet{guo2023causal} and \citet{zheng2025perturbed}; our work is partially inspired by these developments, but differs fundamentally from both. Specifically, \citet{guo2023causal} develops more credible confidence intervals (CIs) for causal structural parameters in the presence of many possibly invalid instrumental variables (IVs) by perturbing IV regression coefficient estimates, constructing a collection of perturbed CIs based on these perturbed estimates, and then forming the final CI as a filtered union of the perturbed CIs. In contrast, \citet{zheng2025perturbed} seeks to relax the convergence-rate requirements on nuisance estimators in double/debiased machine learning \citep{chernozhukov2018double} for inference on low-dimensional functionals in the presence of infinite-dimensional nuisance parameters. Their approach perturbs the nuisance estimation procedure by injecting randomness into the estimation process, with each perturbation yielding a perturbed Wald CI for the target functional, and then constructs the final CI as a filtered union of these perturbed CIs.

At a highly conceptual level, the regeneration step in our proposed propensity score propagation framework may likewise be viewed as perturbing the estimated propensity scores, which are the key nuisance parameters in design-based inference, by injecting randomness into the propensity score estimation procedure. For example, in parametric propensity score propagation, this is done by sampling from a data-dependent asymptotic distribution of the propensity score model parameters, whereas in nonparametric propensity score propagation, it is achieved through subsampling or repeated cross-fitting with flexible propensity score learners. Each regeneration run then yields a corresponding perturbed CI, namely, the CI constructed from the regenerated propensity scores, and the final CI is obtained by taking the union of these perturbed CIs.

However, the proposed propensity score propagation framework is fundamentally different from recent perturbation-based inference approaches (e.g., \citealp{guo2023causal, zheng2025perturbed}) in terms of the target problem, the variables involved, the statistical construction, and the underlying mathematical justification:
\begin{itemize}
    \item \textbf{Target Problem:} As mentioned in Section 3.5, existing perturbation-based inference methods are developed for specific problems in super-population inference. For example, \citet{guo2023causal} studies inference in IV regression with many possibly invalid instruments, whereas \citet{zheng2025perturbed} studies inference for low-dimensional functionals in the presence of infinite-dimensional nuisance parameters within the double/debiased machine learning framework \citep{chernozhukov2018double}. In contrast, our work considers developing a general framework for credible design-based inference with unknown propensity scores.

    \item \textbf{Type of Variables Involved:} Existing perturbation-based inference methods (e.g., \citealp{guo2023causal, zheng2025perturbed}) are mainly designed for settings involving \textit{continuous} outcomes, responses, or treatment variables, together with nuisance parameters. By contrast, the regeneration step in our propensity score propagation framework is introduced to reflect uncertainty in estimating nuisance models for \textit{binary} design variables, such as treatment assignments, sampling inclusion indicators, or missingness indicators. This difference fundamentally changes both the statistical construction and the mathematical justification; see the next two points.

    \item \textbf{Statistical Construction:} Existing perturbation-based inference approaches typically perturb either finite-dimensional parameter estimates or nuisance estimation procedures in models for continuous outcomes, responses, or treatment variables, often by directly injecting randomness into nuisance models or fitted nuisance components associated with continuous variables \citep{guo2023causal, zheng2025perturbed}. By contrast, such direct perturbation is generally neither natural nor stable in our setting, because propensity score propagation targets nuisance models for \textit{binary} design variables. Instead, our framework uses regeneration schemes tailored to the propensity score setting, such as sampling from a data-dependent asymptotic distribution of GLM parameters in the parametric case, and subsampling or repeated cross-fitting with flexible learners in the nonparametric case. Moreover, these nonparametric regeneration schemes differ fundamentally from classical subsampling or re-cross-fitting in super-population inference \citep{politis1994large, politis1999subsampling, bickel2010subsampling, chernozhukov2018double}: they are used to generate multiple plausible propensity score vectors for the observed finite population, rather than to approximate a super-population sampling distribution, and they leave the outcome data untouched throughout.

    \item \textbf{Mathematical Justification:} The mathematical justification of our framework is also fundamentally different from that in the existing perturbation-based inference literature, especially in the nonparametric setting. Existing perturbation-based methods are developed under super-population frameworks and typically rely on arguments tailored to continuous parameter spaces (e.g., \citealp{guo2023causal, zheng2025perturbed}). By contrast, our nonparametric propensity score propagation framework operates under a design-based setting with a fixed finite population, where the randomness is induced by the binary design variable and the random partitioning used for regeneration, rather than by repeated sampling from a super-population. Accordingly, our proof strategy is built on a novel isoperimetric argument over the hypercube $\{0,1\}^{N}$, equivalently, the space of partition indicators, together with the induced push-forward measure on the regenerated estimators. This argument allows us to show that, with sufficiently many regeneration runs, at least one regenerated estimator is sufficiently close to the oracle estimator with high probability. In this sense, the mathematical tools underlying our framework are not a direct extension of existing perturbation-based inference arguments, but are instead specifically designed for the geometry and probabilistic structure of design-based inference with binary design variables and random partitions.
\end{itemize}

\subsection*{B.3: Additional Remarks}

\begin{remark}[The Finite-Population Version of the Ignorability Condition]
In the design-based, or finite-population, inference literature, a rigorous and unified way to express the ignorability condition, also referred to as the no unobserved covariates condition, is as follows. Let $\mathcal{F}_{N, i}$ denote the fixed, outcome-related information for unit $i$, and let $\mathcal{F}_{N}=\{\mathcal{F}_{N, i}:i\in[N]\}$ denote the corresponding collection for the finite population. For example, in observational studies (e.g., the example application in Section~3.4), $\mathcal{F}_{N, i}=(Y_{i}(1),Y_{i}(0))$, so that $\mathcal{F}_{N}=\{(Y_{i}(1),Y_{i}(0)):i\in[N]\}$. In descriptive missing-data problems (e.g., Example Application 2 in Section~5.1), $\mathcal{F}_{N, i}=y_{i}$, so that $\mathcal{F}_{N}=\{y_{i}:i\in[N]\}$ denotes the collection of missing and non-missing outcomes. In design-based DID analyses (e.g., Example Application 3 in Section~5.1), $\mathcal{F}_{N, i}=(\Delta_{i}^{Y}(1),\Delta_{i}^{Y}(0))$, so that $\mathcal{F}_{N}=\{(\Delta_{i}^{Y}(1),\Delta_{i}^{Y}(0)):i\in[N]\}$ denotes the collection of fixed potential values of outcome changes under treatment and control. Then, the ignorability condition (\citealp{rosenbaum1983central, rosenbaum2002observational, rosenbaum2020design}) can be expressed as
\begin{equation}\label{eqn: ignorability general}
    \mathbf{Z}\mid \mathbf{X}, \mathcal{F}_{N} \overset{d}{=}  \mathbf{Z}\mid \mathbf{X}.
\end{equation}
Equivalently, conditional on the observed covariates, the design mechanism for $\mathbf{Z}$ does not depend on the fixed outcome-related quantities in $\mathcal{F}_{N}$. If we further assume conditional independence across units under the design mechanism, the finite-population ignorability condition can be expressed at the individual level as
\begin{equation}\label{eqn: ignorability individual}
    P(Z_{i}=1\mid \mathbf{x}_{i}, \mathcal{F}_{N, i})=P(Z_{i}=1\mid \mathbf{x}_{i}).
\end{equation}
Importantly, the outcome-related quantities in $\mathcal{F}_{N}$ or $\mathcal{F}_{N, i}$ are fixed finite-population values; therefore, this formulation does not require distributional assumptions on the outcomes or potential outcomes (\citealp{rosenbaum2002observational, rosenbaum2020design}).
\end{remark}

\begin{remark}[Some Remarks on the Growth Rate of the Required $M$]
Throughout Section~4, for notational simplicity, we focus on the case $\mu_L\{\mathcal{C}_{1-\alpha}^{\text{oracle}}(\theta)\}\asymp_p N^{-1/2}$, that is, there exist constants $\nu_{1},\nu_{2}>0$ such that $\lim_{N\rightarrow\infty}P\big(\nu_{1}/\sqrt{N}\leq \mu_L\{\mathcal{C}_{1-\alpha}^{\text{oracle}}(\theta)\}\leq \nu_{2}/\sqrt{N}\big)=1$. This is the canonical rate in design-based inference problems with oracle or known propensity scores, including those reviewed in Sections~2.1 and~5.1. This rate holds, for example, when $\mathcal{C}_{1-\alpha}^{\text{oracle}}(\theta)$ is constructed from a regular asymptotically normal design-based statistic with nondegenerate variance \citep{rosenbaum2002observational, li2017general, fogarty2018mitigating}. The proposed methodology and accompanying theoretical arguments extend directly to the more general setting $\mu_L\{\mathcal{C}_{1-\alpha}^{\text{oracle}}(\theta)\}\asymp_p r_N$ for any sequence $r_N\to 0$. The main modification concerns the required growth rate of the number of regeneration runs. Specifically, let $M_N$ be a sequence of positive integers indexed by $N$. By Theorem~1, when the number of regeneration runs satisfies $M\gg \frac{\log N}{(r_{N}\sqrt{N})^{d}}$ (in particular, $M\gg \log N$ when $r_{N}=N^{-1/2}$), we have $\min_{m\in [M]} |\widehat{\theta}^{(m)}-\widehat{\theta}_{\text{oracle}}|=o_{p}(r_{N})$ and $\min_{m\in \mathcal{M}(\alpha^{\prime})} |\widehat{\theta}^{(m)}-\widehat{\theta}_{\text{oracle}}|=o_{p}(r_{N})$. Since this rate is faster than the rate governing the length of $\mathcal{C}^{\text{oracle}}_{1-\alpha}(\theta)$, it follows, roughly speaking, that $\mathcal{C}_{1-\alpha}^{(m^{*})}(\theta)$ has asymptotically the same coverage rate as $\mathcal{C}_{1-\alpha}^{\text{oracle}}(\theta)$, where $m^{*}$ is defined by $|\widehat{\theta}^{(m^{*})}-\widehat{\theta}_{\text{oracle}}|=\min_{m\in [M]} |\widehat{\theta}^{(m)}-\widehat{\theta}_{\text{oracle}}|$ or $|\widehat{\theta}^{(m^{*})}-\widehat{\theta}_{\text{oracle}}|=\min_{m\in \mathcal{M}(\alpha^{\prime})} |\widehat{\theta}^{(m)}-\widehat{\theta}_{\text{oracle}}|$.

\end{remark}

\begin{remark}[Some Remarks on the Matching Design]\label{rem: matching in simulation}
The matching-based approaches examined in our simulation studies, both in Table~2 in Section 2.2 and Table~3 in Section~6, use optimal full matching with a propensity score caliper \citep{rosenbaum1991characterization,hansen2004full,hansen2006optimal}, a widely used matching design that constructs matched sets while retaining all study units. As shown in the covariate balance table in Remark~\ref{rem: further remarks on matching} of Appendix D.3, the average post-matching standardized mean difference between the treated and control units is below 0.035 for all five covariates, substantially smaller than the 0.10 standardized-mean-difference benchmark commonly used in post-matching covariate balance diagnostics \citep{silber2013characteristics,small2024protocols}. Therefore, the matched datasets in our simulation studies would pass standard balance diagnostics, and researchers would typically regard them as adequate for applying the classical randomization inference procedures considered in the simulations. In many practical applications, researchers may further trim units that are difficult to match in order to improve post-matching balance \citep{rosenbaum2002observational,rosenbaum2020design,zubizarreta2012using}. However, such trimming generally changes the finite-population estimand, such as the sample average treatment effect, and therefore complicates comparisons across inference methods. To ensure a fair comparison, we use optimal full matching so that all methods under consideration target the same estimand defined on the full finite population. In practice, when researchers are willing to discard units and thereby change the target estimand, matching can be naturally combined with the proposed propensity score propagation framework by first using matching as a preprocessing step to trim the study population and then applying propensity score propagation to the resulting trimmed population.
\end{remark}

\begin{remark}[The Restricted-Union Strategy for Nonparametric Propensity Score Propagation] Similar to the restricted-union strategy for parametric propensity score propagation introduced in Section~3.2, we may also consider a restricted-union strategy for nonparametric propensity score propagation to further shorten the confidence set $\mathcal{C}_{1-\alpha}$ reported by Algorithm~2. Following the notation in Sections~3 and~4, let $\alpha$ be the prespecified significance level, and let $\alpha'\in(0,\alpha)$ denote the level allocated to screening outlying regeneration runs. Define $\widetilde{\mathcal{M}}(\alpha')=\big\{m\in[M]:|\widehat{\theta}^{(m)}-\widehat{\theta}_{\text{ave}}|\leq q_{1-\alpha'}\big\}$, where $\widehat{\theta}_{\text{ave}}=M^{-1}\sum_{m=1}^{M}\widehat{\theta}^{(m)}$ and $q_{1-\alpha'}$ is the $(1-\alpha')$-quantile of $\big\{|\widehat{\theta}^{(1)}-\widehat{\theta}_{\text{ave}}|,\dots,|\widehat{\theta}^{(M)}-\widehat{\theta}_{\text{ave}}|\big\}$. Under this strategy, the union step in Algorithm~2 is replaced by the restricted union $\widetilde{\mathcal{C}}_{1-\alpha}=\bigcup_{m\in\widetilde{\mathcal{M}}(\alpha')}\mathcal{C}_{1-\alpha_{0}}^{(m)}$, where $\alpha_{0}=\alpha-\alpha'$ is the adjusted significance level used for each retained regenerated confidence set.
\end{remark}

\begin{remark}[Inexact Matching and $Z$-dependence Under General Matching Designs]\label{rem: general matching}
In Appendix B.1, we use matched pairs to introduce the issues of inexact matching and $Z$-dependence. These issues are not specific to pair matching and also arise under more general matching designs. Following the classic framework for matched observational studies \citep{rosenbaum2002observational, rosenbaum2020design}, suppose that the matching procedure retains $I$ matched sets, where matched set $i$ contains $n_i$ units, among which $m_i$ are treated, and $n_i-m_i$ are controls, with $\min\{m_i,n_i-m_i\}=1$. Such matched sets are typically formed by minimizing total within-set covariate distance, such as Mahalanobis distance, subject to prespecified matching constraints on the set sizes and treatment-control ratios, often through optimal matching procedures \citep{rosenbaum1989optimal,rosenbaum1991characterization,rosenbaum2002observational,rosenbaum2020design,hansen2004full,hansen2006optimal}. For unit $j$ in matched set $i$, let $Z_{ij}\in\{0,1\}$ denote the treatment indicator and $\mathbf{x}_{ij}$ the observed covariates. Let $\mathbf{Z}=(Z_{11},\dots,Z_{In_I})$ denote the treatment assignment vector, $\mathbf{X}=(\mathbf{x}_{11},\dots,\mathbf{x}_{In_I})$ the collection of observed covariates, and let $\mathcal{Z}$ denote the realized matched-set structure, encoding both which units are grouped together and the numbers of treated and control units in each matched set. Under exact matching on covariates, that is, $\mathbf{x}_{ij}=\mathbf{x}_{ij'}$ for all $i$ and all $j,j'\in\{1,\dots,n_i\}$, together with ignorability, the post-matching treatment assignment is uniform within each matched set: conditional on the realized matched-set structure $\mathcal{Z}$ and the covariate information $\mathbf{X}$, each subset of $m_i$ treated units among the $n_i$ units in matched set $i$ is equally likely, which in particular implies 
\begin{equation}\label{eqn: post-matching randomization}
    \text{$P(Z_{ij}=1 \mid \mathcal{Z},\mathbf{X})=\frac{m_i}{n_i}$\quad  for $i=1,\dots,I$ and $j=1,\dots,n_i$.}
\end{equation}
In practical research, once a matched dataset satisfies commonly used covariate balance criteria (e.g., requiring the absolute standardized mean difference for each covariate to be below 0.1 or 0.2; \citealp{silber2013characteristics,rosenbaum2020design,small2024protocols}), researchers often proceed by assuming the post-matching uniform propensity condition in \eqref{eqn: post-matching randomization}, treating the matched observational study as a stratified randomized experiment, and applying established design-based inference methods under the post-matching uniform propensity condition; see Remark~\ref{rem: inference under exact matching} for further details. However, this post-matching design-based analysis faces two important limitations. First, exact matching is rarely achievable in practice, especially with continuous or multivariate covariates. Consequently, directly imposing the uniform propensity condition in \eqref{eqn: post-matching randomization} after inexact matching may leave statistically meaningful residual bias and invalidate downstream design-based inference \citep{guo2023statistical,pimentel2024covariate,zhu2025randomization}. Second, existing analyses typically condition on the realized matched-set structure $\mathcal{Z}$ and treat it as fixed. In practice, however, the matching procedure is usually performed after the realized treatment labels $\mathbf{Z}$ have been observed, so the resulting matched-set structure $\mathcal{Z}$ itself depends on $\mathbf{Z}$. This phenomenon, known as $Z$-dependence in the matching literature \citep{pashley2021conditional,pimentel2024covariate,pimentel2024re}, represents an additional source of randomness that is typically ignored in post-matching design-based analyses. As shown by \citet{pimentel2024covariate} and \citet{pimentel2024re}, ignoring such $Z$-dependence can further compromise the validity of downstream design-based inference, including Type-I error control and confidence interval coverage.

\end{remark}

\begin{remark}[Recent Work on Addressing $Z$-Dependence in Matched Observational Studies]
Recent work has begun to address $Z$-dependence in matched observational studies. In particular, \citet{pimentel2024re} propose a computationally efficient procedure, termed \textit{match-adaptive randomization inference}, that accounts for $Z$-dependence by characterizing and sampling from the appropriate conditional distribution of treatment assignments in randomization-based inference after optimal propensity score matching. However, their methods are primarily developed for pair matching and for testing Fisher's sharp nulls, such as constant treatment effects and related variants, and therefore do not directly provide inference for more general matching designs (e.g., full matching or matching with multiple controls) or for sample average treatment effects (i.e., Neyman's weak nulls). Moreover, their theoretical guarantees rely on oracle knowledge of the true propensity scores. Another recent contribution, \citet{li2024sensitivity}, proposes a permutation-based inference approach for addressing $Z$-dependence under flexible matching algorithms, including pair matching, full matching, and matching with multiple controls. However, the inference framework of \citet{li2024sensitivity} takes randomness to arise from random permutations of potential outcomes rather than from random treatment assignments. This differs fundamentally from the design-based perspective adopted in the present paper, where treatment assignment is the only source of randomness and potential outcomes are treated as fixed. In addition, \citet{li2024sensitivity} focuses on Fisher's sharp nulls and related extensions, whereas our framework is not tied to a specific null hypothesis or effect model and can be used for inference on finite-population estimands such as sample average treatment effects.
\end{remark}

\begin{remark}[More Details on Design-Based Difference-in-Differences]
Following arguments similar to those in \citet{heng2025design}, we can show that the weighted difference-in-differences (DID) estimator 
\begin{equation*}
    \widehat{\theta}_{\text{DID}, \text{oracle}}=\frac{1}{N}\sum_{i=1}^{N}\widehat{\theta}_{\text{DID}, \text{oracle}, i}
    =\frac{1}{N}\sum_{i=1}^{N}\left\{\frac{Z_{i}}{p_{i}}(Y_{i1}-Y_{i0})-\frac{1-Z_{i}}{1-p_{i}}(Y_{i1}-Y_{i0})\right\}
\end{equation*}
is an unbiased design-based estimator for \(\theta_{t=1}=N^{-1}\sum_{i=1}^{N}\{Y_{i1}(1)-Y_{i1}(0)\}\). Specifically, define the observed difference in outcome as \(\Delta^{Y}_{i}=Y_{i1}-Y_{i0}\). Under consistency, SUTVA, and no anticipatory effects, we have \(Y_{i0}=Y_{i0}(1)=Y_{i0}(0)\) and \(Y_{i1}=Y_{i1}(Z_i)\) \citep{athey2022design,rambachan2025design, heng2025nonbipartiteDID}. Therefore, \(\Delta^{Y}_{i}=Z_{i}\Delta^{Y}_{i}(1)+(1-Z_{i})\Delta^{Y}_{i}(0)\), where \(\Delta^{Y}_{i}(1)=Y_{i1}(1)-Y_{i0}\) and \(\Delta^{Y}_{i}(0)=Y_{i1}(0)-Y_{i0}\). Suppose further that, conditional on the observed covariates, treatment uptake is independent of the potential outcome differences, with
\begin{equation}\label{eqn: DID randomization}
    P\left(Z_{i}=1\mid \mathbf{x}_{i},\Delta^{Y}_{i}(1),\Delta^{Y}_{i}(0)\right)
    =
    P(Z_i=1\mid \mathbf{x}_i):=p_i,
\end{equation}
and that treatment uptake is independent across units. This condition can be interpreted as a design-based DID version of no time-varying unmeasured confounding: after differencing, time-invariant unmeasured components are removed from the relevant potential outcome differences, so treatment assignment probability may be regarded as depending only on the observed covariates $\mathbf{x}_{i}$. Let \(\mathcal{F}_{N, \Delta}=\{(\Delta_i^Y(1),\Delta_i^Y(0)):i=1,\dots,N\}\). Then, under \eqref{eqn: DID randomization},
\begin{align*}
&\quad \ E\left(\widehat{\theta}_{\text{DID}, \text{oracle}}\mid \mathbf{X},\mathcal{F}_{N, \Delta}\right)\\
&=\frac{1}{N}\sum_{i=1}^{N}E\left\{\frac{Z_i}{p_i}\Delta_i^Y(1)-\frac{1-Z_i}{1-p_i}\Delta_i^Y(0)\mid \mathbf{X},\mathcal{F}_{N, \Delta}\right\} \\
&=\frac{1}{N}\sum_{i=1}^{N}\left\{\frac{P(Z_i=1\mid \mathbf{X},\mathcal{F}_{N, \Delta})}{p_i}\Delta_i^Y(1)-\frac{1-P(Z_i=1\mid \mathbf{X},\mathcal{F}_{N, \Delta})}{1-p_i}\Delta_i^Y(0)\right\} \\
&=\frac{1}{N}\sum_{i=1}^{N}\left\{\Delta_i^Y(1)-\Delta_i^Y(0)\right\}\\
&=\frac{1}{N}\sum_{i=1}^{N}\left\{Y_{i1}(1)-Y_{i1}(0)\right\}=\theta_{t=1}.
\end{align*}
Therefore, \(\widehat{\theta}_{\text{DID}, \text{oracle}}\) is unbiased for the finite-population average treatment effect at time \(t=1\). Moreover, under the same design-based conditions, and following arguments similar to those in \citet{kang2016full} and \citet{zhu2025randomization}, the variance estimator $\widehat{\mathcal{V}}_{\text{DID}, \text{oracle}}=\frac{1}{N(N-1)}\sum_{i=1}^{N}(\widehat{\theta}_{\text{DID},\text{oracle},i}-\widehat{\theta}_{\text{DID}, \text{oracle}})^{2}$ is a valid (conservative) design-based variance estimator for \(\widehat{\theta}_{\text{DID}, \text{oracle}}\), in the sense that \(E(\widehat{\mathcal{V}}_{\text{DID}, \text{oracle}})\geq \text{Var}(\widehat{\theta}_{\text{DID}, \text{oracle}})\).
\end{remark}

\begin{remark}[The Classical Design-Based Inference After Matching]\label{rem: inference under exact matching}
In matched (or stratified) observational studies, researchers commonly use the Neyman (i.e., difference-in-means) estimator to construct design-based confidence intervals for the sample average treatment effect $\theta_{\tau}$ \citep{rosenbaum2002observational, imbens2015causal, fogarty2018mitigating, fogarty2020studentized, fogarty2023testing}. Specifically, suppose that there are $I$ matched sets, and each matched set $i$ has $n_{i}$ units. Therefore, there are $N=\sum_{i=1}^{I}n_{i}$ units in total. Suppose that $i$th matched set has $m_{i}$ treated units and $n_{i}-m_{i}$ control units, where $\min\{m_{i}, n_{i}-m_{i}\}=1$. Then, the Neyman estimator $\widehat{\lambda}$ is defined as
\begin{equation*}\label{eqn: Neyman estimator}
\widehat{\lambda}=\sum_{i=1}^{I}\frac{n_{i}}{N}\widehat{\lambda}_{i} =\frac{1}{I}\sum_{i=1}^{I}\frac{n_{i}I}{N}\widehat{\lambda}_{i}, \text{where $\widehat{\lambda}_{i}=\sum_{j=1}^{n_{i}}\left\{\frac{Z_{ij}Y_{ij}}{m_{i}}-\frac{(1-Z_{ij})Y_{ij}}{n_{i}-m_{i}}\right\}$.}
\end{equation*}
Following \citet{fogarty2018mitigating}, a valid design-based variance estimator for $\widehat{\lambda}$ under the post-matching uniform propensity condition and independence of treatment assignments across matched sets after matching, as well as ignoring $Z$-dependence, is 
\begin{equation*}\label{eqn: Neyman variance}
\widehat{\text{Var}}(\widehat{\lambda})=\frac{1}{I(I-1)}\sum_{i=1}^{I}\left(\frac{n_{i}I}{N}\widehat{\lambda}_{i}-\widehat{\lambda}\right)^{2}.
\end{equation*}
Then, a commonly used design-based confidence interval for $\theta_{\tau}$ in matched observational studies can be expressed as
\begin{equation*}
    \mathcal{C}^{\lambda}_{1-\alpha}(\theta_{\tau})=\left [\widehat{\lambda}-z_{1-\alpha/2}\cdot\sqrt{\widehat{\text{Var}}(\widehat{\lambda})},\  \widehat{\lambda}+z_{1-\alpha/2}\cdot\sqrt{\widehat{\text{Var}}(\widehat{\lambda})} \right].
\end{equation*}
\end{remark}

\begin{remark}[Vanishing Stability Constant $L_{s, N}$]
    In Theorem~7, we strengthen the boundedness requirement on the stability constants $L_{s,N}$ in Condition~4 by requiring $\lim_{N\to\infty} L_{s,N}=0$. A necessary condition for this requirement is that, for any partition indicators $\mathbf S_{1},\mathbf S_{2}\in\{0,1\}^{N}$, the corresponding propensity score estimates satisfy $N^{-1/2}\|\widehat{\mathcal P}_{1}-\widehat{\mathcal P}_{2}\|_{2}\xrightarrow{p}0$ as $N\to\infty$. This necessary condition is implied by the consistency of the propensity score learner: if $\widehat{\mathcal P}_{1}$ and $\widehat{\mathcal P}_{2}$ are both consistent for the true propensity score vector $\mathcal P$, then $N^{-1/2}\|\widehat{\mathcal P}_{1}-\widehat{\mathcal P}_{2}\|_{2} \leq N^{-1/2}\|\widehat{\mathcal P}_{1}-\mathcal P\|_{2}+N^{-1/2}\|\widehat{\mathcal P}_{2}-\mathcal P\|_{2}\xrightarrow{p}0$ as $N\rightarrow \infty$, even when the overlap between $\mathbf S_{1}$ and $\mathbf S_{2}$ is limited, that is, when $N^{-1}d_H(\mathbf S_{1},\mathbf S_{2})$ does not vanish.
\end{remark}

\section*{Appendix C: Some Additional Examples}

\subsection*{C.1: Fisher's Randomization Tests in Observational Studies}

In the main text, we applied our framework to infer the sample average treatment effect $\theta_{\tau}$ (i.e., Neyman's weak null) in observational studies. Here, we show that the same framework can be readily adapted to conduct design-based tests for Fisher's sharp null (i.e., Fisher's randomization tests). Existing Fisher's randomization tests in observational studies are typically based on matching or stratification \citep{rosenbaum2002observational, rosenbaum2020design}, possibly combined with plug-in approaches applied after matching \citep{pimentel2024covariate}. As a result, these methods inherit the limitations of matching- and plug-in-based inference discussed in Section~1.1 and Appendix B.1. In contrast, the proposed propensity score propagation framework provides an alternative approach that can mitigate these limitations. Specifically, consider the notations adopted in the main text. When the true propensity score vector $\mathcal{P}=(p_{1},\dots, p_{N})$ is known, a finite-population-exact $p$-value reported by Fisher's randomization test for Fisher's sharp null hypothesis of no effect, defined as $H_{0,F}: Y_{i}(1)=Y_{i}(0)$ for all $i$, can be expressed as 
\begin{align*}
    p_{F}^{\text{oracle}}&=P\left(T(\mathbf{Z}, \mathbf{Y})\geq t\mid H_{0,F}, \mathcal{P}\right)\\
    &=\sum_{\mathbf{z}=(z_{1},\dots, z_{N})\in \{0,1\}^{N} }\left[P\left(\mathbf{Z}=\mathbf{z}\mid \mathcal{P}\right)\cdot \mathbbm{1}\left\{T(\mathbf{Z}=\mathbf{z}, \mathbf{Y})\geq t\right\}\right]\\
    &=\sum_{\mathbf{z}=(z_{1},\dots, z_{N})\in \{0,1\}^{N} } \left[\prod_{i=1}^{N}p_{i}^{z_{i}}(1-p_{i})^{1-z_{i}}\cdot \mathbbm{1}\left\{T(\mathbf{Z}=\mathbf{z}, \mathbf{Y})\geq t\right\}\right]
\end{align*}
where $T(\mathbf{Z}, \mathbf{Y})$ can be any user-chosen test statistic (e.g., permutation $t$-test, rank-based tests, or $U$-statistics) and $t$ is its observed value. In practice, when the sample size $N$ is large, Fisher's exact $p$-value can be calculated via either the Monte Carlo method or finite-population central limit theorems \citep{rosenbaum2002observational, li2017general}.

Combining these established results under known propensity scores with the proposed propensity score propagation framework, we can conduct Fisher's randomization test in observational studies. Specifically, we can apply either the parametric propensity score propagation or nonparametric propensity score propagation to obtain an estimated propensity score vector $\widehat{\mathcal{P}}^{(m)}=(\widehat{p}^{(m)}_{1},\dots, \widehat{p}^{(m)}_{N})$ from each regeneration run $m \in [M]$, and obtain the corresponding regenerated Fisher's exact $p$-value: 
\begin{align*}
    p_{F}^{(m)}&=P\left(T(\mathbf{Z}, \mathbf{Y})\geq t\mid H_{0,F}, \widehat{\mathcal{P}}^{(m)}\right)\\
    &=\sum_{\mathbf{z}=(z_{1},\dots, z_{N})\in \{0,1\}^{N} }\left[P\left(\mathbf{Z}=\mathbf{z}\mid \widehat{\mathcal{P}}^{(m)}\right)\cdot \mathbbm{1}\left\{T(\mathbf{Z}=\mathbf{z}, \mathbf{Y})\geq t\right\}\right]\\
    &=\sum_{\mathbf{z}=(z_{1},\dots, z_{N})\in \{0,1\}^{N} } \left[\prod_{i=1}^{N}\left(\widehat{p}^{(m)}_{i}\right)^{z_{i}}\left(1-\widehat{p}^{(m)}_{i}\right)^{(1-z_{i})}\cdot \mathbbm{1}\left\{T(\mathbf{Z}=\mathbf{z}, \mathbf{Y})\geq t\right\}\right].
\end{align*}
Then, following a similar spirit to the union step in Algorithm~1 or 2, to combine $p_{F}^{(m)}$ for $m\in [M]$, we define the overall $p$-value reported by the propensity score propagation framework as $p_{F}=\max_{m\in [M]}p_{F}^{(m)}$. Moreover, by inverting the $p$-values $p_{F}$ under various hypothetical effect sizes under Fisher's sharp nulls, we can construct design-based confidence sets for structured effect mechanisms (e.g., constant effect and its extensions) in observational studies \citep{rosenbaum2002observational, rosenbaum2020design}. 

\subsection*{C.2: Design-Based Causal Inference for Randomized Experiments with Missing Data}

Beyond design-based inference for descriptive estimands (e.g., the finite-population mean $\theta_{\bar y}$) with missing data, the proposed framework can be readily extended to design-based inference for causal estimands in the presence of missing data. For example, consider a randomized experiment with outcome missingness. Existing design-based inference methods in this setting predominantly target Fisher's sharp null hypotheses, such as constant treatment effects and their extensions \citep{ivanova2022randomization, heussen2023randomization, heng2025design, li2025randomization}. In contrast, our framework enables design-based inference for Neyman's weak nulls with missing outcomes, thereby substantially broadening the scope of causal questions that can be addressed with missing data. For example, consider a Bernoulli randomized experiment, in which $Z_{i} \overset{\text{iid}}{\sim} \text{Bernoulli}(0.5)$. Let $M_{i}$ denote the missingness indicator for the outcome of unit $i$: $M_{i}=1$ if non-missing and $M_{i}=0$ if missing. That is, under the potential outcomes framework, the observed outcome $Y_{i}=Y_{i}(1)$ (or $Y_{i}(0)$) if $M_{i}=1$ and $Z_{i}=1$ (or $Z_{i}=0$), and $Y_{i}=\text{``NA''}$ if $M_{i}=0$. Under the missing-at-random condition (i.e., the ignorability of outcome missingness conditional on observed covariates), let $p_{i}=P(M_{i}=1\mid \mathbf{x}_{i})$ denote the non-missingness probability for unit $i$. Conditional on known $p_{i}$, it is easy to show that 
\begin{equation*}
    \widehat{\theta}_{*, \text{oracle}}=\frac{1}{N}\sum_{i=1}^{N}\widehat{\theta}_{*, \text{oracle}, i}=\frac{1}{N}\sum_{i=1}^{N}\left\{2p_{i}^{-1}M_{i}Z_{i}Y_{i}-2p_{i}^{-1}M_{i}(1-Z_{i})Y_{i}\right\}
\end{equation*}
is an unbiased estimator for sample average treatment effect $\theta_{\tau}=N^{-1}\sum_{i=1}^{N}\big\{Y_{i}(1)-Y_{i}(0)\big\}$ (where we define $0\times\text{``NA''}=0$), and 
\begin{equation*}
\widehat{\mathcal{V}}_{*, \text{oracle}}=\frac{1}{N(N-1)}\sum_{i=1}^{N}\left(\widehat{\theta}_{*, \text{oracle}, i}-\widehat{\theta}_{*, \text{oracle}}\right)^{2}    
\end{equation*}
is a valid design-based variance estimator for $\widehat{\theta}_{*, \text{oracle}}$. Then, the oracle confidence interval for $\theta_{\tau}$ in this setting is
\begin{equation*}
    \mathcal{C}_{1-\alpha}^{\text{oracle}}=\Lambda_{*}(\mathcal{P}, \mathbf{Z}, \mathbf{Y})=\left[\widehat{\theta}_{*, \text{oracle}}- z_{1-\alpha/2}\cdot \{\widehat{\mathcal{V}}_{*, \text{oracle}}\}^{1/2},\, \widehat{\theta}_{*, \text{oracle}}+z_{1-\alpha/2}\cdot \{\widehat{\mathcal{V}}_{*, \text{oracle}}\}^{1/2}\right].
\end{equation*}

By embedding this established oracle confidence set mapping $\Lambda_{*}$ derived under known $p_{i}$ into the propensity score propagation framework, we can construct a design-based confidence set for $\theta_{\tau}$ (Neyman's weak nulls) in randomized experiments with missing outcomes. Specifically, for each regenerated non-missingness propensity score vector $\widehat{\mathcal{P}}^{(m)}=(\widehat{p}^{(m)}_{1},\ldots,\widehat{p}^{(m)}_{N})$, we construct the regenerated estimator
\begin{equation*}
    \widehat{\theta}_{*}^{(m)}=\frac{1}{N}\sum_{i=1}^{N}\left\{2(\widehat{p}^{(m)}_{i})^{-1}M_{i}Z_{i}Y_{i}-2(\widehat{p}^{(m)}_{i})^{-1}M_{i}(1-Z_{i})Y_{i}\right\}
\end{equation*}
and the corresponding regenerated design-based variance estimator 
\begin{equation*}
    \widehat{\mathcal{V}}_{*}^{(m)}=\frac{1}{N(N-1)}\sum_{i=1}^{N}\left(\widehat{\theta}_{*, i}^{(m)}-\widehat{\theta}_{*}^{(m)}\right)^{2}.
\end{equation*}
This gives the regenerated confidence interval 
\begin{equation*}
    \mathcal{C}^{(m)}_{1-\alpha}=\Lambda_{*}(\widehat{\mathcal{P}}^{(m)}, \mathbf{Z}, \mathbf{Y})=\left[\widehat{\theta}^{(m)}_{*}- z_{1-\alpha/2}\cdot \{\widehat{\mathcal{V}}_{*}^{(m)}\}^{1/2},\, \widehat{\theta}^{(m)}_{*}+z_{1-\alpha/2}\cdot \{\widehat{\mathcal{V}}_{*}^{(m)}\}^{1/2}\right].
\end{equation*}
The final propagation-based confidence set is then $\mathcal{C}_{1-\alpha}=\bigcup_{m=1}^{M}\mathcal{C}^{(m)}_{1-\alpha}$, which propagates uncertainty from estimating the non-missingness probabilities while preserving the design-based, finite-population perspective.

\section*{Appendix D: Additional Simulation Studies and Implementation Details}

\subsection*{D.1: Simulation Studies Under Parametric Propensity Score Settings}

We further illustrate the performance of parametric propensity score propagation (Algorithm~1) under parametric propensity score models. The data-generating procedure is the same as that used in the simulation studies in Section~6 of the main text, except that the treatment assignment mechanism is generated from one of two parametric models. In Setting~1, we consider a linear selection model, where \(Z_{i}=\mathbbm{1}\{\phi(\mathbf{x}_{i})-0.4>\epsilon_{i}^{z}\}\), \(\phi(\mathbf{x}_{i})=-0.5 x_{i1} + 0.5x_{i2} + 0.6 x_{i3} + 0.4 x_{i4} + 0.3 x_{i5}\), and \(\epsilon_{i}^{z} \overset{\text{iid}}{\sim} N(0,1)\). In Setting~2, we consider a linear logistic model, where \(\operatorname{logit} P(Z_{i}=1\mid \mathbf{x}_{i})=-0.5 x_{i1} + 0.5x_{i2} + 0.6 x_{i3} + 0.6 x_{i4} + 0.2 x_{i5} + 0.5\). For each propensity score setting and each effect setting, we generate a finite population of $N=1000$ units using the same covariate and potential outcome generation procedures as in Section~6. Conditional on each fixed finite population $(\mathbf X,\mathbf Y(0),\mathbf Y(1))$, coverage is evaluated over the distribution of treatment assignments $\mathbf Z$. All reported simulation results are averaged over 1000 treatment assignment realizations.

\begin{table}[h]
\centering
\footnotesize
\caption{Simulation results for the parametric propensity score setting with parametric propensity score propagation. The length ratio is computed relative to the corresponding oracle confidence interval length.}
\begin{tabular}{lcccccc}
\toprule
\multirow{2}{*}{} 
& \multicolumn{3}{c}{Propensity Score Setting 1}
& \multicolumn{3}{c}{Propensity Score Setting 2} \\
\cmidrule(lr){2-4} \cmidrule(lr){5-7}
& Coverage & Length & Length Ratio
& Coverage & Length & Length Ratio \\
\midrule
\multicolumn{7}{l}{\textbf{Effect Setting 1}} \\
\midrule
Oracle & 0.954 & 0.633 & 1.000 & 0.958 & 0.422 & 1.000 \\
\midrule
Propagation (\(M=30\)) 
& 1.000 & 0.714 & 1.128 
& 1.000 & 0.624 & 1.479 \\
Propagation (\(M=100\)) 
& 1.000 & 0.783 & 1.237 
& 1.000 & 0.688 & 1.630 \\
\midrule
\multicolumn{7}{l}{\textbf{Effect Setting 2}} \\
\midrule
Oracle & 0.936 & 0.711 & 1.000 & 0.957 & 0.433 & 1.000 \\
\midrule
Propagation (\(M=30\)) 
& 1.000 & 0.744 & 1.046 
& 1.000 & 0.654 & 1.510 \\
Propagation (\(M=100\)) 
& 1.000 & 0.813 & 1.143 
& 1.000 & 0.724 & 1.672 \\
\bottomrule
\end{tabular}
\label{tab:parametric-simulation-n1000-regularized}
\end{table}

Table~\ref{tab:parametric-simulation-n1000-regularized} reports the empirical coverage rate and average confidence interval length for the oracle confidence interval based on the true propensity scores, as well as for parametric propensity score propagation with \(M=30\) and \(M=100\) regeneration runs. We also report the length ratio, defined as the ratio between the average length of the propagation-based confidence interval and that of the corresponding oracle confidence interval. Across the four combinations of propensity score and effect settings, both the oracle and propagation-based confidence intervals achieve nominal coverage. The propagation-based confidence intervals also remain informative: the length ratios range from 1.046 to 1.510 for \(M=30\), and from 1.143 to 1.672 for \(M=100\). Therefore, although the propagation-based confidence intervals are wider than the oracle confidence intervals, as expected because they incorporate additional uncertainty from propensity score estimation, their lengths remain within a moderate multiple of the oracle benchmark.

\subsection*{D.2: Simulation Results for Nonparametric Propensity Score Propagation Under a Small Sample Size} 

Table~\ref{tab:simulation_N500_appendix} reports additional simulation results for $N=500$, using the same propensity score and effect settings as those in Section~6 of the main text. This setting serves as a small-sample stress test for the asymptotic theoretical results for nonparametric propensity score propagation. Under cross-fitting, each propensity score model is trained on approximately half of the sample, so the effective training sample size is only about 250 observations in each fold. Consequently, first-stage propensity score estimation uncertainty is substantial, and ignoring this uncertainty can lead to poorly calibrated downstream design-based inference.

\begin{table}[ht]
\centering
\scriptsize
\caption{Additional simulation results under a small sample size, where the effective training sample size for the propensity score model is approximately $500/2=250$ in each cross-fitting fold.}
\resizebox{\textwidth}{!}{
\begin{tabular}{lcccccccc}
\toprule
& \multicolumn{4}{c}{Propensity Score Setting 1} 
& \multicolumn{4}{c}{Propensity Score Setting 2} \\
\cmidrule(lr){2-5} \cmidrule(lr){6-9}
Approach 
& Coverage & Bias & Length 
& $\frac{\text{Length}}{\text{Length (Oracle)}}$
& Coverage & Bias & Length 
& $\frac{\text{Length}}{\text{Length (Oracle)}}$ \\
\midrule

\multicolumn{9}{l}{\textbf{Effect Setting 1}} \\
Plug-in 
& 0.462 & 0.320 & 0.604 & 0.845
& 0.542 & 0.315 & 0.627 & 1.013 \\
Matching 
& 0.231 & 0.209 & 0.355 & 0.497
& 0.541 & 0.150 & 0.310 & 0.501 \\
Propagation (Our Proposal)
& 0.996 & -- & 1.121 & 1.568
& 0.993 & -- & 1.146 & 1.851 \\
Oracle 
& 0.940 & 0.014 & 0.715 & 1.000
& 0.949 & 0.005 & 0.619 & 1.000 \\

\midrule

\multicolumn{9}{l}{\textbf{Effect Setting 2}} \\
Plug-in 
& 0.584 & 0.329 & 0.678 & 0.919
& 0.603 & 0.344 & 0.701 & 1.075 \\
Matching 
& 0.589 & 0.192 & 0.403 & 0.546
& 0.787 & 0.137 & 0.344 & 0.528 \\
Propagation (Our Proposal)
& 1.000 & -- & 1.276 & 1.729
& 0.991 & -- & 1.300 & 1.994 \\
Oracle 
& 0.956 & 0.012 & 0.738 & 1.000
& 0.958 & 0.005 & 0.652 & 1.000 \\

\bottomrule
\end{tabular}
}
\label{tab:simulation_N500_appendix}
\end{table}

The results in Table~\ref{tab:simulation_N500_appendix} show that the conventional plug-in confidence intervals exhibit substantial under-coverage across all settings. The coverage rate of the plug-in approach ranges from 0.46 to 0.60, far below the nominal level of 0.95. This under-coverage occurs because the plug-in approach treats the estimated propensity scores as if they were known and therefore fails to account for the additional finite-sample impact of propensity score estimation, including both estimation uncertainty and the induced finite-sample bias. In this small-sample regime, such bias is especially non-negligible, as reflected in the substantial estimation bias and compromised coverage rate of the plug-in estimator. 

The matching-based approach also exhibits substantial under-coverage in this small-sample regime. The matching estimator generally has smaller reported bias than the plug-in estimator. This pattern is consistent with the intuition that, when the propensity score is estimated with substantial finite-sample uncertainty, the plug-in estimator can inherit non-negligible bias from first-stage propensity score estimation. This effect is more pronounced here than in the larger-sample simulations in Section~6, where the gap in reported bias between the plug-in and matching estimators becomes smaller. Nevertheless, relative to the oracle estimator, the remaining finite-sample bias of the matching estimator is still substantial. Moreover, its intervals are much shorter than both the oracle and plug-in intervals, with length ratios close to one-half in some settings. As a result, the matching intervals do not provide reliable coverage, with coverage rates ranging from 0.23 to 0.79. Together with the main simulation results in Section~6, these findings are consistent with the theoretical concerns raised in \citet{savje2022inconsistency} and \citet{guo2023statistical} regarding the asymptotic invalidity of matching-based approaches under matching without replacement, which is a canonical matching design for design-based inference.

By contrast, the nonparametric propensity score propagation approach maintains coverage across all simulation settings, even in this small-sample regime. Its coverage rate ranges from 0.99 to 1.00, indicating that it successfully protects against under-coverage in this challenging setting. At the same time, the propagation-based confidence intervals remain informative: their lengths are approximately 1.57 to 1.99 times those of the oracle confidence intervals, even under this small-sample regime. Overall, these results reinforce the main qualitative conclusion of the simulation study: when propensity scores are unknown and estimated flexibly, propagating propensity score uncertainty is crucial for reliable design-based inference.

\subsection*{D.3: Additional Implementation and Data Details}

\begin{remark}[More Details on the Matching Procedure and Results]\label{rem: further remarks on matching} The matching procedure was conducted using optimal full matching based on a rank-based Mahalanobis distance with a propensity score caliper \citep{rosenbaum1991characterization, hansen2004full, hansen2006optimal, rosenbaum2020design}. The caliper was imposed on the logit of the estimated propensity score, with the caliper width set to 0.2 standard deviations, a standard choice in practice. Full matching was implemented using the \textsf{fullmatch()} function from the widely used \textsf{R} package \textsf{optmatch} \citep{hansen2006optimal}. Table~\ref{tab: std} reports the average absolute standardized differences in means for the five covariates between the treatment and control groups, before and after matching, across the 1000 simulated datasets, under both propensity score settings. The average post-matching absolute standardized differences in means are all below 0.035, well below the commonly used threshold of 0.10 for assessing adequate covariate balance after matching \citep{silber2013characteristics, rosenbaum2020design, small2024protocols}. However, as shown in recent work by \citet{pimentel2024covariate} and \citet{zhu2025randomization}, passing these empirical post-matching covariate balance assessments does not guarantee valid design-based inference. In particular, directly applying classical design-based inference methods, such as Fisher's randomization tests or Neyman estimators, may still lead to substantially inflated Type-I error rates and poor confidence interval coverage.
\begin{table}[ht]
\small
\centering
\caption{The average absolute standardized differences in means of the five covariates between the treatment and control groups, before and after matching.}
\begin{tabular}{ccccc}
\toprule
 & \multicolumn{2}{c}{Propensity Score Setting 1} & \multicolumn{2}{c}{Propensity Score Setting 2}\\
 \cmidrule(rl){2-3} \cmidrule(rl){4-5}
 & Before Matching & After Matching & Before Matching & After Matching\\
 \midrule
 $x_{1}$ & 0.402 & 0.033 & 0.265 & 0.029\\
 $x_{2}$ & 0.370 & 0.017 & 0.238 & 0.017\\
 $x_{3}$ & 0.048 & 0.021 & 0.050 & 0.020\\
 $x_{4}$ & 0.077 & 0.019 & 0.095 & 0.020\\
 $x_{5}$ & 0.217 & 0.022 & 0.141 & 0.022\\
\bottomrule
\end{tabular}
\label{tab: std}
\end{table}
\end{remark}

\begin{remark}[Regularization for the IPW Estimator]
If the estimated propensity scores $\widehat{p}_{i}$ involved in the IPW estimator $N^{-1}\sum_{i=1}^{N}\left\{\frac{Z_{i}Y_{i}}{\widehat{p}_{i}}-\frac{(1-Z_{i})Y_{i}}{1-\widehat{p}_{i}}\right\}$ are very close to $0$ or $1$, the corresponding inverse weights $1/\widehat{p}_{i}$ or $1/(1-\widehat{p}_{i})$ can become excessively large, leading to numerical instability and non-informative confidence sets. To improve numerical robustness and mitigate the impact of extreme propensity score values for the IPW estimators used in both the plug-in and propensity score propagation approaches, we implement a commonly used regularization (stabilization) procedure in both the simulation studies in Section 6 and Appendix D and the data application in Section 7 \citep{crump2009dealing, ma2020robust, zhu2025randomization}. Specifically, for a prespecified small constant $\widetilde{\delta}>0$, we define the regularized estimated propensity scores as $\widehat{p}_{i}^{\text{reg}}
= \widehat{p}_{i}\,\mathbbm{1}\{\widetilde{\delta}\leq \widehat{p}_{i}\leq 1-\widetilde{\delta}\}
+ \widetilde{\delta}\,\mathbbm{1}\{\widehat{p}_{i}<\widetilde{\delta} \}
+ (1-\widetilde{\delta})\,\mathbbm{1}\{\widehat{p}_{i}>1-\widetilde{\delta}\}$. For fair comparison across methods, we set $\widetilde{\delta}=0.1$, a commonly used choice as suggested in \citet{crump2009dealing} and \citet{zhu2025randomization}, for the IPW estimators used in both the plug-in and propensity score propagation approaches. In practice, researchers may also consider alternative regularization strategies, such as stabilized IPW estimators that normalize the inverse probability weights \citep{lunceford2004stratification, zhao2019sens_boot}.

\end{remark}

\begin{remark}[Tuning Parameters for the Propensity Score Learner Based on XGBoost]
In practice, to improve the robustness of propensity score estimation and mitigate overfitting, we use cross-validation to select tuning parameters when implementing both the plug-in and propensity score propagation approaches. Table~\ref{tab:tuning-grid} reports the tuning grid used in our simulation studies; these parameter values are standard choices in XGBoost \citep{chen2016xgboost}. Specifically, we adopt Monte Carlo cross-validation with a 50/50 split, repeated 10 times. For each split, the XGBoost model is trained on the training subset, and out-of-sample propensity score predictions are obtained for the held-out subset. Predictive performance is evaluated using the area under the receiver operating characteristic curve (ROC-AUC) \citep{hanley1982meaning}, computed from the out-of-sample estimated propensity scores, which summarizes the ability of the fitted model to discriminate between treated and control units. We compute the ROC-AUC for each of the 10 random splits and select the tuning parameter configuration that maximizes the average ROC-AUC across splits. The same selected tuning parameters are then used for both the plug-in and propensity score propagation approaches, ensuring a fair comparison across methods that rely on estimated propensity scores.

\begin{table}[!htbp]
\centering
\caption{Tuning grid for XGBoost in the simulation settings.}
\label{tab:tuning-grid}
\begin{tabular}{ll}
\hline
Parameter & Candidate Values \\
\hline
Number of boosting rounds ($\texttt{nrounds}$)          & $\{100,\ 200,\ 500\}$ \\
Maximum tree depth ($\texttt{max\_depth}$)              & $\{1,\ 3,\ 5,\ 7\}$ \\
Learning rate (\texttt{eta})                    & $\{0.05,\ 0.10,\ 0.20\}$ \\
Minimum loss reduction to split ($\texttt{gamma}$)      & $\{0,\ 1,\ 5\}$ \\
Minimum sum of instance weight in a child ($\texttt{min\_child\_weight}$) & $\{1,\ 3,\ 5\}$ \\
\hline
\end{tabular}
\end{table}
\end{remark}

\begin{remark}[Variable Descriptions for the Data Application]

Table~\ref{tab:variables} provides detailed descriptions of the variables used in our data reanalysis of \citet{heller2010using}, including the outcome variable, treatment indicator, and 20 covariates. A closely related dataset was also considered in \citet{rouse1995democratization}.

\begin{table}[htbp]
\footnotesize
\centering
\caption{Descriptions of the variables used in the data reanalysis of \citet{heller2010using}.}
\label{tab:variables}
\begin{tabular}{lp{9cm}}
\hline
Variable & Definition \\
\hline
Years of education since 1986 & \texttt{educ86}, outcome variable. \\
Two-year college attendance & \texttt{twoyr}, equal to 1 if the individual attended a two-year college immediately after high school (treatment indicator). \\

West & \texttt{region\_west}, equal to 1 if the high school is located in the West. \\
South & \texttt{region\_south}, equal to 1 if the high school is located in the South. \\
Midwest & \texttt{region\_midwest}, equal to 1 if the high school is located in the Midwest. \\
Urban & \texttt{urban}, equal to 1 if the high school is located in an urban area. \\
Percent white at school (\%) & \texttt{perwhite}, percentage of White students at the individual's high school. \\
Home ownership & \texttt{ownhome}, equal to 1 if the family owns their home. \\
Family income missing & \texttt{fincmiss}, equal to 1 if family income is missing. \\
Family income & \texttt{fincome}, family income. \\

Mother's education missing & \texttt{mommiss}, equal to 1 if mother's education is missing. \\
Mother: college degree & \texttt{momcoll}, equal to 1 if the mother completed a college degree. \\
Mother: some college & \texttt{momsome}, equal to 1 if the mother attended some college. \\
Mother: vocational school & \texttt{momvoc}, equal to 1 if the mother completed vocational school. \\

Father's education missing & \texttt{dadmiss}, equal to 1 if father's education is missing. \\
Father: college degree & \texttt{dadcoll}, equal to 1 if the father completed a college degree. \\
Father: some college & \texttt{dadsome}, equal to 1 if the father attended some college. \\
Father: vocational school & \texttt{dadvoc}, equal to 1 if the father completed vocational school. \\

Baseline test score & \texttt{bytest}, baseline composite test score measured in high school. \\
Hispanic & \texttt{hispanic}, equal to 1 if the individual is Hispanic. \\
Black & \texttt{black}, equal to 1 if the individual is African American/Black. \\
Female & \texttt{female}, equal to 1 if the individual is female. \\
\hline
\end{tabular}
\end{table}

\end{remark}

\section*{Appendix E: Example Constructions of Propagation-Based Sensitivity Analysis}
\label{sec: extensions}

In Section~5.3, we introduce a general formulation of propagation-based sensitivity analysis for hidden bias arising from unobserved covariates (unmeasured confounders). Here, we provide concrete constructions of propagation-based sensitivity sets for several representative problems in design-based inference for observational studies.

Consider conducting a sensitivity analysis for a finite-population causal effect estimand, such as the sample average treatment effect $\theta_{\tau}$, in an observational study. In this setting, a central question is whether a hypothesized value $\tau_{0}$, such as $\tau_{0}=0$, is contained in the resulting sensitivity set. This corresponds to sensitivity analysis for Neyman's weak null hypothesis $H_{\tau_{0}}:\theta_{\tau}=\tau_{0}$ \citep{zhao2019sens_boot, fogarty2020studentized, fogarty2023testing}. Under the above sensitivity analysis procedure for propensity score propagation, it suffices to look at whether there exists any $m$ such that $\tau_{0}\in \mathcal{S}_{1-\alpha, \Gamma}^{(m)}$. To solve this, we first consider a transformation proposed in \citet{zhao2019sens_boot}, in which we let $g^{*}(\mathbf{x}, \mathbf{u})=\text{logit}\{e^{*}(\mathbf{x}, \mathbf{u})\}$, $g(\mathbf{x})=\text{logit}\{e(\mathbf{x})\}$, and $h(\mathbf{x}, \mathbf{u})=g(\mathbf{x})-g^{*}(\mathbf{x}, \mathbf{u})$. Then, we let $\widehat{e}^{(m)}(\mathbf{x}_{i})$ denote the estimated value of $e(\mathbf{x}_{i})$ in the $m$th regeneration run, and $\widehat{g}^{(m)}(\mathbf{x}_{i})=\text{logit}\{\widehat{e}^{(m)}(\mathbf{x}_{i})\}$. Next, we define $t_{i}=\exp\{h(\mathbf{x}_{i}, \mathbf{u}_{i})\}\mathbbm{1}\{Z_{i}=1\}+\exp\{-h(\mathbf{x}_{i}, \mathbf{u}_{i})\}\mathbbm{1}\{Z_{i}=0\} \in [\Gamma^{-1}, \Gamma]$. Under each fixed $\mathbf{t}=(t_{1},\dots, t_{N})\in [\Gamma^{-1}, \Gamma]^{N}$, in each regeneration run $m$, the corresponding estimated propensity score $\widehat{e}^{*(m)}(\mathbf{x}_{i}, \mathbf{u}_{i})=[1+\exp\{h(\mathbf{x}_{i}, \mathbf{u}_{i})-\widehat{g}^{(m)}(\mathbf{x}_{i})\}]^{-1}=[1+t_{i}\exp\{-\widehat{g}^{(m)}(\mathbf{x}_{i})\}]^{-1}$ if $Z_{i}=1$, and $\widehat{e}^{*(m)}(\mathbf{x}_{i}, \mathbf{u}_{i})=[1+t^{-1}_{i}\exp\{-\widehat{g}^{(m)}(\mathbf{x}_{i})\}]^{-1}$ if $Z_{i}=0$. In each regeneration run $m$, by plugging $\widehat{e}^{*(m)}(\mathbf{x}_{i}, \mathbf{u}_{i})$ into the oracle IPW estimator $\widehat{\theta}_{\tau, \text{oracle}}$ and its corresponding design-based variance estimator $\widehat{\mathcal{V}}_{\tau, \text{oracle}}$, we obtain the regenerated design-based confidence interval $\mathcal{C}^{(m)}_{1-\alpha, \mathbf{t}}$ under each fixed $\mathbf{t}=(t_{1},\dots, t_{N})\in [\Gamma^{-1}, \Gamma]^{N}$:
\begin{equation*}
 \mathcal{C}^{(m)}_{1-\alpha, \mathbf{t}}=\left [ \widehat{\theta}_{\tau, \mathbf{t}}^{(m)}- z_{1-\alpha/2}\cdot\left\{ \widehat{\mathcal{V}}^{(m)}_{\tau, \mathbf{t}}\right\}^{1/2}, \, \widehat{\theta}_{\tau, \mathbf{t}}^{(m)}+ z_{1-\alpha/2}\cdot\left\{ \widehat{\mathcal{V}}^{(m)}_{\tau, \mathbf{t}}\right\}^{1/2} \right ],
\end{equation*}
where 
\begin{align*}
    \widehat{\theta}_{\tau, \mathbf{t}}^{(m)}&=N^{-1}\sum_{i=1}^{N}\widehat{\theta}^{(m)}_{\tau, \mathbf{t}, i}\\
    &=N^{-1}\sum_{i=1}^{N}Z_{i}Y_{i}\left[1+t_{i}\exp\{-\widehat{g}^{(m)}(\mathbf{x}_{i})\}\right] - (1-Z_{i})Y_{i}\left[1+t_{i}\exp\{\widehat{g}^{(m)}(\mathbf{x}_{i})\}\right]
\end{align*}
and 
\begin{equation*}
    \widehat{\mathcal{V}}^{(m)}_{\tau, \textbf{t}}=\frac{1}{N(N-1)}\sum_{i=1}^{N}\left(\widehat{\theta}^{(m)}_{\tau, \mathbf{t}, i}- \widehat{\theta}^{(m)}_{\tau, \mathbf{t}}\right)^{2}.
\end{equation*}
Then, in each regeneration run $m$, the corresponding sensitivity set $\mathcal{S}^{(m)}_{1-\alpha, \Gamma}$ is
\begin{equation*}
    \mathcal{S}_{1-\alpha, \Gamma}^{(m)}=\left[\min_{\mathbf{t}\in [\Gamma^{-1}, \Gamma]^{N}}\widehat{\theta}_{\tau, \mathbf{t}}^{(m)}- z_{1-\alpha/2}\cdot \left\{\widehat{\mathcal{V}}^{(m)}_{\tau, \mathbf{t}}\right\}^{1/2}, \,\max_{\mathbf{t}\in [\Gamma^{-1}, \Gamma]^{N}}\widehat{\theta}_{\tau, \mathbf{t}}^{(m)}+ z_{1-\alpha/2}\cdot \left\{\widehat{\mathcal{V}}^{(m)}_{\tau, \mathbf{t}}\right\}^{1/2} \right],
\end{equation*}
and the final sensitivity set 
\begin{equation*}
    \mathcal{S}_{1-\alpha, \Gamma}=\bigcup_{m=1}^{M}\mathcal{S}_{1-\alpha, \Gamma}^{(m)}.
\end{equation*}
Therefore, to determine whether a hypothetical effect size $\tau_{0}\in \mathcal{S}_{1-\alpha, \Gamma}$ or not, we just need to look at whether there exists some $m\in [M]$ such that the optimal value (denoted as $d_{*}^{(m)}$) of the following quadratic program with box constraints is non-positive:
\begin{equation*}
     \begin{split}
        \underset{\mathbf{t}}{\text{minimize}} \quad & \big(\widehat{\theta}_{\tau, \mathbf{t}}^{(m)}-\tau_{0}\big)^{2}-z^{2}_{1-\alpha/2}\cdot \widehat{\mathcal{V}}^{(m)}_{\tau, \mathbf{t}}  \\
         \text{subject to}\quad & t_{i}\in [\Gamma^{-1},\, \Gamma], \quad \forall i.
     \end{split}
 \end{equation*}
If so, we have $\tau_{0}\in \mathcal{S}_{1-\alpha, \Gamma}$; otherwise, we have $\tau_{0}\notin \mathcal{S}_{1-\alpha, \Gamma}$. Then, researchers can repeat the above procedure for various values of $\Gamma$ and identify the largest $\Gamma\geq 1$ such that $\tau_{0}\notin \mathcal{S}_{1-\alpha, \Gamma}$; such $\Gamma$ is also referred to as the sensitivity value $\Gamma^{*}$ \citep{zhao2019sensitivityvalue}. That is, if $\tau_{0}\notin \mathcal{C}_{1-\alpha}=\mathcal{S}_{1-\alpha, \Gamma=1}$, the sensitivity value $\Gamma^{*}$ represents the largest $\Gamma>1$ (i.e., the largest magnitude of hidden bias) under which we can still reject the null hypothesis $H_{\tau_{0}}:\theta_{\tau}=\tau_{0}$ in a sensitivity analysis, which serves as a transparent summary measure of the robustness of the evidence against $H_{\tau_{0}}$ in the presence of potential hidden bias. 

In addition to sensitivity analysis for Neyman's weak nulls, another important sensitivity analysis setting in design-based inference concerns Fisher's sharp nulls. For example, consider testing Fisher's sharp null hypothesis of no treatment effect $H_{0,F}: Y_i(1)=Y_i(0)$ for all $i$, using a sum test statistic $T(\mathbf Z,\mathbf Y)=\sum_{i=1}^N Z_i q_{i}$, where each $q_i$ is a test score calculated from the observed outcomes $\mathbf Y$ and is therefore fixed under $H_{0,F}$ \citep{rosenbaum1984conditional, rosenbaum1987sensitivity, rosenbaum2002observational}. For instance, when $q_i=Y_i$, $T$ is the permutational $t$-test statistic. When $q_i=\operatorname{rank}(Y_i)$, $T$ is the Wilcoxon rank-sum statistic; see \citet{rosenbaum2002observational} for additional examples of sum test statistics. Invoking the finite-population central limit theorem \citep{rosenbaum2002observational, li2017general}, under $H_{0,F}$, we have $\{T-E(T)\}/\sqrt{\text{Var}(T)}\xrightarrow{d}N(0,1)$, where $E(T)=\sum_{i=1}^N q_{i}e_i^*$ and $\text{Var}(T)=\sum_{i=1}^N q_i^{2}e_i^*(1-e_{i}^{*})$, with $e_i^*=e^*(\mathbf x_i,\mathbf u_i)$.

Now consider the marginal sensitivity model described in Section 5.3 of the main text. For notational simplicity, let $e_i=e(\mathbf x_i)$ and $e^{*}_i=e^*(\mathbf x_i,\mathbf u_i)$. Under this sensitivity model, for some prespecified sensitivity parameter $\Gamma\geq 1$, we have
\[
    \Gamma^{-1}
    \leq
    \frac{e^{*}_i/(1-e^{*}_i)}{e_i/(1-e_i)}
    \leq
    \Gamma,
    \quad i=1,\ldots,N.
\]
Equivalently, each $e_i^{*}$ satisfies the following box constraint
\begin{equation}\label{eqn: equivalent sensitivity bounds}
    e^{*}_i\in
        \left[
        \frac{e_{i}}{(1-\Gamma)e_{i}+\Gamma },\,
        \frac{\Gamma e_{i}}{(\Gamma-1)e_{i}+1}
        \right],
    \quad i=1,\ldots,N.
\end{equation}
Given the observed value of the test statistic $T$ and a prespecified value of $\Gamma$, under the true values of $\mathcal E=\{e(\mathbf x_1),\ldots,e(\mathbf x_N)\}=\{e_{1},\dots, e_{N}\}$, a sensitivity analysis rejects $H_{0,F}$ if and only if the worst-case (maximal) two-sided or one-sided $p$-value is below the significance level $\alpha$. 

\textit{Two-sided Sensitivity Analysis:} For a sensitivity analysis based on the worst-case (maximal) two-sided $p$-value, it rejects $H_{0,F}$ if and only if the worst-case (minimal) squared deviate exceeds the rejection threshold $z^2_{1-\alpha/2}$ under the constraint (\ref{eqn: equivalent sensitivity bounds}). Specifically, this is equivalent to determining whether the following quadratic program with box constraints has a positive optimal value:
\begin{equation}\label{eqn: sharp null sensitivity analysis}
     \begin{split}
        \underset{\mathbf{e}^{*}}{\operatorname{minimize}} \quad
        & \big\{T-E_{\mathbf{e}^{*}}(T)\big)\}^2-z^2_{1-\alpha/2}\cdot \text{Var}_{\mathbf{e}^{*}}(T) \\
        \operatorname{subject\ to}\quad
        & e^{*}_i\in
        \left[
        \frac{e_{i}}{(1-\Gamma)e_{i}+\Gamma },\,
        \frac{\Gamma e_{i}}{(\Gamma-1)e_{i}+1}
        \right],
        \quad i=1,\ldots,N,
     \end{split}
 \end{equation}
where $\mathbf{e}^{*}=(e^{*}_1,\dots,e^{*}_N)$, $E_{\mathbf{e}^{*}}(T)=\sum_{i=1}^N q_ie^{*}_i$, and $\text{Var}_{\mathbf{e}^{*}}(T)=\sum_{i=1}^N q_i^2 e^{*}_i(1-e^{*}_{i})$. Let $c_{*}$ denote the optimal value of \eqref{eqn: sharp null sensitivity analysis}. We reject $H_{0,F}$ if $c_{*}>0$ and fail to reject otherwise. In practice, the true values of $\mathcal E=\{e(\mathbf x_1),\ldots,e(\mathbf x_N)\}=\{e_{1},\dots, e_{N}\}$ are unknown. Researchers can therefore use the proposed propensity score propagation framework to conduct a sensitivity analysis for $H_{0,F}$. Specifically, in each regeneration run, after obtaining the corresponding estimated values of $\mathcal E$, denoted by $\widehat{\mathcal E}^{(m)}=\{\widehat e^{(m)}(\mathbf x_1),\ldots,\widehat e^{(m)}(\mathbf x_N)\}=\{\widehat e^{(m)}_1,\ldots,\widehat e^{(m)}_N\}$, we replace $e_{i}$ in \eqref{eqn: sharp null sensitivity analysis} with $\widehat e_i^{(m)}$ and obtain the corresponding optimal value $c_{*}^{(m)}$. Finally, we reject $H_{0,F}$ if and only if $\min_{m\in[M]}c_{*}^{(m)}>0$.

\textit{One-Sided Sensitivity Analysis:} For a sensitivity analysis based on the worst-case (maximal) one-sided $p$-value, instead of solving the corresponding optimization problem, we can directly obtain an analytical solution. Specifically, consider finding the one-sided greater-than $p$-value under the constraint (\ref{eqn: equivalent sensitivity bounds}). Consider a sum test statistic $T=\sum_{i=1}^{N}T_{i}$ with each $T_{i}=q_{i}Z_{i}$. Then, we define $T^{+}=\sum_{i=1}^{N}T_{i}^{+}$, where each $T_{i}^{+}=q_{i}B_{i}^{+}$ and $\{B_{i}^{+}\}_{i=1}^{N}$ are independent Bernoulli variables such that, if $q_{i}\geq 0$, then $P(B_{i}^{+}=1)=\frac{\Gamma e_{i}}{(\Gamma-1)e_{i}+1}$ and $P(B_{i}^{+}=0)=\frac{1-e_{i}}{(\Gamma-1)e_{i}+1}$; and if $q_{i}<0$, then $P(B_{i}^{+}=1)=\frac{e_{i}}{(1-\Gamma)e_{i}+\Gamma }$ and $P(B_{i}^{+}=0)=\frac{\Gamma(1-e_{i})}{(1-\Gamma)e_{i}+\Gamma}$. Following an argument similar to Proposition 13 in \citet{rosenbaum2002observational}, for any $t\in \mathbbm{R}$, we have $P(T^{+}\geq t)\geq P(T\geq t)$ under the constraint (\ref{eqn: equivalent sensitivity bounds}), and the inequality is sharp. Under the mild regularity conditions required by finite-population central limit theorems \citep{rosenbaum2002observational, li2017general, zhao2019sensitivityvalue}, we have 
\begin{equation*}
   \frac{T^{+}-E(T^{+})}{\sqrt{\text{Var}(T^{+})}}\xrightarrow{d} N(0,1). 
\end{equation*}
Therefore, given the observed value $t$ of $T$, the worst-case (maximal) one-sided greater-than $p$-value under $H_{0,F}$ is
\begin{equation}\label{eqn: worst-case greater than p-value}
    P(T^{+}\geq t)= P\left(\frac{T^{+}-E(T^{+})}{\sqrt{\text{Var}(T^{+})}}\geq  \frac{t-E(T^{+})}{\sqrt{\text{Var}(T^{+})}}\right)\simeq 1- \Phi\left( \frac{t-E(T^{+})}{\sqrt{\text{Var}(T^{+})}}\right),
\end{equation}
where
\begin{equation*}
    E(T^{+})=\sum_{i=1}^{N}\left(q_{i} \cdot \frac{\Gamma e_{i}}{(\Gamma-1)e_{i}+1}\cdot \mathbbm{1}\{q_{i}\geq 0\}+ q_{i} \cdot \frac{e_{i}}{(1-\Gamma)e_{i}+\Gamma }\cdot \mathbbm{1}\{q_{i}< 0\}  \right),
\end{equation*}
and 
\begin{equation*}
    \text{Var}(T^{+})=\sum_{i=1}^{N}\left(q^{2}_{i} \cdot \frac{\Gamma e_{i}(1-e_{i})}{\{(\Gamma-1)e_{i}+1\}^{2}}\cdot \mathbbm{1}\{q_{i}\geq 0\}+ q^{2}_{i} \cdot \frac{\Gamma e_{i}(1-e_{i})}{\{(1-\Gamma)e_{i}+\Gamma\}^{2} }\cdot \mathbbm{1}\{q_{i}< 0\}  \right).
\end{equation*}

Similarly, consider finding the one-sided less-than $p$-value under the constraint (\ref{eqn: equivalent sensitivity bounds}), based on a sum test statistic $T=\sum_{i=1}^{N}T_{i}$ with each $T_{i}=q_{i}Z_{i}$. Then, we define $T^{-}=\sum_{i=1}^{N}T_{i}^{-}$, where each $T_{i}^{-}=q_{i}B_{i}^{-}$ and $\{B_{i}^{-}\}_{i=1}^{N}$ are independent Bernoulli variables such that, if $q_{i}\geq 0$, then $P(B_{i}^{-}=1)=\frac{e_{i}}{(1-\Gamma)e_{i}+\Gamma }$ and $P(B_{i}^{-}=0)=\frac{\Gamma(1-e_{i})}{(1-\Gamma)e_{i}+\Gamma }$; and if $q_{i}<0$, then $P(B_{i}^{-}=1)=\frac{\Gamma e_{i}}{(\Gamma-1)e_{i}+1}$ and $P(B_{i}^{-}=0)=\frac{1-e_{i}}{(\Gamma-1)e_{i}+1}$. Following an argument similar to Proposition 13 in \citet{rosenbaum2002observational}, for any $t\in \mathbbm{R}$, we have $P(T^{-}\leq t)\geq P(T\leq t)$ under the constraint (\ref{eqn: equivalent sensitivity bounds}), and the inequality is sharp. Under the mild regularity conditions required by finite-population central limit theorems \citep{rosenbaum2002observational, li2017general, zhao2019sensitivityvalue}, we have 
\begin{equation*}
   \frac{T^{-}-E(T^{-})}{\sqrt{\text{Var}(T^{-})}}\xrightarrow{d} N(0,1). 
\end{equation*}
Therefore, given the observed value $t$ of $T$, the worst-case (maximal) one-sided less-than $p$-value under $H_{0,F}$ is
\begin{equation}\label{eqn: worst-case less than p-value}
    P(T^{-}\leq t)= P\left(\frac{T^{-}-E(T^{-})}{\sqrt{\text{Var}(T^{-})}}\leq  \frac{t-E(T^{-})}{\sqrt{\text{Var}(T^{-})}}\right)\simeq \Phi\left( \frac{t-E(T^{-})}{\sqrt{\text{Var}(T^{-})}}\right),
\end{equation}
where
\begin{equation*}
    E(T^{-})=\sum_{i=1}^{N}\left(q_{i} \cdot \frac{e_{i}}{(1-\Gamma)e_{i}+\Gamma }\cdot \mathbbm{1}\{q_{i}\geq 0\}+ q_{i} \cdot \frac{\Gamma e_{i}}{(\Gamma-1)e_{i}+1} \cdot \mathbbm{1}\{q_{i}< 0\}  \right),
\end{equation*}
and 
\begin{equation*}
    \text{Var}(T^{-})=\sum_{i=1}^{N}\left(q^{2}_{i} \cdot \frac{\Gamma e_{i}(1-e_{i})}{\{(1-\Gamma)e_{i}+\Gamma\}^{2} }  \cdot \mathbbm{1}\{q_{i}\geq 0\}+ q^{2}_{i} \cdot \frac{\Gamma e_{i}(1-e_{i})}{\{(\Gamma-1)e_{i}+1\}^{2}}\cdot \mathbbm{1}\{q_{i}< 0\}  \right).
\end{equation*}

After establishing the oracle formulas for the worst-case one-sided greater-than and less-than \(p\)-values in (\ref{eqn: worst-case greater than p-value}) and (\ref{eqn: worst-case less than p-value}) under oracle knowledge of \(\mathcal E=\{e(\mathbf x_1),\ldots,e(\mathbf x_N)\}=\{e_1,\ldots,e_N\}\), we can apply the propensity score propagation framework to conduct valid one-sided sensitivity analysis for \(H_{0,F}\) when \(\mathcal E\) is unknown and must be estimated from the observed data. Specifically, in each regeneration run, after obtaining the regenerated estimates \(\widehat{\mathcal E}^{(m)}=\{\widehat e^{(m)}(\mathbf x_1),\ldots,\widehat e^{(m)}(\mathbf x_N)\}=\{\widehat e^{(m)}_1,\ldots,\widehat e^{(m)}_N\}\), we replace each \(e_i\) in (\ref{eqn: worst-case greater than p-value}) and (\ref{eqn: worst-case less than p-value}) by \(\widehat e_i^{(m)}\), yielding the regenerated worst-case one-sided greater-than and less-than \(p\)-values \(p_g^{(m)}\) and \(p_l^{(m)}\), respectively. Finally, we reject \(H_{0,F}\) if and only if \(\max_{m\in[M]} p_g^{(m)}<\alpha\) for the one-sided greater-than test, or \(\max_{m\in[M]} p_l^{(m)}<\alpha\) for the one-sided less-than test.

\endgroup

\end{document}